\documentclass[a4paper,12pt,english]{article}
\usepackage[english]{babel}

\usepackage[T1]{fontenc}
\usepackage[utf8]{inputenc}
\usepackage{lmodern}
\usepackage{mathtools}
\usepackage{amsfonts}
\usepackage{amsmath}
\usepackage{amssymb}
\usepackage{amsthm}
\usepackage[round,authoryear]{natbib}
\usepackage{graphics}
\usepackage{xcolor}
\usepackage[hidelinks]{hyperref}
\usepackage{bbm}
\usepackage{booktabs} 			  
\usepackage{adjustbox}
\usepackage[flushleft]{threeparttable}
\usepackage{float}
\usepackage{subcaption}
\usepackage{pdfpages}
\usepackage{xspace}
\usepackage[toc,page]{appendix}
\usepackage[onehalfspacing]{setspace}
\usepackage[paper=a4paper,top=1in,bottom=1in,left=1in,right=1in]{geometry}
\usepackage{authblk}
\usepackage{enumitem}
\usepackage[multiple]{footmisc}
\usepackage{soul}

\usepackage[loadshadowlibrary]{todonotes}

\newcommand{\mR}{\ensuremath{{\mathbb R}}}
\newcommand{\mRquer}{\ensuremath{\overline{\mR}}}
\newcommand{\mZ}{\ensuremath{{\mathbb Z}}}

\newcommand{\mN}{\ensuremath{{\mathbb N}}}
\newcommand{\mE}{\ensuremath{{\mathbb E}}}
\newcommand{\mP}{\ensuremath{{\mathbb P}}}

\newcommand{\eps}{\ensuremath{\varepsilon}}
\newcommand{\ind}{\ensuremath{\mathbbm{1}}}

\newcommand{\floor}[1]{\lfloor#1\rfloor}

\newcommand{\dlim}{\ensuremath{\Rightarrow}}
\newcommand{\plim}{\ensuremath{\overset{p}{\longrightarrow}}}
\newcommand{\argmin}{\operatorname{argmin}}
\newcommand{\sign}{\text{\rm sign}}
\newcommand{\diag}{\text{\rm diag}}
\newcommand{\AL}{{\scriptstyle\text{\rm\tiny AL}}}
\newcommand{\betaAL}{\ensuremath{\hat\beta_\AL}}
\newcommand{\betaALj}{\ensuremath{\hat\beta_{\AL,j}}}
\newcommand{\betaALi}{\ensuremath{\hat\beta_{\AL,i}}}
\newcommand{\cA}{\mathcal{A}}
\newcommand{\cI}{\mathcal{I}}
\newcommand{\cJ}{\mathcal{J}}
\newcommand{\betaALA}{\ensuremath{\hat\beta_{\AL,\cA}}}
\newcommand{\betaALAc}{\ensuremath{\hat\beta_{\AL,\cA^c}}}
\newcommand{\betaALAz}{\ensuremath{\hat\beta_{\AL,\cA_0}}}
\newcommand{\betaALAcz}{\ensuremath{\hat\beta_{\AL,\cA_0^c}}}
\newcommand{\Mhateps}{\widehat{\mathcal{M}}_T(\eps)}
\newcommand{\Mhatzero}{\widehat{\mathcal{M}}_T(0)}
\newcommand{\Mc}{\mathcal{M}^c}

\newtheorem{remark}{Remark}
\newtheorem{assumption}{Assumption}

\newtheorem{corollary}{Corollary}
\newtheorem{proposition}{Proposition}
\newtheorem{lemma}{Lemma}
\newtheorem{theorem}{Theorem}

\setlist[enumerate]{label=(\alph*),ref=(\alph*)}

\setuptodonotes{fancyline, color=cyan!10}

\author{Karsten Reichold\footnote{Correspondence to: Karsten Reichold,
Institute of Statistics and Mathematical Methods in Economics, TU
Wien, Wiedner Hauptstr.~8--10, A--1040 Vienna, Austria.\newline E-mail
addresses:
\href{mailto:karsten.reichold@tuwien.ac.at}{karsten.reichold@tuwien.ac.at}, 
\href{mailto:ulrike.schneider@tuwien.ac.at}{ulrike.schneider@tuwien.ac.at}}\ } 

\author{Ulrike Schneider} 
\affil{\small{Institute of Statistics and Mathematical Methods in Economics, TU Wien, 
Vienna, Austria}}

\title{Beyond the Oracle Property: Adaptive LASSO in
Cointegrating Regressions with Local-to-Unity Regressors}

\date{\vspace{-0.2cm} \today \vspace{0.2cm}}

\begin{document}
		
\maketitle

\vspace{-1cm}


\begin{abstract} 
\noindent This paper derives new asymptotic results for the adaptive
LASSO estimator in cointegrating regressions, allowing for uncertainty
about whether the regressors are exact unit root processes. We study
model selection probabilities, estimator consistency, and limiting
distributions under standard and moving-parameter asymptotics. We
further derive uniform convergence rates and the fastest local-to-zero
rates detectable by the estimator under conservative and consistent
tuning. For consistent tuning, we construct confidence regions that
are easy to implement, uniformly valid over the parameter space, and
achieve sure asymptotic coverage without requiring knowledge or
estimation of local-to-unity or long-run covariance parameters.
Simulation results reveal that the finite-sample distribution of the
adaptive LASSO estimator can deviate substantially from the oracle
property, whereas moving-parameter asymptotics provide much more
accurate approximations. Consequently, in addition to being infeasible
in applications due to their dependence on non-estimable nuisance
parameters, oracle-based confidence regions are often too small to
achieve adequate coverage in empirically relevant scenarios with small
but non-zero coefficients. In contrast, the proposed confidence
regions are always feasible and deliver reliable coverage across the
parameter space. An empirical application to predicting the U.S.
unemployment rate illustrates their practical usefulness for
quantifying uncertainty around adaptive LASSO estimates.

\smallskip

	
\smallskip


\end{abstract}

\section{Introduction} \label{sec:intro}

In recent years, the availability of large datasets comprising
numerous economic and financial variables has become the rule rather
than the exception. Consequently, practitioners using traditional
methods are frequently confronted with the challenge of selecting a
small number of relevant variables from an extensive pool of potential
covariates. In this context, statistical methods that simultaneously
perform estimation and variable selection, such as variants of the
least absolute shrinkage and selection operator (LASSO) introduced in
\citet{Tibshirani96}, are becoming increasingly popular in
econometrics.

While contributions such as \cite{WangEtAl07b}, \cite{RenZhang10},
\citet{MedeirosMendes16}, \citet{AdamekEtAl23}, and \citet{ChenEtAl25}
examine the use of LASSO-type estimators in models with (locally)
stationary time series, a growing body of recent research considers
models with highly persistent and endogenous regressors exhibiting
unit root or local-to-unit root behavior. For example,
\cite{LiaoPhillips15} propose adaptive shrinkage methods to estimate
vector error correction models. \citet{KooEtAl20} and \citet{MeiShi24}
consider so-called predictive regressions with high-dimensional
stationary and unit root regressors and derive certain asymptotic
properties of LASSO-type procedures in this context.
\citet{SmeekesWijler21} consider asymptotic properties of a LASSO-type
estimator in a high-dimensional error correction model.
\citet{Schweikert22} proposes an adaptive group LASSO method to
estimate structural breaks in cointegrating regressions and
\citet{TuXie23} follow a similar agenda in predictive regressions with
a fixed number of highly persistent regressors. \citet{LeeEtAl22}
derive asymptotic properties of LASSO-type estimators in regressions
with a fixed number of stationary and (potentially cointegrated)
local-to-unity regressors. \citet{GonzaloPitarakis25} propose a test
for cointegration in a regression model with a fixed number of
regressors based on the residuals of the adaptive LASSO estimator.

Typically, these papers consider regression models where the
regressors can be split into a set of relevant regressors, i.e., those
with a non-zero regression coefficient, and a set of irrelevant
regressors, i.e., those with a regression coefficient being exactly
equal to zero. Then these articles, among other things, often derive
model selection probabilities and sometimes also the limiting
distribution of the LASSO-type estimator for the set of non-zero
coefficients. If the procedure identifies zero coefficients with
probability approaching one and if the limiting distribution for the
non-zero coefficients coincides with that of the ordinary least
squares (OLS) estimator applied to the true model, the estimator is
said to possess the ``oracle property'' \citep{FanLi01}. While the
oracle property certainly appears convenient, it has to be interpreted
with extreme caution. The oracle property primarily characterizes the
asymptotic behavior of the penalized estimation method under the
assumption that coefficients are either exactly zero or sufficiently
large (in absolute value, relative to the sample size). It thus offers
very limited guidance in empirically relevant situations where some
coefficients are small (in absolute value, relative to sample size),
but not exactly zero. To study the asymptotic properties when
coefficients are allowed to be small but unequal to zero, one has to
let the true coefficients depend on sample size.

The contributions of \citet{LeeEtAl22} and \citet{TuXie23} take a step
forward by providing such moving-parameter asymptotic properties of
the LASSO-type procedures under consideration. Their results, however,
are restricted to specific sequences, as they focus on particular
rates at which the true coefficients go to zero, and also couple these
rates to the choice of the tuning parameter. We discuss the
implications of these restrictions at several points in this paper.

As a first illustration that these restrictions are not innocuous, we
refer to \citet{Kock16}, who analyzes the asymptotic properties of the
adaptive LASSO estimator in stationary and non-stationary
autoregressions. For example, in the AR(1) case
$$
\Delta y_t = \rho_T y_{t-1} + \eps_t,
$$
with i.i.d.~errors $\eps_t$, the estimator's behavior critically
depends on both the value of $\rho_T$ as well as the choice of the
tuning parameter. In particular, when $\rho_T$ approaches zero at rate
$1/T$, the procedure fails to detect the coefficient as non-zero if
the tuning parameter diverges, whereas it succeeds with positive
probability if the tuning parameter remains bounded. Moreover, if the
tuning parameter converges to zero as well, the procedure detects the
coefficient as non-zero with probability approaching one. What remains
unresolved, however, is the cut-off rate for the local-to-zero
coefficient that can still be detected by the estimator when the
tuning parameter diverges.

These so-called local-to-zero rates are motivated not only by
theoretical considerations but also by their relevance in empirical
applications. For example, they provide a natural framework for
modeling weak signal-to-noise ratios, which play an important role in,
e.g., the analysis of stock return predictability \citep[see,
e.g.,][]{CampbellYogo06,Campbell08,Phillips15,DemetrescuEtAl22}. In
this context, a detailed analysis of the properties of LASSO-type
estimators under general sequences at which the true coefficients are
allowed to go to zero, reveals the fastest local-to-zero rates that
can still be detected by the estimator.

In this paper, we provide a comprehensive analysis of the asymptotic
properties of the adaptive LASSO estimator \citep{Zou06} applied to
cointegrating regressions with potentially local-to-unity regressors.
Allowing for deviations from exact unit roots is important, as
macroeconomic variables often display high persistence without being
unit-root processes, see, e.g., \citet{Jensen09} for evidence on
inflation, and \citet{HwangValdes24} for further discussion. In
particular, we derive model selection probabilities, estimator
consistency, and limiting distributions, while allowing the true
coefficients to freely move through the parameter space along
arbitrary sequences. In addition, we derive uniform convergence rates
and the local-to-zero rates of the true coefficients that can still be
detected by the estimator. These findings will shed light on, e.g.,
the signal-to-noise ratios practitioners can accept when employing the
adaptive LASSO estimator in empirical applications. We complete our
theoretical analysis with providing uniformly valid confidence regions
in the typical regime when the tuning parameter diverges.

Before we discuss our results in more detail, note that in the context
of classical linear regression models with non-stochastic regressors,
\citet{PoetscherSchneider09} provide a comprehensive analysis of the
adaptive LASSO estimator, examining both its asymptotic behavior and
finite-sample properties. In particular, they derive uniform
convergence rates and highlight how the choice of tuning parameters
influences the performance of the estimation procedure. In contrast to
the analysis of \cite{PoetscherSchneider09}, we have to overcome
several difficulties to derive the results of this paper. First, we
are dealing with different convergence rates of the estimators and
second, we have to account for the stochastic and non-stationary
nature of the regressors, which leads to the occurrence of stochastic
integrals in the limit. Moreover, the second-order bias terms in the
limiting distribution of the OLS estimator also affect the limiting
distribution of the adaptive LASSO estimator. Finally, we also provide
results for the multivariate case which is not addressed in the
aforementioned article.

Based on the asymptotic study of model selection probabilities, we
distinguish between two regimes determined by the large-sample
behavior of the tuning parameter: consistent model selection (or
``consistent tuning''), where zero coefficients are found with
asymptotic probability equal to one, and conservative model selection
(``conservative tuning''), where zero coefficients are detected with
asymptotic probability less than one. The asymptotic properties of the
adaptive LASSO estimator differ substantially between these two cases,
with the main massages as follows: In the conservatively tuned case,
the estimator is uniformly $T$-consistent for parameter estimation and
the cut-off rate for local-to-zero coefficients that can be detected
by the procedure is $1/T$. In the consistently tuned case, the uniform
convergence rate depends on the tuning parameter and is slower than
$1/T$. Deviations of the true parameter from zero of rate $1/T$ cannot
be discovered by the estimator. The fastest local-to-zero rate that is
still detectable with positive probability again depends on the tuning
parameter and is slower than $1/T$. Moreover, in the consistently
tuned case, the detailed theoretical analysis of the adaptive LASSO
estimator allows us to construct confidence regions that have coverage
probability approaching one uniformly over the parameter space,
without requiring any knowledge or estimation of local-to-unity or
long-run covariance parameters. Although inspired by the
non-stochastic regressor case considered in \cite{AmannSchneider23},
extending the construction of such regions to the stochastic
regressors present in the unit-root or local-to-unity setting
substantially changes the nature of the problem.

The theoretical analysis is complemented by an extensive simulation
study. The results show that the finite-sample distribution of the
adaptive LASSO estimator often deviates substantially from what is
suggested by the oracle property, whereas the limiting distributions
derived under moving-parameter asymptotics capture the finite-sample
properties of the procedure more closely. Moreover, the poor
approximation quality of the oracle property to the finite-sample
distribution of the adaptive LASSO estimator is also reflected in the
performance of the confidence regions based on the oracle property. In
particular, we find that the oracle-based confidence regions are often
too small to achieve adequate coverage in empirically relevant
scenarios with small but non-zero coefficients, whereas the confidence
regions proposed in this paper achieve adequate coverage across the
entire parameter space.

Taken together, the theoretical and simulation results indicate that
the oracle property provides an incomplete characterization of both
the asymptotic and finite-sample properties of the adaptive LASSO
estimator. A full understanding instead requires a moving-parameter
framework of the type considered in this paper. Moving beyond the
oracle property in this way enables the construction of uniformly
valid confidence regions under consistent tuning.

Finally, an empirical application to predicting the U.S. unemployment
rate illustrates the usefulness of the proposed confidence regions for
quantifying uncertainty around adaptive LASSO estimates.

The paper is organized as follows. Section~\ref{sec:setting}
introduces the model and states the assumptions.
Section~\ref{sec:asymptotics} contains our theoretical contributions:
Section~\ref{subsec:fixed} derives the large-sample properties of the
adaptive LASSO estimator in a fixed-parameter asymptotic framework in
the univariate regressor case, Section~\ref{subsec:unif} considers a
moving-parameter asymptotic framework in the univariate regressor
case, and Section~\ref{subsec:multi} extends the results to the
multivariate case. Section~\ref{subsec:ci} then derives the uniform
confidence region based on the adaptive LASSO estimator.
Section~\ref{sec:simulation} presents the simulation results and
Section~\ref{sec:emp} contains the empirical application.
Section~\ref{sec:conclusion} summarizes and concludes. All proofs are
provided in the appendix, which also contains additional simulation
and empirical results.

We use the following notation: $\floor{x}$ denotes the integer part of
$x \in \mR$, $L$ is the backward-shift operator, $\text{diag}(\cdot)$
denotes a diagonal matrix with elements specified throughout, and
$\mRquer \coloneqq \mR \cup \{-\infty,\infty\}$. With $\dlim$ and
$\plim$ we denote weak convergence and convergence in probability,
respectively, and all limits apply as the sample size $T$ tends to
infinity. We denote a normal distribution with mean $\mu$ and
covariance matrix $\Sigma$ as $\mathcal{N}\left(\mu,\Sigma\right)$.
The symbol $I_k$ denotes the $k$-dimensional identity matrix. For any
event $E$, the indicator function $\ind\{E\}$ equals one if $E$ occurs
and zero otherwise. If a sequence $a_T$ is identical to $a\in\mR$ for
all $T$, we write $a_T \equiv a$. By $\omega$ we denote an element of
the sample space of the underlying probability space and $(\omega)$
attached to a random variable denotes its realization for this
particular $\omega$.

\section{Setting and Assumptions} \label{sec:setting}

As motivated in the introduction, we consider a cointegrating
regression model with local-to-unity regressors of the form
\begin{align}
y_t &= x_t'\beta_T + u_t,\label{eq:y}\\
x_t &= \left(I_k - T^{-1}c\right) x_{t-1} + v_t, \label{eq:x}
\end{align}
for $t = 1,\dots,T$, where $c \coloneqq \diag(c_1,\ldots,c_k)$ with
$c_j\geq0$, and $x_0 = O_{\mP}(1)$. For $c=0$, the model encompasses
classical cointegrating regressions with unit root regressors.
Following \cite{LeeEtAl22}, we treat the number of regressors $k$ as
fixed. For notational brevity, we
exclude deterministic components from~\eqref{eq:y}. For $\{w_t\}_{t
\in \mZ} \coloneqq \{[u_t,v_t']'\}_{t \in \mZ}$ we impose the
following assumption.
\begin{assumption} \label{ass:w1}
Let $w_t = \Psi(L) \eps_t = \sum_{j=0}^\infty \Psi_j \eps_{t-j}$, with
$\sum_{j=0}^\infty j \Vert \Psi_j\Vert < \infty$ and $\det(\Psi(1)) \neq 0$,
where $\{\eps_t\}_{t \in \mZ}$ is a $(1+k)$-dimensional strictly
stationary ergodic martingale difference sequence with natural
filtration $\mathcal{F}_t \coloneqq
\sigma\left(\{\eps_s\}_{-\infty}^t\right)$, conditional covariance
matrix $\Sigma \coloneqq \mE(\eps_t\eps_t'| \mathcal{F}_{t-1}) >
0$, and $\sup_{t \geq 1}\mE(\Vert\eps_t\Vert^r|\mathcal{F}_{t-1})
< \infty$ a.s.\ for some $r > 4$.
\end{assumption}
Conditions similar to Assumption~\ref{ass:w1} are common in the
cointegrating regression literature, see, e.g., \cite{WagnerHong16}
for a detailed discussion. In particular, Assumption~\ref{ass:w1} allows for regressor
endogeneity and error serial correlation, but excludes cointegration among
the elements of $x_t$.\footnote{As shown by \cite{LeeEtAl22}, allowing for cointegration among the regressors requires a different estimation strategy, termed the twin adaptive LASSO. Extending our analysis to this estimator in the more general setting is left for future research.} Under Assumption~\ref{ass:w1}, the process
$\{w_t\}_{t \in \mZ}$ fulfills a functional central limit theorem of
the form
\begin{align} \label{eq:FCLT}
T^{-1/2} \sum_{t=1}^{\floor{rT}} w_t \Rightarrow B(r) = 
\begin{bmatrix} B_u(r) \\ B_v(r) \end{bmatrix} = 
\Omega^{1/2} W(r), \quad 0 \leq r \leq 1,
\end{align}
where $W(r) = [W_{u \cdot v}(r),W_v(r)']'$ is a $(1 + k)$-dimensional
vector of independent standard Brownian motions and $\Omega \coloneqq
\sum_{h = -\infty}^\infty \mE(w_0 w_h')>0$ denotes the long-run
covariance matrix of $\{w_t\}_{t \in \mZ}$. Note that the results in
this paper also hold under alternative sets of assumptions  as long as
they imply the functional central limit theorem for $\{w_t\}_{t \in
\mZ}$ given in \eqref{eq:FCLT}, see, e.g., \cite{IbragimovPhillips08} and 
\cite{deJong03UP} for possible other conditions.

The target of our investigation, the adaptive LASSO estimator
\citep{Zou06} of $\beta_T$ in~\eqref{eq:y}, is defined as
\begin{align} \label{eq:ALASSO}
\betaAL \coloneqq \argmin_{b \in \mR^k} \left\{\sum_{t=1}^T \left(y_t - x_t'b\right)^2 + 
\lambda_T \sum_{j=1}^k|\hat\beta_j^0|^{-\gamma} |b_j|\right\},
\end{align}
where $\lambda_T > 0$ and $\gamma \geq 1$ are tuning parameters. In
contrast to the classical LASSO estimator of \cite{Tibshirani96}, the
penalty term for the $j$-th coefficient in \eqref{eq:ALASSO} contains
the reciprocal of the absolute value of a preliminary estimator
$\hat\beta_j^0$ of $\beta_{T,j}$, where $\beta_{T,j}$ denotes the
$j$-th element of $\beta_T$. Its aim is to increase the penalty term
if $\beta_{T,j}$ seems small to encourage shrinking, and to penalize
less if $\beta_{T,j}$ appears to be large in order to reduce the bias.
In practice, $\gamma$ is often chosen as 1 or 2 and $\lambda_T$ is
typically selected based on cross-validation or information criteria.
In line with the recommendations in \cite{LeeEtAl22}, we set $\gamma =
1$ and $\hat\beta^0 = \hat\beta$, where $\hat\beta$ denotes the
OLS estimator of $\beta_T$ in~\eqref{eq:y}.
Under Assumption~\ref{ass:w1}, the limiting distribution of
$\hat\beta$ is given by
\begin{align}\label{eq:OLS}
T\left(\hat\beta - \beta_T\right) \Rightarrow \mathcal{Z}^c \coloneqq 
\left(\int_0^1 J_v^c(r)J_v^c(r)'dr\right)^{-1}\left(\int_0^1 J_v^c(r)dB_u(r) + \Delta_{vu}\right), 
\end{align}
where $J_v^c(r) \coloneqq \int_0^r e^{(r-s)c}dB_v(s)$ and $\Delta_{vu}
\coloneqq \sum_{h=0}^\infty \mE(v_0u_h)$, see, \cite{Phillips88}. To
simplify notation, we define $\zeta_{vv}^c \coloneqq \int_0^1
J_v^c(r)J_v^c(r)'dr$.

If at least one regressor is endogenous, the limiting distribution of
the OLS estimator is contaminated by second-order bias terms. In
contrast to the classical cointegrating regressions with unit root
regressors, where such bias terms can be addressed, for example, using
the fully modified approach of \cite{PhillipsHansen90}, in the
local-to-unity regressor case, these bias terms additionally depend on
the unknown local-to-unity parameters $c_j$. Since these parameters
are not consistently estimable, the resulting bias is difficult to
correct, see, e.g., \citet{Phillips23} and the references therein for
a detailed discussion. Consequently, constructing asymptotically valid
confidence intervals or hypothesis tests in cointegrating regressions
with local-to-unity regressors is non-trivial, see, e.g.,
\citet{MagdalinosPhillips09} for an instrumental variables approach
and \citet{HwangValdes24} for a modified low-frequency transformed and
augmented OLS method. In Section~\ref{subsec:ci}, we show that
asymptotically valid uniform confidence regions can nevertheless be
constructed from the consistently tuned adaptive LASSO estimator which
are straightforward to implement, and do not require any knowledge of
local-to-unity parameters.

\begin{remark} \label{rem:PRED}
Our results also extend to the predictive regression setting
\citep[see, e.g.,][]{KooEtAl20,MeiShi24}, where the regressor at time
$t$ is $x_{t-1}$ rather than $x_t$. In this setting, the limiting
distribution of the OLS estimator coincides with $\mathcal{Z}^c$,
except that $\Delta_{vu}$ needs to be replaced by $\sum_{h=1}^\infty
\mE(v_0u_h)$.
\end{remark}

Before deriving the asymptotic properties of the adaptive LASSO
estimator in the multivariate case, where no closed-form solution of
the minimization problem in~\eqref{eq:ALASSO} is available, we
consider the univariate case, where the minimization problem has an
explicit solution of the form
\begin{equation} \label{eq:betaAL}
\betaAL = \begin{cases}
\hat\beta - \tilde{\lambda}_T \hat\beta^{-1} & \text{if }|\hat\beta| > \sqrt{\tilde\lambda_T} \\
0 & \text{otherwise},
\end{cases}
\end{equation}
with $\tilde\lambda_T \coloneqq 0.5\lambda_T (\sum_{t=1}^T
x_t^2)^{-1}$ and $\sum_{t=1}^T x_t^2 = O_\mP(T^{-2})$, compare
\cite{PoetscherSchneider09}. Analyzing the univariate case in detail
facilitates a transparent derivation of the estimator’s asymptotic
properties and provides insights into the underlying  mechanisms,
helping to clarify the multivariate results.

Equation~\eqref{eq:betaAL} reveals that in the univariate case the
adaptive LASSO estimator can be represented solely in terms of the OLS
estimator and that the tuning parameter $\lambda_T$ only affects the
estimator in its ``standardized'' version $\tilde\lambda_T$, where the
term $\sum_{t=1}^T x_t^2$ can be viewed as a measure of variation in
the regressor. This explains why for the asymptotic study in the
subsequent section, the large-sample behavior of $\lambda_T$ relative
to $T^{-2}$ is important -- which is the rate at which $\sum_{t=1}^T
x_t^2$ stabilizes.

Figure~\ref{fig:AL} illustrates the
relationship between $\betaAL$ and $\hat\beta$ for different values of
$\tilde\lambda_T$.
\begin{figure}[!ht]	
\begin{center}
\includegraphics[trim={0cm 0.2cm 1cm 0.8cm},width=0.4\textwidth,clip]{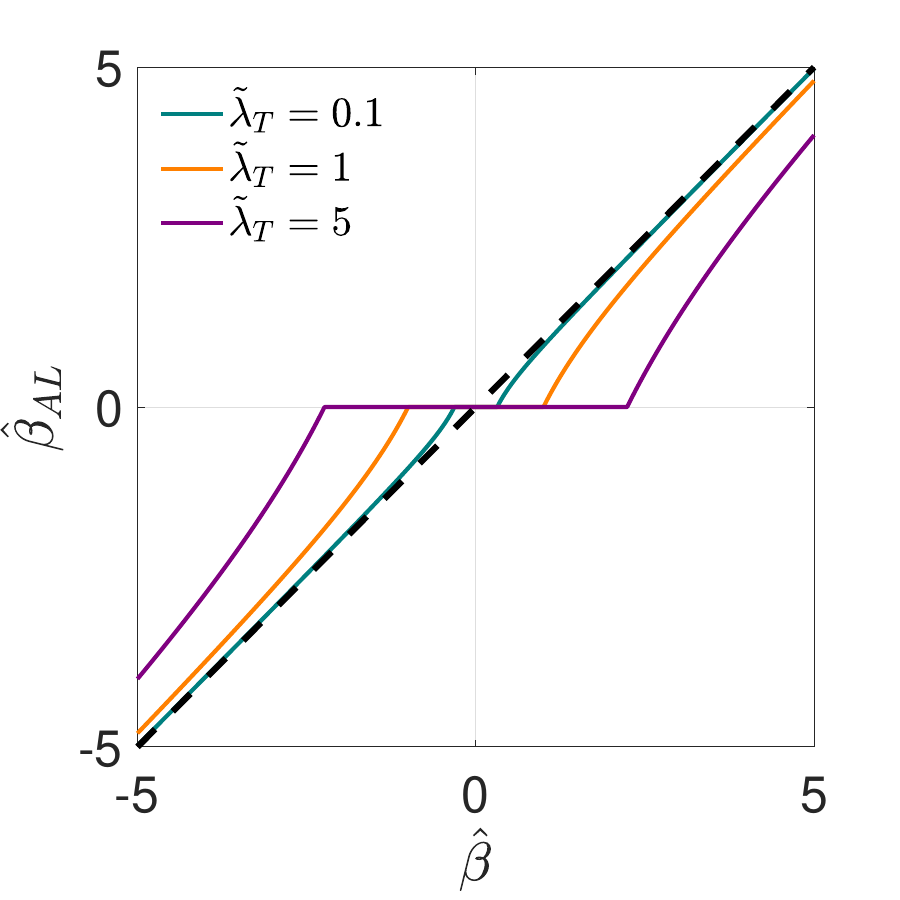}
\end{center} \vspace{-0.2cm}
\caption{\label{fig:AL} Relationship between $\betaAL$ and $\hat\beta$
for different values of $\tilde{\lambda}_T$.}
\end{figure}
%

\section{Asymptotic Theory} \label{sec:asymptotics}

In this section, we investigate the large-sample behavior of the
adaptive LASSO estimator under two different asymptotic regimes
regarding the model selection properties of the procedure. We speak of
consistent model selection (or ``consistent tuning'') if all zero
coefficients are detected with asymptotic probability equal to one,
whereas the case where at least one zero coefficient is set to zero by
the estimator with limiting probability strictly less than one is
referred to as conservative model selection (``conservative tuning'').
Formally, this definition only relates to zero coefficients and poses
no requirement on the non-zero coefficients in the model. Which regime
applies depends on the limiting behavior of the tuning parameter
sequence $\lambda_T$, as will be clarified below.

We first consider the univariate regressor case. In
Section~\ref{subsec:fixed}, we set $\beta_T \equiv \beta$ and study
how the behavior of the tuning parameter sequence $\lambda_T$ affects
both model selection and parameter estimation. In addition, we derive
the asymptotic distribution of the estimator when $\beta_T \equiv
\beta$ is fixed under both model selection regimes. The insights from
Section~\ref{subsec:fixed} into which limiting behavior of $\lambda_T$
leads to what type of model selection regime then serve as a starting
point for the detailed analysis of the large-sample behavior of the
adaptive LASSO estimator in Section~\ref{subsec:unif}. In that
section, we adopt a moving-parameter framework in which the true
parameter $\beta_T$ may vary with the sample size $T$. This allows to
determine local-to-zero and uniform convergence rates, and to derive
asymptotic distributions that, as will become apparent in the
simulation study in Section~\ref{subsec:sim_fs}, more accurately
capture the estimator’s finite-sample properties. The analysis is
again conducted under both conservative and consistent model
selection.

After utilizing the explicit expression of the adaptive LASSO
estimator in the univariate case, we turn to the multivariate case in
Section~\ref{subsec:multi}. In this case, the absence of a closed-form
solution necessitates different techniques for deriving the asymptotic
properties. Unlike in the univariate case, we do not separately
present results under fixed-parameter asymptotics, since these are
encompassed by the moving-parameter framework and do not provide
additional insights beyond those already established in the univariate
analysis. As before, we study the asymptotic behavior of the estimator
under both conservative and consistent model selection.

Finally, in Section~\ref{subsec:ci}, we use the insights from
Section~\ref{subsec:multi} to construct asymptotically valid uniform
confidence regions for the consistently tuned adaptive LASSO
estimator.

\subsection{Fixed-Parameter Asymptotics in the Univariate Case} \label{subsec:fixed}

As outlined above, we begin by deriving asymptotic results for the
adaptive LASSO estimator in the univariate regressor case under a
fixed-parameter framework, i.e., by setting $\beta_T \equiv \beta$
in~\eqref{eq:y} with $\beta \in \mR$ fixed. We first examine the
large-sample properties of the estimator with respect to model
selection.

\begin{proposition}[Model selection] \label{prop:ms-fixed}
Let $\{y_t\}_{t\in\mZ}$ and $\{x_t\}_{t\in\mZ}$ be generated
by~\eqref{eq:y}~and~\eqref{eq:x} with $k=1$ and $\beta_T\equiv\beta$, and let
$\{w_t\}_{t\in\mZ}$ satisfy Assumption~\ref{ass:w1}. 
\begin{enumerate}
		
\item \label{enum:basic-noFN} Let $\beta\neq 0$. If $T^{-2}\lambda_T \to 0$,
then $\mP\left(\betaAL = 0\right) \to 0$.
		
\item Let $\beta=0$. 
\begin{enumerate}
			
\item[(b1)] If $\lambda_T \to \lambda_0$, $0 \leq
\lambda_0 < \infty$, then
$$
\mP\left(\betaAL = 0\right) \to 
\mP\left((\zeta_{vv}^c)^{1/2}\left|\mathcal{Z}^c\right| \leq \sqrt{\frac{\lambda_0}{2}}\right) 
< 1.
$$
			
\item[(b2)] If $\lambda_T \to \infty$, then $\mP\left(\betaAL = 0\right) \to 1$.
			
\end{enumerate}
		
\end{enumerate}
	
\end{proposition}

Proposition~\ref{prop:ms-fixed} reveals the role of the tuning
parameter sequence $\lambda_T$ for model selection of the adaptive
LASSO: In case $\lambda_T \to \lambda_0$ with $0 \leq \lambda_0 <
\infty$, the estimator detects zero coefficients with probability
smaller than one asymptotically, resulting in conservative model
selection. In contrast, when $\lambda_T \to \infty$, the estimator
sets zero coefficients equal to zero with probability approaching one
and consequently leads to consistent model selection. In the
following, we therefore refer to the case $\lambda_T \to \lambda_0$,
$0 \leq \lambda_0 < \infty$, as \emph{conservative tuning}, whereas
the case $\lambda_T\to\infty$ is termed \emph{consistent
tuning}.\footnote{Under conservative tuning with $\lambda_0 = 0$, the
adaptive LASSO estimator is equivalent to OLS in the sense that
$T(\betaAL - \hat\beta) = o_\mP(1)$. This follows directly from the
proof of
Theorem~\ref{thm:param_consist-multi}\ref{enum:rate_conserv-unif} in
Section~\ref{subsec:multi}. The result also holds in a
moving-parameter framework and extends to the multivariate case.}
Moreover, the condition $T^{-2}\lambda_T \to 0$ is a basic requirement
for the tuning parameter as it ensures that the probability of the
adaptive LASSO estimator incorrectly setting the coefficient to zero
vanishes asymptotically. While this condition is automatically
fulfilled under conservative tuning, it controls the rate at which
$\lambda_T$ may diverge under consistent tuning. We will assume this
condition in all subsequent statements in this section.

We now derive the asymptotic properties of the adaptive LASSO
estimator with respect to parameter estimation.

\begin{proposition}[Parameter estimation] \label{prop:param_consist-fixed}	
Let $\{y_t\}_{t \in \mZ}$ and $\{x_t\}_{t \in \mZ}$ be generated
by~\eqref{eq:y}~and~\eqref{eq:x} with $k=1$ and $\beta_T \equiv \beta$, let $\{w_t\}_{t \in \mZ}$ satisfy Assumption~\ref{ass:w1}, and assume $T^{-2}\lambda_T \to 0$.
\begin{enumerate}
		
\item \label{enum:param_consist-fixed} $\betaAL - \beta = o_\mP(1)$.
		
\item \label{enum:rateT-fixed} If $T^{-1}\lambda_T \to
\tilde\lambda_0$, $0 \leq \tilde\lambda_0 < \infty$, then $T(\betaAL -
\beta) = O_\mP(1)$.

\item \label{enum:slower_rate-fixed} If $T^{-1}\lambda_T \to \infty$ then
$\lambda_T^{-1}T^2(\betaAL - \beta) = O_\mP(1)$.
		
\end{enumerate}
	
\end{proposition}

From Proposition~\ref{prop:param_consist-fixed}\ref{enum:param_consist-fixed}, 
we learn that the basic condition $T^{-2}\lambda_T \to 0$, which
ensures that non-zero coefficients are not falsely put to zero as
shown in Proposition~\ref{prop:ms-fixed}\ref{enum:basic-noFN}, also
guarantees that the procedure is consistent for $\beta$ with respect
to parameter estimation.
Proposition~\ref{prop:param_consist-fixed}\ref{enum:rateT-fixed} shows
that in case of conservative tuning or for consistent tuning with a
slowly diverging tuning parameter sequence (in the sense that
$T^{-1}\lambda_T$ stays bounded), the convergence rate of the adaptive
LASSO estimator is $T^{-1}$ and coincides with the rate of OLS.
However, when the estimator is tuned consistently and $\lambda_T$
tends to infinity fast enough so that $T^{-1}\lambda_T$ diverges also,
Proposition~\ref{prop:ms-fixed}(c) reveals that the convergence rate
of the adaptive LASSO estimator is only $T^{-2}\lambda_T$, which is
slower than $T^{-1}$ in this case.

We now derive the limiting distribution of the adaptive LASSO
estimator.

\begin{proposition}[Limiting distribution]\label{prop:ls_dist-fixed}
Let $\{y_t\}_{t \in \mZ}$ and $\{x_t\}_{t \in \mZ}$ be generated
by~\eqref{eq:y}~and~\eqref{eq:x} with $k=1$ and $\beta_T \equiv \beta$, let $\{w_t\}_{t \in \mZ}$ satisfy Assumption~\ref{ass:w1}, and assume $T^{-2}\lambda_T \to 0$.
\begin{enumerate}
		
\item \label{enum:ls_dist_conserv-fixed} Let $\lambda_T \to \lambda_0$, $0 \leq \lambda_0 < \infty$.
		
\begin{itemize}

\item[(a1)] If $\beta \neq 0$, then $T(\betaAL - \beta) \Rightarrow
\mathcal{Z}^c$.
			
\item[(a2)] If $\beta = 0$, then
\begin{align*}
T(\betaAL - \beta) \Rightarrow \ind\left\{(\zeta_{vv}^c)^{1/2}\left|\mathcal{Z}^c\right| >
\sqrt{\frac{\lambda_0}{2}}\right\}\left(\mathcal{Z}^c -
\frac{\lambda_0}{2\zeta_{vv}^c}(\mathcal{Z}^c)^{-1}\right).
\end{align*} 

\end{itemize}
		
\item \label{enum:ls_dist_consist-fixed} Let $\lambda_T \to \infty$.

\begin{itemize}

\item[(b1)] If $T^{-1} \lambda_T \to
\tilde\lambda_0$, $0 \leq \tilde\lambda_0 < \infty$, then
\begin{align*}
T(\betaAL - \beta) \Rightarrow \ind\left\{\beta \neq 0\right\} 
\left(\mathcal{Z}^c - (\zeta_{vv}^c)^{-1}\frac{\tilde\lambda_0}{2\beta}\right).
\end{align*}

\item[(b2)] If $T^{-1} \lambda_T \to \infty$, then 
\begin{align*}
\lambda_T^{-1}T^2(\betaAL - \beta) \Rightarrow - \ind\left\{\beta \neq 0\right\} 
(\zeta_{vv}^c)^{-1}\frac{1}{2\beta}.
\end{align*}

\end{itemize}
		
\end{enumerate}
	
\end{proposition}

Under conservative tuning,
Proposition~\ref{prop:ls_dist-fixed}\ref{enum:ls_dist_conserv-fixed}
reveals that the asymptotic distribution of the adaptive LASSO
estimator coincides with the one of OLS if $\beta\neq 0$. In case
$\beta = 0$, however, the limiting distribution consists of an atomic
part at zero incurred by the positive probability of the estimator
being equal to zero, and an absolutely continuous part for when the
estimator is not equal to zero.

Proposition~\ref{prop:ls_dist-fixed}\ref{enum:ls_dist_consist-fixed}
further shows that under consistent tuning, the limiting distribution
of the adaptive LASSO estimator fully collapses to pointmass at zero
whenever $\beta = 0$.\footnote{For completeness, note that the proofs
of (b1) and (b2) in Proposition~\ref{prop:ls_dist-fixed} reveal that
$\delta_T(\betaAL - \beta) \Rightarrow 0$ for any sequence $\delta_T \to \infty$
if $\beta = 0$ and $\lambda_T \to \infty$.} In case $\beta \neq 0$,
both the convergence rate of the adaptive LASSO estimator and its
limiting distribution depend on how $\lambda_T$ diverges in relation
to $T$. If $\lambda_T$ tends to infinity slower than $T$, the limiting
distribution of $T(\betaAL - \beta)$ coincides with the one of OLS. If
$\lambda_T$ diverges at rate $T$, rate-$T$ consistency is still
maintained as is already shown in
Proposition~\ref{prop:param_consist-fixed}(b1), but the asymptotic
distribution now deviates from the one of OLS by a random shift that
is inversely proportional to and has the opposite sign of $\beta$.
This random shift in the limiting distribution of the adaptive LASSO
estimator is not detected in \citet{LeeEtAl22}, as their assumptions
imply that $\tilde\lambda_0 = 0$. Moreover, in contrast to the second
order bias term in the limiting distribution of the OLS estimator, the
random shift in the limiting distribution of the adaptive LASSO
estimator does not vanish if the regressor is exogenous. Lastly, if
$\lambda_T$ diverges faster than $T$, the convergence rate becomes
slower than $T$, as shown in
Proposition~\ref{prop:param_consist-fixed}(b2). In this case, for the
appropriately scaled estimator $\lambda_T^{-1}T^2(\betaAL - \beta)$,
the term ``corresponding to OLS'' no longer appears in the limit.
Finally, Proposition~\ref{prop:ls_dist-fixed} confirms that the rates
established in Proposition~\ref{prop:param_consist-fixed} are indeed
sharp.

The results in this section can be used to explicitly formulate a
condition for the so-called ``oracle property'' of the adaptive LASSO
estimator, a term coined by \cite{FanLi01}, established for the
adaptive LASSO estimator in a classical linear regression model in
\cite{Zou06}.

\begin{corollary}[``Oracle property'']\label{cor:summary_fixed}
Let $\{y_t\}_{t \in \mZ}$ and $\{x_t\}_{t \in \mZ}$ be generated
by~\eqref{eq:y}~and~\eqref{eq:x} with $k=1$ and $\beta_T \equiv
\beta$, let $\{w_t\}_{t \in \mZ}$ satisfy Assumption~\ref{ass:w1}, and
assume $T^{-1} \lambda_T + \lambda_T^{-1} \to 0$. Then
$\mP\left(\betaAL = 0\right) \to \ind\{\beta=0\}$ and $T(\betaAL
-\beta) \Rightarrow \ind\left\{\beta \neq 0\right\}\mathcal{Z}^c$.
\end{corollary}
Corollary~\ref{cor:summary_fixed} states that, under consistent
tuning with $\lambda_T$ diverging more slowly than $T$, the adaptive
LASSO estimator identifies non-zero coefficients and sets null
coefficients to zero with probability approaching one. Moreover, its
limiting distribution coincides with that of the OLS estimator
whenever $\beta\neq 0$.
	
Results similar to Corollary~\ref{cor:summary_fixed} often constitute
the main asymptotic findings in the literature on LASSO-type
estimators across various contexts, see, e.g.,
\citet[Theorem~1]{MedeirosMendes16},
\citet[Corollary~1]{SmeekesWijler21}, \citet[Theorem~4]{Schweikert22},
\citet[Theorems~1 and 3]{LeeEtAl22}, \citet[Theorem~3.2]{TuXie23}, and
\citet[Theorem~4.2]{ChenEtAl25}. However, although such results may
seem convenient, they have to be interpreted with extreme caution.
While the ``oracle property'' represents the large-sample performance
of the estimator in situations where regression coefficients are
either equal to zero or ``relatively large'' (in absolute value and in
relation to sample size), it does not shed light on the empirically
relevant case where some coefficients are ``relatively small'' rather
than being exactly equal to zero. In Section~\ref{subsec:unif}, we
analyze the large-sample properties of the adaptive LASSO estimator
within an asymptotic framework that also accommodates this case.

\subsection{Moving-Parameter Asymptotics in the Univariate Case} \label{subsec:unif}

In this section, we study the asymptotic behavior of the adaptive
LASSO estimator in the univariate regressor case within a
moving-parameter framework, where the unknown coefficient $\beta_T$
may vary with $T$. This framework overcomes the limitations of the
fixed-parameter setting, in which the true coefficient is restricted
to being either exactly zero or asymptotically large relative to
sample size. Such a dichotomy is unsatisfactory, since in finite
samples the coefficient may be non-zero yet small, especially when the
signal-to-noise ratio is low. The smaller the non-zero coefficient that
an estimator can still reliably detect, the better its performance. A
key advantage of the moving-parameter framework is that it reveals the
local-to-zero rate at which the estimator can detect non-zero
coefficients. 

In the following theorem, we derive the model selection
probabilities of the adaptive LASSO under both conservative and
consistent tuning.

\begin{theorem}[Model selection] \label{thm:ms-unif}

Let $\{y_t\}_{t \in \mZ}$ and $\{x_t\}_{t \in \mZ}$ be generated
by~\eqref{eq:y}~and~\eqref{eq:x} with $k=1$, and let $\{w_t\}_{t \in \mZ}$ satisfy
Assumption~\ref{ass:w1}.
\begin{enumerate}

\item \label{enum:ms_conserv-unif} If $\lambda_T \to \lambda_0$, $0
\leq \lambda_0< \infty$, and $T\beta_T \to \beta_0 \in \mRquer$, then
$$
\mP\left(\betaAL = 0\right) \to \mP\left((\zeta_{vv}^c)^{1/2}\left|\mathcal{Z}^c +
\beta_0\right| \leq \sqrt{\frac{\lambda_0}{2}}\right) < 1.
$$

\item \label{enum:ms_consist-unif} If $\lambda_T \to \infty$ and
$\lambda_T^{-1/2}T\beta_T \to \tilde\beta_0 \in \mRquer$, then
$$
\mP\left(\betaAL = 0\right) \to \mP\left((\zeta_{vv}^c)^{1/2} \leq
\frac{1}{\sqrt{2}}\left| \tilde\beta_0 \right|^{-1}\right).
$$

\end{enumerate}

\end{theorem}

\begin{remark} \label{rem:sequences}
Theorem~\ref{thm:ms-unif} describes the asymptotic behavior of the
model selection probabilities for \emph{arbitrary} sequences of
$\beta_T$ in the sense that all accumulation points of the selection
probabilities can be obtained in the following way: apply the result
to subsequences and observe that for every such subsequence, we can
select a further subsequence such that relevant quantities, i.e.,
$T\beta_T$ or $\lambda_T^{-1/2}T\beta_T$, converge to a limit in
$\mRquer$. A similar comment also applies to
Theorems~\ref{thm:param_consist-unif}--\ref{thm:ls_dist-unif} below.
\end{remark}

Part~(a) of Theorem~\ref{thm:ms-unif} shows that under
conservative tuning, if $\beta_T$ is bounded away from zero or
converges to zero at rate slower than $T^{-1}$, i.e., $|\beta_0| =
\infty$, the estimator can detect the coefficient as non-zero with
asymptotic probability equal to one. If $\beta_T \equiv 0$ or $\beta_T$
converges to zero at rate $T^{-1}$ or faster, i.e., $\beta_0 \in
\mR$, the estimator will set the coefficient equal to zero with positive
probability less than one even asymptotically.

To interpret the results in (b) of Theorem~\ref{thm:ms-unif} in a
meaningful way, we suppose that the basic condition $T^{-2}\lambda_T
\to 0$ from Section~\ref{subsec:fixed} holds, which we also assume to
hold in all subsequent statements in this section. Part~(b) of
Theorem~\ref{thm:ms-unif} then reveals that if $\beta_T$ is bounded
away from zero or converges to zero at rate slower than
$T^{-1}\lambda_T^{1/2}$, i.e., $|\tilde\beta_0| = \infty$, the
estimator can detect the coefficient as non-zero with asymptotic
probability equal to one. If $\beta_T$ converges to zero exactly at
rate $T^{-1}\lambda_T^{1/2}$, i.e., $\tilde\beta_0 \in \mR$,
$\tilde\beta_0 \neq 0$, the estimator will set the coefficient equal
to zero with positive probability less than one asymptotically.
Finally, if $\beta_T \equiv 0$ or $\beta_T$ converges to zero with
rate faster than $T^{-1}\lambda_T^{1/2}$, the estimator will set the
coefficient equal to zero with asymptotic probability equal to one.

\begin{remark} \label{rem:lee-local_to_zero}
As discussed above, Theorem~\ref{thm:ms-unif} shows that in the
consistently tuned case, the relevant local-to-zero rate is
$T^{-1}\lambda_T^{1/2}$ in the sense that coefficients that converge
to zero slower than that will be detected as non-zero with asymptotic
probability equal to one and coefficients that converge to zero faster
will be detected as non-zero with asymptotic probability equal to
zero. In \cite{LeeEtAl22}, it appears that coefficients converging to
zero at rate $T^{-\delta}$ for any $\delta \in (0,1)$ pose no
difficulty for the consistently tuned adaptive Lasso in the sense that
they will be detected as non-zero with asymptotic probability equal to
one. This, however, is made possible by an assumption that links the
true coefficient to the tuning parameter (through the parameter
$\delta$), thereby masking the dependence of the local-to-zero rate on
the tuning parameter.
\end{remark}

Next, we analyze estimation consistency of the adaptive
LASSO estimator for the parameter $\beta_T$.

\begin{theorem}[Parameter estimation] \label{thm:param_consist-unif}
Let $\{y_t\}_{t \in \mZ}$ and $\{x_t\}_{t \in \mZ}$ be generated
by~\eqref{eq:y}~and~\eqref{eq:x} with $k=1$, let $\{w_t\}_{t \in \mZ}$ satisfy
Assumption~\ref{ass:w1}, and assume $T^{-2}\lambda_T \to 0$.

\begin{enumerate}

\item \label{enum:param_consist-unif} It holds that $\betaAL - \beta_T = o_\mP(1)$. 
		
\item \label{enum:rate_conserv-unif} If $\lambda_T \to \lambda_0$, $0 \leq
\lambda_0 < \infty$, then $T(\betaAL - \beta_T) = O_\mP(1)$.
		
\item \label{enum:rate_consist-unif} If $\lambda_T \to \infty$, then
$\lambda_T^{-1/2}T(\betaAL - \beta_T) = O_\mP(1)$.

\end{enumerate}

\end{theorem}

Theorem~\ref{thm:param_consist-unif}(a) shows that if $T^{-2}\lambda_T
\to 0$, the adaptive LASSO estimator is not only consistent (cf.\
Proposition~\ref{prop:param_consist-fixed}), but also uniformly
consistent. Parts (b) and (c) reveal that the uniform convergence rate
depends on the tuning regime. Under conservative tuning, the estimator
is rate-$T$ consistent, whereas under consistent tuning, it is only
rate-$T\lambda_T^{-1/2}$ consistent.

We now derive the limiting distribution of the adaptive LASSO
estimator under arbitrary sequences of $\beta_T$.

\begin{theorem}[Limiting distribution] \label{thm:ls_dist-unif}

Let $\{y_t\}_{t \in \mZ}$ and $\{x_t\}_{t \in \mZ}$ be generated
by~\eqref{eq:y}~and~\eqref{eq:x} with $k=1$, let $\{w_t\}_{t \in \mZ}$ satisfy
Assumption~\ref{ass:w1}, and assume $T^{-2}\lambda_T \to 0$.
	
\begin{enumerate}

\item \label{enum:ls_dist_conserv-unif} If $\lambda_T \to \lambda_0$, $0 \leq \lambda_0
< \infty$, and $T\beta_T \to \beta_0 \in \mRquer$,
then
\begin{align*}
T(\betaAL - \beta_T) \Rightarrow & \ind\left\{(\zeta_{vv}^c)^{1/2}\left|\mathcal{Z}^c +
\beta_0\right| > \sqrt{\frac{\lambda_0}{2}}\right\} 
\left(\mathcal{Z}^c - \frac{\lambda_0}{2\zeta_{vv}^c}(\mathcal{Z}^c + \beta_0)^{-1}\right) \\ 
& - \ind\left\{(\zeta_{vv}^c)^{1/2}\left|\mathcal{Z}^c +
\beta_0\right| \leq \sqrt{\frac{\lambda_0}{2}}\right\}\beta_0.
\end{align*}
		
\item \label{enum:ls_dist_consist-unif} If $\lambda_T \to \infty$ and
$\lambda_T^{-1/2}T\beta_T \to \tilde\beta_0 \in \mRquer$, then $\lambda_T^{-1/2}T(\betaAL - 
\beta_T) \Rightarrow$ 
\begin{align*}
\begin{cases}
-\ind\left\{(\zeta_{vv}^c)^{1/2} > a_0\right\}(2\tilde\beta_0\zeta_{vv}^c)^{-1}
-\ind\left\{(\zeta_{vv}^c)^{1/2} \leq a_0\right\}\tilde\beta_0 & 
\text{if } 0 < |\tilde\beta_0| < \infty\\
0 & \text{ otherwise},
\end{cases}
\end{align*}
where $a_0 \coloneqq 1/(\sqrt{2}|\tilde\beta_0|)$.

\end{enumerate}
	
\end{theorem}

From Theorem~\ref{thm:ls_dist-unif}\ref{enum:ls_dist_conserv-unif}, we
learn that under conservative tuning, if $\beta_T$ is bounded away
from zero or converges to zero at rate slower than $T^{-1}$, i.e.,
$|\beta_0| = \infty$, the limiting distribution of $T(\betaAL -
\beta_T)$ is $\mathcal{Z}^c$ and thus coincides with the limiting
distribution of OLS. If $\beta_T \equiv 0$ or $\beta_T$ converges to
zero at rate $T^{-1}$ or faster, i.e., $\beta_0 \in \mR$, the limiting
distribution consists of an atomic as well as an absolutely continuous
part.

Part~\ref{enum:ls_dist_consist-unif} of the above theorem shows that
under consistent tuning and using the correct scaling factor, the
limit of $\lambda_T^{-1/2}T(\betaAL - \beta_T)$ collapses to zero if
$\beta_T \equiv 0$ or $\beta_T$ converges to zero faster than
$T\lambda_T^{-1/2}$, i.e. $\tilde\beta_0 = 0$, or if $\beta_T$ is
bounded away from zero or converges to zero slower than
$T\lambda_T^{-1/2}$, i.e., $|\tilde\beta_0| = \infty$. However, if
$\beta_T$ converges to zero exactly at rate $T^{-1}\lambda_T^{1/2}$,
i.e., $0 < |\tilde\beta_0| < \infty$, the limit is random and contains
an atomic as well as an absolutely continuous part. Interestingly, all
remaining randomness originates from the regressor $x_t$, but not from
the errors $u_t$. The dependence on $u_t$ disappears because the
influence of the rate-$T$ consistent OLS estimator vanishes
asymptotically if $\hat \beta - \beta_T$ is scaled by
$\lambda_T^{-1/2}T$ rather than $T$, whereas the dependence on $x_t$
appears in the limit through $\tilde\lambda_T$.\footnote{For more
details, we refer to Lemma~\ref{lem:fs_results}(c) in
Appendix~\ref{app:prep}.}
		
\begin{remark} \label{rem:ls_dist_consist_rateT-unif}
Theorem~\ref{thm:ls_dist-unif}(b) shows that the uniform convergence
rate for the adaptive Lasso estimator under consistent tuning is,
indeed, $T^{-1}\lambda_T^{1/2}$ and that scaling the estimation error
with the larger factor $T$ will result in a stochastically unbounded
sequence if $\beta_T$ converges to zero at rate
$T^{-1}\lambda_T^{1/2}$. For completeness, we also list the limiting
distribution of $T(\betaAL - \beta_T)$ for arbitrary sequences of
$\beta_T$: If $\lambda_T \to \infty$, such that $T^{-2}\lambda_T \to
0$, and $\lambda_T^{-1/2}T\beta_T \to \tilde\beta_0 \in \mRquer$, then
\begin{align*}
T(\betaAL - \beta_T) \Rightarrow \begin{cases}
-\beta_0 & \text{ if } \tilde\beta_0 = 0 \\
-\sign(\tilde\beta_0)\infty & \text{ if } 0 < |\tilde\beta_0| < \infty \\
\mathcal{Z}^c - 0.5(\zeta_{vv}^c\bar\beta_0)^{-1} & \text{ if } |\tilde\beta_0| = \infty,
\end{cases}
\end{align*}
where $T\beta_T \to \beta_0 \in \mRquer$ and $\lambda_T^{-1}T\beta_T
\to \bar\beta_0 \in \mRquer$. Hence, $T(\betaAL - \beta_T)$ collapses
to pointmass at $-\beta_0$ whenever $\beta_T \equiv 0$ or $\beta_T$
converges to zero at rate $T^{-1}$ or faster, i.e., $\beta_0 \in \mR$
and $\tilde\beta_0 = 0$. The limiting distribution of $T(\betaAL -
\beta_T)$ is random if $\beta_T$ is bounded away from zero or
converges to zero at rate $T^{-1}\lambda_T$ or slower, i.e.,
$\bar\beta_0 \neq 0$ and $|\tilde\beta_0| = \infty$. In this case, the
limiting distribution coincides with the one of OLS if $\beta_T$
converges to zero slower than $T^{-1}\lambda_T$, i.e., $|\bar\beta_0|
= |\tilde\beta_0| = \infty$. However, if $|\bar\beta_0| < \infty$, the
limiting distribution of the adaptive LASSO estimator deviates from
the one of OLS by a random shift that is inversely proportional to and
has the opposite sign of $\bar\beta_0$, analogously to what we have
seen under fixed-parameter asymptotics in
Proposition~\ref{prop:ls_dist-fixed}(b1). Again, the random shift is
not detected in \citet{LeeEtAl22}, and, in contrast to the second
order bias term in the limiting distribution of the OLS estimator, it
does not vanish if the regressor is exogeneous. In all other cases,
the total mass of $T(\betaAL - \beta_T)$ escapes to $-\infty$ or
$\infty$.
\end{remark}

\begin{remark} \label{rem:lee-ls_dist}
Remark~\ref{rem:ls_dist_consist_rateT-unif} illustrates that in the
consistently tuned case, while the adaptive Lasso estimator can detect
local-to-zero rates of any order greater than $T^{-1}\lambda_T^{1/2}$,
in order to obtain the same limiting distribution as OLS, the true
coefficient must be of even larger order of magnitude, i.e., greater
than $T^{-1}\lambda_T$. In the setting of \cite{LeeEtAl22} with
$\beta_T = \beta T^{-\delta}$ for $\delta \in (0,1)$, $\beta \neq 0$,
and $T^{-(1-\delta)}\lambda_T \to 0$, it automatically holds that
$\lambda_T^{-1}T\beta_T \to \infty$.
\end{remark}

\subsection{The Multivariate Case} \label{subsec:multi}

We now turn to the multivariate regressor case and investigate the
asymptotic properties of the adaptive LASSO estimator within a
moving-parameter framework. Since detailed results for the
fixed-parameter framework have already been presented and discussed in
the univariate case, our focus here remains on the more general
setting, where the true coefficients are allowed to vary with sample
size, see also the discussion at the beginning of
Section~\ref{sec:asymptotics}. With respect to notation, please note
that the subscript $j$ continues to denote the $j$-th element of the
vector to which it is attached, e.g., $\betaALj$ denotes the $j$-th
component of $\betaAL$.

We start by deriving model selection probabilities for the adaptive
LASSO estimator under conservative as well as consistent tuning for
certain relevant sequences of $\beta_T$.

\begin{theorem}[Model selection] \label{thm:ms-multi}

Let $\{y_t\}_{t \in \mZ}$ and $\{x_t\}_{t \in \mZ}$ be generated
by~\eqref{eq:y}~and~\eqref{eq:x}, and let $\{w_t\}_{t \in \mZ}$ satisfy
Assumption~\ref{ass:w1}.

\begin{enumerate}

\item \label{enum:ms_conserv-multi} If $\lambda_T \to \lambda_0$, $0
\leq \lambda_0 < \infty$, and $T\beta_T \to \beta_0 \in \mRquer^k$, then
$$
\mP\left(\betaALj = 0\right) \to  0,
$$
if $|\tilde\beta_{0,j}| = \infty$.

\item \label{enum:ms_consist-multi} If $\lambda_T \to \infty$ and
$\lambda_T^{-1/2}T\beta_T \to \tilde\beta_0 \in \mRquer^k$ then
$$
\mP\left(\betaALj = 0\right) \to \begin{cases}
1 & \text{ if } \tilde\beta_{0,j} = 0 \\
0 & \text{ if } |\tilde\beta_{0,j}| = \infty.
\end{cases}
$$

\end{enumerate}

\end{theorem}

Before we discuss the results in Theorem~\ref{thm:ms-multi} in detail,
we extend the statement in
Theorem~\ref{thm:ms-multi}\ref{enum:ms_conserv-multi} in the following
remark to obtain a more comprehensive picture for the model selection
properties in the conservatively tuned case.

\begin{remark} \label{rem:ms_conserv-multi}
We point out two additional special cases for the model selection
properties in the conservatively tuned case. Let $\lambda_T \to
\lambda_0$, $0 \leq \lambda_0 < \infty$. Then:

\begin{enumerate}

\item \label{enum:ms_conserv_ge0-multi} For $T\beta_T \to \beta_0 \in
\mRquer^k$ with $\beta_{0,j} = 0$ we have
$$
\liminf_{T \rightarrow \infty} \mP\left(\betaALj = 0\right) > 0.
$$

\item  \label{enum:ms_conserv_le1-multi} For $\beta_T \equiv \beta \in
\mR^k$, $\cA \coloneqq \{j: \beta_j \neq 0\}$, and $\hat\cA \coloneqq
\{j: \betaALj \neq 0\}$ we have
$$
\limsup_{T \rightarrow \infty} \mP\left(\hat\cA = \cA \right) < 1.
$$

\end{enumerate}

\end{remark}

In line with the results from the univariate case,
Theorem~\ref{thm:ms-multi}\ref{enum:ms_conserv-multi} shows that under
conservative tuning, the estimator can detect coefficients with
local-to-zero rates of order greater than $T^{-1}$ as non-zero with
asymptotic probability equal to one. Importantly, in the multivariate
case, this property depends solely on the rate of the coefficient
under consideration and is unaffected by the behavior of the other
components of $\beta_T$. Moreover, coefficients converging to zero
with rate faster than $T^{-1}$ will be set to zero with positive
asymptotic probability as can be seen in
Remark~\ref{rem:ms_conserv-multi}\ref{enum:ms_conserv_ge0-multi}. The
smallest detectable local-to-zero rate under conservative tuning
therefore remains $T^{-1}$.
Remark~\ref{rem:ms_conserv-multi}\ref{enum:ms_conserv_le1-multi}
illustrates that this tuning regime is indeed conservative, also in
the multivariate case.

For the consistently tuned case, a meaningful interpretation again
requires the basic condition $T^{-2}\lambda_T \to 0$, which we assume
to hold throughout this section. As in the univariate case,
Theorem~\ref{thm:ms-multi}\ref{enum:ms_consist-multi} then reveals
that the estimator detects coefficients with local-to-zero rates of
order greater than $T^{-1}\lambda_T^{1/2}$ as non-zero with asymptotic
probability equal to one, while coefficients converging to zero with
rate faster than $T^{-1}\lambda_T^{1/2}$ will always be set to zero
with asymptotic probability equal to one. As under conservative
tuning, these properties only depend on the rate of the coefficient
under consideration and are not affected by the behavior of the
remaining components of $\beta_T$.  Hence, in the multivariate setting
the smallest detectable local-to-zero rate under consistent tuning
continues to be $T^{-1}\lambda_T^{1/2}$.

We turn to parameter estimation consistency in the following theorem.

\begin{theorem}[Parameter estimation] \label{thm:param_consist-multi}

Let $\{y_t\}_{t \in \mZ}$ and $\{x_t\}_{t \in \mZ}$ be generated
by~\eqref{eq:y}~and~\eqref{eq:x}, let $\{w_t\}_{t \in \mZ}$ satisfy
Assumption~\ref{ass:w1}, and assume $T^{-2}\lambda_T \to 0$.

\begin{enumerate}

\item \label{enum:param_consist-multi} It holds that $\betaAL - \beta_T = o_\mP(1)$.

\item \label{enum:rate_conserv-multi} If $\lambda_T \to \lambda_0$, $0
\leq \lambda_0< \infty$, then $T(\betaAL - \beta_T) = O_\mP(1)$.

\item \label{enum:rate_consist-multi}  If $\lambda_T \to \infty$, then  $\lambda_T^{-1/2}T(\betaAL -
\beta_T) = O_\mP(1)$.

\end{enumerate}

\end{theorem}

Theorem~\ref{thm:param_consist-multi} shows that if $T^{-2}\lambda_T
\to 0$, the adaptive LASSO estimator is uniformly consistent also in
the multivariate case and its uniform convergence rate depends on the
tuning regime. Under conservative tuning, the estimator is rate-$T$
consistent, whereas under consistent tuning, the rate decreases to
$T\lambda_T^{-1/2}$, just as in the univariate case.

We now derive the limiting distribution of the adaptive LASSO
estimator under arbitrary sequences of $\beta_T$.

\begin{theorem}[Limiting distribution] \label{thm:ls_dist-multi}

Let $\{y_t\}_{t \in \mZ}$ and $\{x_t\}_{t \in \mZ}$ be generated
by~\eqref{eq:y}~and~\eqref{eq:x}, let $\{w_t\}_{t \in \mZ}$ satisfy
Assumption~\ref{ass:w1}, and assume $T^{-2}\lambda_T \to 0$.

\begin{enumerate}

\item \label{enum:ls_dist_conserv-multi} If $\lambda_T \to \lambda_0$,
$0 \leq \lambda_0 < \infty$ and $T\beta_T \to \beta_0 \in \mRquer^k$,
then $T(\betaAL - \beta_T) \Rightarrow \argmin_{z \in \mR^k} V_{\beta_0}^c(z)$,
where 
$$
V_{\beta_0}^c(z) \coloneqq z'\zeta_{vv}^cz - 2z'\left(\int_0^1 J_v^c(r)dB_u(r) 
+ \Delta_{vu}\right) + \lambda_0 \sum_{j=1}^k A_j(z_j,\beta_{0,j})
$$
and
$$
A_j(z_j,\beta_{0,j}) \coloneqq \begin{cases}
0 & \text{ if } |\beta_{0,j}| = \infty \text{ or } z_j = 0  \\
\frac{|z_j|}{|\mathcal{Z}^c_j|} & \text{ if } \beta_{0,j} = 0 \text{ and } z_j\neq 0 \\
\frac{|\beta_{0,j} + z_j| - |\beta_{0,j}|}{|\beta_{0,j} + \mathcal{Z}^c_j|} & \text{ otherwise.}
\end{cases}
$$

\item \label{enum:ls_dist_consist-multi} If $\lambda_T \to \infty$ and
$\lambda_T^{-1/2}T\beta_T \to \tilde\beta_0 \in \mRquer^k$, then
$\lambda_T^{-1/2}T(\betaAL - \beta_T) \Rightarrow \argmin_{z \in \mR^k}
\tilde{V}_{\tilde\beta_0}^c(z)$, where 
\begin{align*}
\tilde{V}_{\tilde\beta_0}^c(z) \coloneqq z'\zeta_{vv}^cz 
+ \sum_{j=1}^k \tilde{A}_j(z_j,\tilde\beta_{0,j})
\end{align*}
and
$$
\tilde{A}_j(z_j,\tilde\beta_{0,j}) \coloneqq \begin{cases}
0 & \text{ if } |\tilde\beta_{0,j}| = \infty \text{ or } z_j = 0 \\
\infty & \text{ if }\tilde\beta_{0,j} = 0 \text{ and } z_j \neq 0 \\
\frac{|z_j + \tilde\beta_{0,j}|}{|\tilde\beta_{0,j}|} - 1 & \text{ otherwise.}
\end{cases}
$$

\item \label{enum:ls_dist_consist_rateT-multi} If $\lambda_T \to
\infty$, $T\beta_T \to \beta_0 \in \mRquer^k$, and
$\lambda_T^{-1}T\beta_T \to \bar\beta_0 \in \mRquer^k$, then
$T(\betaAL - \beta_T) \Rightarrow \argmin_{z \in \mR^k} \bar{V}_{\bar\beta_0}^c(z)$,
where 
$$
\bar{V}_{\bar\beta_0}^c(z) \coloneqq z'\zeta_{vv}^cz 
- 2z'\left(\int_0^1 J_v^c(r)dB_u(r) + \Delta_{vu}\right) 
+ \sum_{j=1}^k \bar{A}_j(z_j,\beta_{0,j},\bar\beta_{0,j})
$$
and
$$
\bar{A}_j(z_j,\beta_{0,j},\bar{\beta}_{0,j}) \coloneqq \begin{cases}
0 & \text{ if } |\bar\beta_{0,j}| = \infty \text{ or } z_j = 0 \\
\infty & \text{ if } \bar\beta_{0,j} = \beta_{0,j} = 0  \text{ and } z_j \neq 0 \\
\sign(z_j + 2\beta_{0,j})\sign(z_j)\infty & \text{ if } \bar\beta_{0,j} = 0, \, 0 < |\beta_{0,j}| < \infty, 
\text{ and } z_j \neq 0 \\
\sign(z_j)\sign(\beta_{0,j})\infty & \text{ if } \bar\beta_{0,j} = 0, \, |\beta_{0,j}| = \infty, 
\text{ and } z_j \neq 0 \\
\frac{-\sign(\bar\beta_{0,j})z_j}{|\bar\beta_{0,j}|} & \text{ otherwise.}
\end{cases}
$$
\end{enumerate}

\end{theorem}

The limiting distributions presented in
Theorem~\ref{thm:ls_dist-multi} are defined implicitly. While we
cannot explicitly minimize $V_{\beta_0}^c(z)$,
$\tilde{V}_{\tilde\beta_0}^c(z)$, and $\bar{V}_{\bar\beta_0}^c(z)$ for
fixed $\beta_0$, $\tilde\beta_0$, and $\bar\beta_0$ in general, there
are a number of special cases worth pointing out. First, in line with
Theorem~\ref{thm:ls_dist-unif}\ref{enum:ls_dist_conserv-unif},
Theorem~\ref{thm:ls_dist-multi}\ref{enum:ls_dist_conserv-multi} shows
that under conservative tuning with either $\lambda_0 = 0$ or
$|\beta_{0,j}| = \infty$ for all $j = 1,\ldots,k$, the limiting
distribution of $T(\betaAL - \beta_T)$ is $\mathcal{Z}^c$ and thus
coincides with the limiting distribution of OLS. Second, in line with
Theorem~\ref{thm:ls_dist-unif}\ref{enum:ls_dist_consist-unif},
part~\ref{enum:ls_dist_consist-multi} shows that under consistent
tuning, the limit of $\lambda_T^{-1/2}T(\betaAL - \beta_T)$ collapses
to zero whenever $\tilde\beta_{0,j} = 0$ or $|\tilde\beta_{0,j}| = \infty$
for all $j = 1,\ldots,k$. Moreover, whenever the limit is stochastic,
all randomness originates from the regressors $x_t$, but not from the
errors $u_t$ (see also the discussion in the univariate case).
Finally, part~\ref{enum:ls_dist_consist_rateT-multi} reveals that
under consistent tuning, the limiting distribution of $T(\betaAL -
\beta_T)$ coincides with the limiting distribution of OLS if
$|\bar{\beta}_{0,j}| = \infty$ for all $j = 1,\ldots,k$. Conversely,
the limit of $T(\betaAL - \beta_T)$ collapses to zero whenever
$|\bar\beta_{0,j}| = 0$ for all $j = 1,\ldots,k$. Both results are
consistent with Remark~\ref{rem:ls_dist_consist_rateT-unif}.

\begin{remark} \label{rem:sequences2}
A similar comment as in Remark~\ref{rem:sequences} also applies to
Theorems~\ref{thm:ms-multi}--\ref{thm:ls_dist-multi}.
\end{remark}
	
The following proposition offers additional insights into the limiting
distribution of $\lambda_T^{-1/2}T(\betaAL - \beta_T)$ under
consistent tuning, as derived in
Theorem~\ref{thm:ls_dist-multi}\ref{enum:ls_dist_consist-multi}.

\begin{proposition}\label{prop:M-argmin}

For a fixed $\omega$ in the sample space of the underlying probability
space, the point $m = m(\omega) \in \mR^k$ is a minimizer of
$\tilde{V}_{\tilde\beta_0}^c(z)(\omega)$ if and only if
\begin{align*}
\begin{cases} m_j = 0 & \text{ if } \tilde\beta_{0,j} = 0 \\
(\zeta_{vv}^c(\omega)m)_j = 0 & \text{ if } |\tilde\beta_{0,j}| = \infty \\
(\zeta_{vv}^c(\omega)m)_j = - \frac{\sign(m_j+\tilde\beta_{0,j})}{2|\tilde\beta_{0,j}|} & 
\text{ if } 0 < |\tilde\beta_{0,j}| < \infty \text{ and } m_j\neq - \tilde\beta_{0,j} \\
|(\zeta_{vv}^c(\omega)m)_j| \leq \frac{1}{2|\tilde\beta_{0,j}|} & 
\text{ if } 0 < |\tilde\beta_{0,j}| < \infty \text{ and } m_j = - \tilde\beta_{0,j},
\end{cases}
\end{align*}
where $\zeta_{vv}^c$ is the same random matrix as in the definition of
$\tilde{V}_{\tilde\beta_0}^c(z)$ in
Theorem~\ref{thm:ls_dist-multi}\ref{enum:ls_dist_consist-multi}.

\end{proposition}

Building on Proposition~\ref{prop:M-argmin}, the following theorem
shows that the set of minimizers of $\tilde{V}_{\tilde\beta_0}^c(z)$ taken over over all $\tilde\beta_0 \in \mRquer^k$ is contained in a random set that does not depend on $\tilde\beta_0$.

\begin{theorem} \label{thm:setM}
Define the random set
$$
\Mc \coloneqq \left\{m\in\mR^k:m_j(\zeta_{vv}^cm)_j\leq \frac{1}{2},\,j=1,\ldots,k\right\},
$$
where $\zeta_{vv}^c$ is the same as in the definition of
$\tilde{V}_{\tilde\beta_0}^c(z)$ in
Theorem~\ref{thm:ls_dist-multi}\ref{enum:ls_dist_consist-multi}. Then,
$$
\underset{\tilde\beta_0 \in \mRquer^k}{\bigcup} \argmin_{z \in \mR^k}\tilde V_{\tilde\beta_0}^c(z)
\subseteq \Mc,
$$
where the set inclusion holds surely, i.e., for all $\omega$ in the
sample space of the underlying probability space.

\end{theorem}

\begin{remark} \label{rem:skorohod}
For later use, we will assume that the random matrix $\zeta_{vv}^c$ in
Theorem~\ref{thm:ls_dist-multi}\ref{enum:ls_dist_consist-multi}
satisfies $\lim_{T \to \infty} T^{-2} \sum_{t=1}^Tx_t(\omega)x_t'(\omega) =
\zeta_{vv}^c(\omega)$ for all $\omega$. This can be achieved using
Skorohod's representation theorem.
\end{remark}

Theorem~\ref{thm:setM} together with Remark~\ref{rem:skorohod} shows
that the union of limits of $\lambda_T^{-1/2} T(\betaAL - \beta_T)$
over all possible parameter sequences is contained in the set $\Mc$,
which is a compact set for each realization of $\zeta_{vv}^c$. This
observation allows us to construct uniformly valid confidence regions
centered at the adaptive LASSO estimator, as developed in the
following subsection. Before doing so, let us examine the set $\Mc$ in
more detail.

Clearly, all randomness in $\Mc$ stems from the regressors $x_t$,
i.e., no randomness arises from the regression errors $u_t$. For
expositional convenience, we focus on the pure unit root case ($c=0$):
There, $\zeta_{vv}^c$ has expectation $0.5\,\Omega_{vv}$, where
$\Omega_{vv}$, given by the $k \times k$ bottom-right block of
$\Omega$, the long-run covariance matrix of $\{v_t\}_{t \in
\mathbb{Z}}$. Hence, on average, $\Mc$ is given by $\{m\in \mR^k : m_j
(\Omega_{vv} m)_j \leq 1,\,j=1,\ldots,k \}$. Thus, on average, the set
$\Mc$ becomes smaller as the variability of $\{v_t\}_{t \in
\mathbb{Z}}$ increases. In the univariate regressor case,
$\Omega_{vv}$ reduces to the long-run variance of $v_t$. Normalizing
this variance to one implies that, on average, $\Mc$ coincides with
the interval $[-1,1]$. Consequently, in one dimension and on average,
we recover the same interval as \cite{PoetscherSchneider09} and
\cite{AmannSchneider23}. Note, however, that the corresponding sets in
these two papers are non-random, as the regressors are treated as
deterministic.

\subsection{A Universal Confidence Region Under Consistent Tuning}\label{subsec:ci}

We now use the observation from Theorem~\ref{thm:setM} to construct a
confidence region that has asymptotic coverage probability equal to
one. To this end, we define a ``slightly larger'' finite-sample
analogue $\Mhateps \coloneqq \{m\in\mR^k :
m_j((T^{-2}\sum_{t=1}^Tx_tx_t')m)_j \leq \frac{1}{2} +
\eps,\,j=1,\ldots,k\}$ of $\Mc$, where $\eps > 0$ but arbitrarily
small. The following theorem shows that this set can be used to
construct a confidence region based on the adaptive LASSO estimator
that asymptotically holds any prescribed coverage level.

\begin{theorem}[Confidence regions] \label{thm:confM}
Let $\{y_t\}_{t \in \mZ}$ and $\{x_t\}_{t \in \mZ}$ be generated
by~\eqref{eq:y}~and~\eqref{eq:x}, let $\{w_t\}_{t \in \mZ}$ satisfy
Assumption~\ref{ass:w1}, and let $T^{-2}\lambda_T \to 0$ and $\lambda_T
\to \infty$. Then
$$
\lim_{T \to \infty} \inf_{\beta \in \mR^k} \mP_\beta\left(\beta \in
\betaAL - T^{-1}\lambda_T^{1/2}\Mhateps\right) = 1
$$
for any $\eps > 0$.
\end{theorem}

Theorem~\ref{thm:confM} delivers valid confidence regions, as the
coverage probability holds \emph{uniformly} over the parameter space.
Technically, this is achieved by taking the infimum over the parameter
space \emph{prior} to letting $T$ tend to infinity. The key underlying
idea is that $\beta \in \{\betaAL - \lambda_T^{1/2}T^{-1}m : m\in
\Mhateps\}$ holds if and only if $\lambda_T^{-1/2}T(\betaAL - \beta)
\in \Mhateps$. The latter event can then be approximated by
$\argmin_{z \in \mR^k}\tilde V_{\tilde\beta_0}^c(z) \in \Mc$ for some
$\tilde\beta_0 \in \mRquer^k$, and this inclusion surely holds for any
$\tilde\beta_0$ by Theorem~\ref{thm:setM}.
%

In practice, we propose to construct the confidence region for $\beta$
based on $\Mhatzero$. Since no closed-form solution is available,
$\Mhatzero$ can be computed numerically following the the description
below, and the confidence region is then given by $\{\betaAL -
\lambda_T^{1/2}T^{-1}m : m\in \Mhatzero\}$.

If confidence intervals for individual components $\beta_j$ are of
interest, these can be obtained by 
$$
[\betaALj - \lambda_T^{1/2}T^{-1}\overline{m}_j, \betaALj -
\lambda_T^{1/2}T^{-1}\underline{m}_j],
$$ 
where $\overline{m}_j\coloneqq \max\{m_j : m\in \Mhatzero\}$ and
$\underline{m}_j\coloneqq \min\{m_j : m\in
\Mhatzero\}=-\overline{m}_j$. The quantity $\overline{m}_j$ can be
computed numerically using sequential quadratic programming without
explicitly constructing the set $\Mhatzero$. Specifically, we solve
the constrained maximization problem defining $\overline{m}_j$
directly. To reduce the risk of convergence to local optima, the
algorithm can be initialized from multiple random starting values,
retaining the largest value of $m_j$ obtained across these runs.

The coordinate-wise intervals can also be used to construct the set
$\Mhatzero$. Since $\Mhatzero$ is contained in the Cartesian product
of these intervals, i.e., $\Mhatzero \subseteq \mathcal{B} \coloneqq
\prod_{j=1}^{k}[\underline{m}_j, \overline{m}_j]$, one can either
construct a grid within $\mathcal{B}$ or, to speed up computation for
large $k$, randomly sample vectors from $\mathcal{B}$, retaining only
those that satisfy the constraints defining
$\Mhatzero$.\footnote{MATLAB code for computing $\overline{m}_j$ and
$\underline{m}_j$ for all $j=1,\ldots,k$, as well as $\Mhatzero$ for
general $k \times k$ symmetric positive definite matrices, is
available on the first author's personal website.}

It may appear curious to consider confidence regions whose asymptotic
coverage probability equals one. To explain this, recall that the
consistently tuned adaptive LASSO estimator exhibits the behavior that
only randomness stemming from the regressors $x_t$ can persist
asymptotically. This occurs because the scaling induced by the uniform
convergence rate is not sufficiently large for stochastic variation
from the error terms $u_t$ to survive in the limit. When the
regressors are treated as non-random, this even manifests in entirely
non-random limits which are contained in a compact set.
\cite{AmannSchneider23} show that this non-random set can be utilized
in a similar way to construct confidence regions with uniform
asymptotic coverage probability equal to one, but that any slightly
smaller region has asymptotic coverage probability equal to zero.

When the regressors $x_t$ are random, a slightly different picture
arises. All possible limits of $\lambda_T^{-1/2}T(\betaAL - \beta)$
are contained in a set that is compact for a fixed realization of the
limiting regressor matrix $\zeta_{vv}^c$. As a result, confidence
regions can be constructed that still achieve uniform asymptotic
coverage equal to one. However, it cannot be shown anymore that this
probability will drop to zero when the regions are made smaller, an
effect of $x_t$ being stochastic.

It may be possible to exploit the remaining randomness to construct
confidence regions with asymptotic coverage strictly less than one,
but we leave a formal investigation of this question for future
research. Nevertheless, the regions proposed here possess a key
advantage: they do not rely on the asymptotic distribution of either
the adaptive LASSO estimator or the rescaled regressors. As a
consequence, their construction avoids the need for any knowledge or
estimation of local-to-unity and long-run covariance parameters, as
well as for accounting for second-order bias terms in the limiting
distribution of the adaptive LASSO estimator. This universality
distinguishes our approach from existing methods and, to the best of
our knowledge, represents the first construction of LASSO-based
uniformly valid confidence regions in regressions with unit root or
local-to-unity regressors.

\section{Simulation Results}
\label{sec:simulation}

This section presents simulation results. Section~\ref{subsec:sim_fs}
analyzes the approximation quality of the asymptotic results to the
finite-sample distribution of the adaptive LASSO estimator, whereas
Section~\ref{subsec:sim_ci} focuses on the empirical coverage
probabilities of the uniform confidence regions.

\subsection{Finite-Sample Distributions}\label{subsec:sim_fs}

We investigate the approximation quality of our theoretical results to
the finite-sample distribution of the adaptive LASSO estimator under
both conservative and consistent tuning for various sequences
$\beta_T$ and different sample sizes. We also compare the
finite-sample distribution of the adaptive LASSO to that of OLS to
analyze how much it deviates from what is suggested by the oracle
property in empirically relevant scenarios where some coefficients are
small rather than exactly equal to zero.

We generate data according to~\eqref{eq:y} and~\eqref{eq:x} for the
univariate unit root case with
$[u_t,v_t]'\sim\mathcal{N}\left(0,I_2\right)$ i.i.d.~across $t$, and
present results for $\beta_T\in\{0.1\beta, \beta/T^{1/2}, \beta/T,
\lambda_T^{1/2}\beta/T\}$, with $\beta=1$,
$\lambda_T\in\{1,T^{1/4},T^{1/2},T\}$, and
$T\in\{25,50,100,250,1000\}$. All results are based on $10{,}000$
Monte Carlo replications.\footnote{To focus on the main effects, we
omit error serial correlation and regressor endogeneity from the
model. While our empirical findings remain qualitatively similar when
these features are included -- as well as under changes in the
variances of $u_t$ and $v_t$ or deviations from normality -- the
approximation quality of the limiting distributions derived in
Theorem~\ref{thm:ls_dist-unif} and
Remark~\ref{rem:ls_dist_consist_rateT-unif} may be reduced,
particularly in small- to medium-sized samples.} The choices for
$\beta_T$ cover the cases where $\beta_T$ is bounded away from zero,
converges to zero at rate $T^{1/2}$ \citep[the same rate as used in
the simulations in][]{LeeEtAl22}, converges to zero at rate $T^{-1}$
(the cut-off rate under conservative tuning), and converges to zero at
the slower rate $T^{-1}\lambda_T^{1/2}$ (the cut-off rate under
consistent tuning). With respect to the tuning parameter $\lambda_T$,
the choice $\lambda_T\equiv1$ leads to conservative tuning, while the
other three choices lead to consistent tuning. Importantly, in case
$\beta_T=\beta/T^{1/2}$, only $\lambda_T=T^{1/4}$ fulfills the
condition in \cite{LeeEtAl22} that $T^{-1/2}\lambda_T + \lambda_T^{-1}
\to 0$.

Separately for the four choices of $\beta_T$,
Figures~\ref{fig:densities_thm3_1}--\ref{fig:densities_thm3_4} display
the finite-sample distributions of $T(\betaAL - \beta_T)$ (under
conservative tuning) and $\lambda_T^{-1/2}T(\betaAL - \beta_T)$ (under
consistent tuning). The distributions consist of an atomic mass, drawn
at the height corresponding to the relative frequency $p$ of the event
$\betaAL = 0$, and a continuous component (rescaled to integrate to $1
- p$), representing the density of the non-zero
estimates.\footnote{All displayed densities are smoothed by using the
Gaussian kernel.} The figures also display the finite-sample
distribution of the OLS estimator, $T(\hat\beta - \beta_T)$, as well
as the case-specific limiting distribution of the adaptive LASSO
estimator from Theorem~\ref{thm:ls_dist-unif} evaluated at
$\beta_{0,T} \coloneqq T\beta_T$ (under conservative tuning) and
$\tilde\beta_{0,T} \coloneqq \lambda_T^{-1/2}T\beta_T$ (under
consistent tuning).\footnote{Densities of limiting distributions are
obtained by simulation, where Brownian motions are approximated by
normalized sums of $10{,}000$ i.i.d.~standard normal random variables
and stochastic integrals are approximated accordingly.} Replacing the
limiting parameters $\beta_0$ and $\tilde \beta_0$ with their
finite-sample counterparts allows us to account for the size of
$\beta_T$ relative to the sample size when evaluating the
approximation quality of the limiting distributions derived in
Theorem~\ref{thm:ls_dist-unif} for the finite-sample distributions.

Figure~\ref{fig:densities_thm3_1} presents the results for $\beta_T
\equiv 0.1\beta$. In general, the adaptive LASSO estimator identifies
the true coefficient as non-zero with probability approaching one as
sample size $T$ increases. However, by construction, the empirical
probability of incorrectly setting the non-zero coefficient to zero
increases with the order at which $\lambda_T$ diverges. Under
conservative tuning, the finite-sample distribution of the estimator
approaches that of the OLS estimator, with the two distributions
becoming virtually indistinguishable already for $T = 100$. Notably,
the limiting distribution derived in Theorem~\ref{thm:ls_dist-unif}(a)
evaluated at $\beta_{0,T}$ already provides a good approximation to
the finite-sample distribution of the adaptive LASSO estimator for
small $T$, e.g., $T = 25$. In the context of
Lemma~\ref{lem:fs_results} in Appendix~\ref{app:prep}, this indicates
that $\mathcal{Z}^c_T$ and $\zeta_{vv,T}$ converge quickly to their
asymptotic counterparts, such that the finite-sample distribution of
the procedure is effectively governed by $\beta_{0,T}$. Under
consistent tuning, the figure shows that the scaling factor implied by
the uniform rate causes all mass of the distribution of the adaptive
LASSO estimator to collapse at zero as $T$ increases. Moreover, the
limiting distribution derived in
Theorem~\ref{thm:ls_dist-unif}\ref{enum:ls_dist_consist-unif},
evaluated at $\tilde\beta_{0,T}$, approximates the finite-sample
distribution more accurately the larger the order of $\lambda_T$. For
$\lambda_T = T$, the approximation is already quite accurate for small
$T$.
	
We now turn to Figure~\ref{fig:densities_thm3_2}, which presents the
results for $\beta_T=\beta/T^{1/2}$. Under conservative tuning, the
adaptive LASSO estimator still identifies the smaller coefficient as
non-zero with probability approaching one as $T$ increases and its
distribution quickly approaches that of the OLS estimator. Under
consistent tuning, however, its properties now depend on the order of
$\lambda_T$. As before, for $\lambda_T\in\{T^{1/4}, T^{1/2}\}$, the
procedure correctly identifies the true coefficient as non-zero with
probability approaching one and the scaling factor implied by the
uniform rate causes all mass to collapse at zero as $T$ increases. For
$\lambda_T = T$, however, the estimator incorrectly sets the true
coefficient to zero with probability approaching $0.68$. As a result,
even for large $T$, the distribution of the estimator consists of both
an atomic and a continuous part. Nevertheless, the limiting
distribution derived in
Theorem~\ref{thm:ls_dist-unif}\ref{enum:ls_dist_consist-unif},
evaluated at $\tilde\beta_{0,T}$, provides a good approximation to the
finite-sample distribution even for small $T$.
	
Figure~\ref{fig:densities_thm3_3} presents the results for $\beta_T =
\beta/T$, which is the cut-off rate for detection of non-zero
coefficients under conservative tuning. As a result, under
conservative tuning, the adaptive LASSO estimator incorrectly sets the
true coefficient to zero with probability approaching $0.43$ and its
distribution -- which is approximated very well by the limiting
distribution derived in
Theorem~\ref{thm:ls_dist-unif}\ref{enum:ls_dist_conserv-unif},
evaluated at $\beta_{0,T}$, already for small $T$ -- consists of both
an atomic and a continuous part. Under consistent tuning, the
estimator incorrectly sets the true coefficient to zero with
probability approaching one, such that its distribution collapses to
an atomic part at $-\tilde\beta_{0,T}$.
	
Finally, Figure~\ref{fig:densities_thm3_4} presents the results for
$\beta_T = \sqrt{\lambda_T}\beta/T$, which is the cut-off rate for the
detection of non-zero coefficients under consistent
tuning.\footnote{Since $\lambda_T \equiv 1$ implies $\beta_T =
\beta/T$, the results under conservative tuning coincide with those
shown in Figure~\ref{fig:densities_thm3_3}.} As a result, under
consistent tuning, the adaptive LASSO estimator incorrectly sets the
true coefficient to zero with probability approaching $0.68$ and its
distribution consists of both an atomic and a continuous part. The
limiting distribution derived in
Theorem~\ref{thm:ls_dist-unif}\ref{enum:ls_dist_consist-unif},
evaluated at $\tilde\beta_{0,T}$, approximates the finite-sample
distribution of the estimator more accurately the larger the order of
$\lambda_T$. For $\lambda_T = T$, the approximation is already quite
accurate for small $T$.

Now, we analyze the finite-sample distribution of $T(\hat\beta_{\AL} -
\beta_T)$ under consistent tuning, which is presented in
Figures~\ref{fig:densities_rem5_1}--\ref{fig:densities_rem5_4} in
Appendix~\ref{app:simulation} separately for the four choices of
$\beta_T$ alongside the corresponding limits from
Remark~\ref{rem:ls_dist_consist_rateT-unif} and the finite-sample
distribution of the OLS estimator.

Figure~\ref{fig:densities_rem5_1} presents the results for $\beta_T
\equiv 0.1\beta$. As seen previously in
Figure~\ref{fig:densities_thm3_1}, the adaptive LASSO estimator
identifies the true coefficient as non-zero with probability
approaching one as $T$ increases. However, the order of $\lambda_T$
has a substantial impact on the estimator’s distribution. For
$\lambda_T \in \{T^{1/4}, T^{1/2}\}$, its distribution approaches that
of the OLS estimator. In contrast, for $\lambda_T = T$, the
distribution approaches to that of the OLS estimator shifted to the
left. This stochastically bounded shift significantly distorts the
distribution even for small $T$. As a result, when $\lambda_T = T$,
the adaptive LASSO estimator fails to exhibit the oracle property,
despite correctly identifying the true coefficient as non-zero with
probability approaching one.
	
We now turn to Figure~\ref{fig:densities_rem5_2}, which presents the
results for $\beta_T = \beta/T^{1/2}$. For $\lambda_T \in \{T^{1/4},
T^{1/2}\}$, the adaptive LASSO estimator identifies the true
coefficient as non-zero with probability approaching one as $T$
increases (c.f.~Figure~\ref{fig:densities_thm3_2}), but only for
$\lambda_T = T^{1/4}$ does its distribution approach that of OLS. For
$\lambda_T = T^{1/2}$, on the other hand, the distribution of the
procedure approaches the one of OLS shifted to the left and it thus
loses its oracle property. For $\lambda_T = T$, the adaptive LASSO
estimator incorrectly sets the true coefficient to zero with
probability approaching $0.68$ and its distribution consists of both
an atomic part and a continuous part. As $T$ increases, both the
location of the atomic part and the region where the continuous part
of the distribution has mass shift toward $-\infty$. While the
behavior of the estimator in case $\lambda_T = T^{1/4}$ is already
described in \citet{LeeEtAl22}, the cases $\lambda_T \in \{T^{1/2},
T\}$ are not covered by their results, but are in line with our
asymptotic results in Remark~\ref{rem:ls_dist_consist_rateT-unif}.
	
Figure~\ref{fig:densities_rem5_3} presents the results for
$\beta_T=\beta/T$. In this case, the adaptive LASSO estimator
incorrectly sets the true coefficient to zero with probability
approaching one (see also Figure~\ref{fig:densities_thm3_3}), causing
its distribution to collapse to an atomic part at $-\beta$. The
collapse occurs more rapidly the larger the order at which $\lambda_T$
diverges.
	
Finally, Figure~\ref{fig:densities_rem5_4} presents the results for
$\beta_T = \sqrt{\lambda_T}\beta/T$. As this is the cut-off rate for
the detection of non-zero coefficients under consistent tuning, the
adaptive LASSO estimator incorrectly sets the true coefficient to zero
with probability approaching $0.68$ (see also
Figure~\ref{fig:densities_thm3_4}), and its distribution consists of
both an atomic and a continuous part. As $T$ increases, both the
location of the atomic part and the region where the continuous part
of the distribution has mass shift toward $-\infty$.

Overall, the simulation results reveal that the finite-sample
distribution of the adaptive LASSO estimator can deviate substantially
from what is suggested by the oracle property, especially under
consistent tuning. By contrast, the limiting distributions derived in
under moving-parameter asymptotics capture the finite-sample
properties much more accurately. They not only provide reasonable
approximations for the absolutely continuous part of the estimator's
distribution but also convey information on the relative frequency
with which the coefficient is set to zero.

\subsection{Coverage Probabilities}\label{subsec:sim_ci} 

We now study the coverage probabilities of the confidence regions
introduced in Section~\ref{subsec:ci} and benchmark them against
confidence regions based on the oracle property.

Data are generated as in the previous subsection. In the univariate
case, the confidence region from Theorem~\ref{thm:confM} simplifies to
the interval $\betaAL \mp \sqrt{\lambda_T}/\sqrt{2\sum_{t=1}^T
x_t^2}$, which we label the ``Uniform CI''. The oracle-based interval
(``Oracle CI'') is given by $[\betaAL - q_{1-\alpha/2}/T, \betaAL -
q_{\alpha/2}/T]$, where $q_{\alpha}$ denotes the $\alpha$-quantile of
$\mathcal{Z}^c$ defined in Equation~\eqref{eq:OLS} in
Section~\ref{sec:setting}. The quantiles are obtained by simulation as
described in the previous subsection. We report results for
$\alpha=0.05$ and $\alpha=0.01$, corresponding to nominal $95\%$ and
$99\%$ oracle confidence intervals, respectively. In addition, we
consider the interval $\betaAL \mp \sqrt{\lambda_T}/T$, labeled
``asymptotic Uniform CI'', which replaces $T^{-2}\sum_{t=1}^T x_t^2$
in the definition of $\Mhatzero$ by the expectation of its limit
$\zeta_{vv}^c$, which is equal to 1 in this case.

For each interval, we compute the empirical coverage probability for
the true parameter value $\beta \in [-0.6,0.6]$ using an equidistant
grid with step size $0.01$. As the results are symmetric in $\beta$,
we report them only for $|\beta|$. Practitioners following the oracle
property are typically interested in confidence intervals only when
$\betaAL\neq 0$, as they assume $\beta=0$ otherwise. Accordingly, when
$\betaAL = 0$ we collapse the Oracle CI to the singleton $\{0\}$,
which covers the true parameter $\beta$ only if $\beta =
0$.\footnote{Results are qualitatively similar without this
restriction.} The Uniform CI is constructed the same way for all
values of $\betaAL$. As before, all results reported in this
subsection are based on $10{,}000$ Monte Carlo replications.

Figure~\ref{fig:coverage_lambda} reports coverage probabilities for
those choices of tuning parameters $\lambda_T$ from the previous
subsection that lead to consistent tuning, i.e.,
$\lambda_T\in\{T^{1/4},T^{1/2},T\}$. Since the oracle property
requires $T^{-1}\lambda_T+\lambda_T^{-1}\to 0$ (see, e.g.,
Corollary~\ref{cor:summary_fixed}), it does not hold for
$\lambda_T=T$. Consequently, the Oracle CI is asymptotically invalid
in this case, whereas the Uniform CI remains asymptotically valid. The
figure shows that the coverage probabilities of all confidence
intervals are close to one when $\beta=0$. However, for small
deviations of $\beta$ away from zero, the coverage probability of the
Oracle CI drops sharply, often to values below $0.5$, while the
coverage probability of the Uniform CI remains relatively stable. As
$\beta$ moves further away from zero, the coverage probability of the
Oracle CI eventually recovers. Nevertheless, it performs poorly in the
region of primary interest, namely when $\beta$ is small but non-zero.
More generally, we find that the higher the rate at which $\lambda_T$
diverges, or the larger the sample size $T$, the more stable the
coverage probability of the Uniform CI across values of $\beta$. When
$\lambda_T=T$, the asymptotically invalid Oracle CI can exhibit
coverage probabilities close to zero even in large samples, whereas
the Uniform CI remains asymptotically valid and performs well even in
small samples. Finally, the coverage probabilities of the asymptotic
Uniform CI are generally much lower than those of the Uniform CI. This
reflects the point already emphasized in Section~\ref{subsec:ci} that
the confidence sets must be constructed using the realized regressors
$x_t$, rather than quantities based on (limiting) distributional
properties only.

We next repeat the analysis after scaling $\lambda_T$ by a factor of
four in each case, leaving its rate of divergence unchanged.
Figure~\ref{fig:coverage_4lambda} presents the results. Scaling
$\lambda_T$ in this way further improves the coverage probability of
the Uniform CI across all sample sizes and divergence rates, while the
coverage probability of the Oracle CI deteriorates further. This
outcome reflects two opposing effects of increasing $\lambda_T$.
First, more estimates $\betaAL$ are shrunk to zero, which worsens the
performance of the Oracle CI. Second, the width of the Uniform CI
increases, which improves its coverage.

Finally, we relax the assumption of i.i.d.\ standard normal regression
errors and regressor innovations and examine the performance of the
confidence intervals under error serial correlation and regressor
endogeneity. Specifically, we generate $u_t=\rho_1
u_{t-1}+e_t+\rho_2\nu_t$ and $v_t=\nu_t+0.5\nu_{t-1}$, where
$[e_t,\nu_t]'\sim\mathcal{N}(0,(4/9)I_2)$ i.i.d.\ across
$t$.\footnote{Rescaling the variance ensures that $\Omega_{vv}=1$ and
thus simplifies the comparison with the previous results.} The
parameters $\rho_1$ and $\rho_2$ govern the degree of error serial
correlation and regressor endogeneity, respectively.
Figure~\ref{fig:coverage_4lambda_rho6} reports results for
$\rho_1=\rho_2=0.6$ and the scaled tuning parameters.\footnote{When
simulating the quantiles of $\mathcal{Z}^c$ for the oracle intervals,
we use the true long-run covariance parameters to capture the
dependence structure in the data. In applications, these quantities
are unknown and must be estimated. In the presence of local-to-unity
regressors this approach becomes infeasible as $\mathcal{Z}^c$ also
depends on the local-to-unity parameters which are not consistently
estimable. In contrast, the uniform confidence intervals are
unaffected by these issues.} The figure shows that the previous
findings remain intact in the presence of serial correlation and
regressor endogeneity.

To complete the analysis, Tables~\ref{tab:CI_length_iid} and
\ref{tab:CI_length_corr} in Appendix~\ref{app:simulation} report the
lengths of the confidence intervals underlying
Figures~\ref{fig:coverage_lambda}--\ref{fig:coverage_4lambda_rho6}. As
expected, the Uniform CI is typically longer than the Oracle CI, but
the difference is usually moderate and diminishes further as the
sample size increases. In some instances, however, the Uniform CI can
be substantially longer than the Oracle CI, but this occurs either for
unfavorable realizations of $x_t$ or in settings where the Oracle CI
is asymptotically invalid and exhibits coverage probabilities close to
zero. We therefore conclude that the uniform confidence region
proposed in Section~\ref{subsec:ci} constitutes a useful tool for
quantifying uncertainty around adaptive LASSO estimates under
consistent tuning in empirical applications.

\begin{figure}[ht]
\begin{center}
	\caption*{$\lambda_T \equiv 1$}
	\begin{subfigure}{0.2\textwidth}
		\centering
		\caption*{$T = 25$}
		\vspace{-1.5ex}
		\includegraphics[trim={0cm 0cm 0.50cm 0.5cm},width=\textwidth,clip]
		{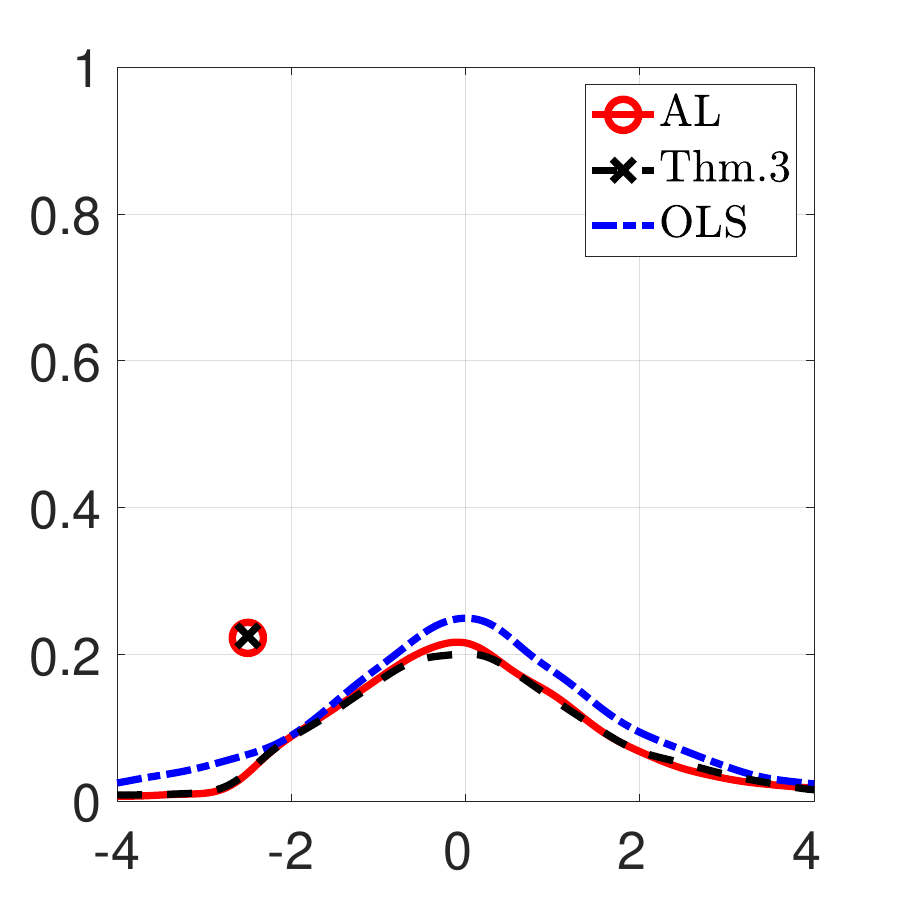}
	\end{subfigure}\begin{subfigure}{0.2\textwidth}
		\centering
		\caption*{$T = 50$}
		\vspace{-1.5ex}
		\includegraphics[trim={0cm 0cm 0.50cm 0.5cm},width=\textwidth,clip]
		{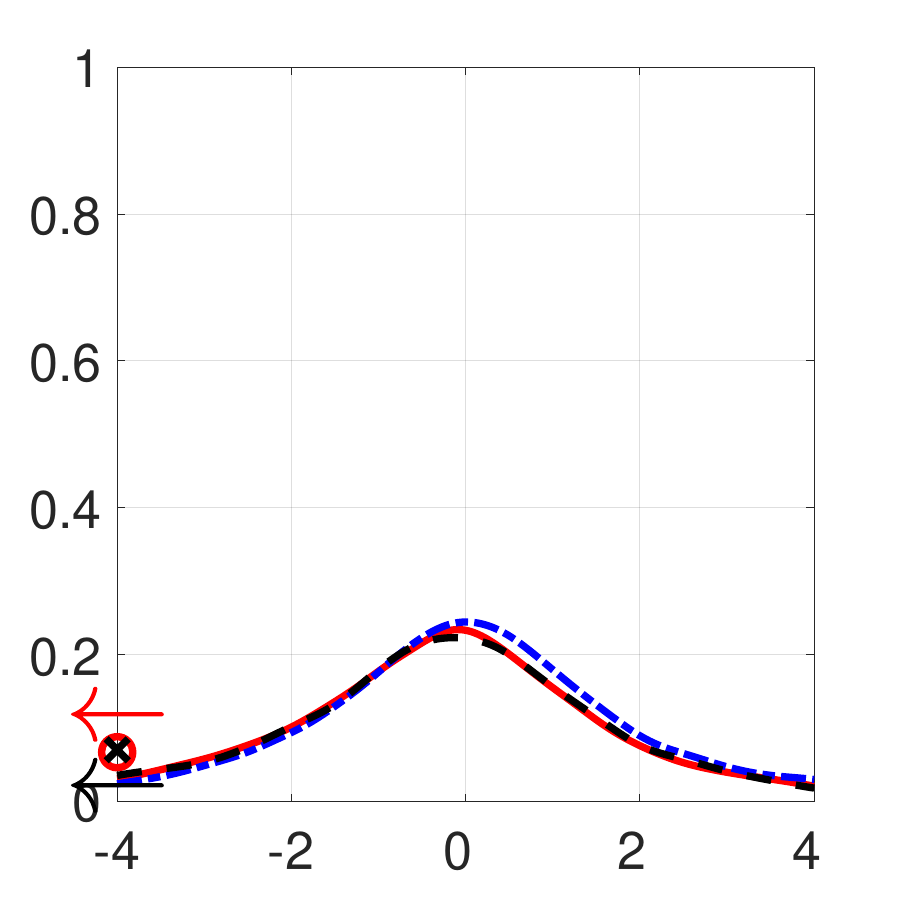}
	\end{subfigure}\begin{subfigure}{0.2\textwidth}
		\centering
		\caption*{$T = 100$}
		\vspace{-1.5ex}
		\includegraphics[trim={0cm 0cm 0.50cm 0.5cm},width=\textwidth,clip]
		{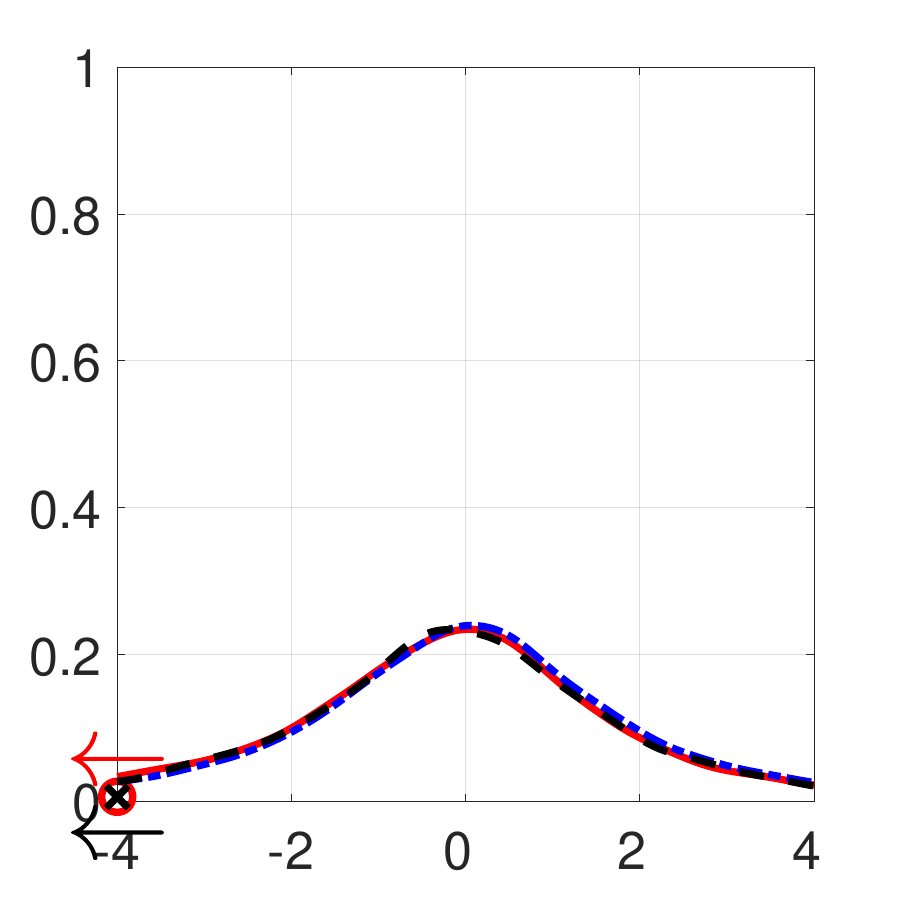}
	\end{subfigure}\begin{subfigure}{0.2\textwidth}
		\centering
		\caption*{$T = 250$}
		\vspace{-1.5ex}
		\includegraphics[trim={0cm 0cm 0.50cm 0.5cm},width=\textwidth,clip]
		{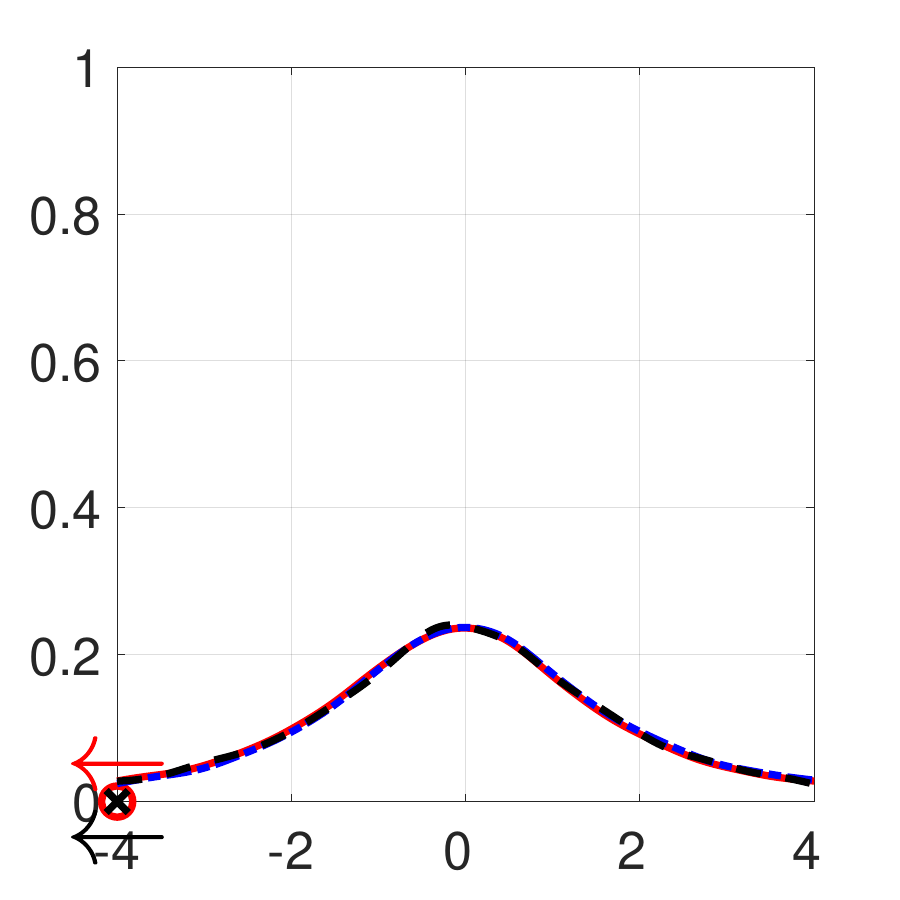}
	\end{subfigure}\begin{subfigure}{0.2\textwidth}
		\centering
		\caption*{$T = 1000$}
		\vspace{-1.5ex}
		\includegraphics[trim={0cm 0cm 0.50cm 0.5cm},width=\textwidth,clip]
		{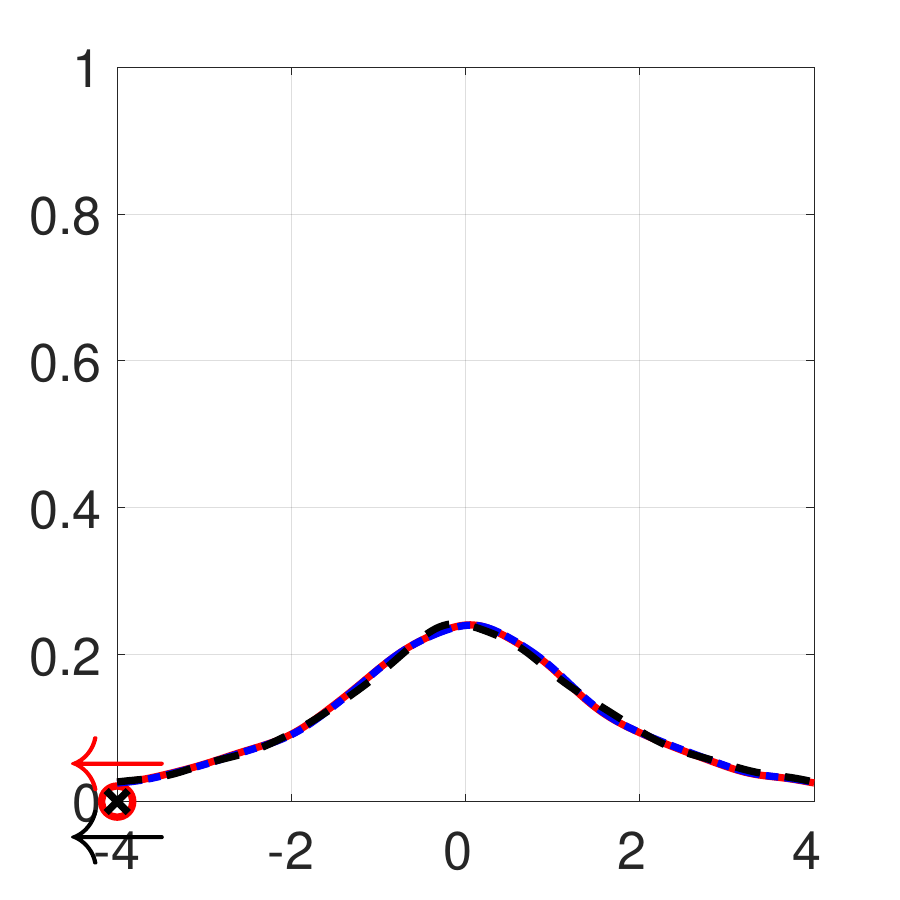}
	\end{subfigure}
	
	\caption*{$\lambda_T = T^{1/4}$}
	\vspace{-1.5ex}
	\begin{subfigure}{0.2\textwidth}
		\centering
		\includegraphics[trim={0cm 0cm 0.50cm 0.5cm},width=\textwidth,clip]
		{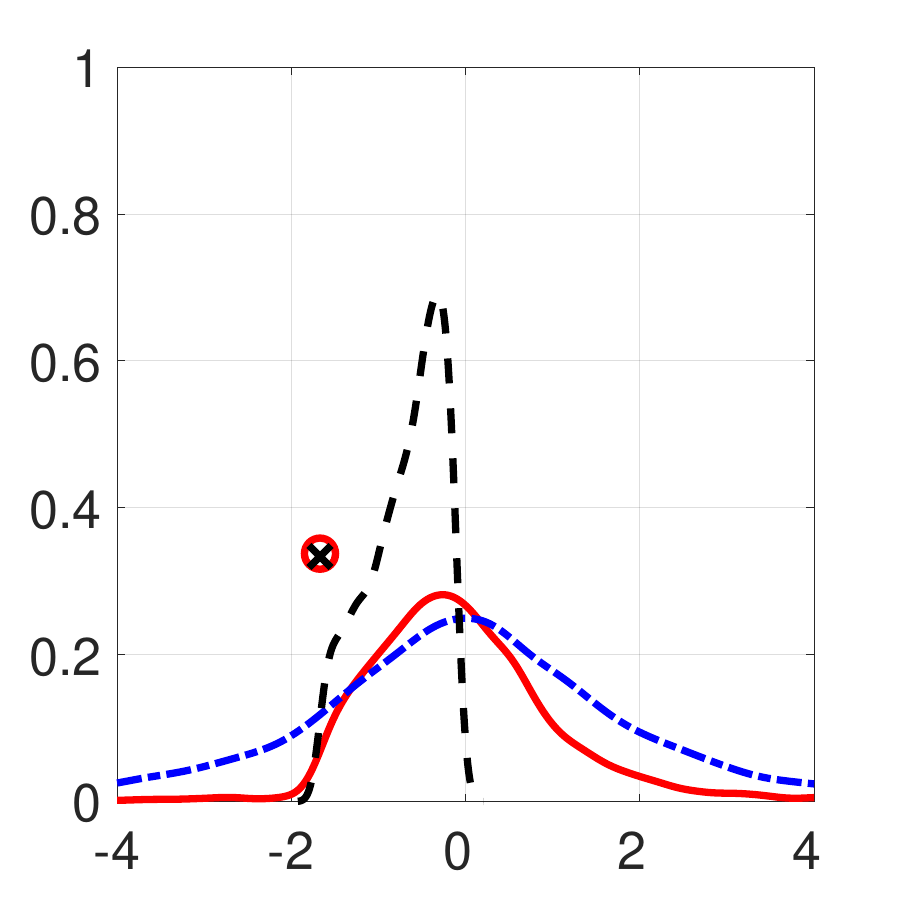}
	\end{subfigure}\begin{subfigure}{0.2\textwidth}
		\centering
		\includegraphics[trim={0cm 0cm 0.50cm 0.5cm},width=\textwidth,clip]
		{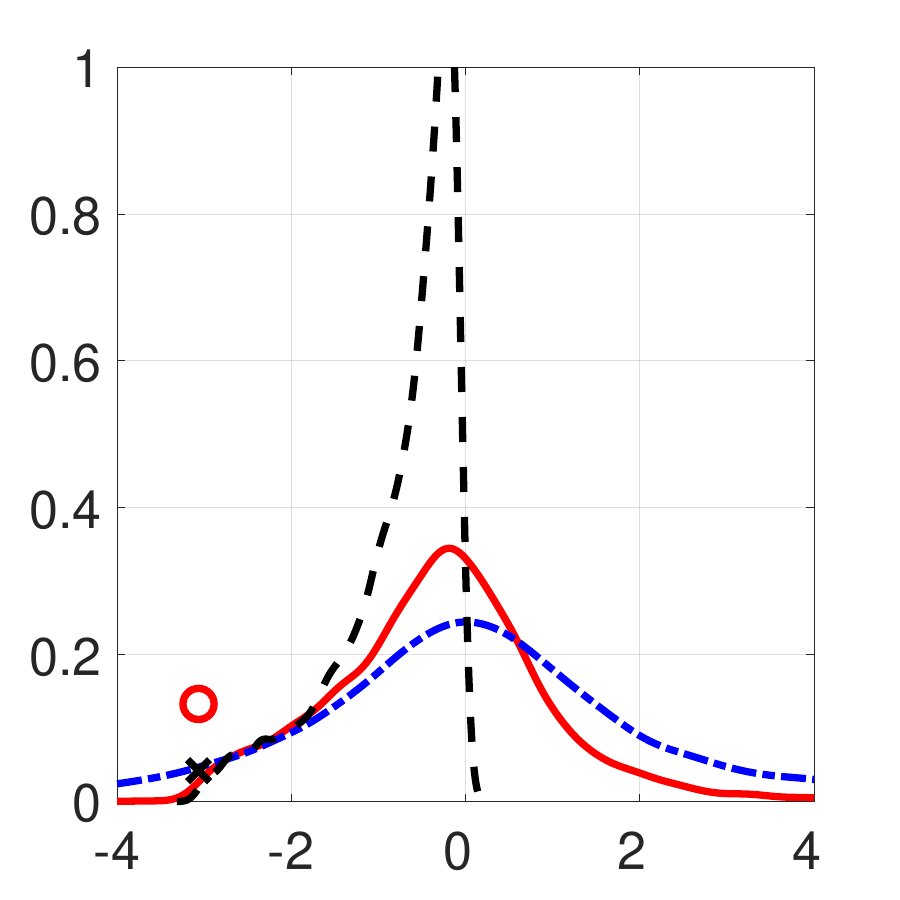}
	\end{subfigure}\begin{subfigure}{0.2\textwidth}
		\centering
		\includegraphics[trim={0cm 0cm 0.50cm 0.5cm},width=\textwidth,clip]
		{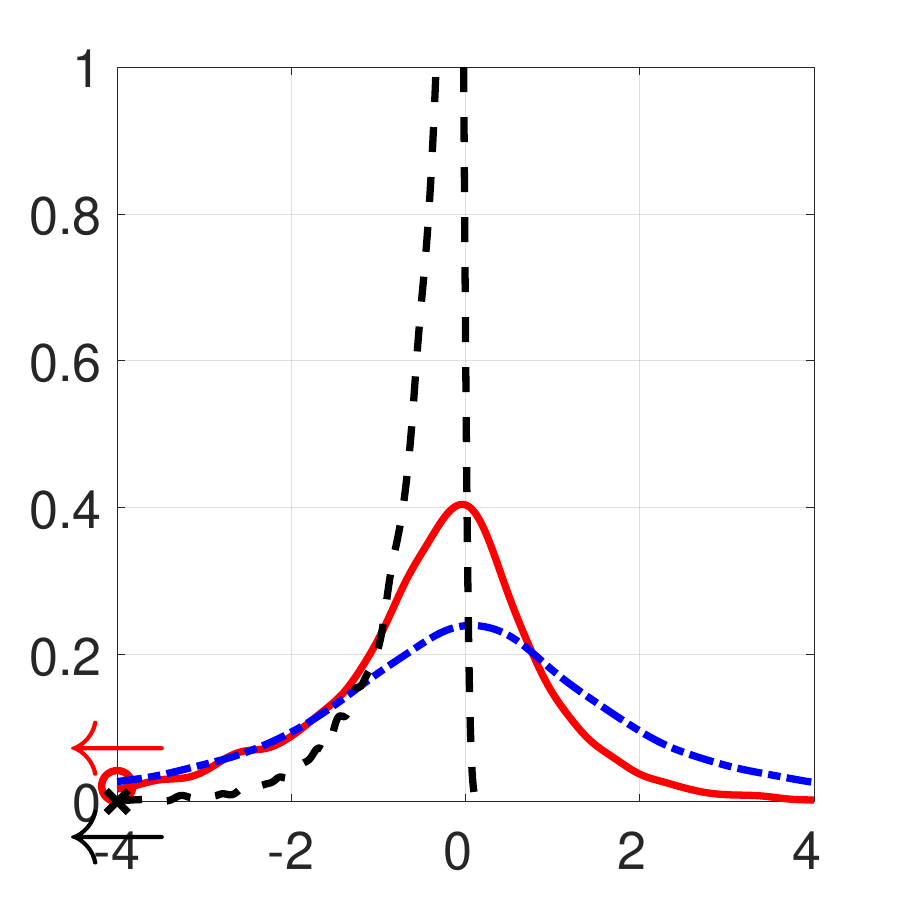}
	\end{subfigure}\begin{subfigure}{0.2\textwidth}
		\centering
		\includegraphics[trim={0cm 0cm 0.50cm 0.5cm},width=\textwidth,clip]
		{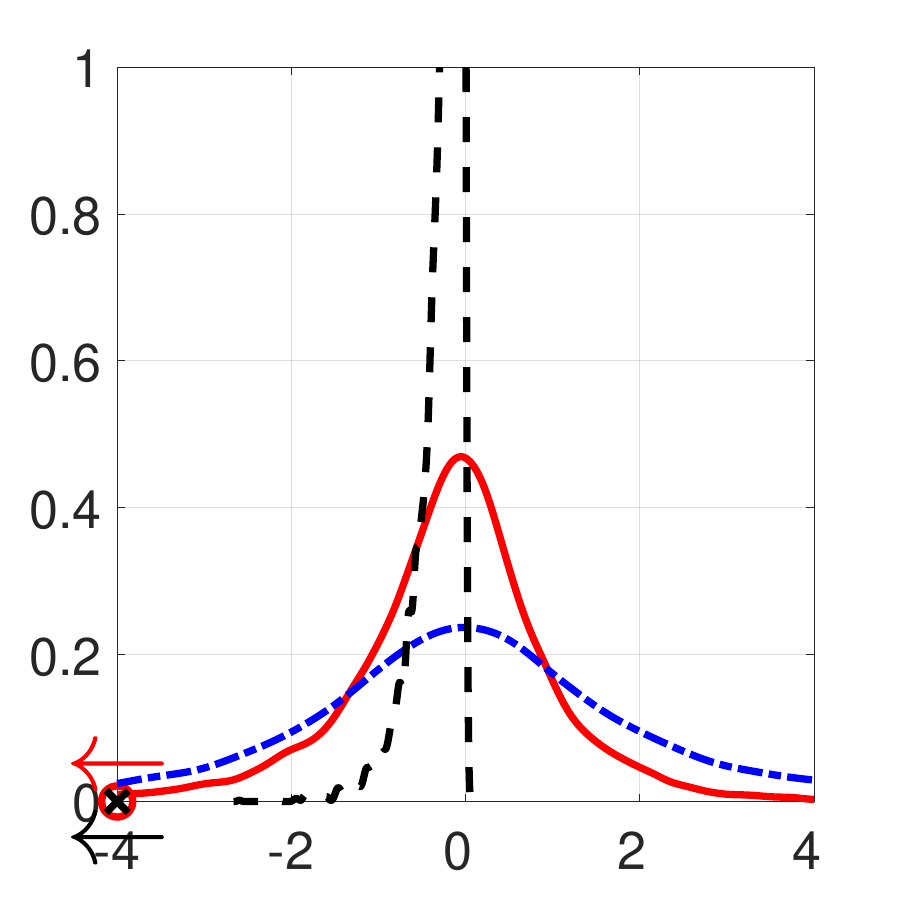}
	\end{subfigure}\begin{subfigure}{0.2\textwidth}
		\centering
		\includegraphics[trim={0cm 0cm 0.50cm 0.5cm},width=\textwidth,clip]
		{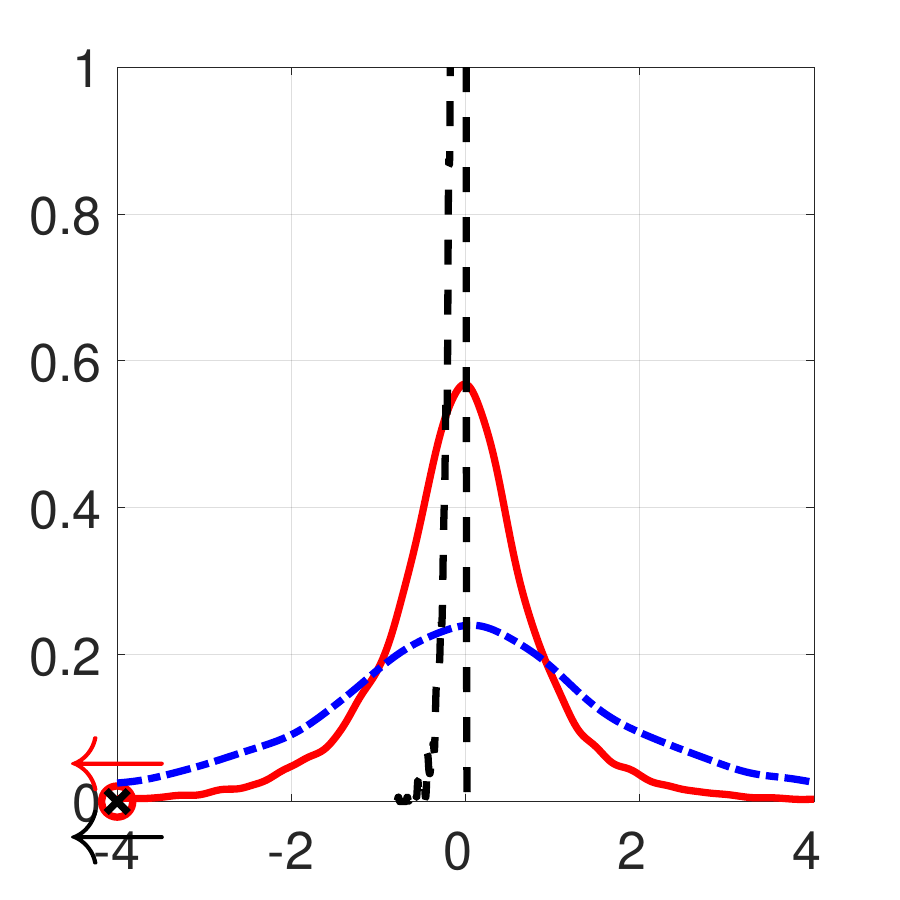}
	\end{subfigure}
	
	\caption*{$\lambda_T = T^{1/2}$}
	\vspace{-1.5ex}
	\begin{subfigure}{0.2\textwidth}
		\centering
		\includegraphics[trim={0cm 0cm 0.50cm 0.5cm},width=\textwidth,clip]
		{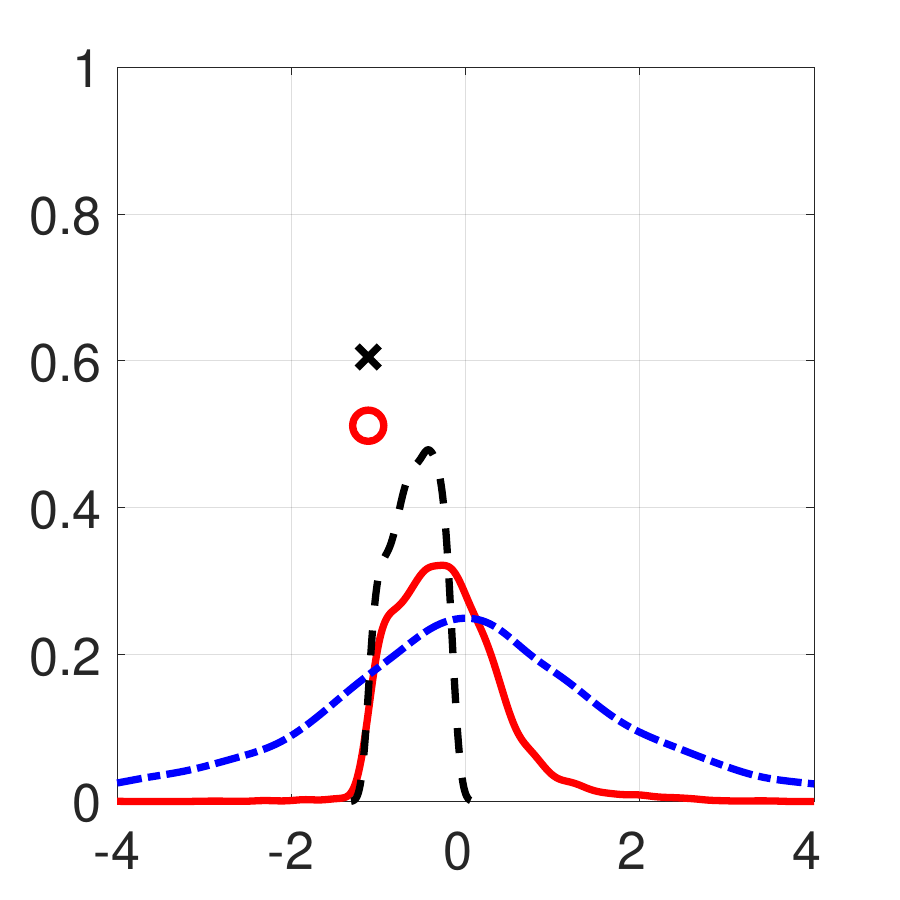}
	\end{subfigure}\begin{subfigure}{0.2\textwidth}
		\centering
		\includegraphics[trim={0cm 0cm 0.50cm 0.5cm},width=\textwidth,clip]
		{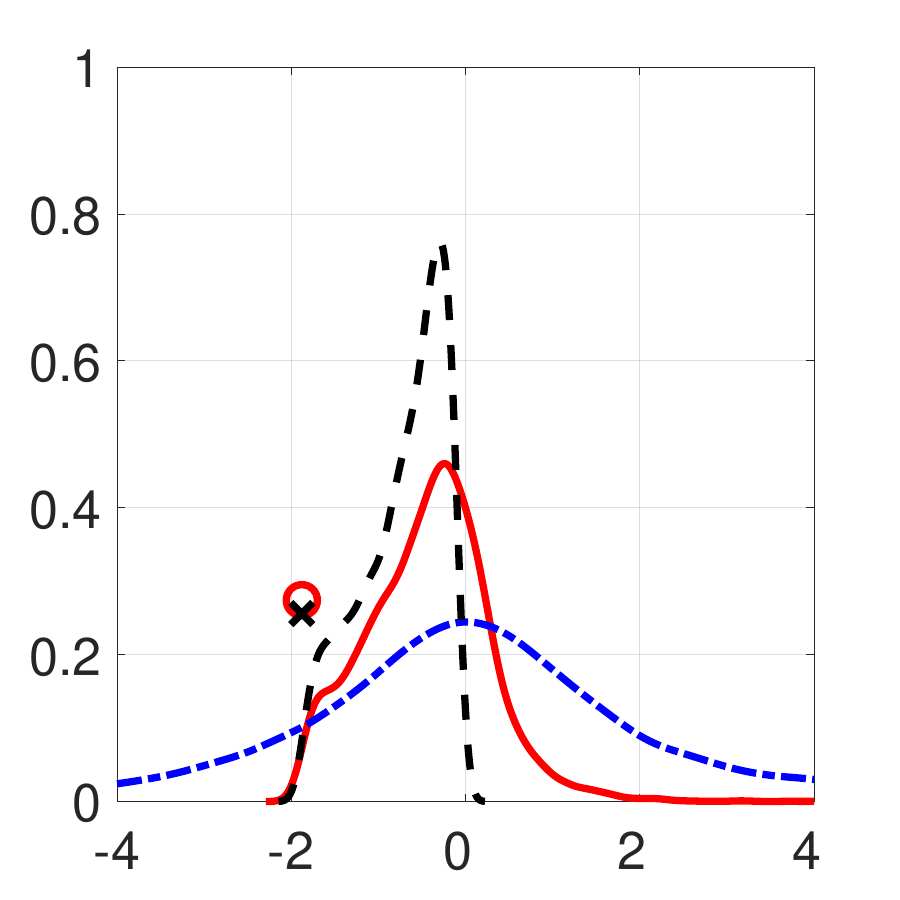}
	\end{subfigure}\begin{subfigure}{0.2\textwidth}
		\centering
		\includegraphics[trim={0cm 0cm 0.50cm 0.5cm},width=\textwidth,clip]
		{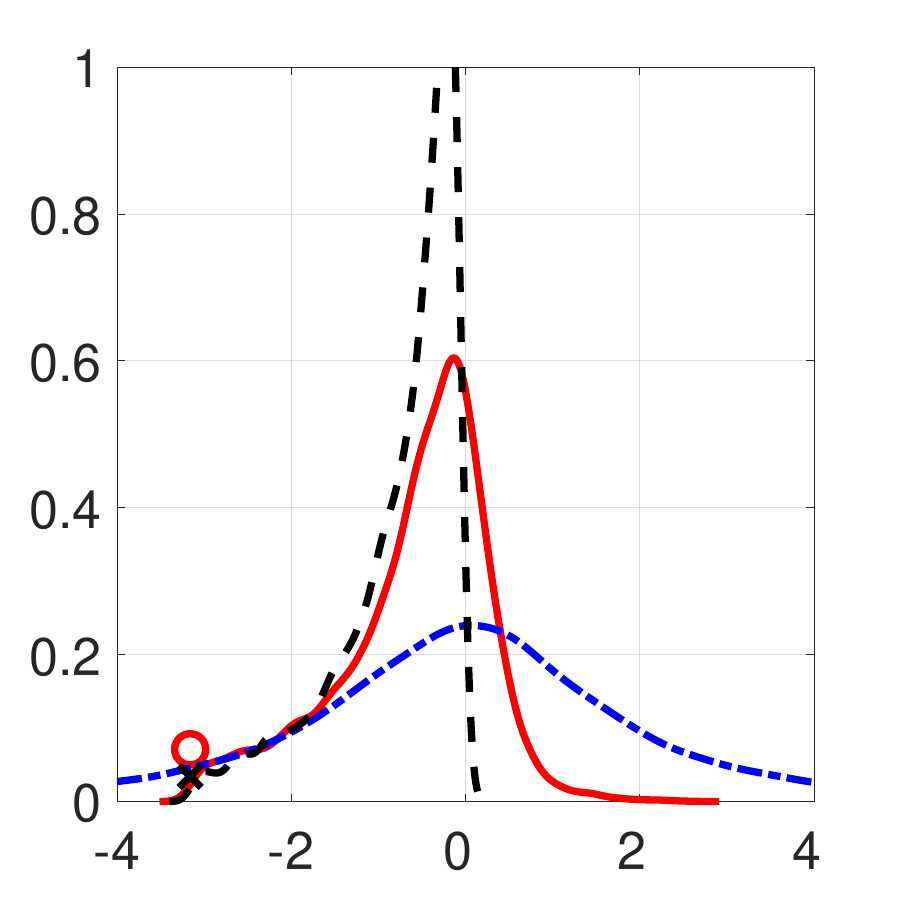}
	\end{subfigure}\begin{subfigure}{0.2\textwidth}
		\centering
		\includegraphics[trim={0cm 0cm 0.50cm 0.5cm},width=\textwidth,clip]
		{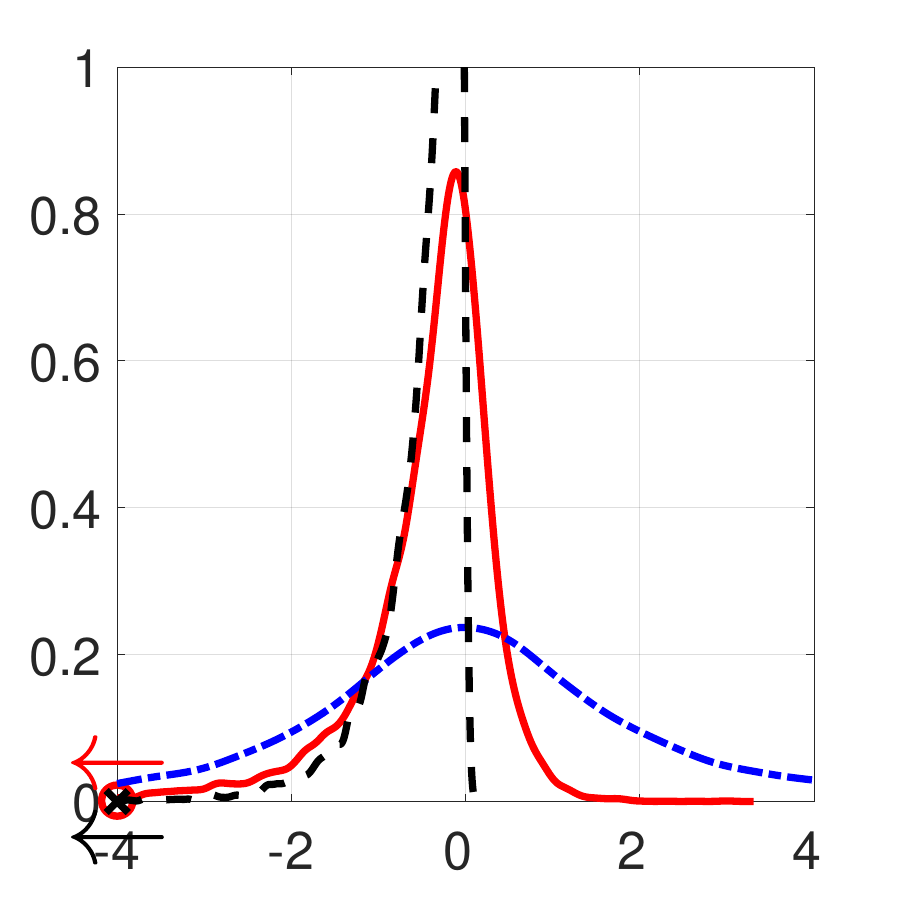}
	\end{subfigure}\begin{subfigure}{0.2\textwidth}
		\centering
		\includegraphics[trim={0cm 0cm 0.50cm 0.5cm},width=\textwidth,clip]
		{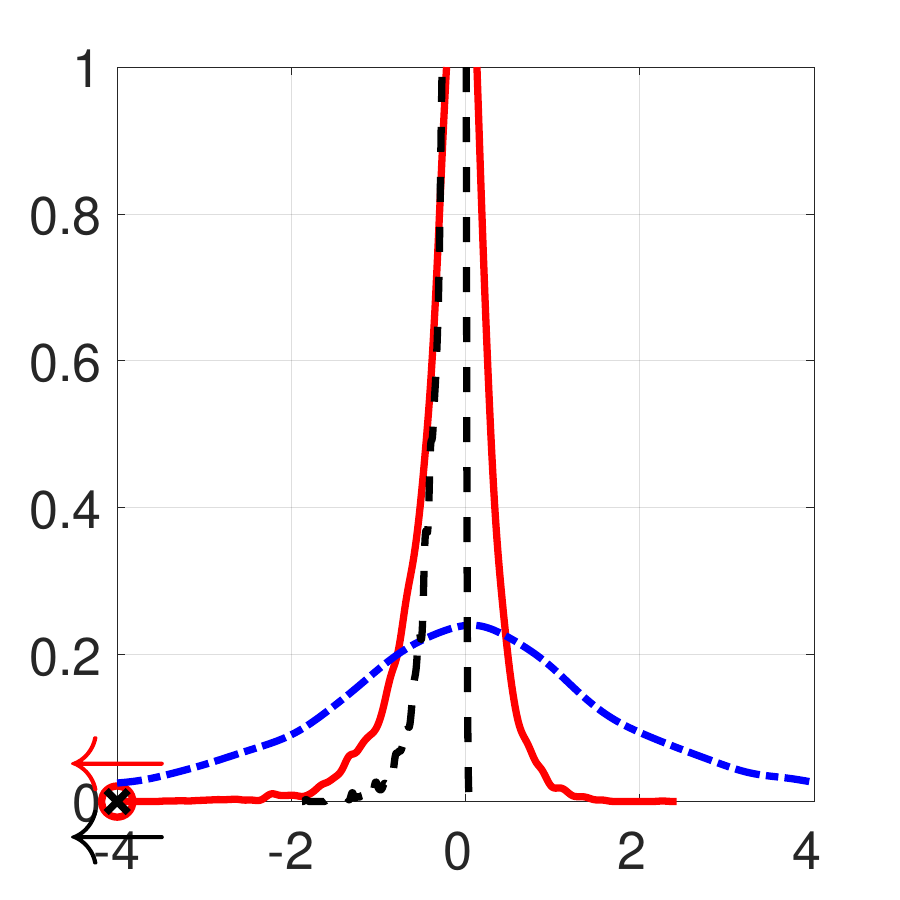}
	\end{subfigure}
	
	\caption*{$\lambda_T = T$}
	\vspace{-1.5ex}
	\begin{subfigure}{0.2\textwidth}
		\centering
		\includegraphics[trim={0cm 0cm 0.50cm 0.5cm},width=\textwidth,clip]
		{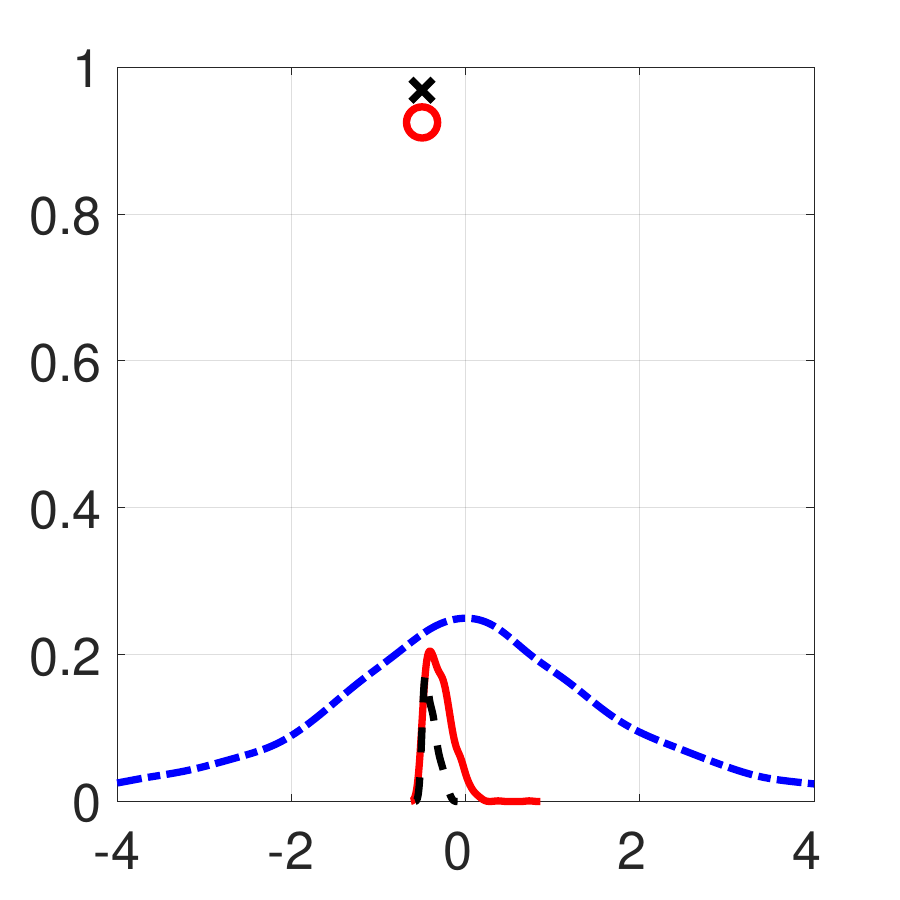}
	\end{subfigure}\begin{subfigure}{0.2\textwidth}
		\centering
		\includegraphics[trim={0cm 0cm 0.50cm 0.5cm},width=\textwidth,clip]
		{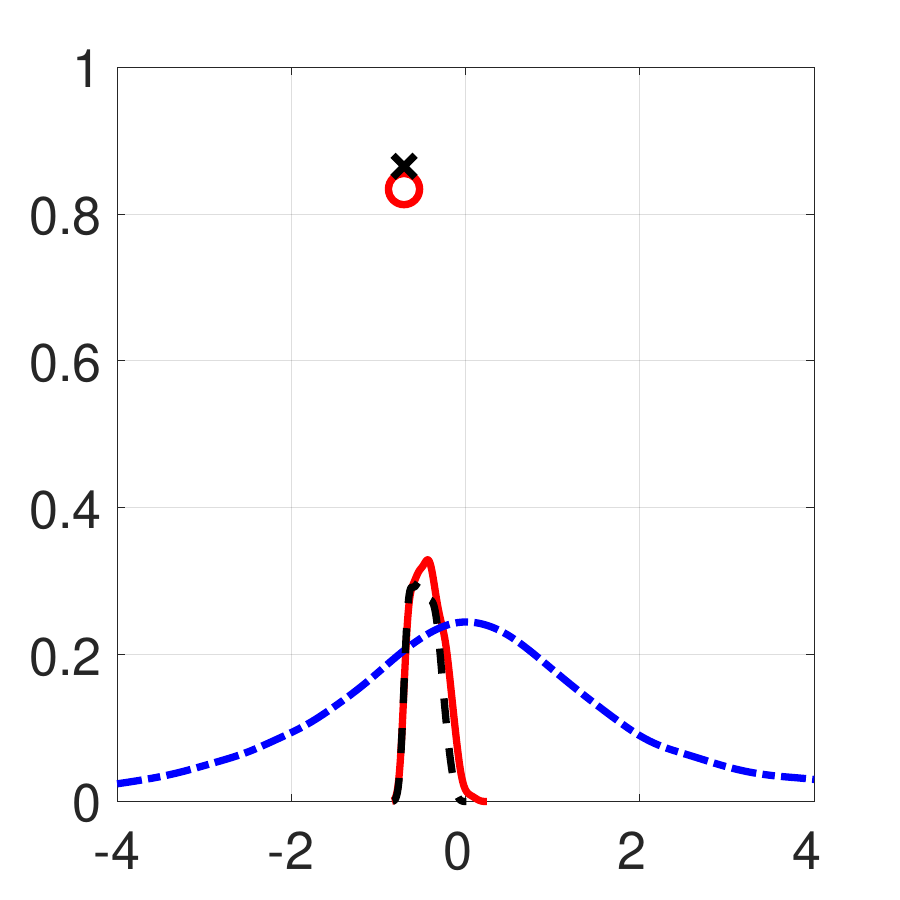}
	\end{subfigure}\begin{subfigure}{0.2\textwidth}
		\centering
		\includegraphics[trim={0cm 0cm 0.50cm 0.5cm},width=\textwidth,clip]
		{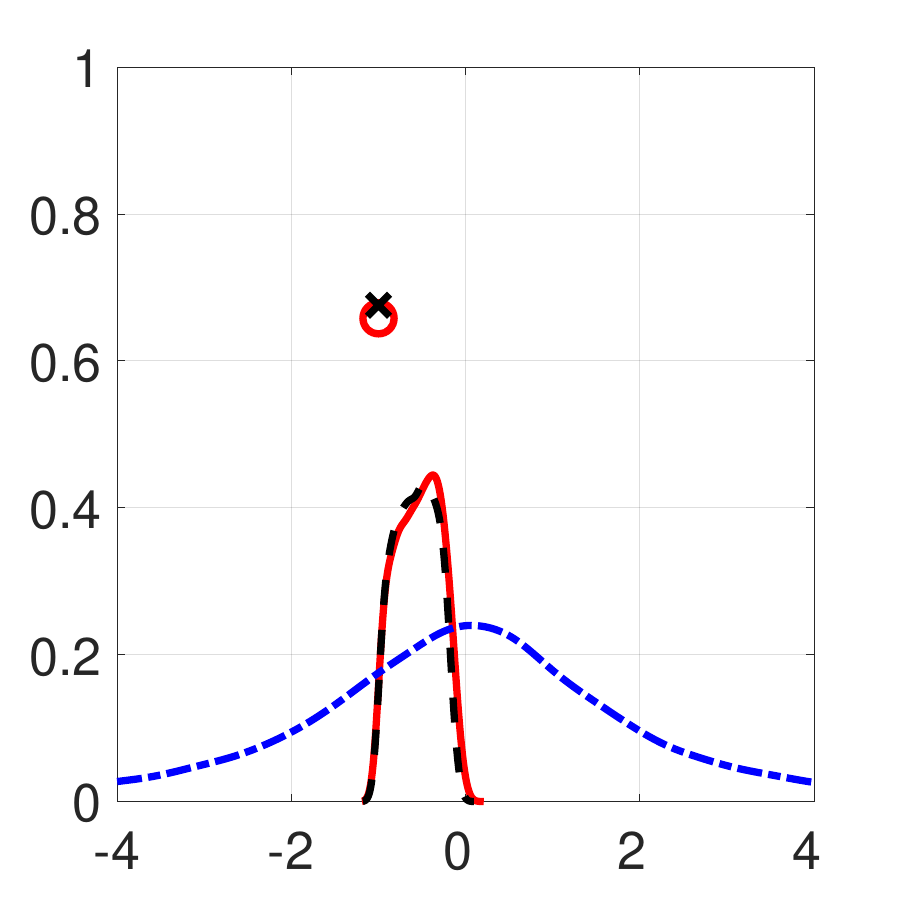}
	\end{subfigure}\begin{subfigure}{0.2\textwidth}
		\centering
		\includegraphics[trim={0cm 0cm 0.50cm 0.5cm},width=\textwidth,clip]
		{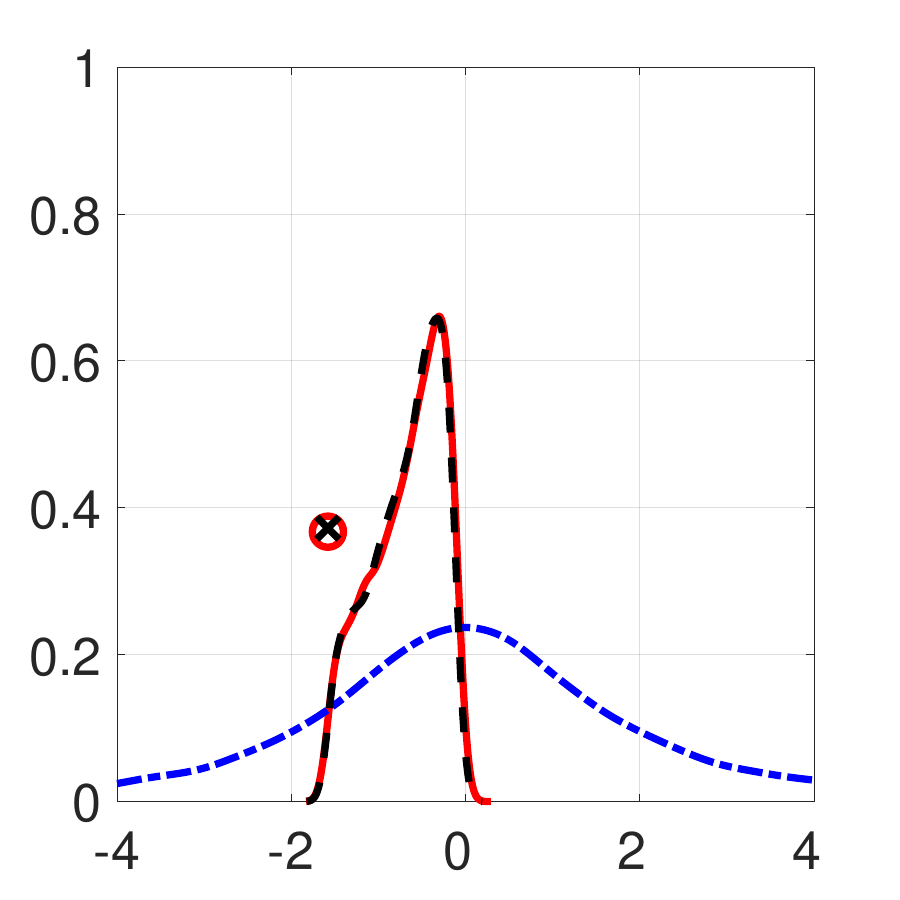}
	\end{subfigure}\begin{subfigure}{0.2\textwidth}
		\centering
		\includegraphics[trim={0cm 0cm 0.50cm 0.5cm},width=\textwidth,clip]
		{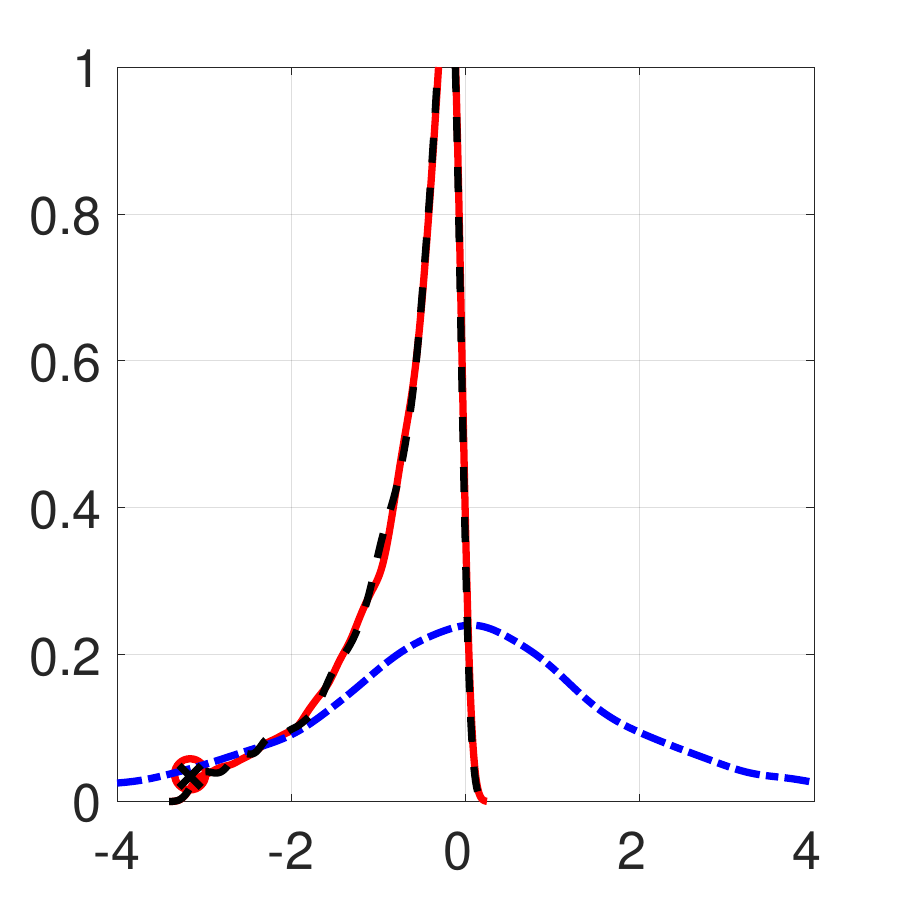}
	\end{subfigure}
	
\end{center}

\vspace{-2ex} 

\caption{Finite-sample distributions of $T(\betaAL - \beta_T)$ (under
conservative tuning, in the first row) and $\lambda_T^{-1/2}T(\betaAL
- \beta_T)$ (under consistent tuning, in the remaining rows) in case
$\beta_T \equiv 0.1\beta$ (labeled ``AL''), and case-specific limiting
distribution from Theorem~\ref{thm:ls_dist-unif}, evaluated at sample
counterparts of limiting parameters (labeled ``Thm.3''). \emph{Notes}:
``OLS'' denotes the finite-sample distribution of the OLS estimator.
If the correct location of the atomic part of a density is smaller
than $-4$, it is plotted at $-4$ with an arrow pointing to the left.}

\label{fig:densities_thm3_1}

\end{figure}

\begin{figure}[ht]
\begin{center}
	\caption*{$\lambda_T \equiv 1$}
	\begin{subfigure}{0.2\textwidth}
		\centering
		\caption*{$T = 25$}
		\vspace{-1.5ex}
		\includegraphics[trim={0cm 0cm 0.50cm 0.5cm},width=\textwidth,clip]
		{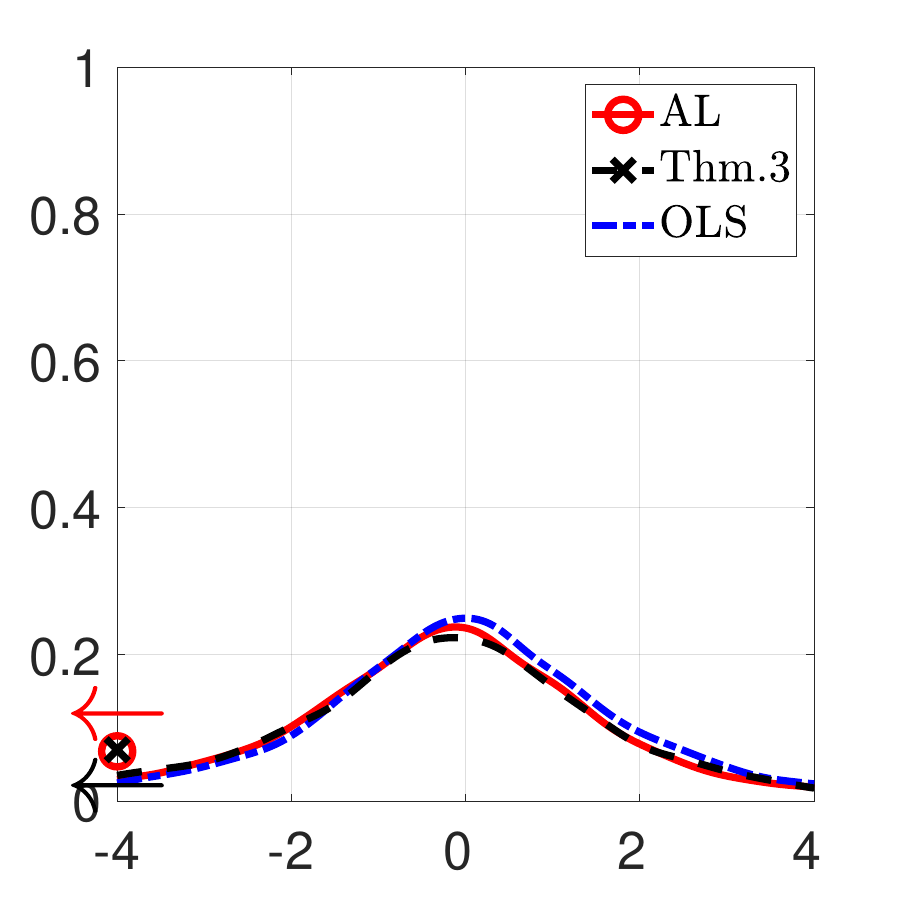}
	\end{subfigure}\begin{subfigure}{0.2\textwidth}
		\centering
		\caption*{$T = 50$}
		\vspace{-1.5ex}
		\includegraphics[trim={0cm 0cm 0.50cm 0.5cm},width=\textwidth,clip]
		{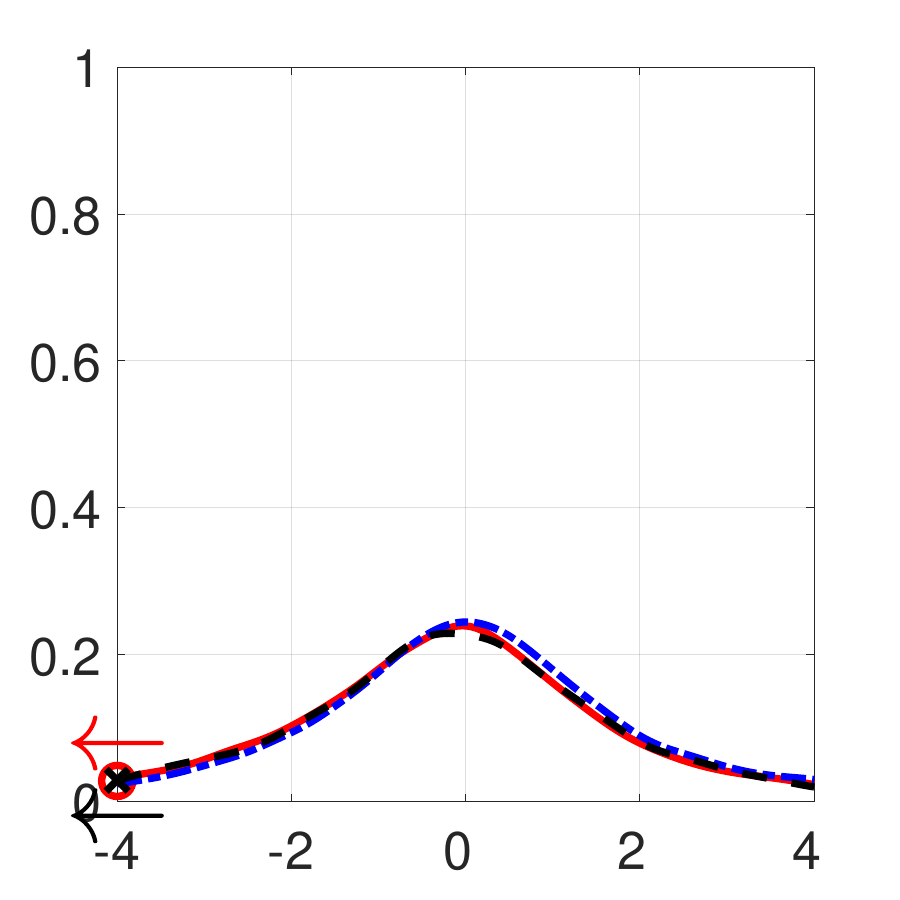}
	\end{subfigure}\begin{subfigure}{0.2\textwidth}
		\centering
		\caption*{$T = 100$}
		\vspace{-1.5ex}
		\includegraphics[trim={0cm 0cm 0.50cm 0.5cm},width=\textwidth,clip]
		{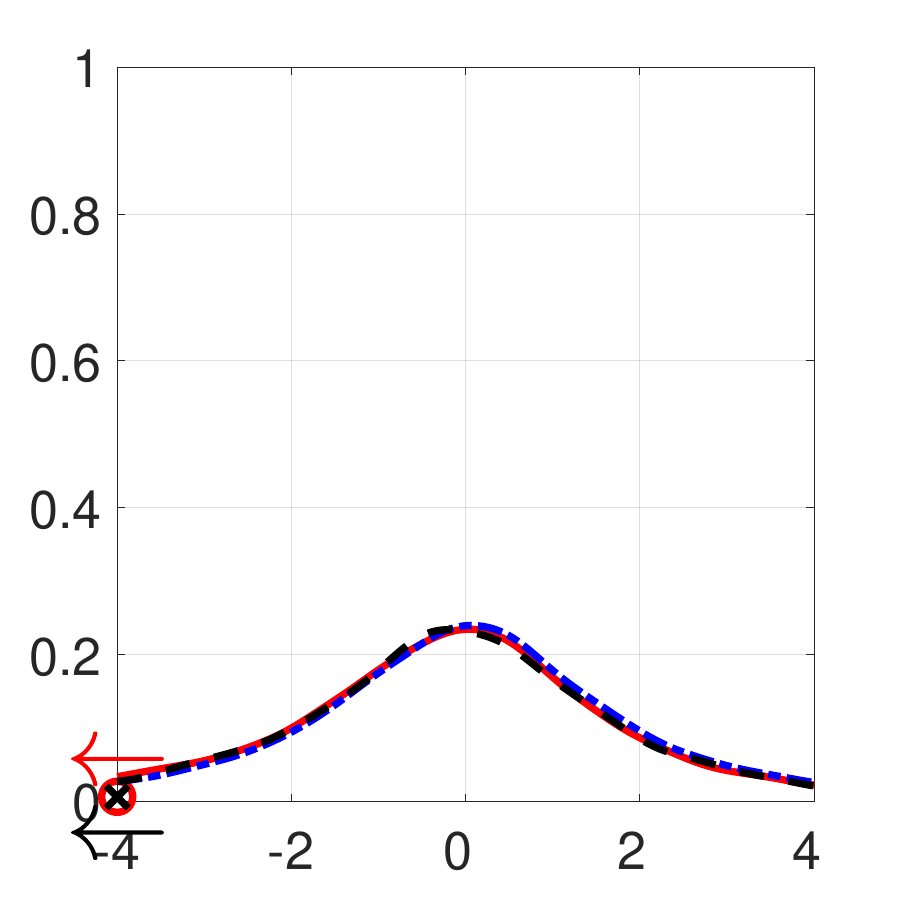}
	\end{subfigure}\begin{subfigure}{0.2\textwidth}
		\centering
		\caption*{$T = 250$}
		\vspace{-1.5ex}
		\includegraphics[trim={0cm 0cm 0.50cm 0.5cm},width=\textwidth,clip]
		{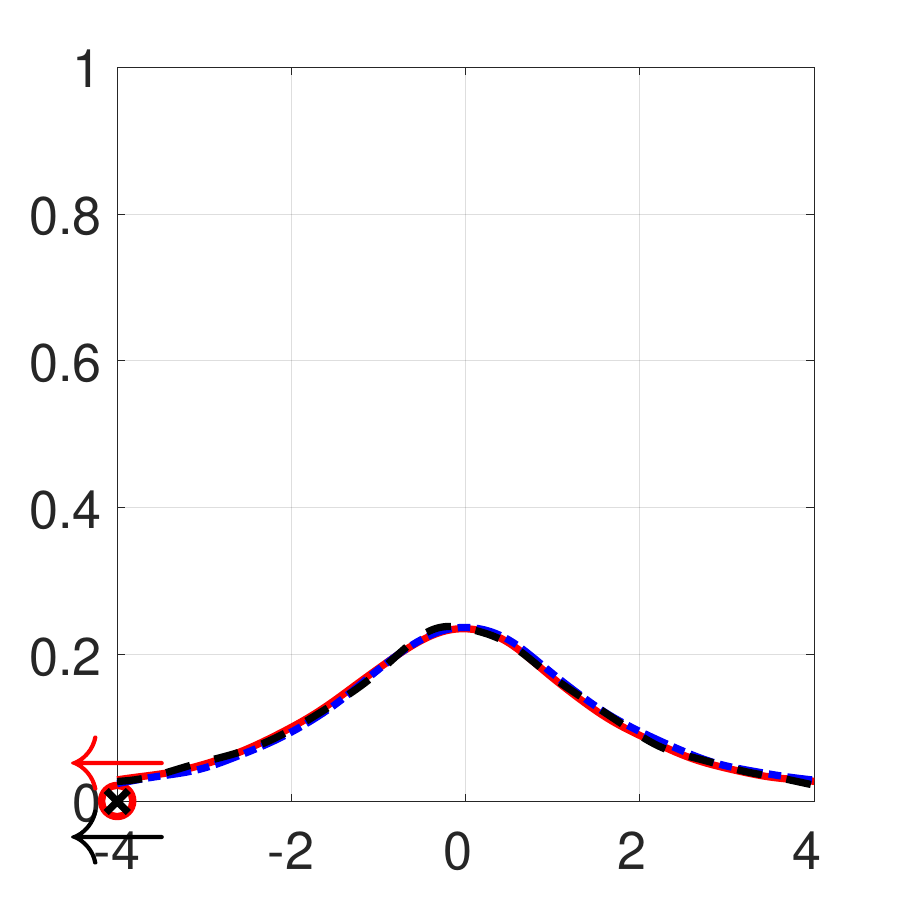}
	\end{subfigure}\begin{subfigure}{0.2\textwidth}
		\centering
		\caption*{$T = 1000$}
		\vspace{-1.5ex}
		\includegraphics[trim={0cm 0cm 0.50cm 0.5cm},width=\textwidth,clip]
		{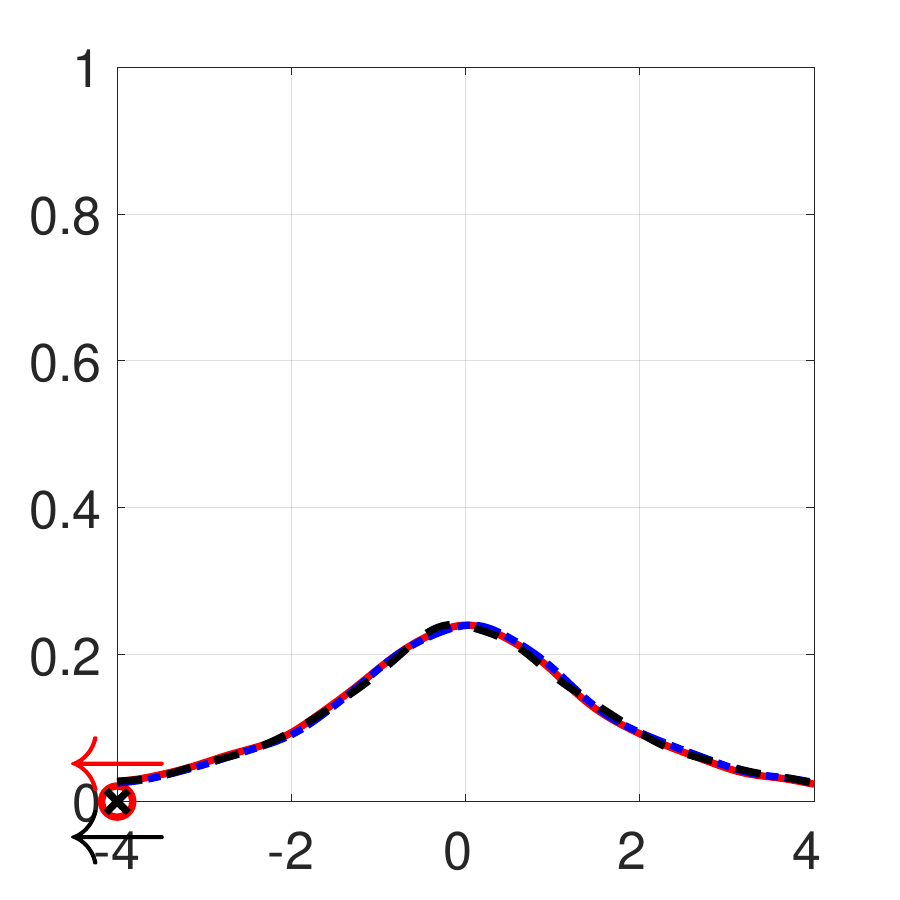}
	\end{subfigure}
	
	\caption*{$\lambda_T = T^{1/4}$}
	\vspace{-1.5ex}
	\begin{subfigure}{0.2\textwidth}
		\centering
		\includegraphics[trim={0cm 0cm 0.50cm 0.5cm},width=\textwidth,clip]
		{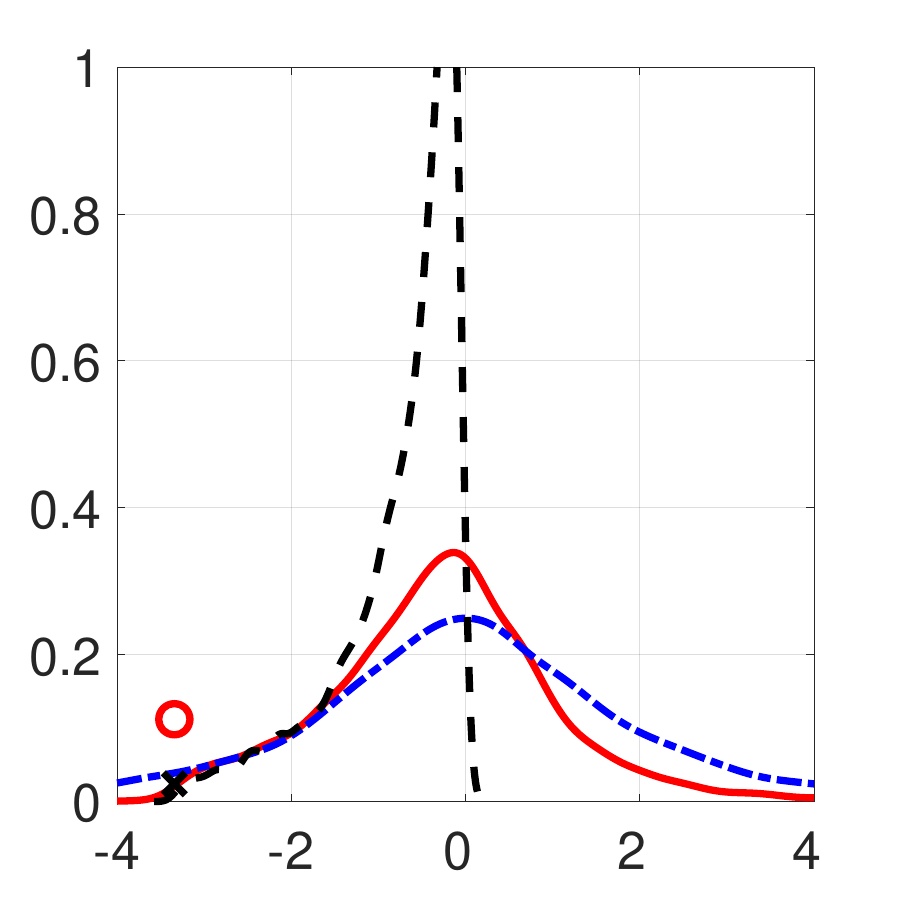}
	\end{subfigure}\begin{subfigure}{0.2\textwidth}
		\centering
		\includegraphics[trim={0cm 0cm 0.50cm 0.5cm},width=\textwidth,clip]
		{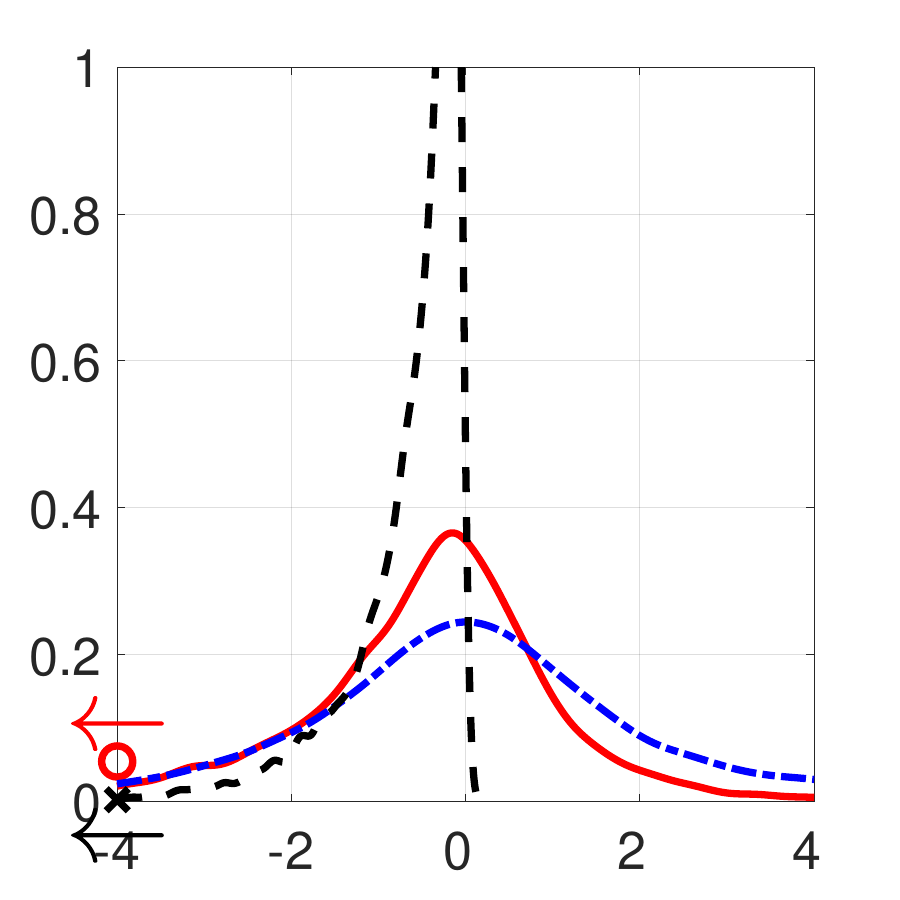}
	\end{subfigure}\begin{subfigure}{0.2\textwidth}
		\centering
		\includegraphics[trim={0cm 0cm 0.50cm 0.5cm},width=\textwidth,clip]
		{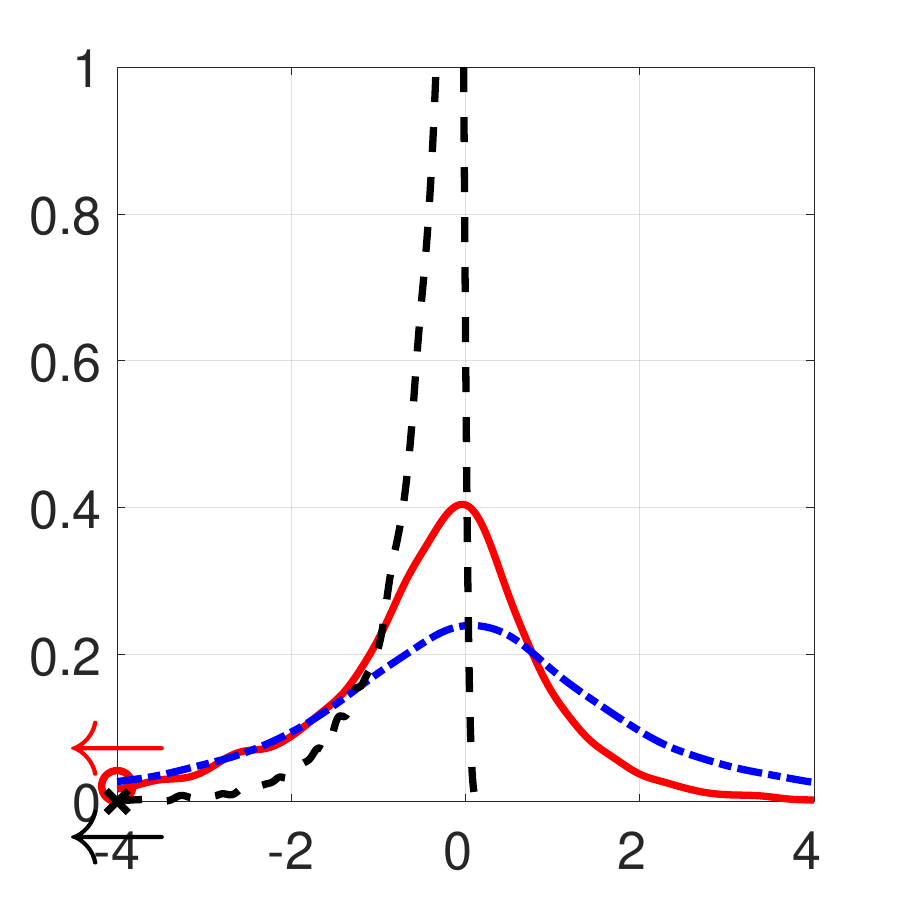}
	\end{subfigure}\begin{subfigure}{0.2\textwidth}
		\centering
		\includegraphics[trim={0cm 0cm 0.50cm 0.5cm},width=\textwidth,clip]
		{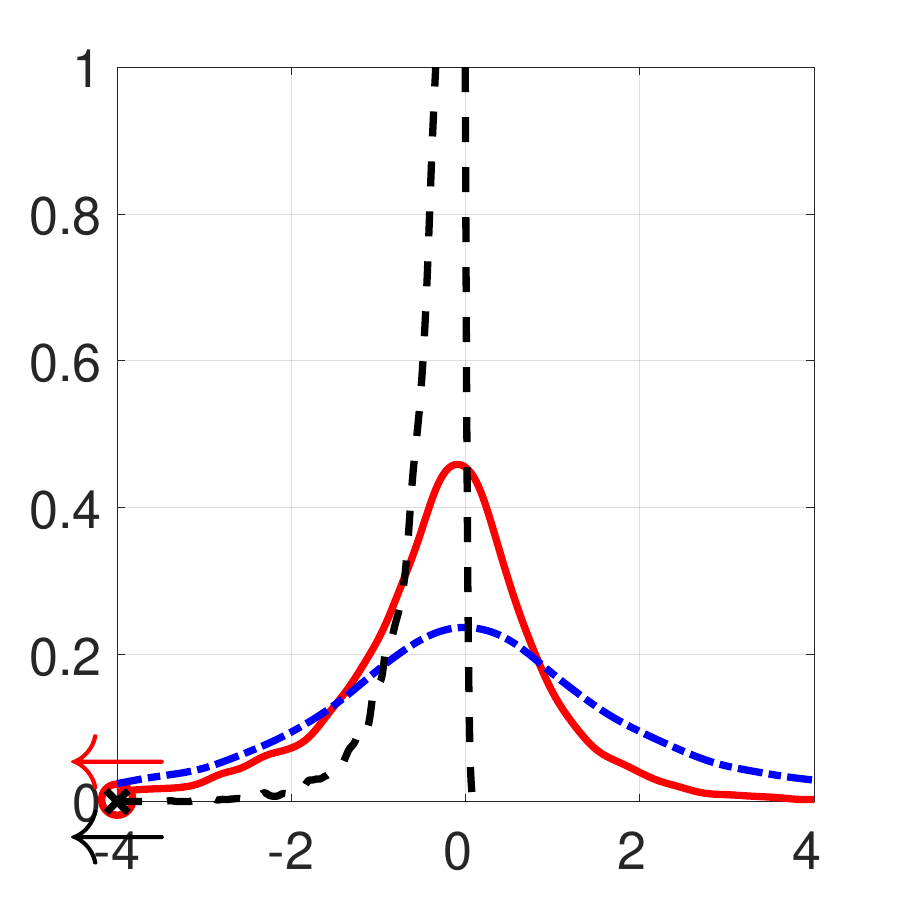}
	\end{subfigure}\begin{subfigure}{0.2\textwidth}
		\centering
		\includegraphics[trim={0cm 0cm 0.50cm 0.5cm},width=\textwidth,clip]
		{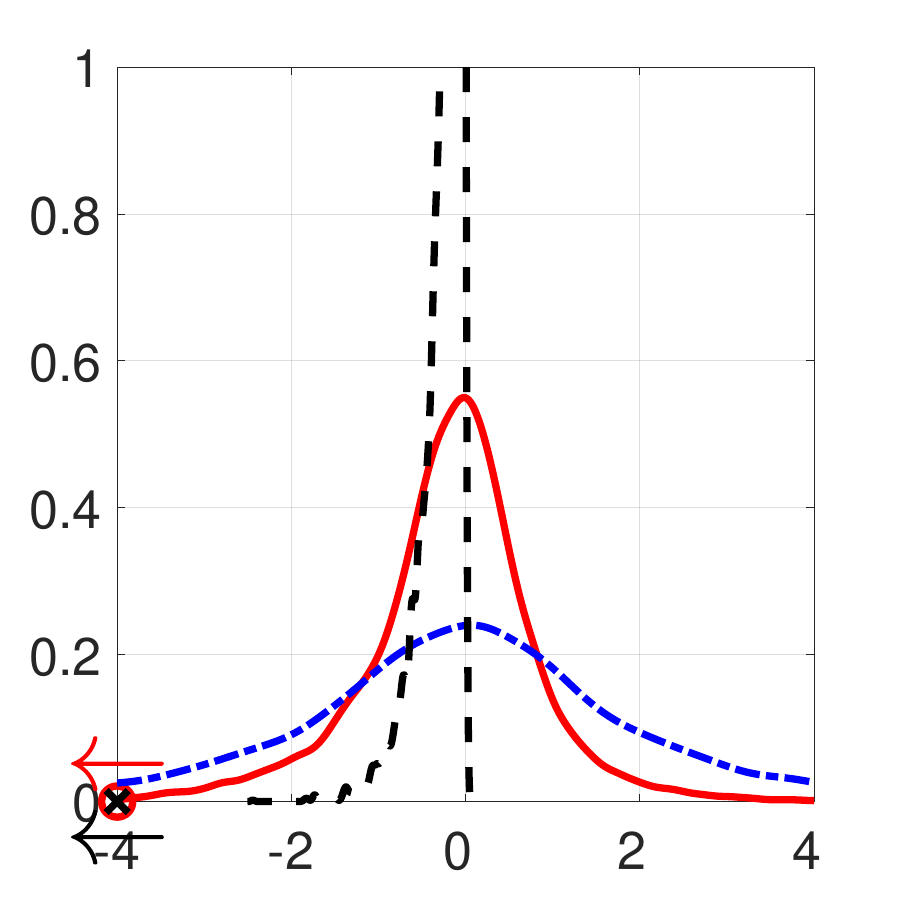}
	\end{subfigure}
	
	\caption*{$\lambda_T = T^{1/2}$}
	\vspace{-1.5ex}
	\begin{subfigure}{0.2\textwidth}
		\centering
		\includegraphics[trim={0cm 0cm 0.50cm 0.5cm},width=\textwidth,clip]
		{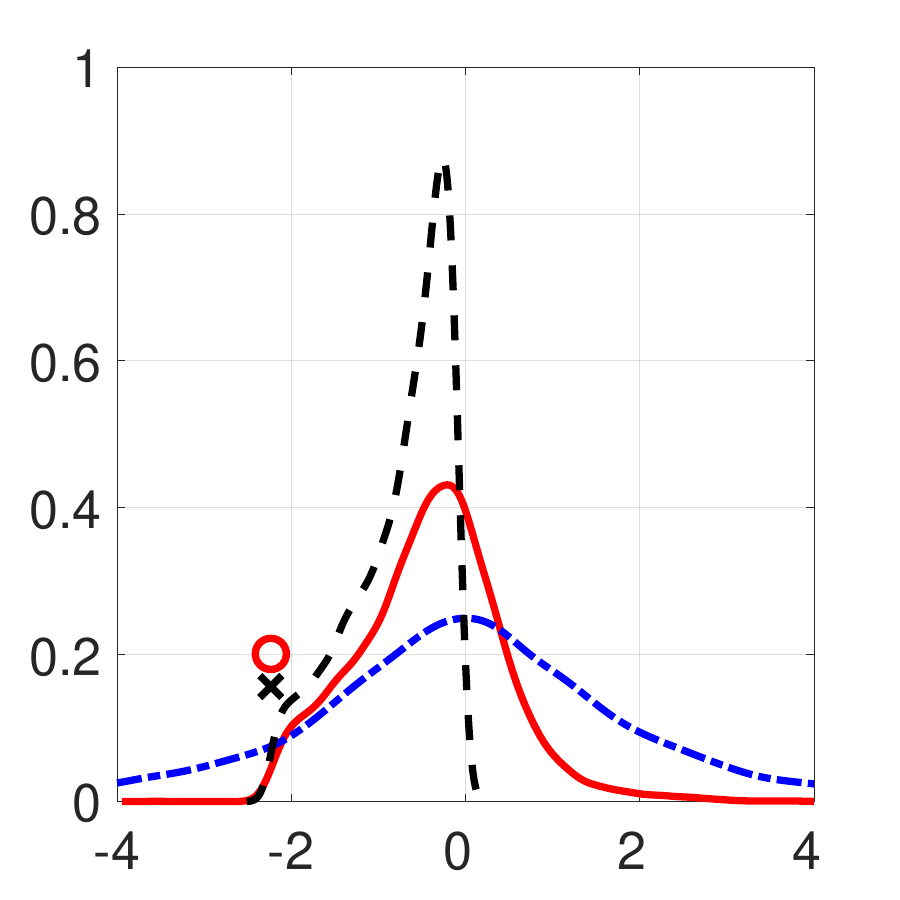}
	\end{subfigure}\begin{subfigure}{0.2\textwidth}
		\centering
		\includegraphics[trim={0cm 0cm 0.50cm 0.5cm},width=\textwidth,clip]
		{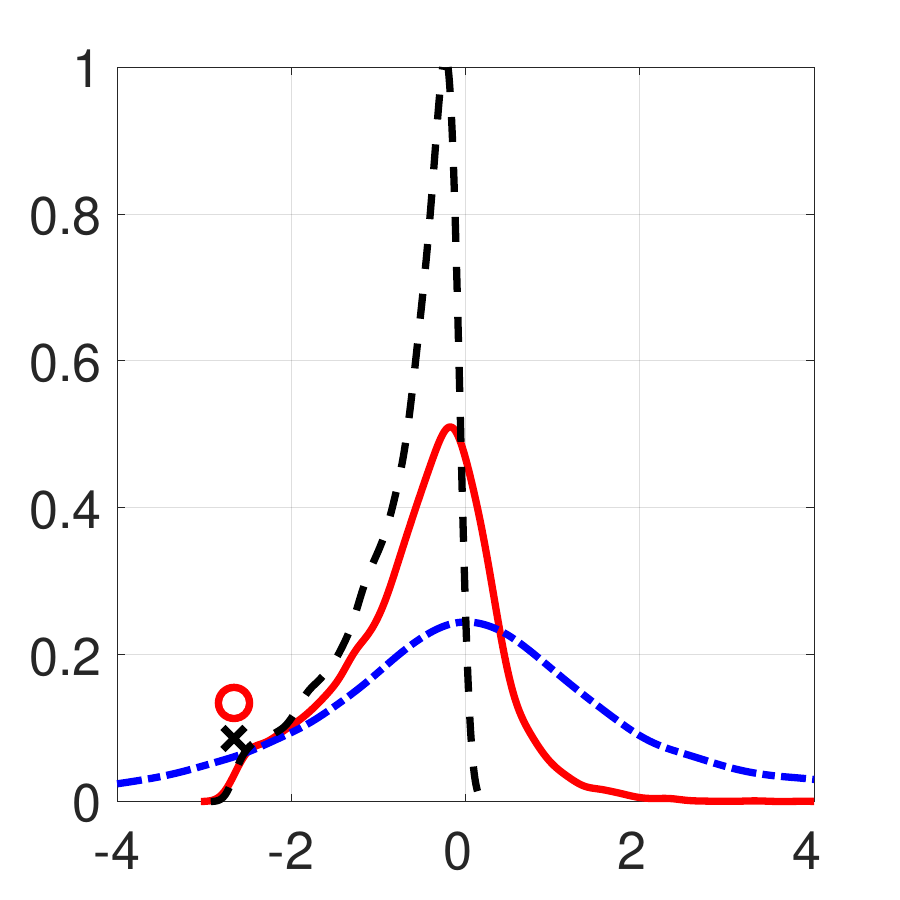}
	\end{subfigure}\begin{subfigure}{0.2\textwidth}
		\centering
		\includegraphics[trim={0cm 0cm 0.50cm 0.5cm},width=\textwidth,clip]
		{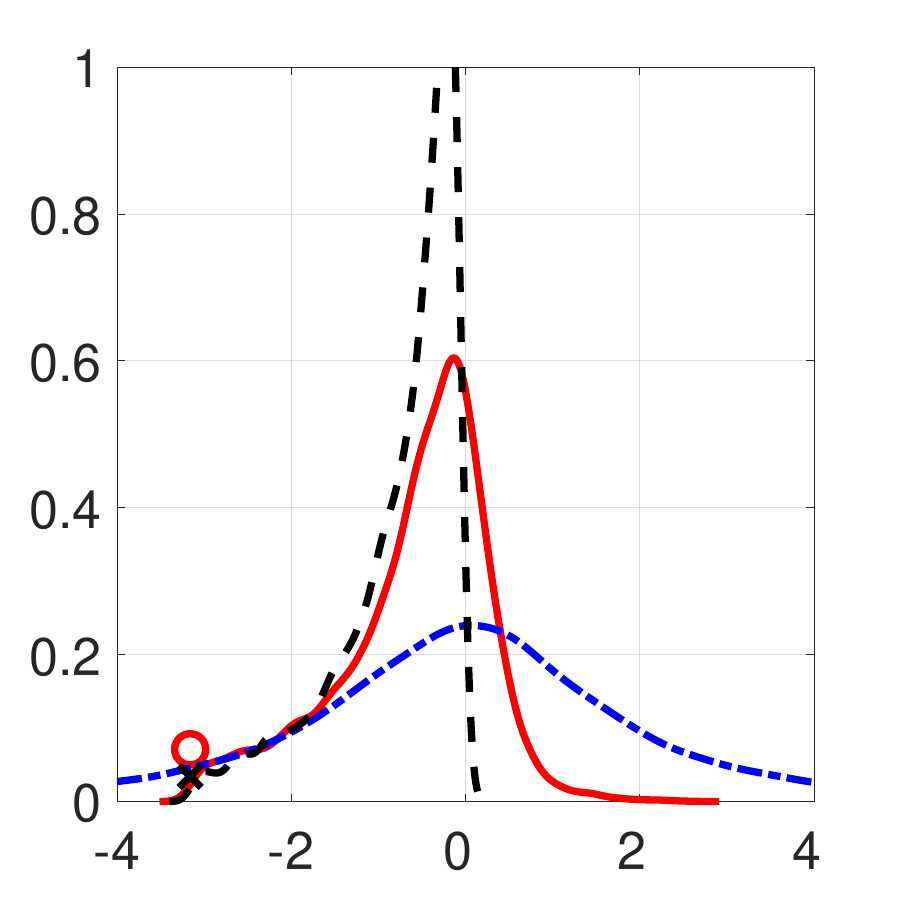}
	\end{subfigure}\begin{subfigure}{0.2\textwidth}
		\centering
		\includegraphics[trim={0cm 0cm 0.50cm 0.5cm},width=\textwidth,clip]
		{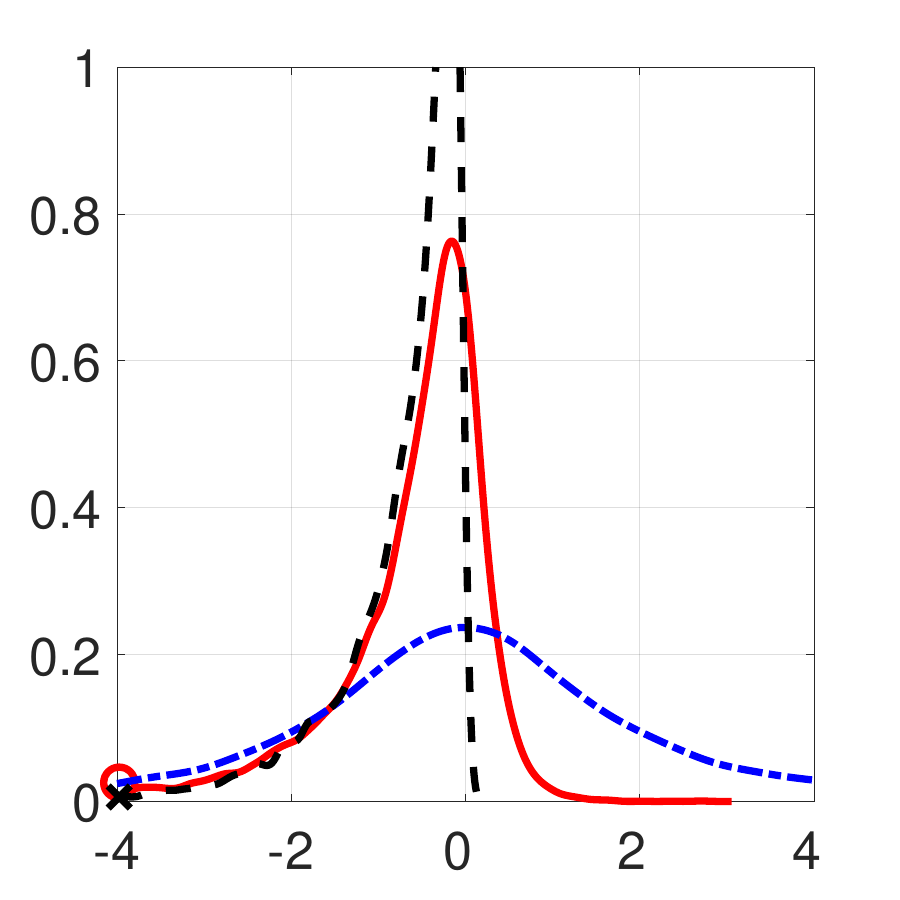}
	\end{subfigure}\begin{subfigure}{0.2\textwidth}
		\centering
		\includegraphics[trim={0cm 0cm 0.50cm 0.5cm},width=\textwidth,clip]
		{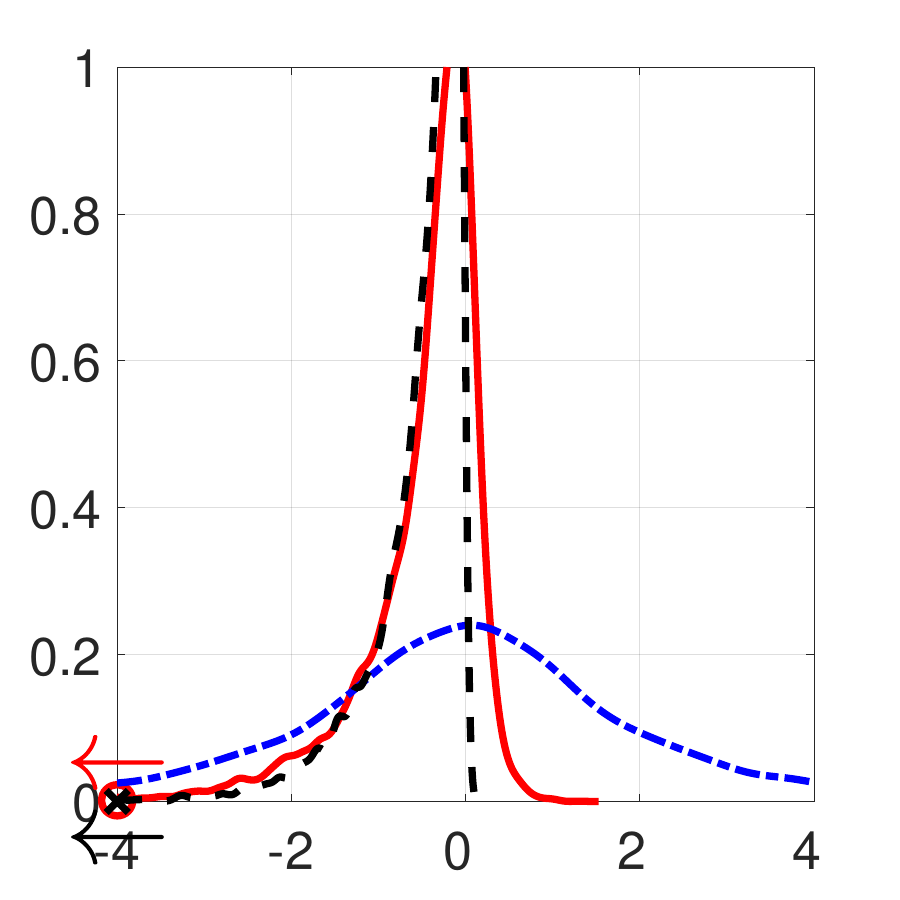}
	\end{subfigure}
	
	\caption*{$\lambda_T = T$}
	\vspace{-1.5ex}
	\begin{subfigure}{0.2\textwidth}
		\centering
		\includegraphics[trim={0cm 0cm 0.50cm 0.5cm},width=\textwidth,clip]
		{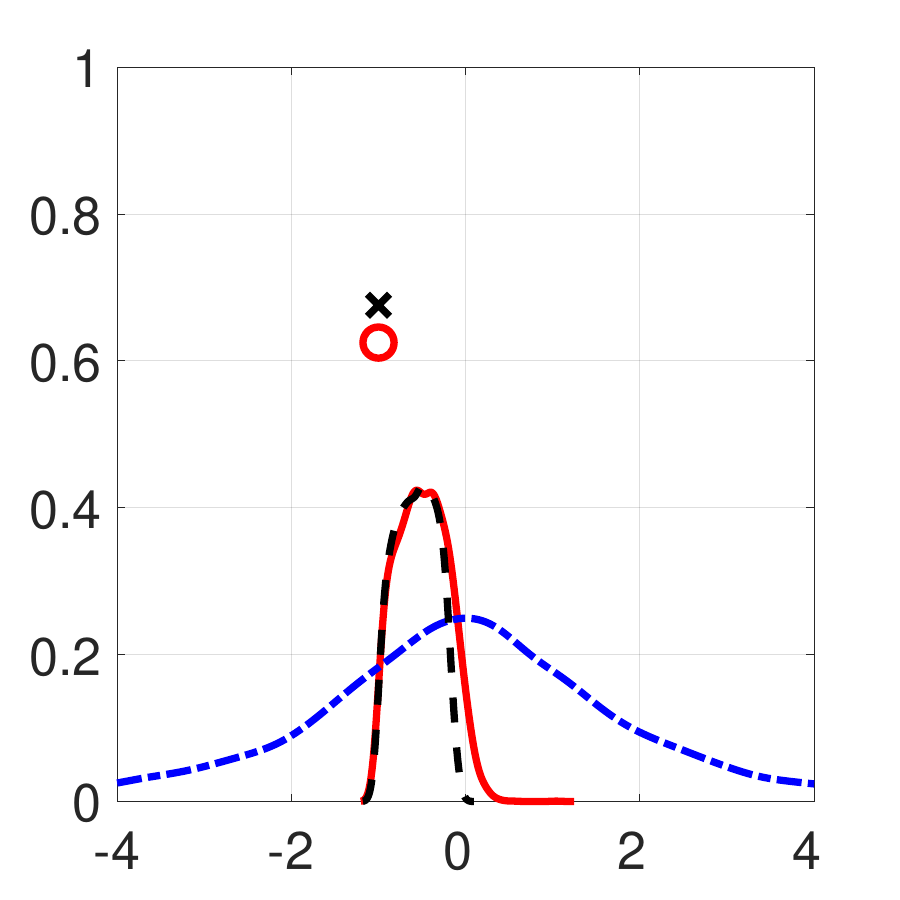}
	\end{subfigure}\begin{subfigure}{0.2\textwidth}
		\centering
		\includegraphics[trim={0cm 0cm 0.50cm 0.5cm},width=\textwidth,clip]
		{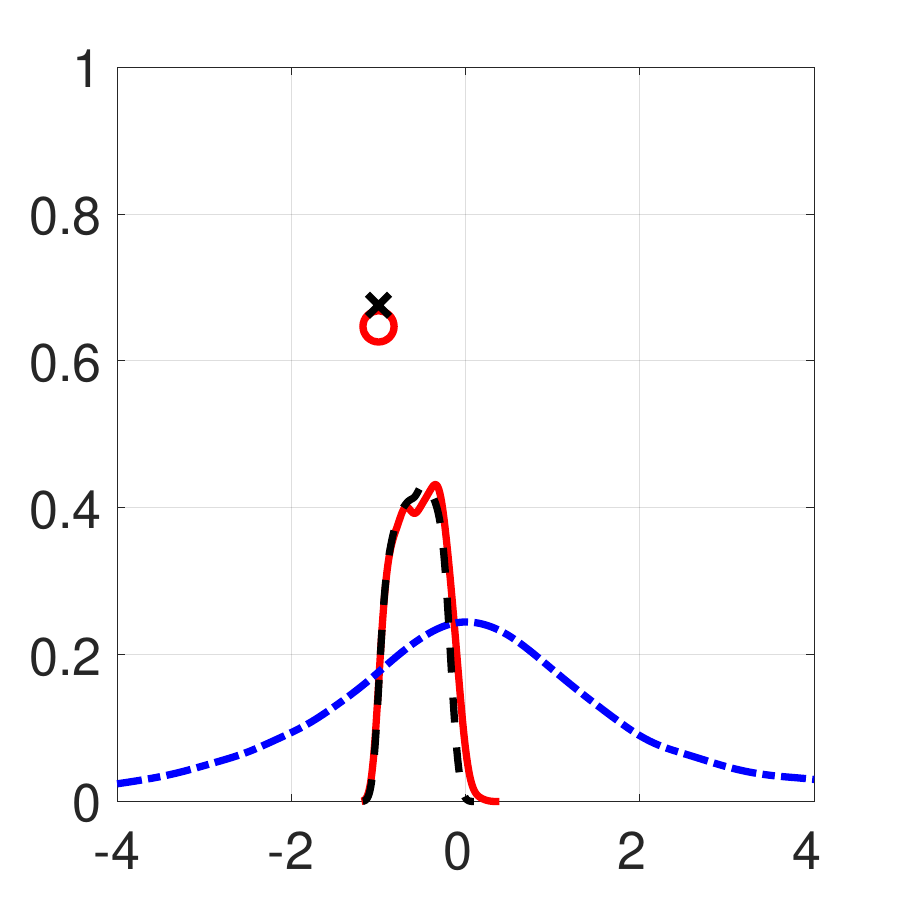}
	\end{subfigure}\begin{subfigure}{0.2\textwidth}
		\centering
		\includegraphics[trim={0cm 0cm 0.50cm 0.5cm},width=\textwidth,clip]
		{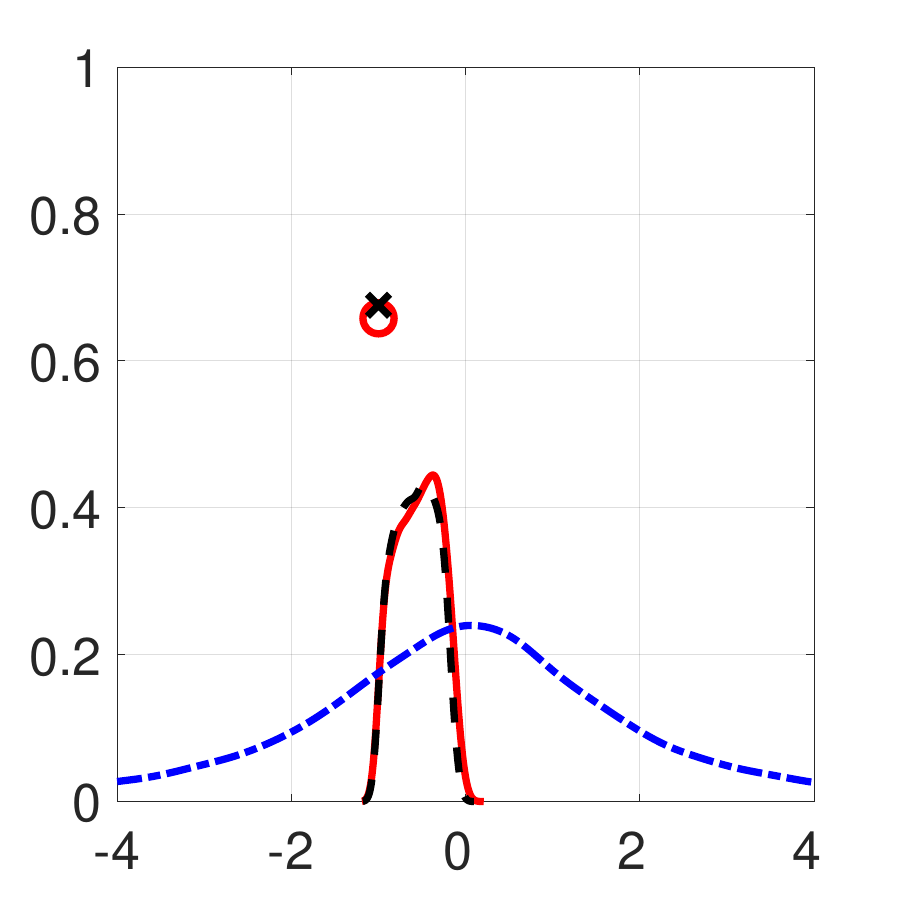}
	\end{subfigure}\begin{subfigure}{0.2\textwidth}
		\centering
		\includegraphics[trim={0cm 0cm 0.50cm 0.5cm},width=\textwidth,clip]
		{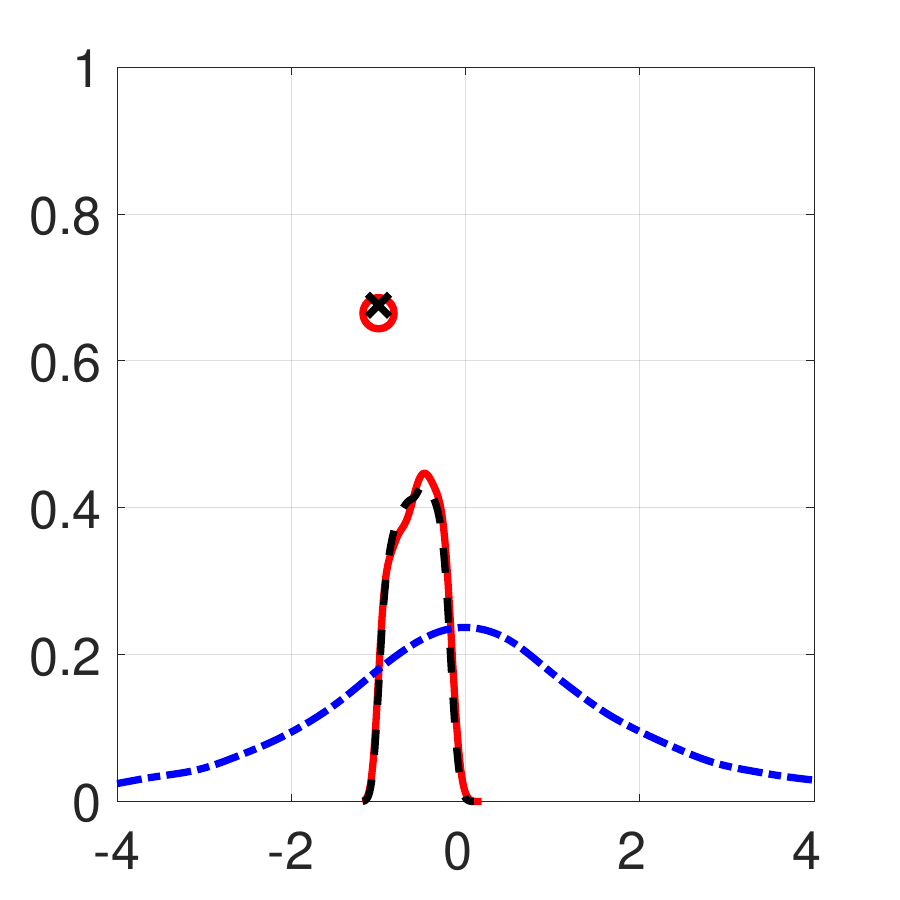}
	\end{subfigure}\begin{subfigure}{0.2\textwidth}
		\centering
		\includegraphics[trim={0cm 0cm 0.50cm 0.5cm},width=\textwidth,clip]
		{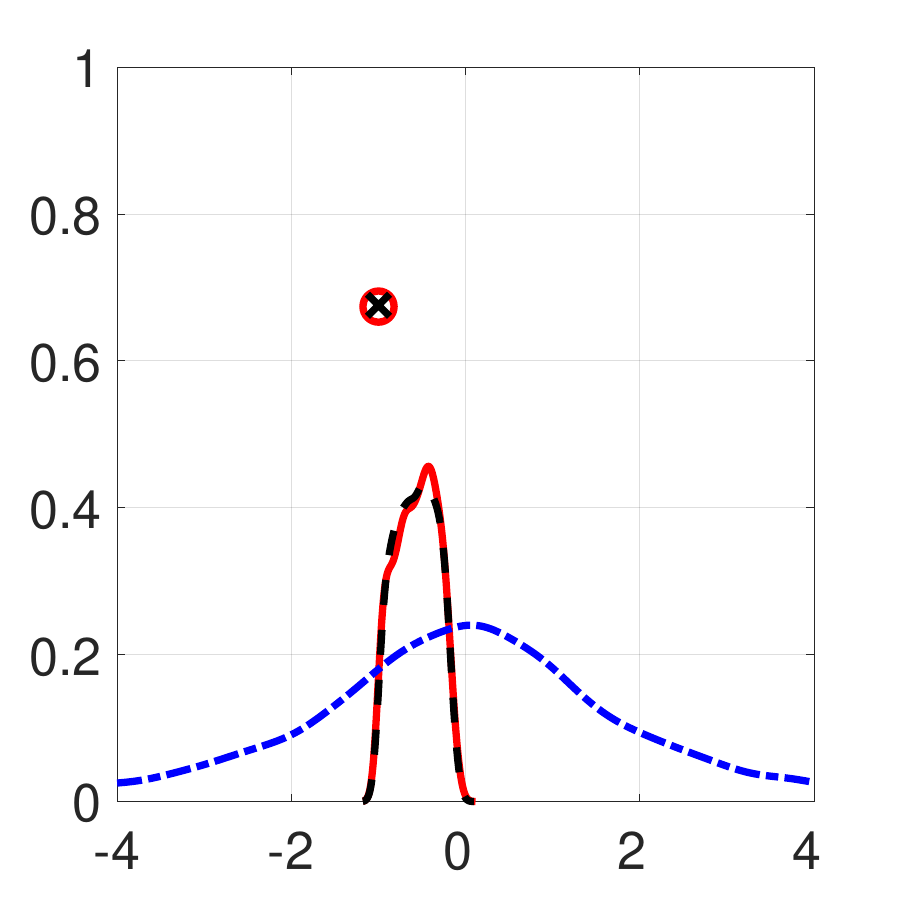}
	\end{subfigure}
	
\end{center}

\vspace{-2ex} 

\caption{Finite-sample distributions of $T(\betaAL - \beta_T)$ (under
conservative tuning, in the first row) and $\lambda_T^{-1/2}T(\betaAL
- \beta_T)$ (under consistent tuning, in the remaining rows) in case
$\beta_T = \beta/T^{1/2}$ (labeled ``AL''), and case-specific limiting
distribution from Theorem~\ref{thm:ls_dist-unif}, evaluated at sample
counterparts of limiting parameters (labeled ``Thm.3''). \emph{Notes}:
See notes to Figure~\ref{fig:densities_thm3_1}.}

\label{fig:densities_thm3_2}

\end{figure}

\begin{figure}[ht]
\begin{center}
	\caption*{$\lambda_T \equiv 1$}
	\begin{subfigure}{0.2\textwidth}
		\centering
		\caption*{$T = 25$}
		\vspace{-1.5ex}
		\includegraphics[trim={0cm 0cm 0.50cm 0.5cm},width=\textwidth,clip]
		{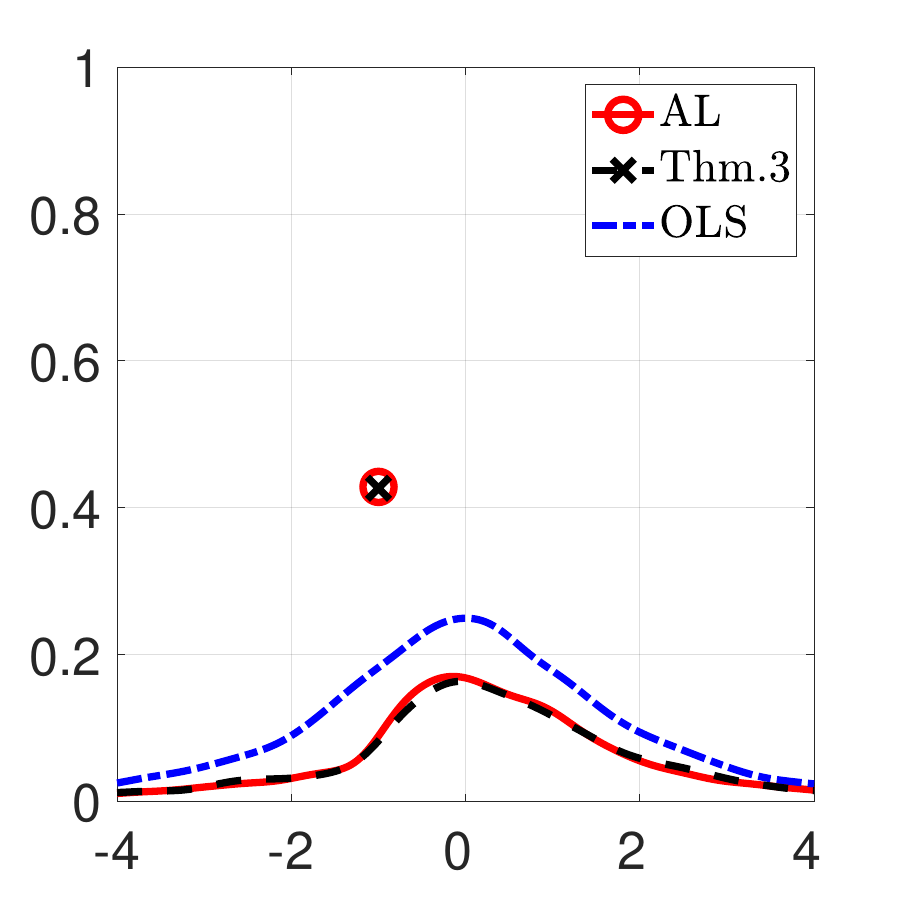}
	\end{subfigure}\begin{subfigure}{0.2\textwidth}
		\centering
		\caption*{$T = 50$}
		\vspace{-1.5ex}
		\includegraphics[trim={0cm 0cm 0.50cm 0.5cm},width=\textwidth,clip]
		{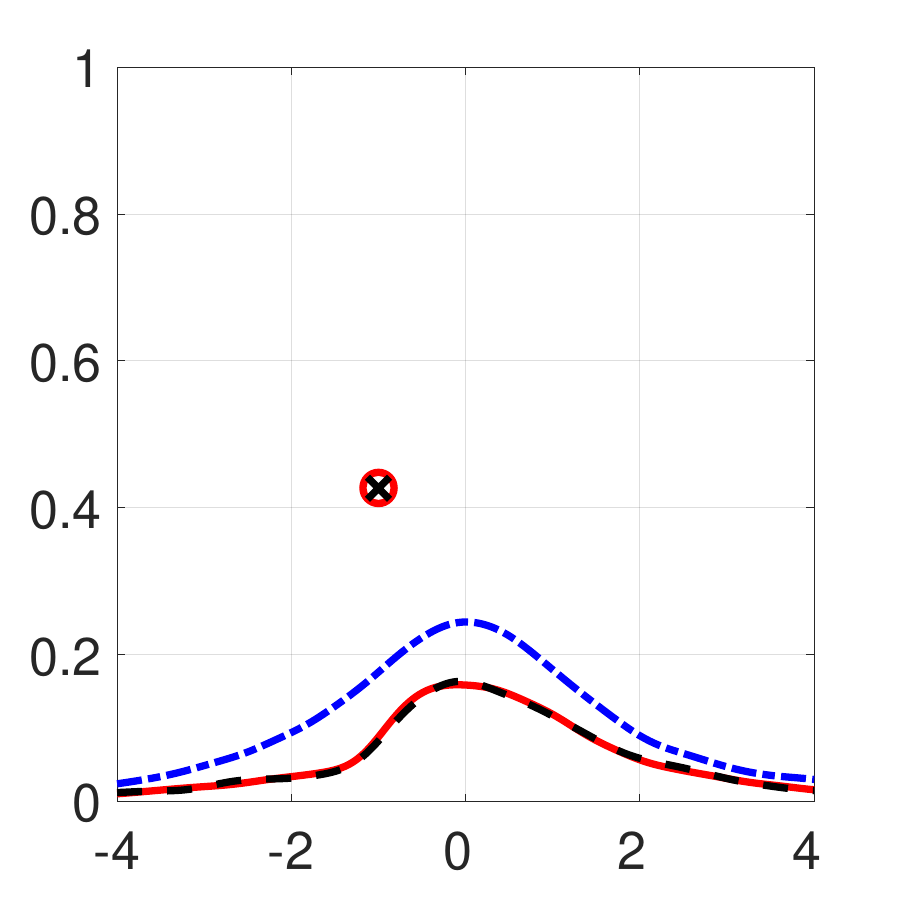}
	\end{subfigure}\begin{subfigure}{0.2\textwidth}
		\centering
		\caption*{$T = 100$}
		\vspace{-1.5ex}
		\includegraphics[trim={0cm 0cm 0.50cm 0.5cm},width=\textwidth,clip]
		{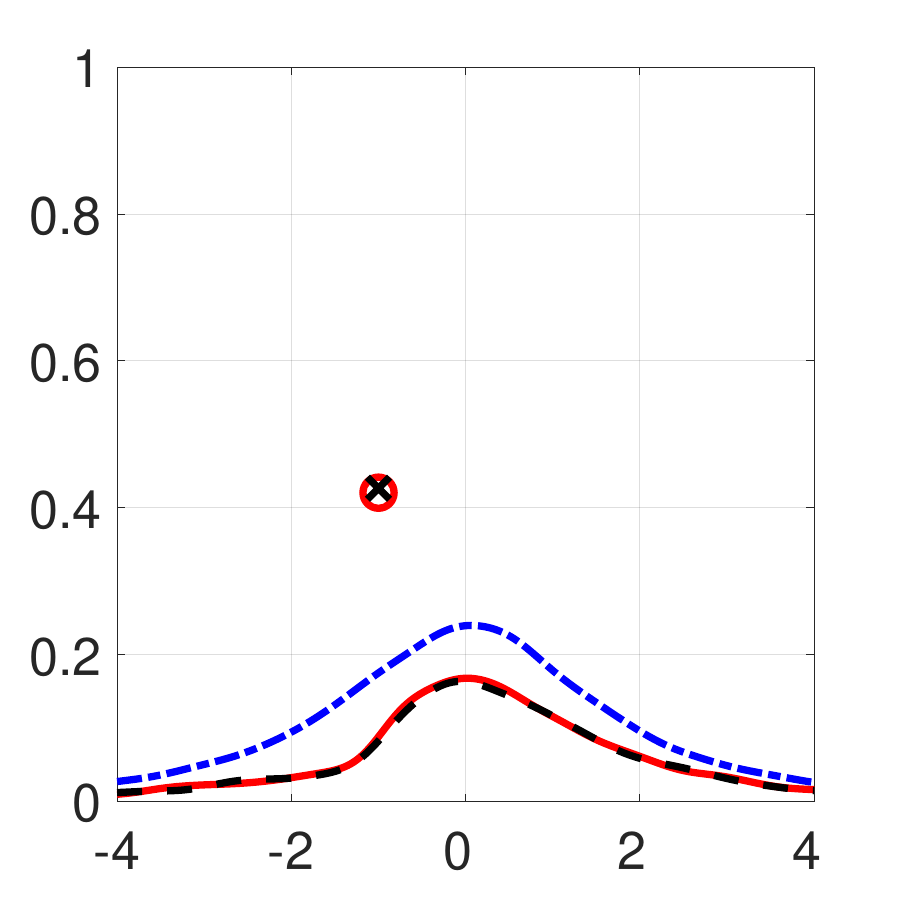}
	\end{subfigure}\begin{subfigure}{0.2\textwidth}
		\centering
		\caption*{$T = 250$}
		\vspace{-1.5ex}
		\includegraphics[trim={0cm 0cm 0.50cm 0.5cm},width=\textwidth,clip]
		{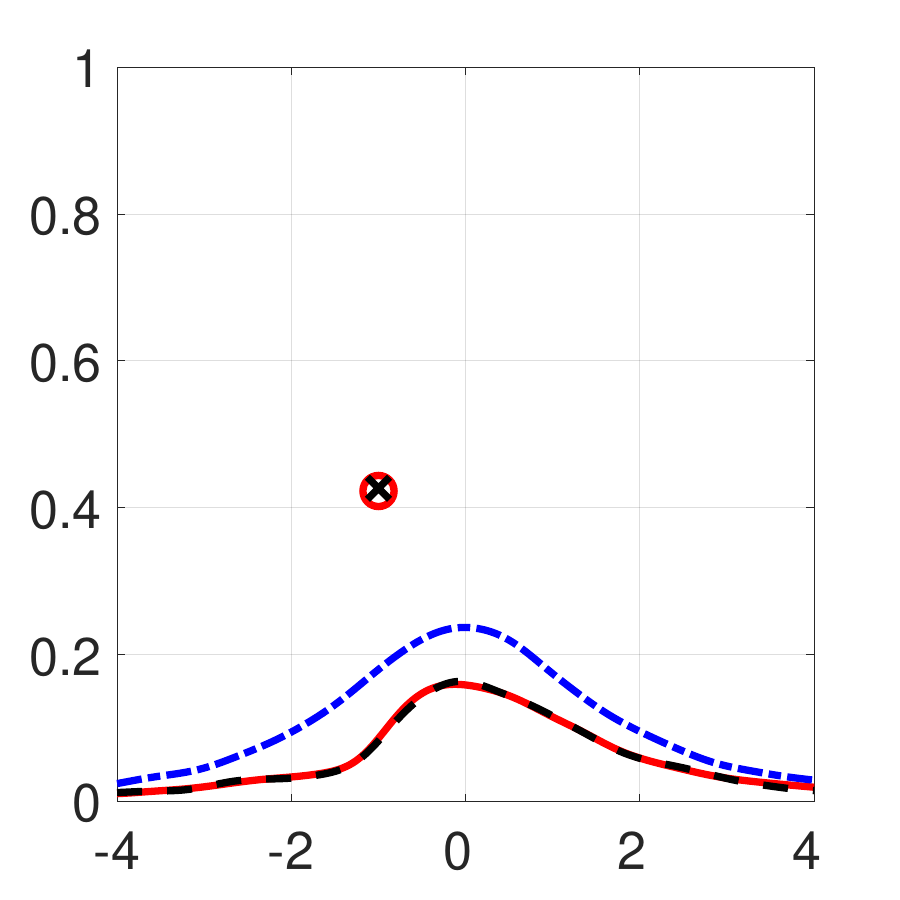}
	\end{subfigure}\begin{subfigure}{0.2\textwidth}
		\centering
		\caption*{$T = 1000$}
		\vspace{-1.5ex}
		\includegraphics[trim={0cm 0cm 0.50cm 0.5cm},width=\textwidth,clip]
		{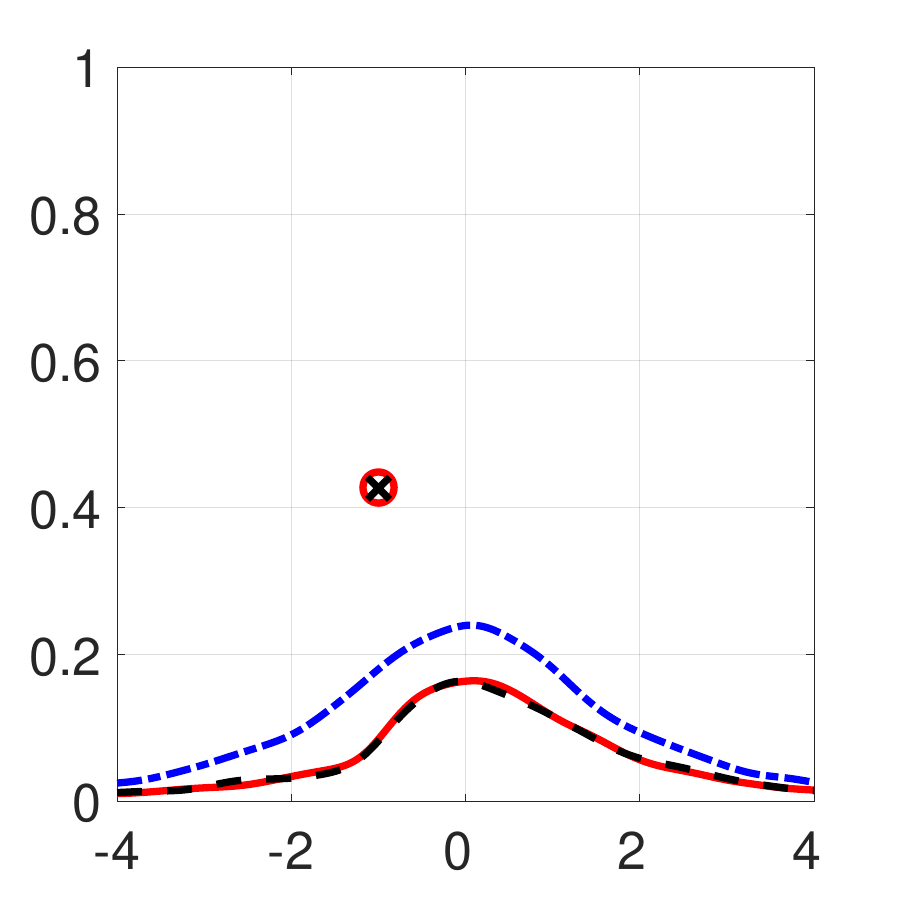}
	\end{subfigure}
	
	\caption*{$\lambda_T = T^{1/4}$}
	\vspace{-1.5ex}
	\begin{subfigure}{0.2\textwidth}
		\centering
		\includegraphics[trim={0cm 0cm 0.50cm 0.5cm},width=\textwidth,clip]
		{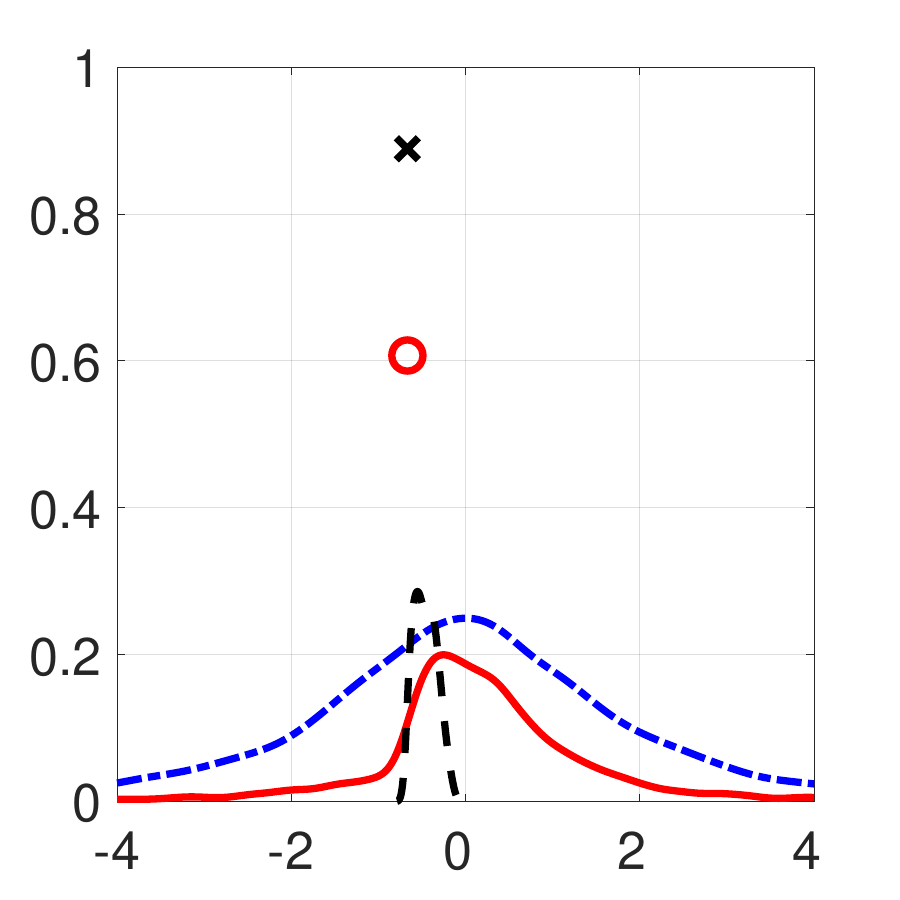}
	\end{subfigure}\begin{subfigure}{0.2\textwidth}
		\centering
		\includegraphics[trim={0cm 0cm 0.50cm 0.5cm},width=\textwidth,clip]
		{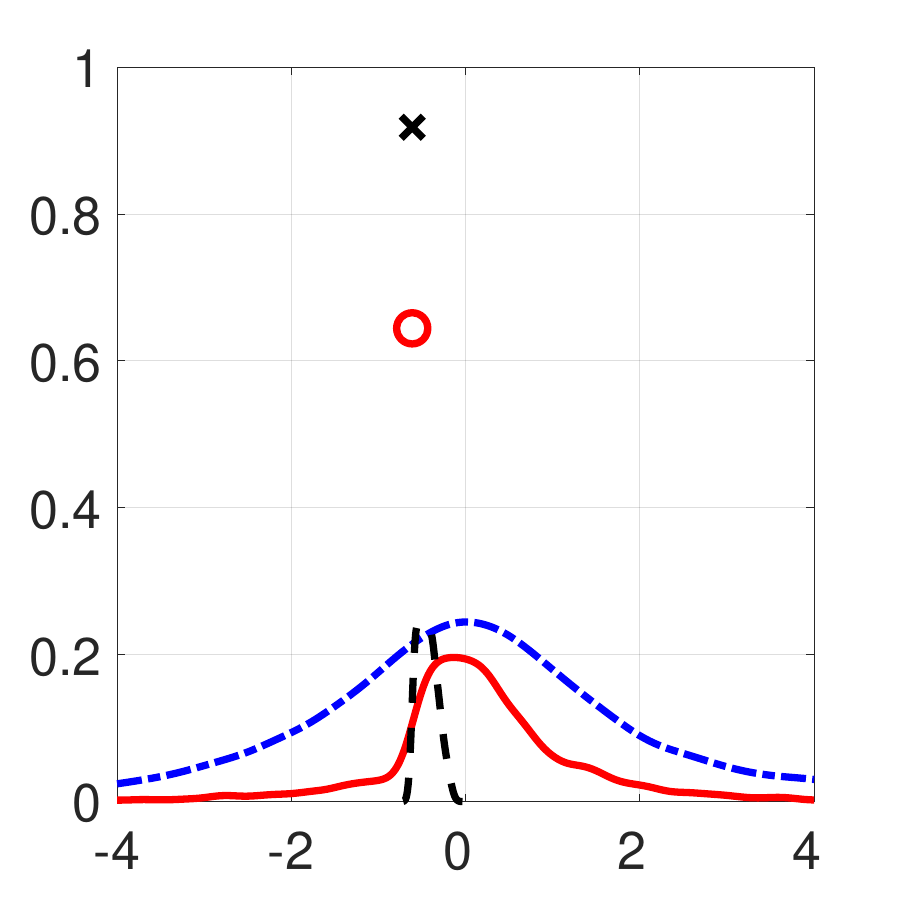}
	\end{subfigure}\begin{subfigure}{0.2\textwidth}
		\centering
		\includegraphics[trim={0cm 0cm 0.50cm 0.5cm},width=\textwidth,clip]
		{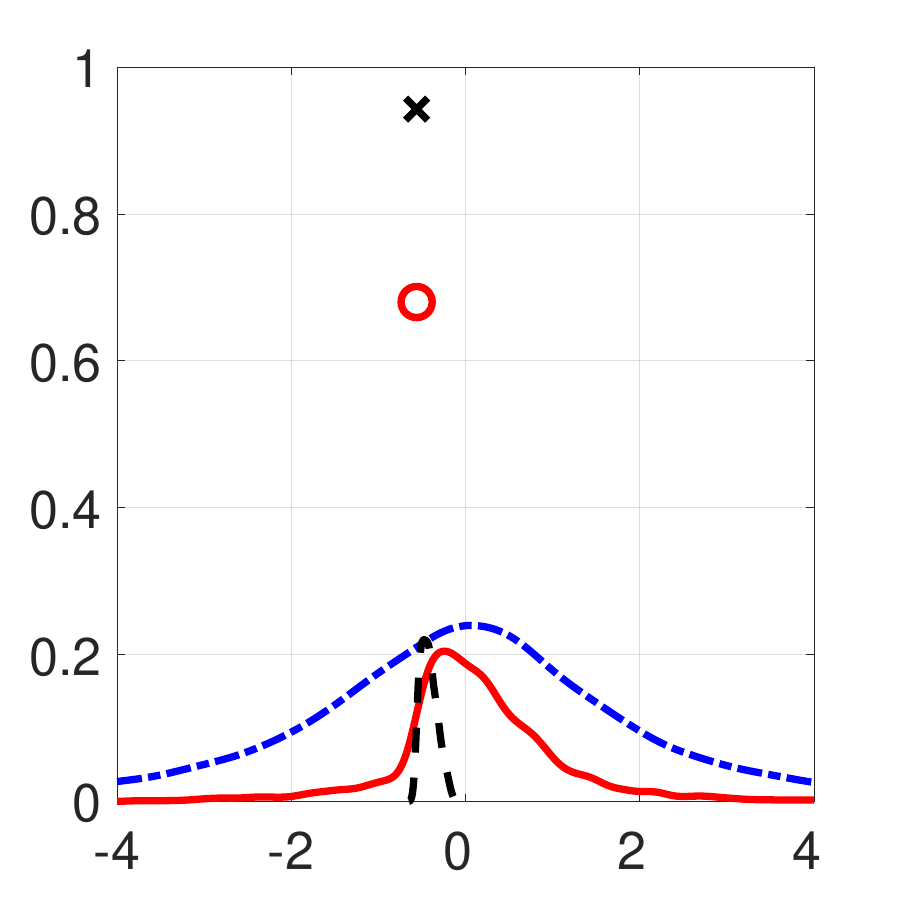}
	\end{subfigure}\begin{subfigure}{0.2\textwidth}
		\centering
		\includegraphics[trim={0cm 0cm 0.50cm 0.5cm},width=\textwidth,clip]
		{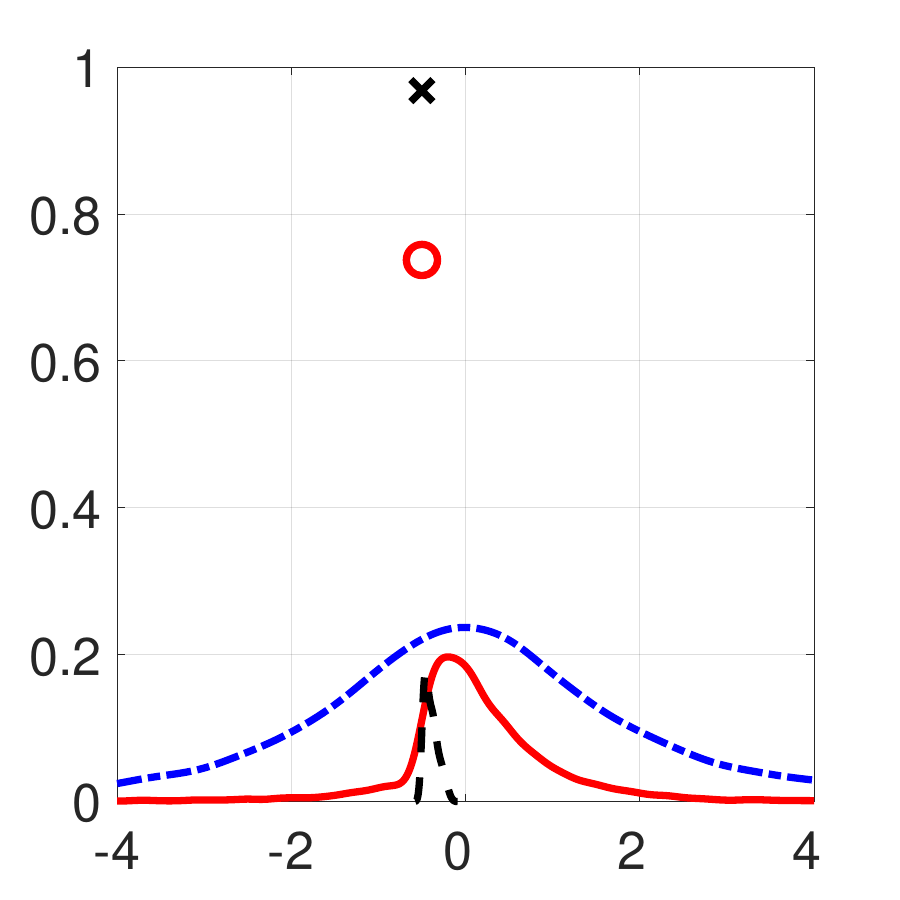}
	\end{subfigure}\begin{subfigure}{0.2\textwidth}
		\centering
		\includegraphics[trim={0cm 0cm 0.50cm 0.5cm},width=\textwidth,clip]
		{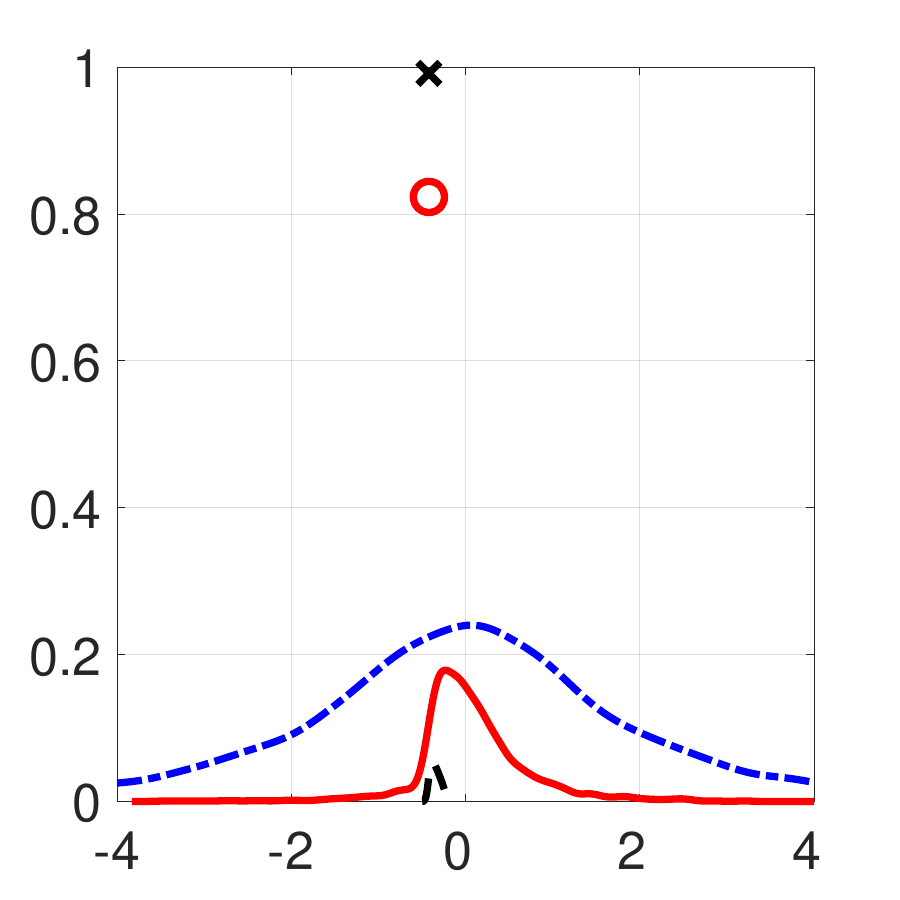}
	\end{subfigure}
	
	\caption*{$\lambda_T = T^{1/2}$}
	\vspace{-1.5ex}
	\begin{subfigure}{0.2\textwidth}
		\centering
		\includegraphics[trim={0cm 0cm 0.50cm 0.5cm},width=\textwidth,clip]
		{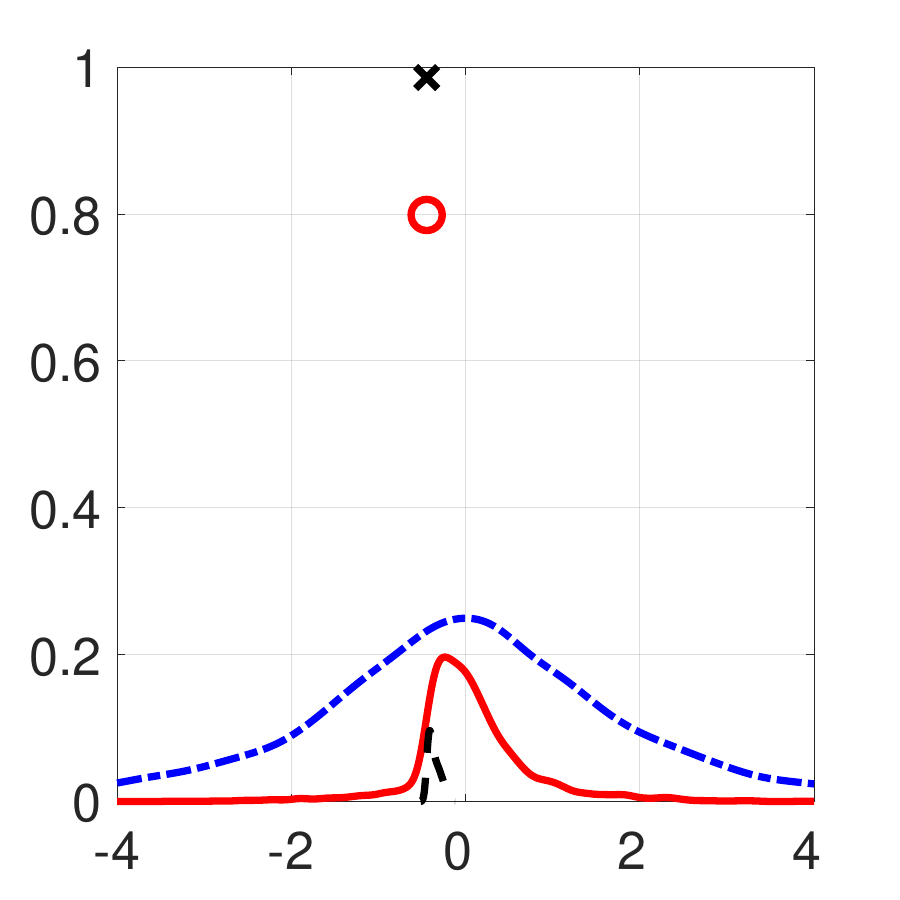}
	\end{subfigure}\begin{subfigure}{0.2\textwidth}
		\centering
		\includegraphics[trim={0cm 0cm 0.50cm 0.5cm},width=\textwidth,clip]
		{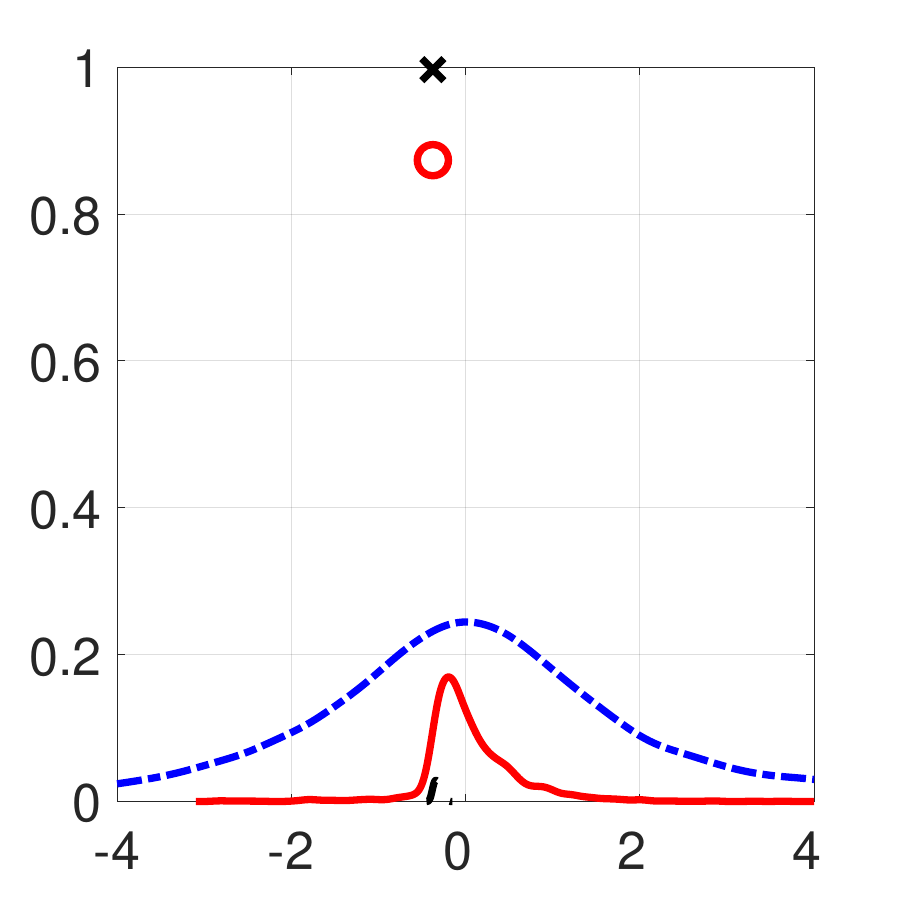}
	\end{subfigure}\begin{subfigure}{0.2\textwidth}
		\centering
		\includegraphics[trim={0cm 0cm 0.50cm 0.5cm},width=\textwidth,clip]
		{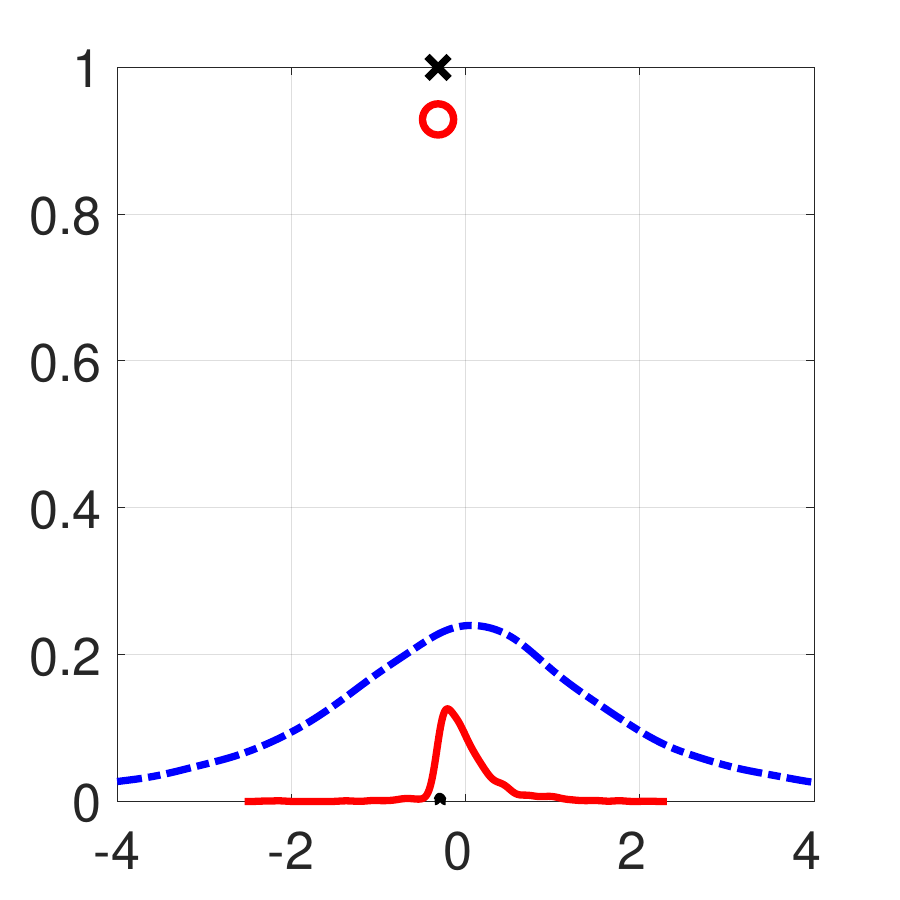}
	\end{subfigure}\begin{subfigure}{0.2\textwidth}
		\centering
		\includegraphics[trim={0cm 0cm 0.50cm 0.5cm},width=\textwidth,clip]
		{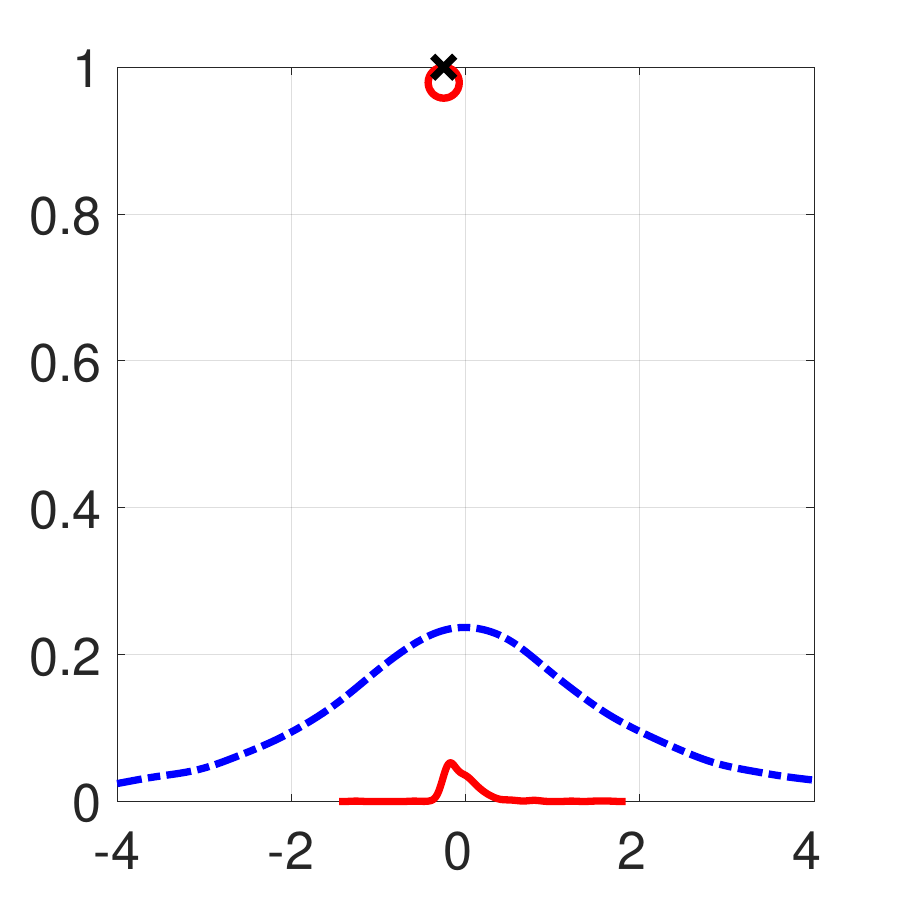}
	\end{subfigure}\begin{subfigure}{0.2\textwidth}
		\centering
		\includegraphics[trim={0cm 0cm 0.50cm 0.5cm},width=\textwidth,clip]
		{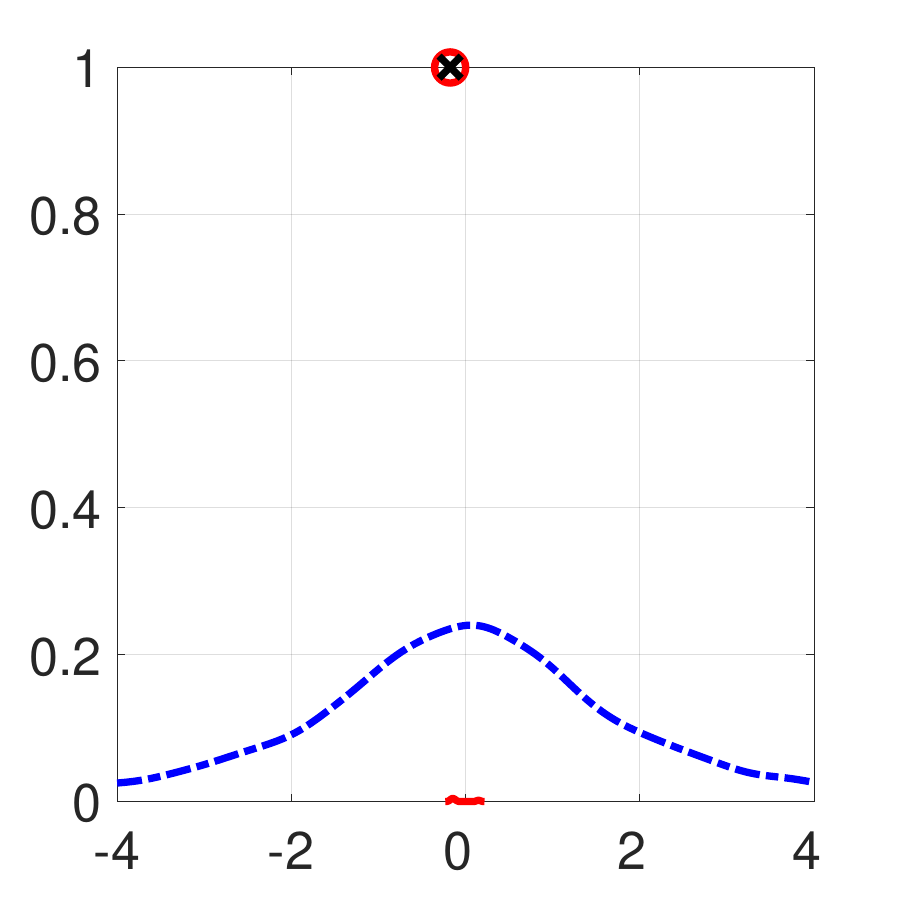}
	\end{subfigure}
	
	\caption*{$\lambda_T = T$}
	\vspace{-1.5ex}
	\begin{subfigure}{0.2\textwidth}
		\centering
		\includegraphics[trim={0cm 0cm 0.50cm 0.5cm},width=\textwidth,clip]
		{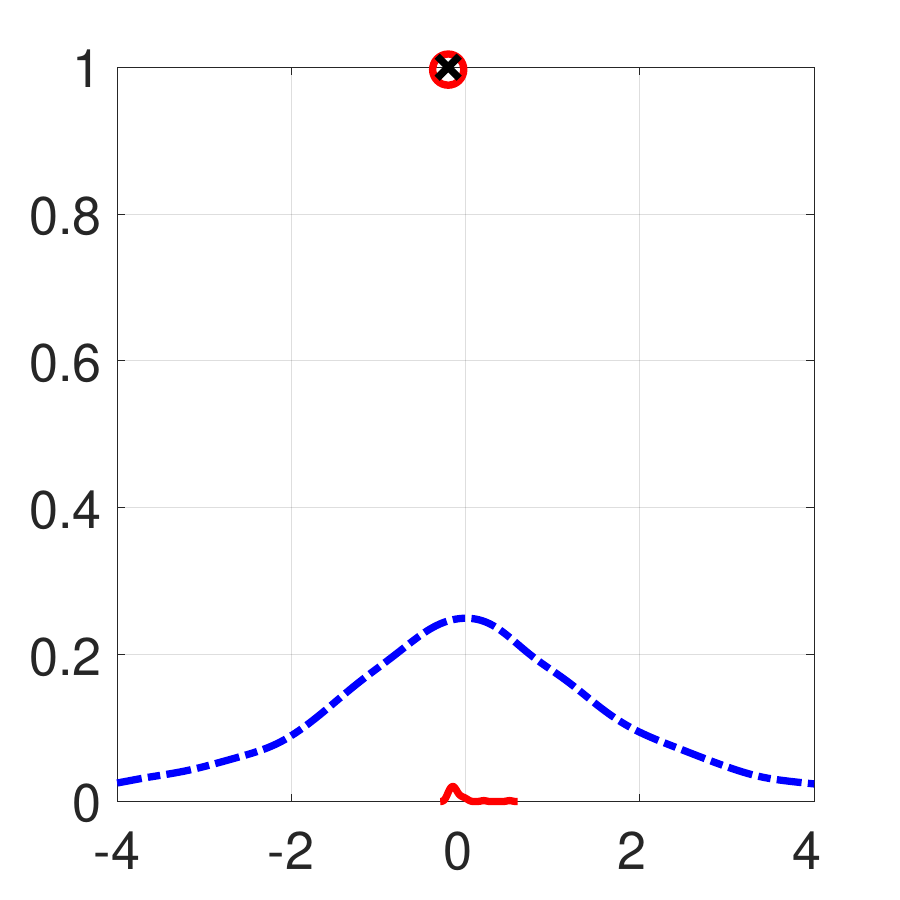}
	\end{subfigure}\begin{subfigure}{0.2\textwidth}
		\centering
		\includegraphics[trim={0cm 0cm 0.50cm 0.5cm},width=\textwidth,clip]
		{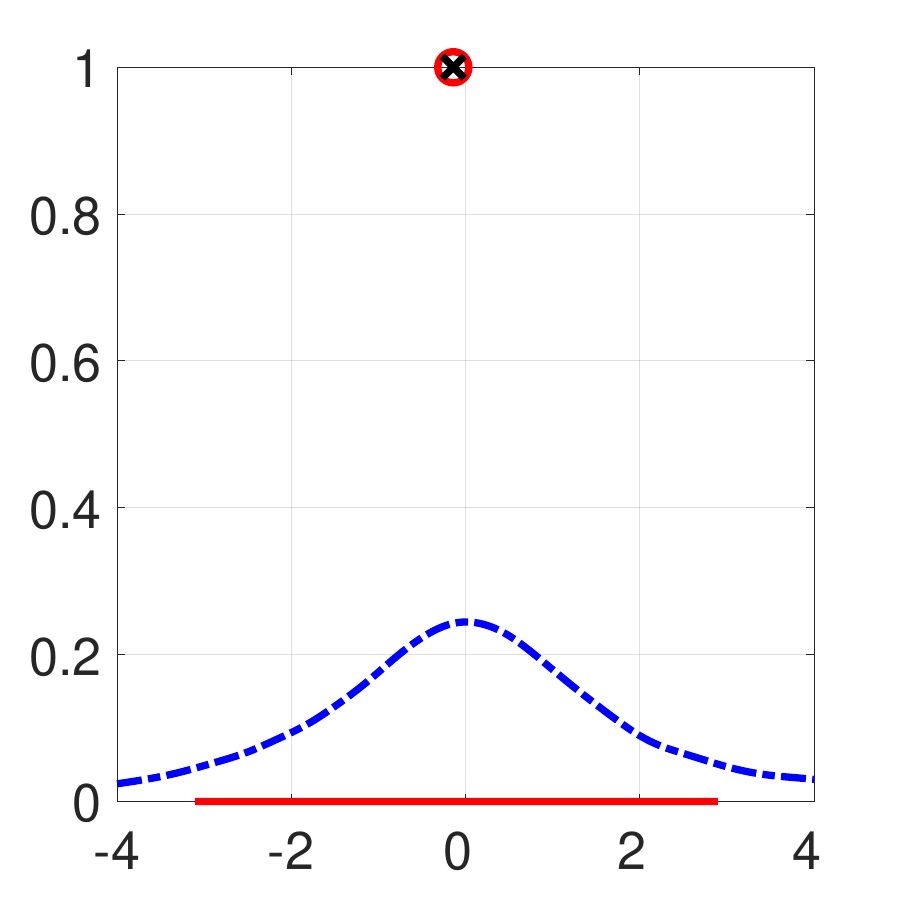}
	\end{subfigure}\begin{subfigure}{0.2\textwidth}
		\centering
		\includegraphics[trim={0cm 0cm 0.50cm 0.5cm},width=\textwidth,clip]
		{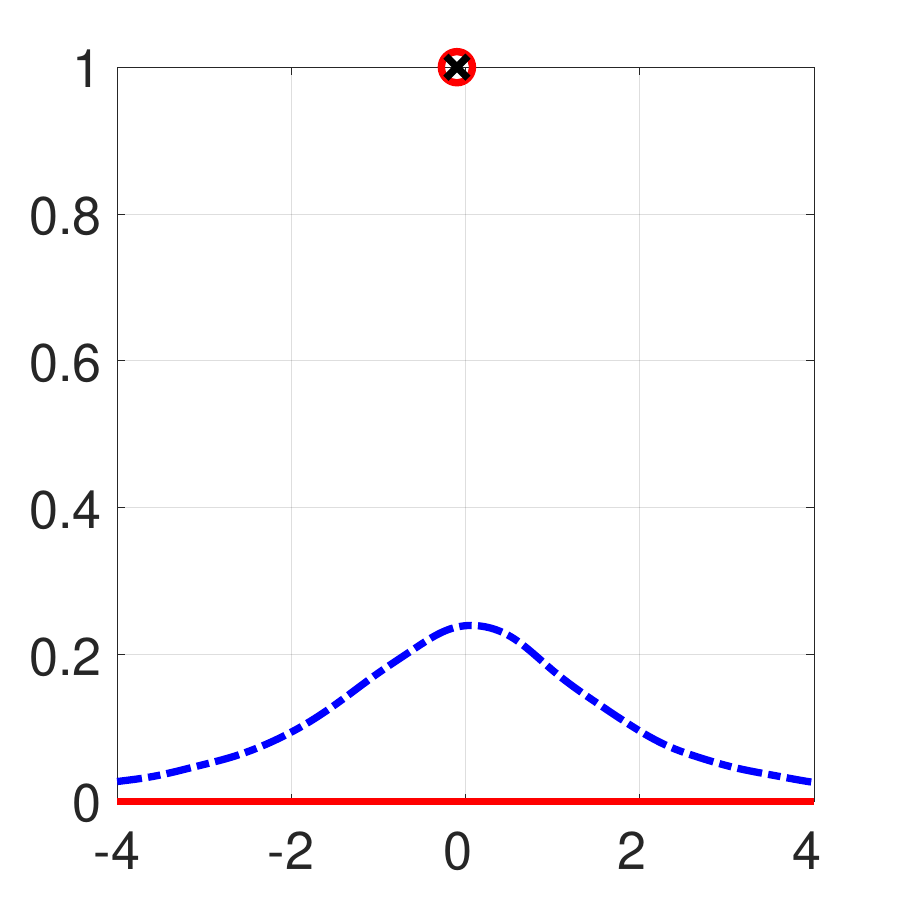}
	\end{subfigure}\begin{subfigure}{0.2\textwidth}
		\centering
		\includegraphics[trim={0cm 0cm 0.50cm 0.5cm},width=\textwidth,clip]
		{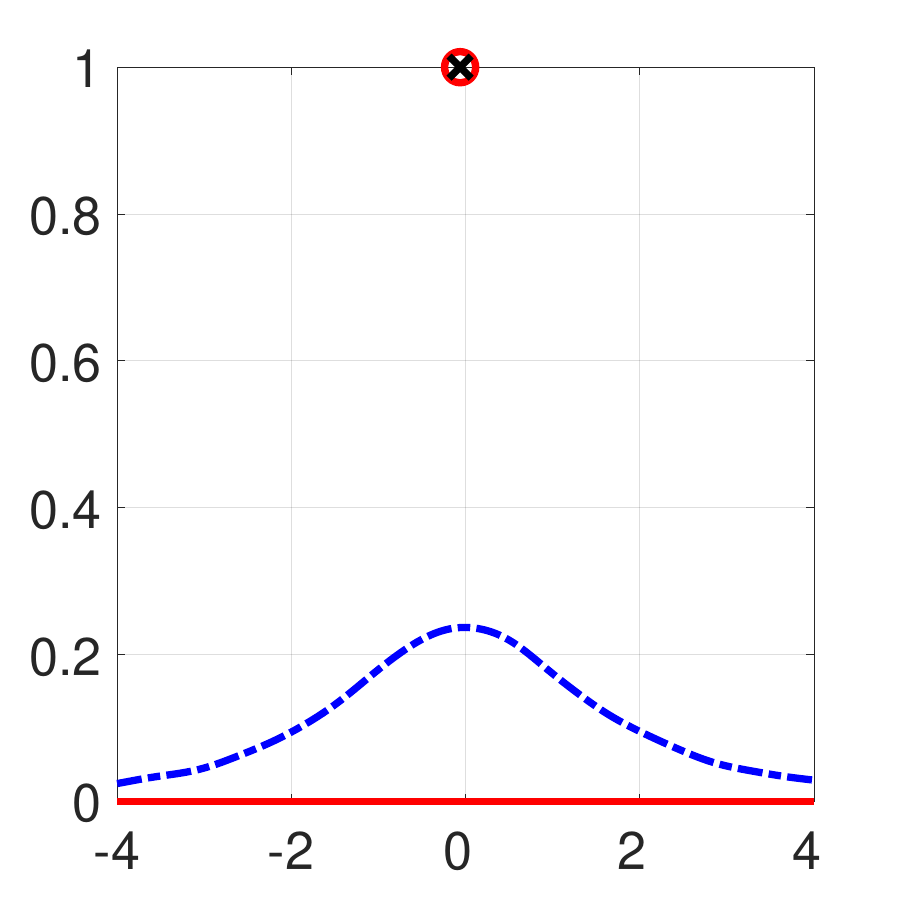}
	\end{subfigure}\begin{subfigure}{0.2\textwidth}
		\centering
		\includegraphics[trim={0cm 0cm 0.50cm 0.5cm},width=\textwidth,clip]
		{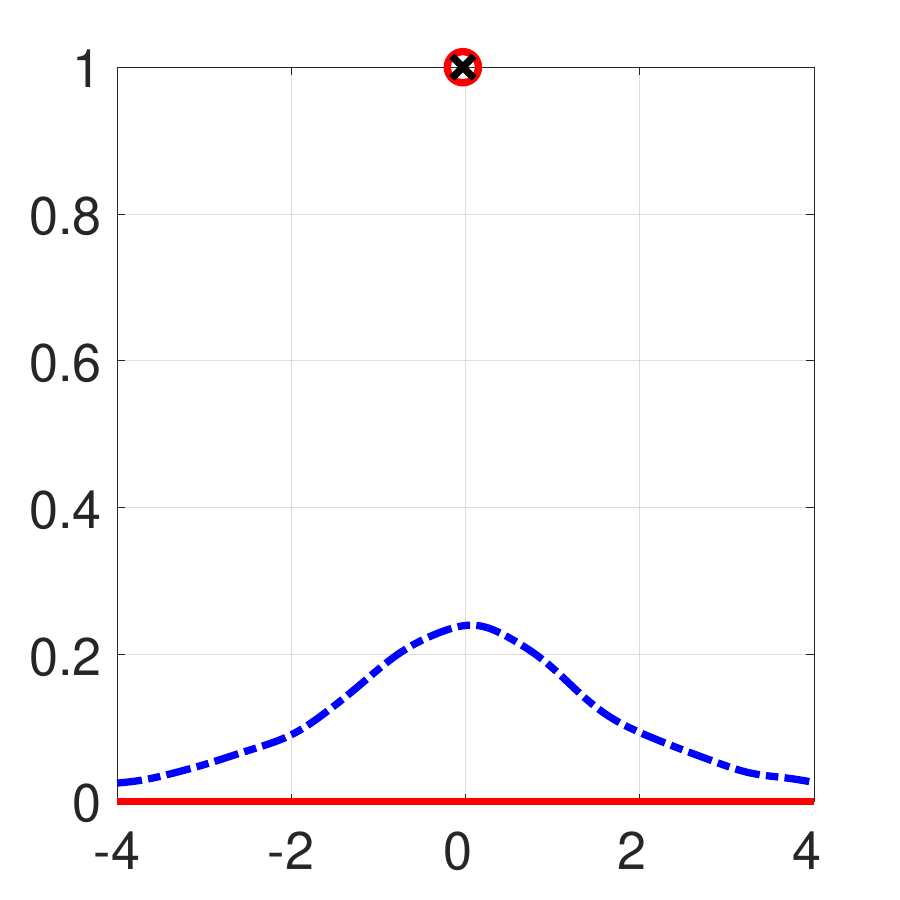}
	\end{subfigure}
	
\end{center}

\vspace{-2ex} 

\caption{Finite-sample distributions of $T(\betaAL - \beta_T)$ (under
conservative tuning, in the first row) and $\lambda_T^{-1/2}T(\betaAL
- \beta_T)$ (under consistent tuning, in the remaining rows) in case
$\beta_T = \beta/T$ (labeled ``AL''), and case-specific limiting
distribution from Theorem~\ref{thm:ls_dist-unif}, evaluated at sample
counterparts of limiting parameters (labeled ``Thm.3''). \emph{Notes}:
See notes to Figure~\ref{fig:densities_thm3_1}.}

\label{fig:densities_thm3_3}

\end{figure}

\begin{figure}[ht]
\begin{center}
	\caption*{$\lambda_T \equiv 1$}
	\begin{subfigure}{0.2\textwidth}
		\centering
		\caption*{$T = 25$}
		\vspace{-1.5ex}
		\includegraphics[trim={0cm 0cm 0.50cm 0.5cm},width=\textwidth,clip]
		{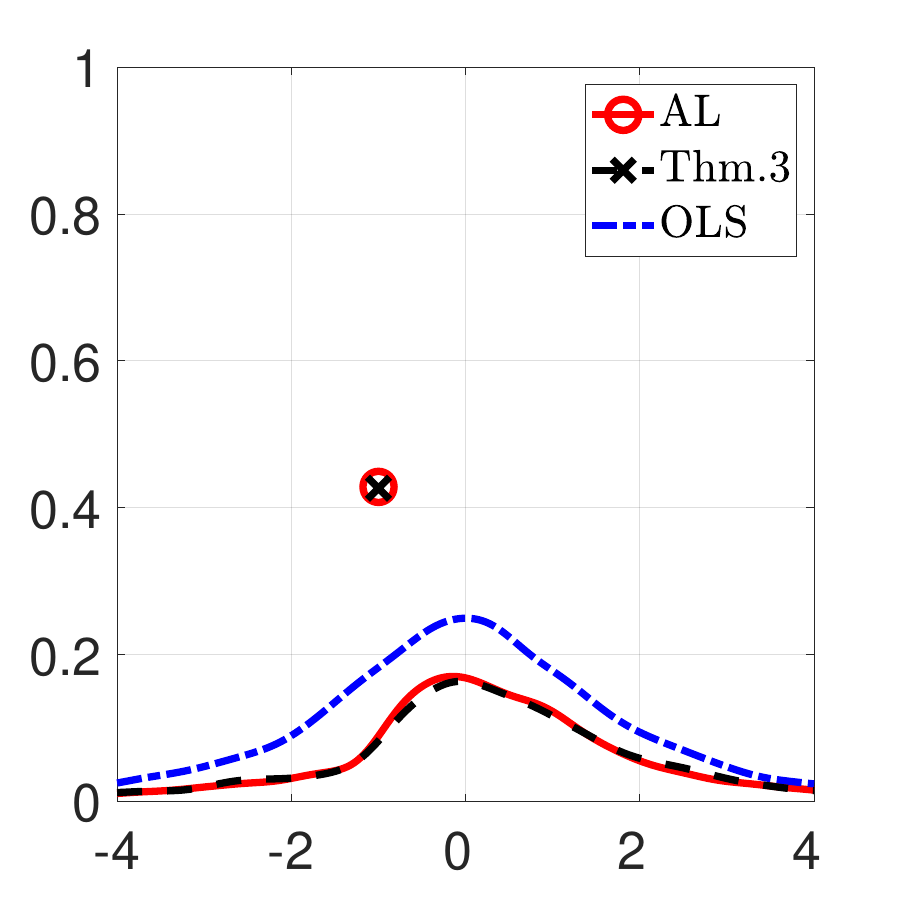}
	\end{subfigure}\begin{subfigure}{0.2\textwidth}
		\centering
		\caption*{$T = 50$}
		\vspace{-1.5ex}
		\includegraphics[trim={0cm 0cm 0.50cm 0.5cm},width=\textwidth,clip]
		{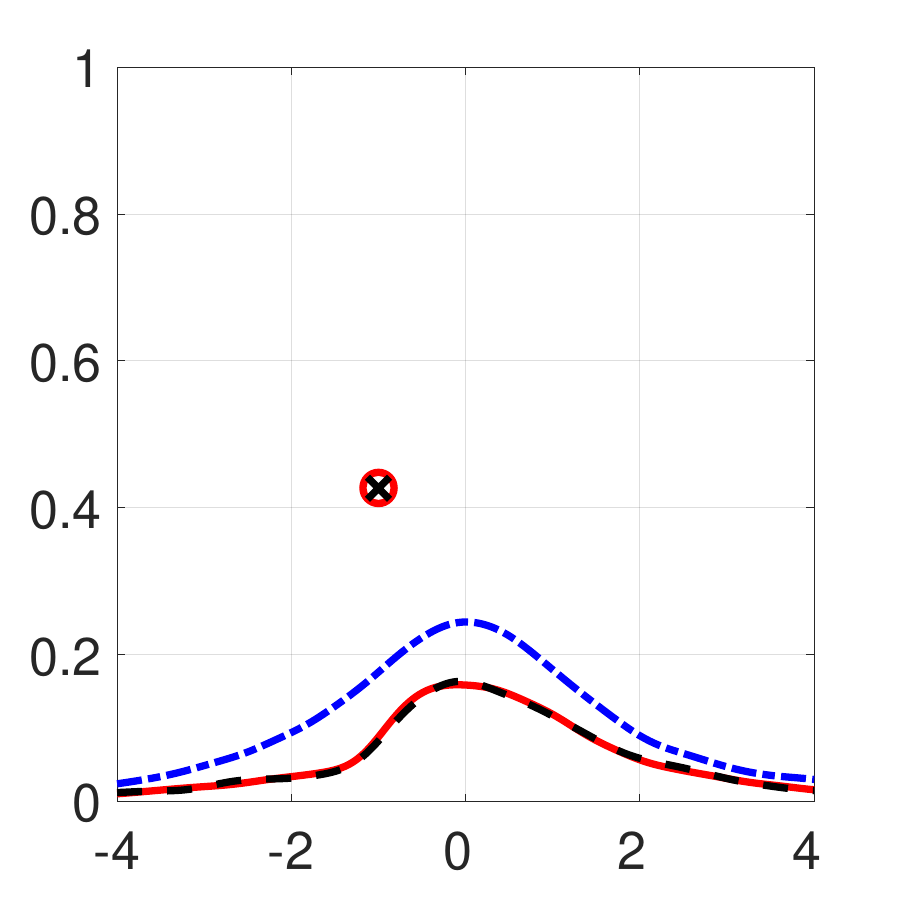}
	\end{subfigure}\begin{subfigure}{0.2\textwidth}
		\centering
		\caption*{$T = 100$}
		\vspace{-1.5ex}
		\includegraphics[trim={0cm 0cm 0.50cm 0.5cm},width=\textwidth,clip]
		{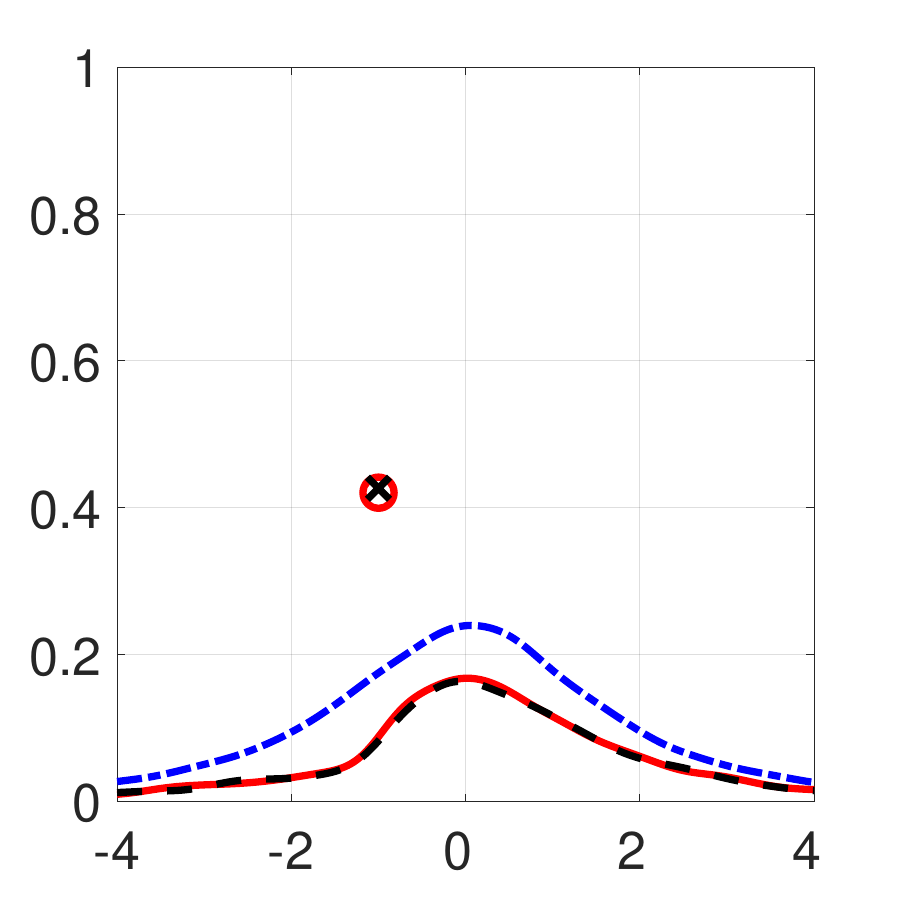}
	\end{subfigure}\begin{subfigure}{0.2\textwidth}
		\centering
		\caption*{$T = 250$}
		\vspace{-1.5ex}
		\includegraphics[trim={0cm 0cm 0.50cm 0.5cm},width=\textwidth,clip]
		{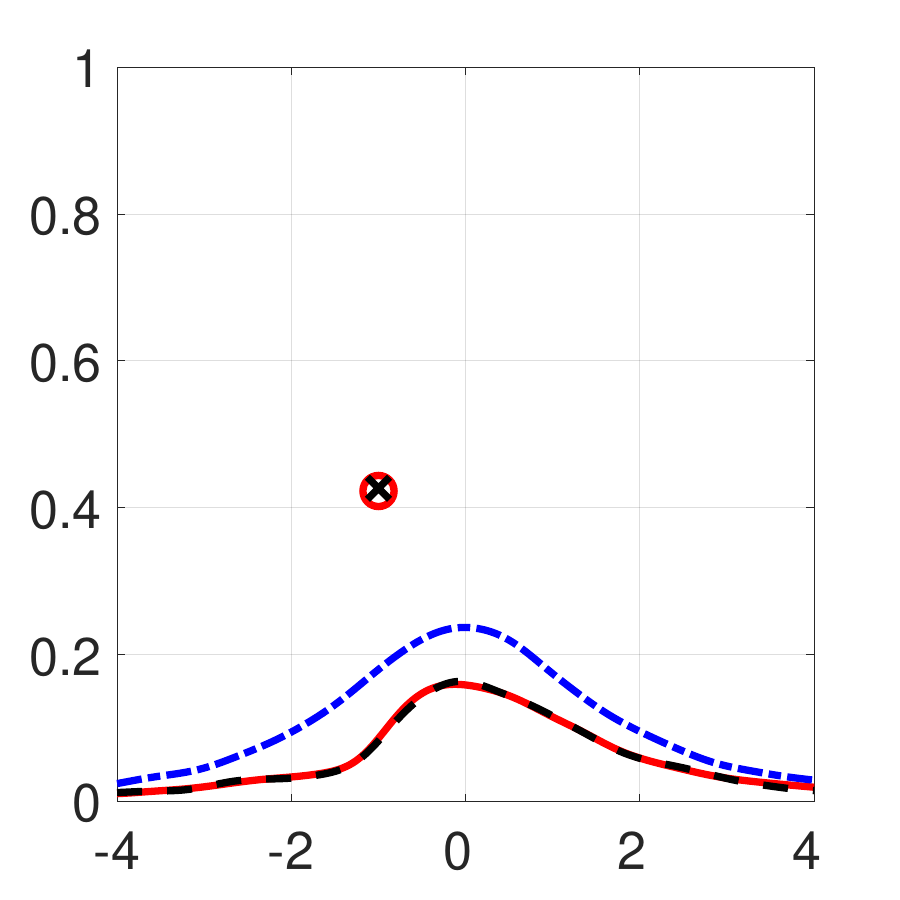}
	\end{subfigure}\begin{subfigure}{0.2\textwidth}
		\centering
		\caption*{$T = 1000$}
		\vspace{-1.5ex}
		\includegraphics[trim={0cm 0cm 0.50cm 0.5cm},width=\textwidth,clip]
		{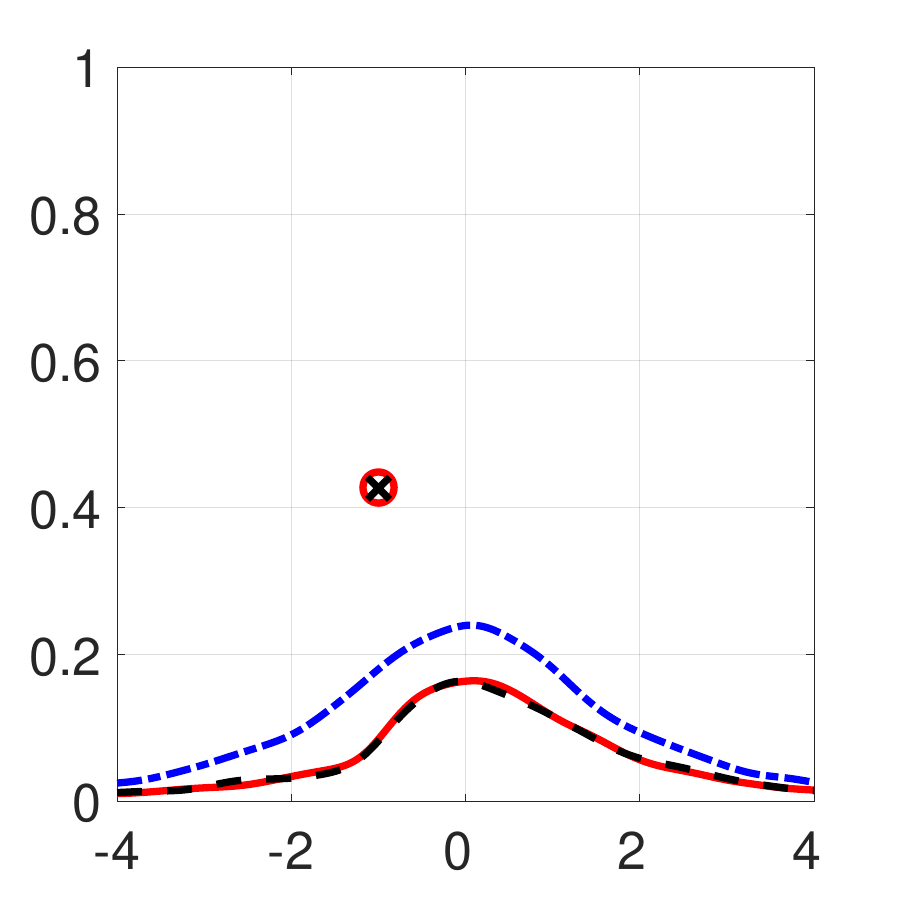}
	\end{subfigure}
	
	\caption*{$\lambda_T = T^{1/4}$}
	\vspace{-1.5ex}
	\begin{subfigure}{0.2\textwidth}
		\centering
		\includegraphics[trim={0cm 0cm 0.50cm 0.5cm},width=\textwidth,clip]
		{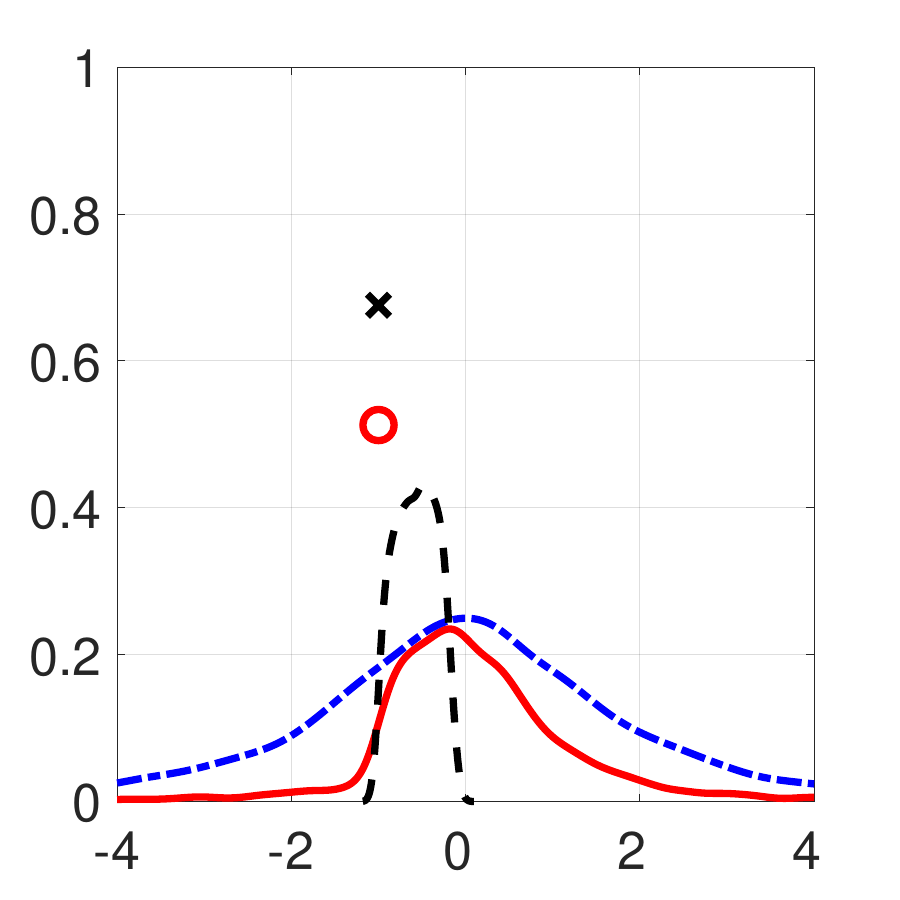}
	\end{subfigure}\begin{subfigure}{0.2\textwidth}
		\centering
		\includegraphics[trim={0cm 0cm 0.50cm 0.5cm},width=\textwidth,clip]
		{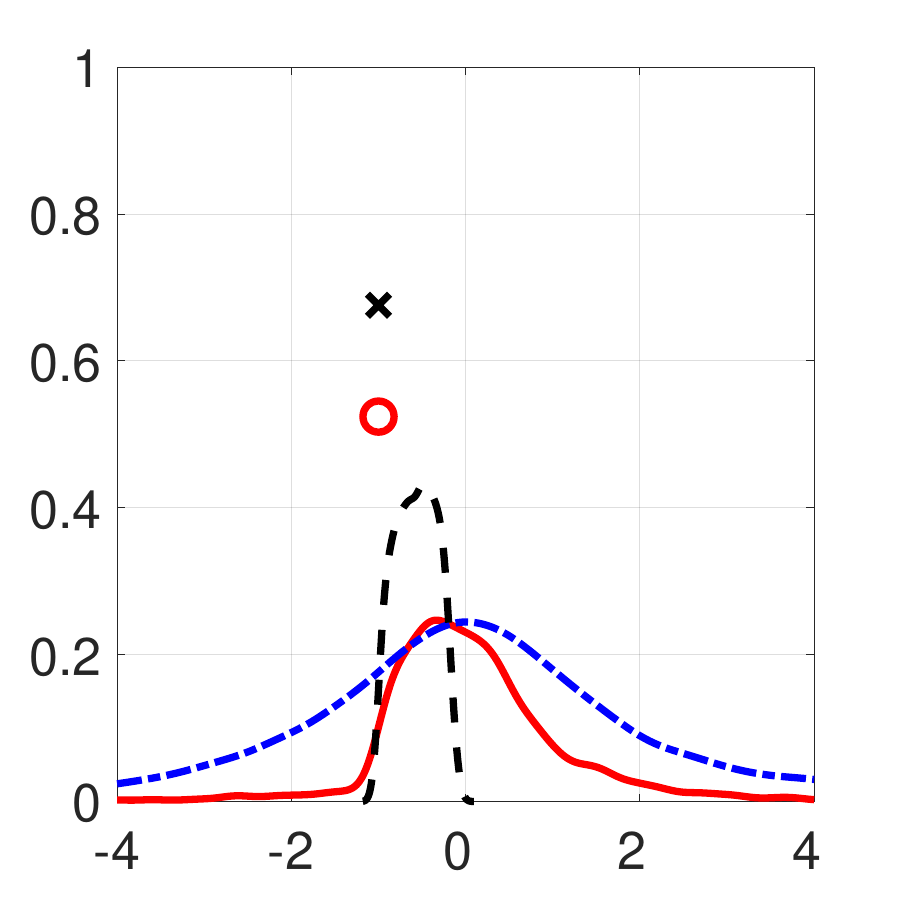}
	\end{subfigure}\begin{subfigure}{0.2\textwidth}
		\centering
		\includegraphics[trim={0cm 0cm 0.50cm 0.5cm},width=\textwidth,clip]
		{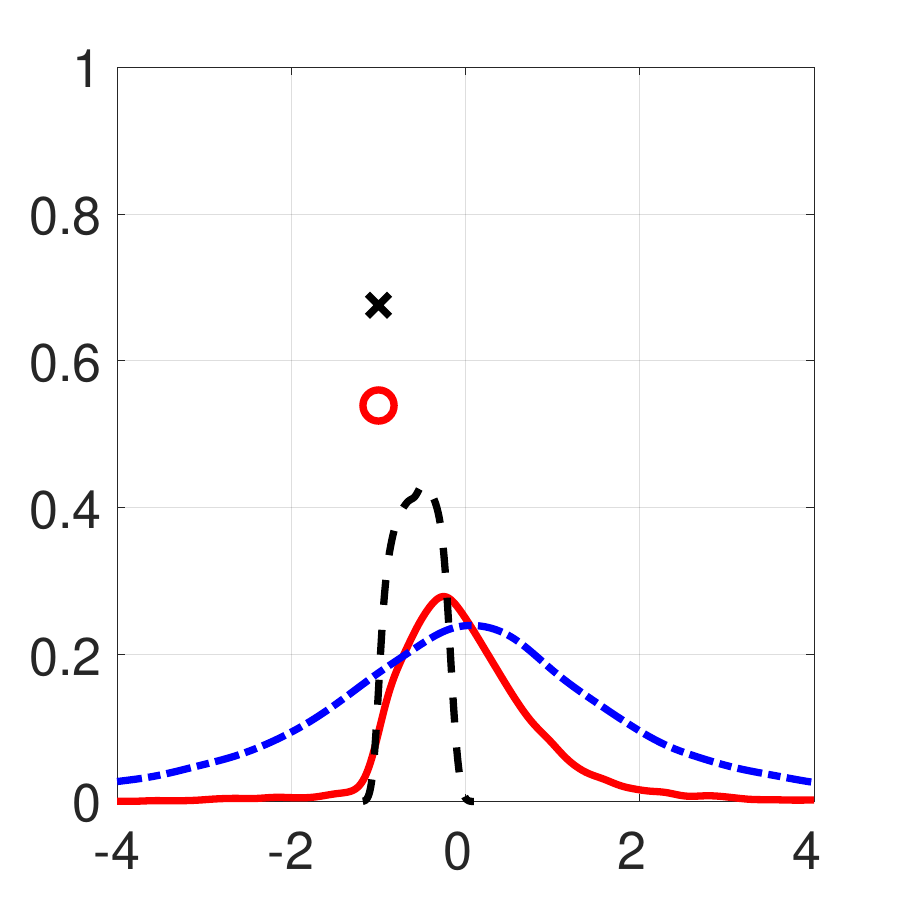}
	\end{subfigure}\begin{subfigure}{0.2\textwidth}
		\centering
		\includegraphics[trim={0cm 0cm 0.50cm 0.5cm},width=\textwidth,clip]
		{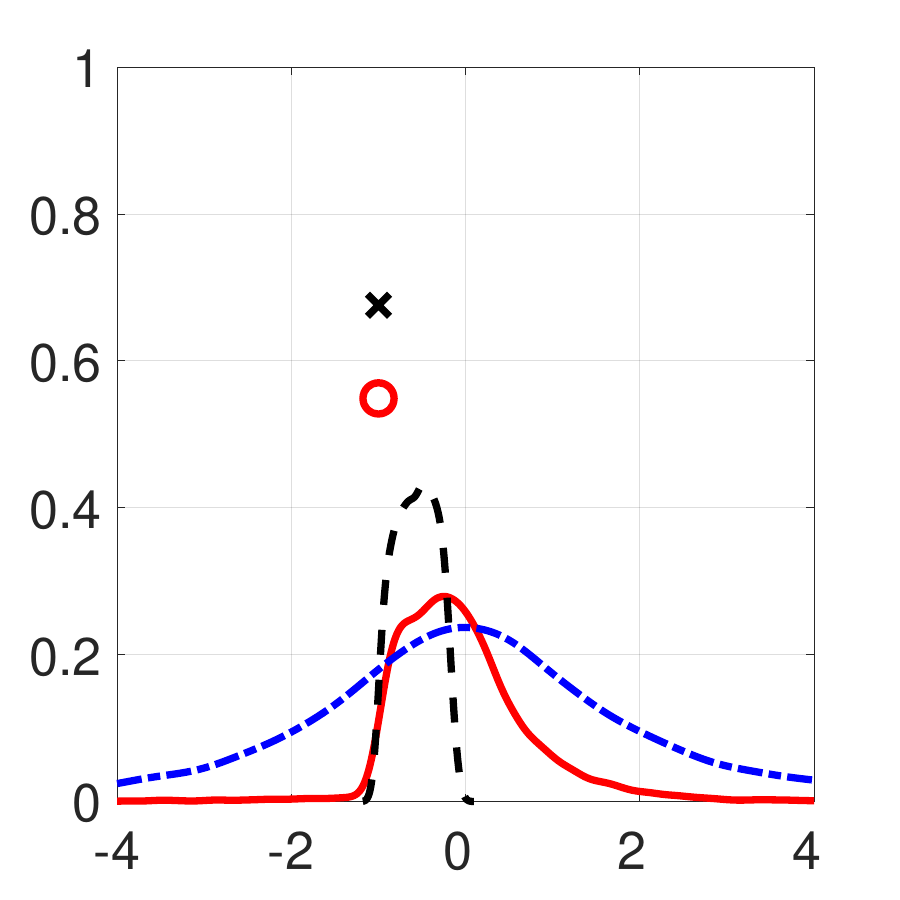}
	\end{subfigure}\begin{subfigure}{0.2\textwidth}
		\centering
		\includegraphics[trim={0cm 0cm 0.50cm 0.5cm},width=\textwidth,clip]
		{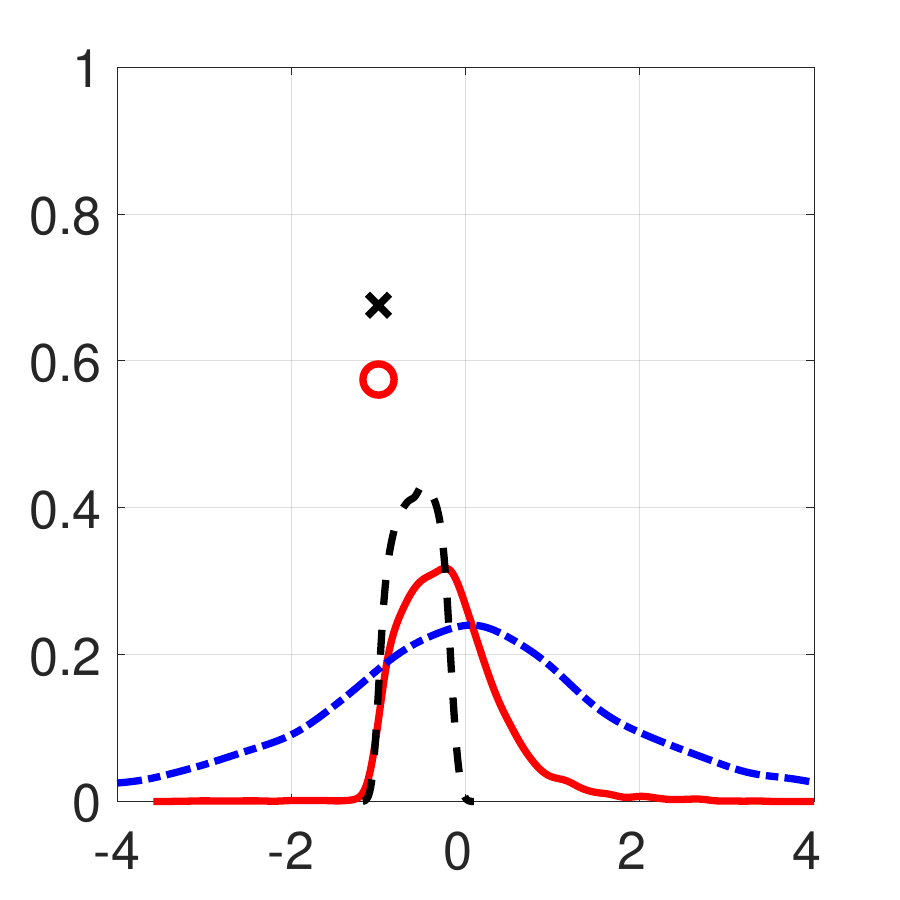}
	\end{subfigure}
	
	\caption*{$\lambda_T = T^{1/2}$}
	\vspace{-1.5ex}
	\begin{subfigure}{0.2\textwidth}
		\centering
		\includegraphics[trim={0cm 0cm 0.50cm 0.5cm},width=\textwidth,clip]
		{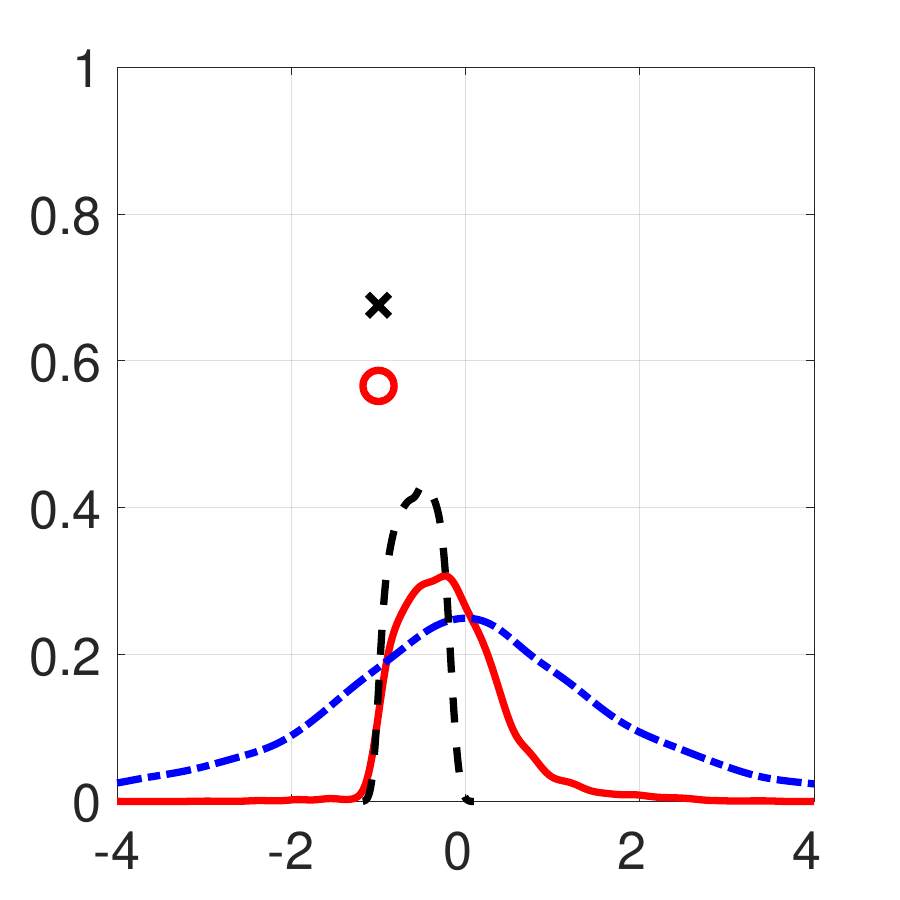}
	\end{subfigure}\begin{subfigure}{0.2\textwidth}
		\centering
		\includegraphics[trim={0cm 0cm 0.50cm 0.5cm},width=\textwidth,clip]
		{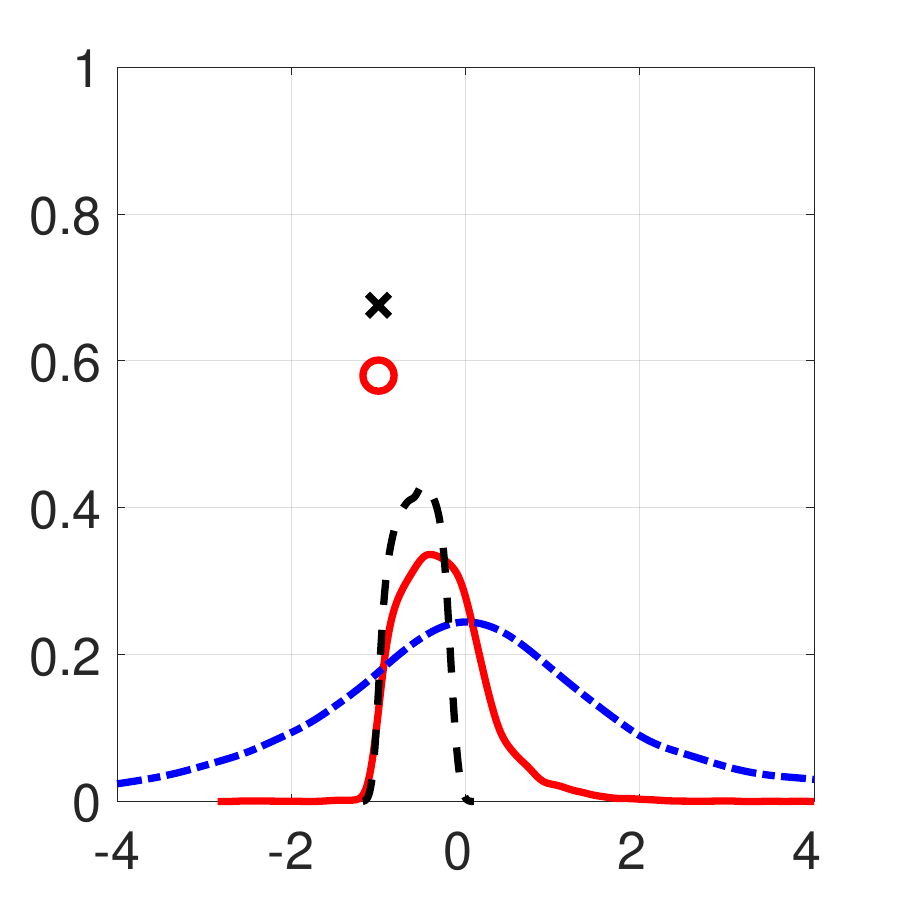}
	\end{subfigure}\begin{subfigure}{0.2\textwidth}
		\centering
		\includegraphics[trim={0cm 0cm 0.50cm 0.5cm},width=\textwidth,clip]
		{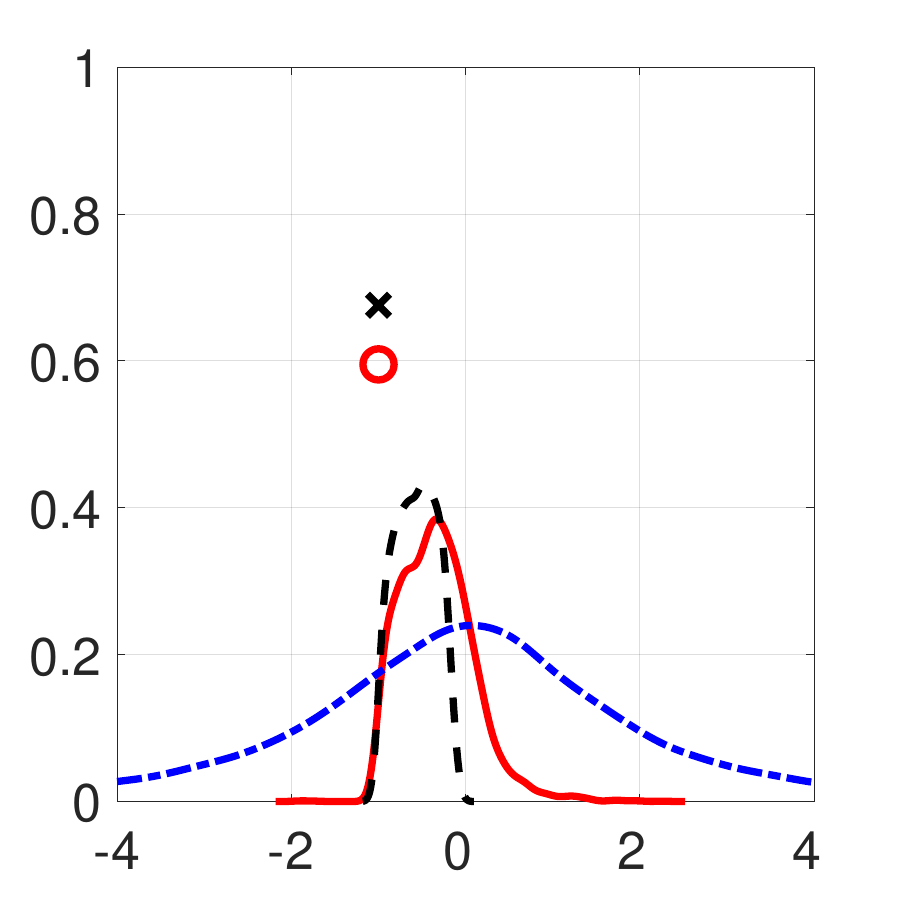}
	\end{subfigure}\begin{subfigure}{0.2\textwidth}
		\centering
		\includegraphics[trim={0cm 0cm 0.50cm 0.5cm},width=\textwidth,clip]
		{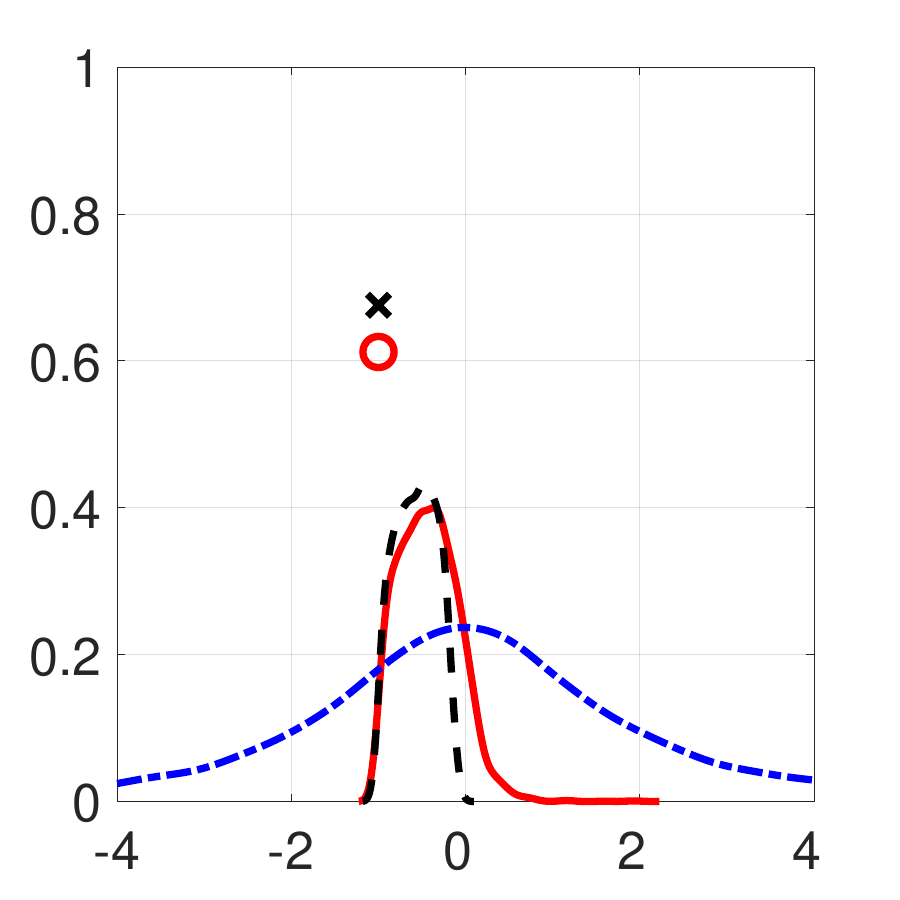}
	\end{subfigure}\begin{subfigure}{0.2\textwidth}
		\centering
		\includegraphics[trim={0cm 0cm 0.50cm 0.5cm},width=\textwidth,clip]
		{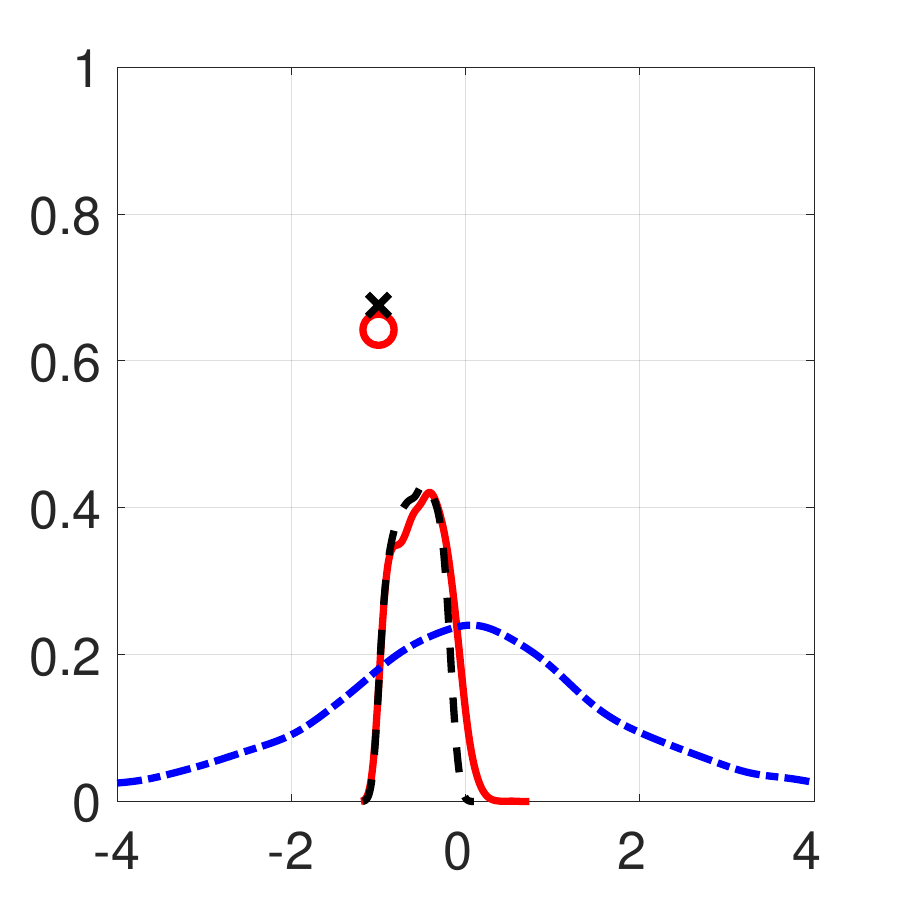}
	\end{subfigure}
	
	\caption*{$\lambda_T = T$}
	\vspace{-1.5ex}
	\begin{subfigure}{0.2\textwidth}
		\centering
		\includegraphics[trim={0cm 0cm 0.50cm 0.5cm},width=\textwidth,clip]
		{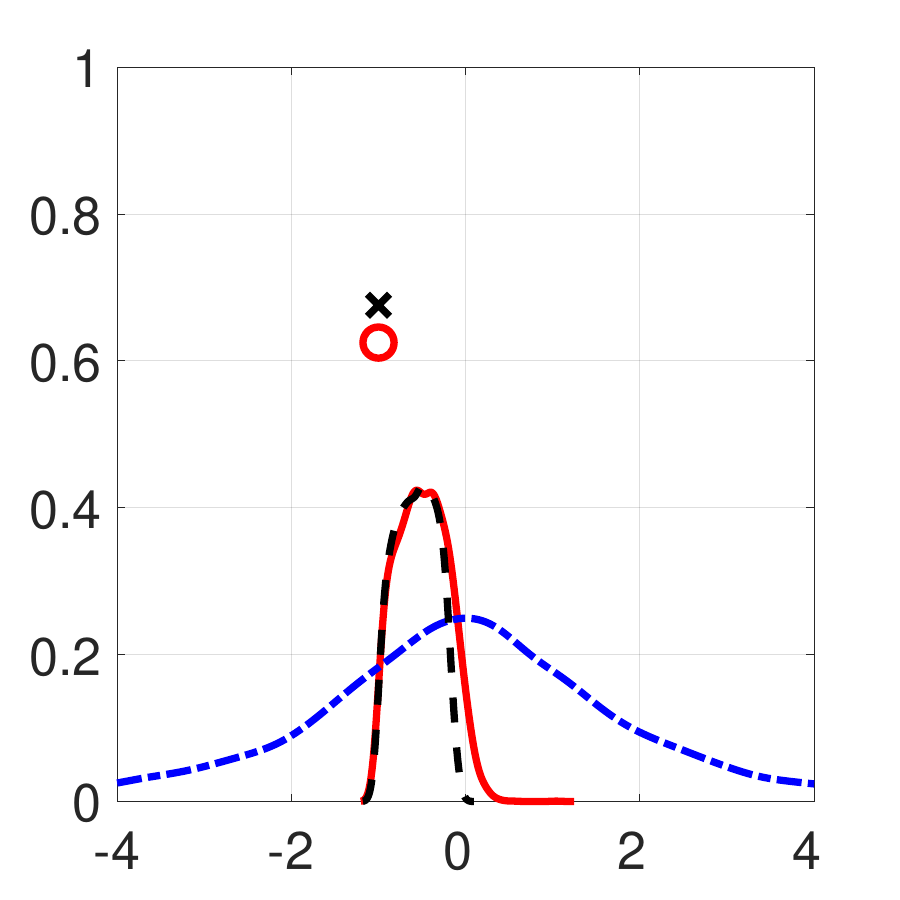}
	\end{subfigure}\begin{subfigure}{0.2\textwidth}
		\centering
		\includegraphics[trim={0cm 0cm 0.50cm 0.5cm},width=\textwidth,clip]
		{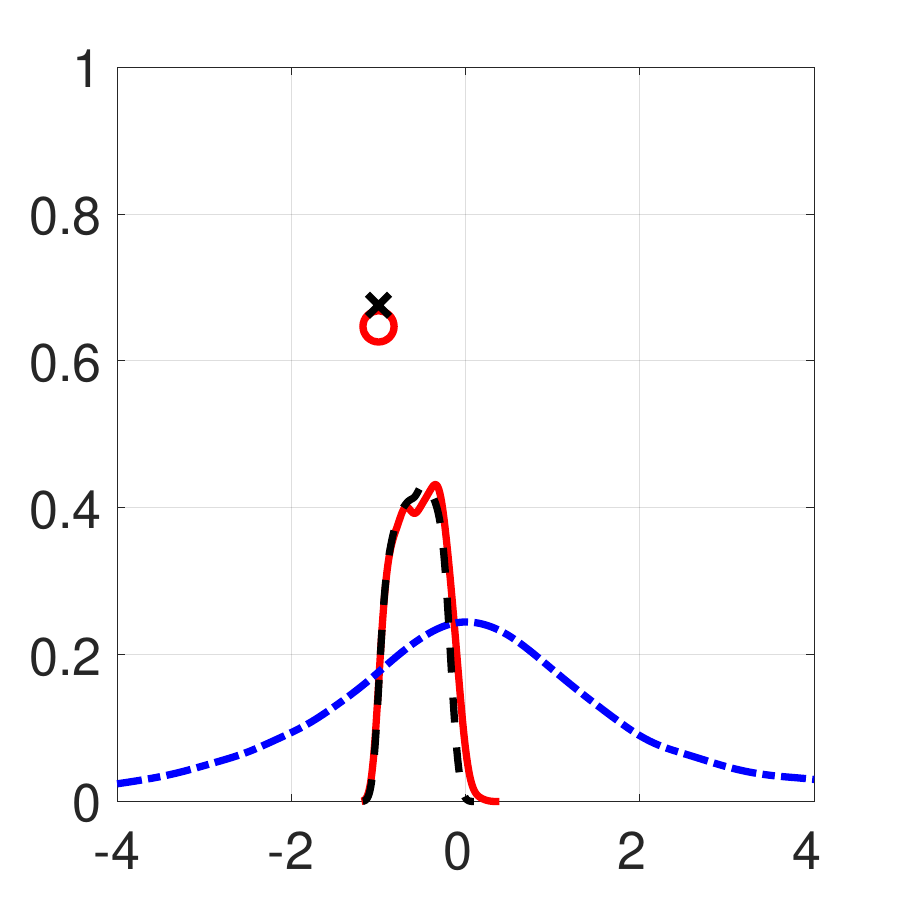}
	\end{subfigure}\begin{subfigure}{0.2\textwidth}
		\centering
		\includegraphics[trim={0cm 0cm 0.50cm 0.5cm},width=\textwidth,clip]
		{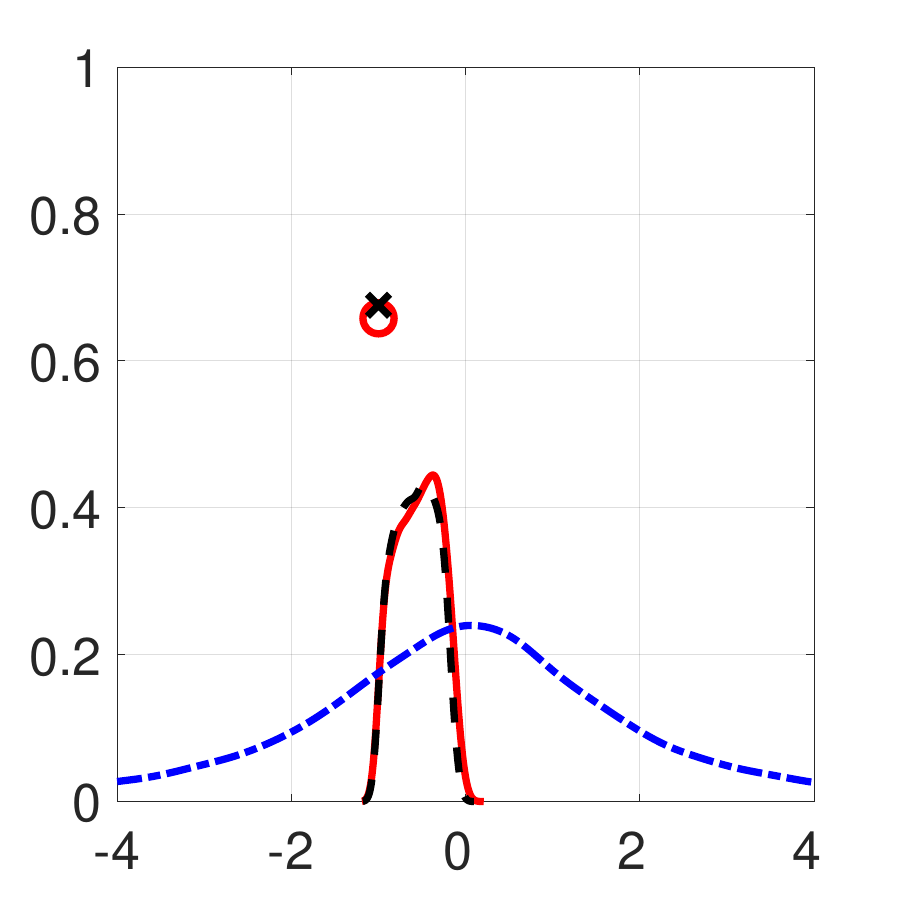}
	\end{subfigure}\begin{subfigure}{0.2\textwidth}
		\centering
		\includegraphics[trim={0cm 0cm 0.50cm 0.5cm},width=\textwidth,clip]
		{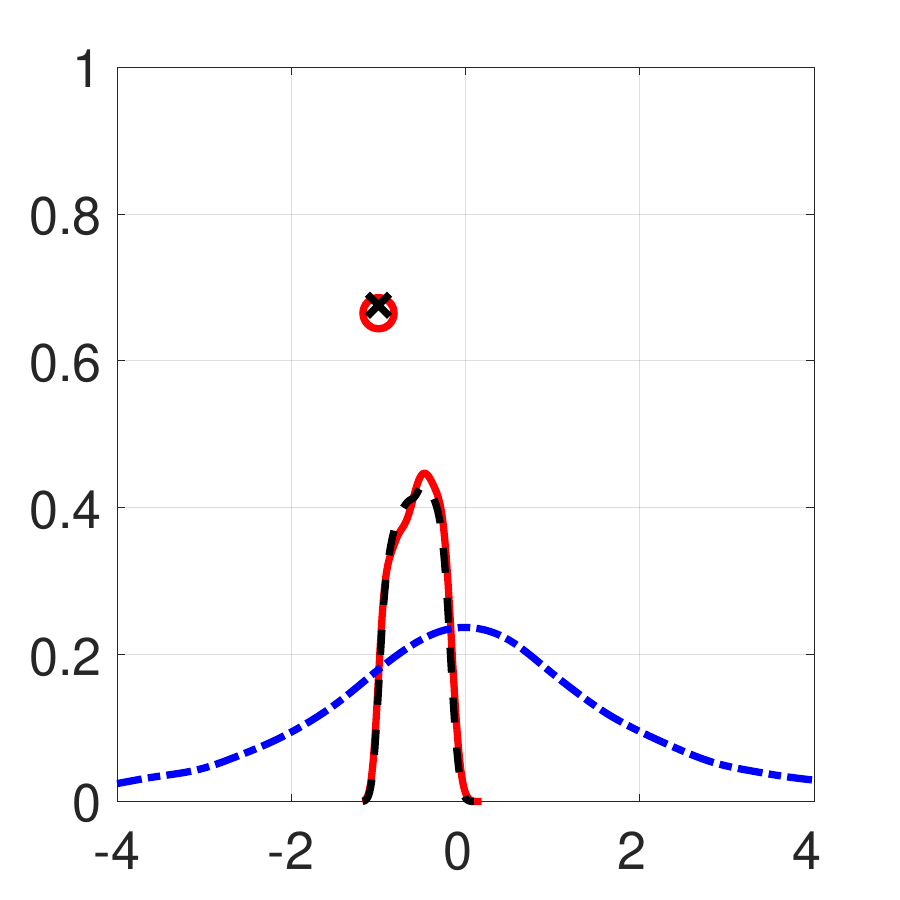}
	\end{subfigure}\begin{subfigure}{0.2\textwidth}
		\centering
		\includegraphics[trim={0cm 0cm 0.50cm 0.5cm},width=\textwidth,clip]
		{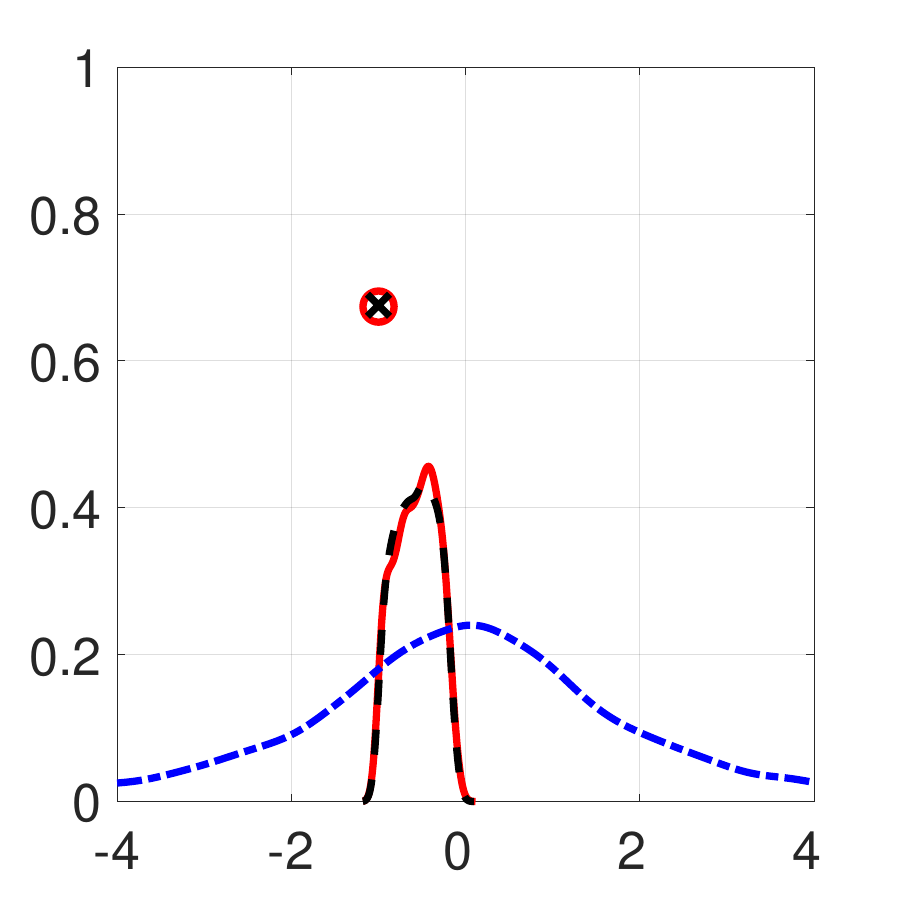}
	\end{subfigure}
	
\end{center}

\vspace{-2ex} 

\caption{Finite-sample distributions of $T(\betaAL - \beta_T)$ (under
conservative tuning, in the first row) and $\lambda_T^{-1/2}T(\betaAL
- \beta_T)$ (under consistent tuning, in the remaining rows) in case
$\beta_T = \sqrt{\lambda_T}\beta/T$ (labeled ``AL''), and
case-specific limiting distribution from
Theorem~\ref{thm:ls_dist-unif}, evaluated at sample counterparts of
limiting parameters (labeled ``Thm.3''). \emph{Notes}: See notes to
Figure~\ref{fig:densities_thm3_1}.}

\label{fig:densities_thm3_4}

\end{figure}

\begin{figure}[ht]
\begin{center}
	\caption*{$\lambda_T = T^{1/4}$}
	\vspace{-1.5ex}
	\begin{subfigure}{0.2\textwidth}
		\centering
		\caption*{$T = 25$}
		\vspace{-1.5ex}
		\includegraphics[trim={0cm 0cm 0.50cm 0.5cm},width=\textwidth,clip]
		{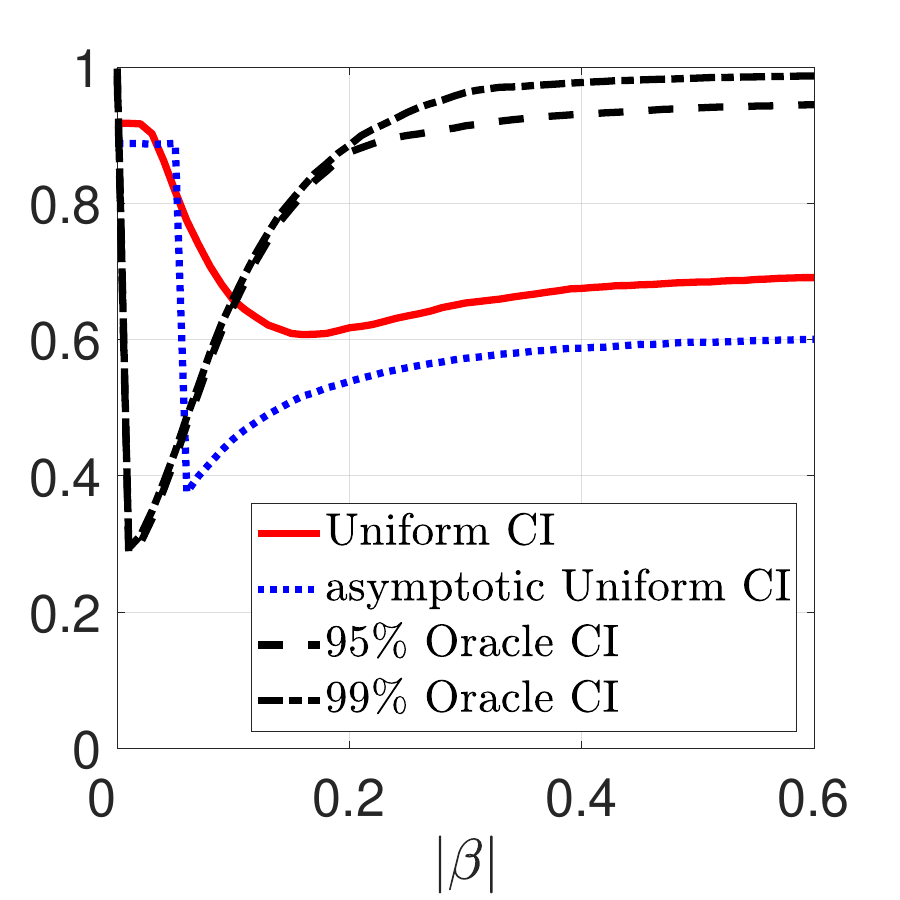}
	\end{subfigure}\begin{subfigure}{0.2\textwidth}
		\centering
		\caption*{$T = 50$}
		\vspace{-1.5ex}
		\includegraphics[trim={0cm 0cm 0.50cm 0.5cm},width=\textwidth,clip]
		{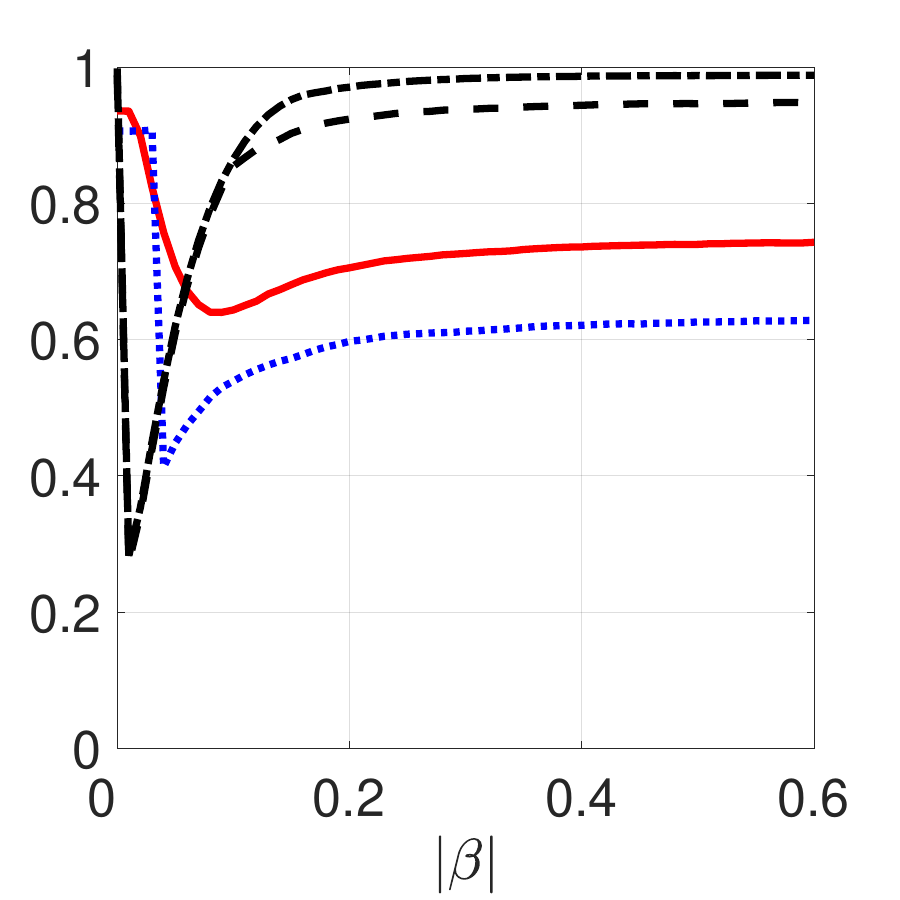}
	\end{subfigure}\begin{subfigure}{0.2\textwidth}
		\centering
		\caption*{$T = 100$}
		\vspace{-1.5ex}
		\includegraphics[trim={0cm 0cm 0.50cm 0.5cm},width=\textwidth,clip]
		{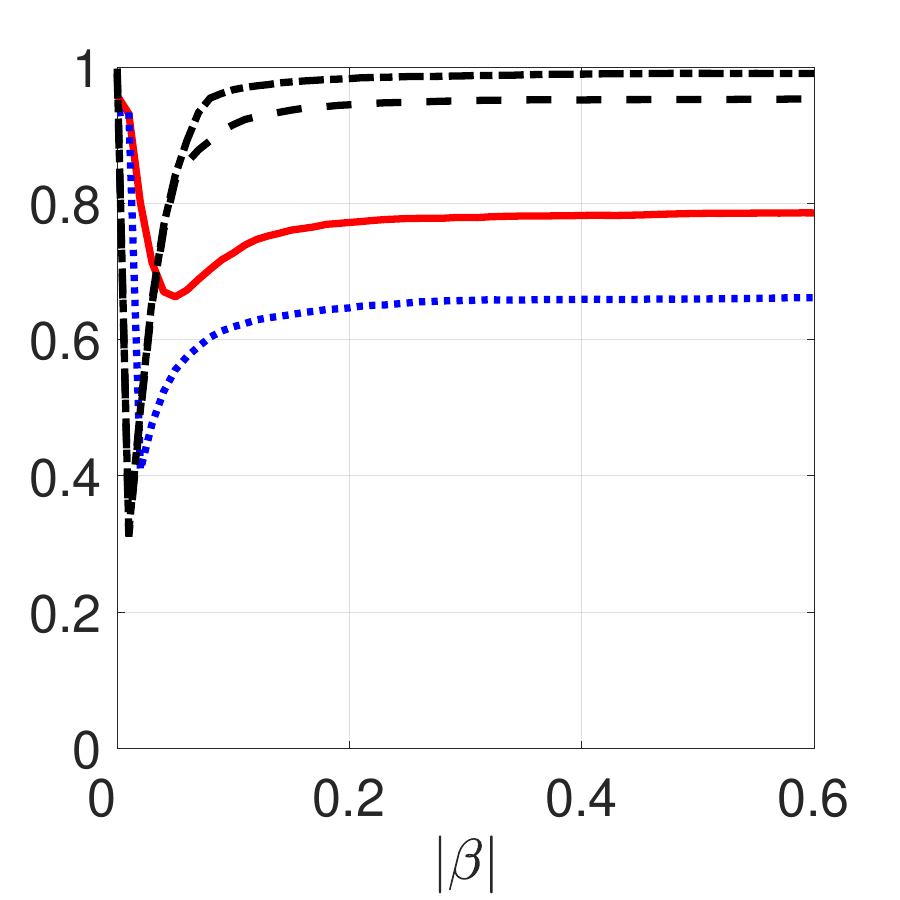}
	\end{subfigure}\begin{subfigure}{0.2\textwidth}
		\centering
		\caption*{$T = 250$}
		\vspace{-1.5ex}
		\includegraphics[trim={0cm 0cm 0.50cm 0.5cm},width=\textwidth,clip]
		{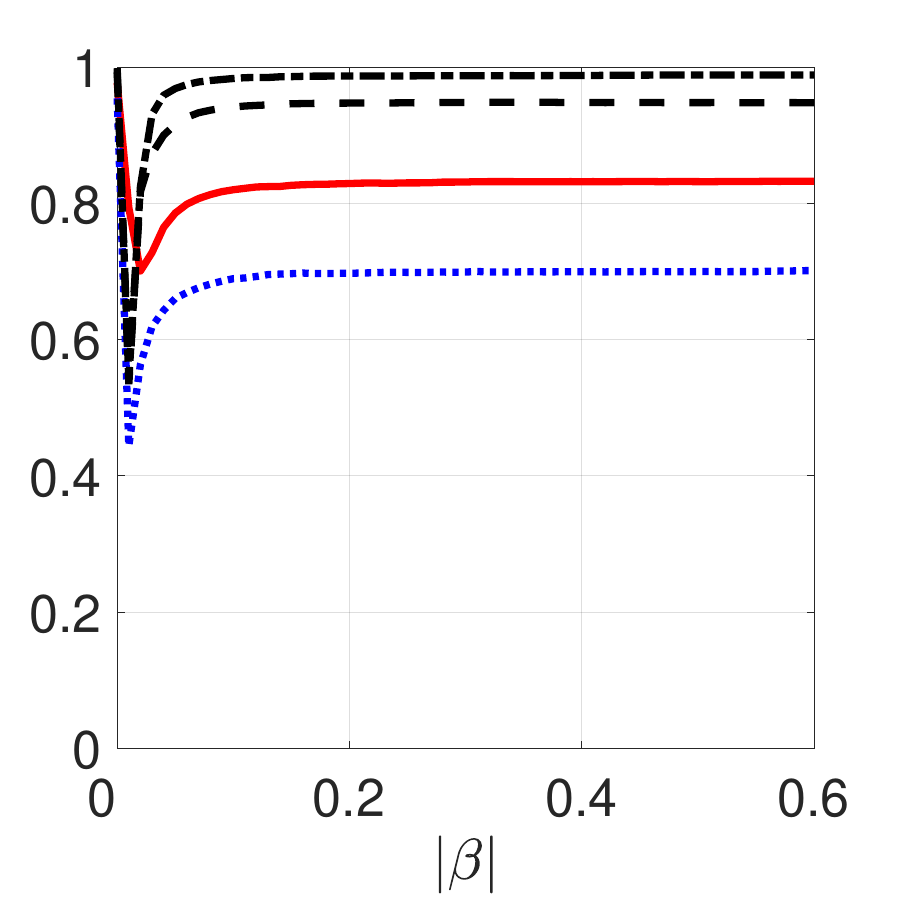}
	\end{subfigure}\begin{subfigure}{0.2\textwidth}
		\centering
		\caption*{$T = 1000$}
		\vspace{-1.5ex}
		\includegraphics[trim={0cm 0cm 0.50cm 0.5cm},width=\textwidth,clip]
		{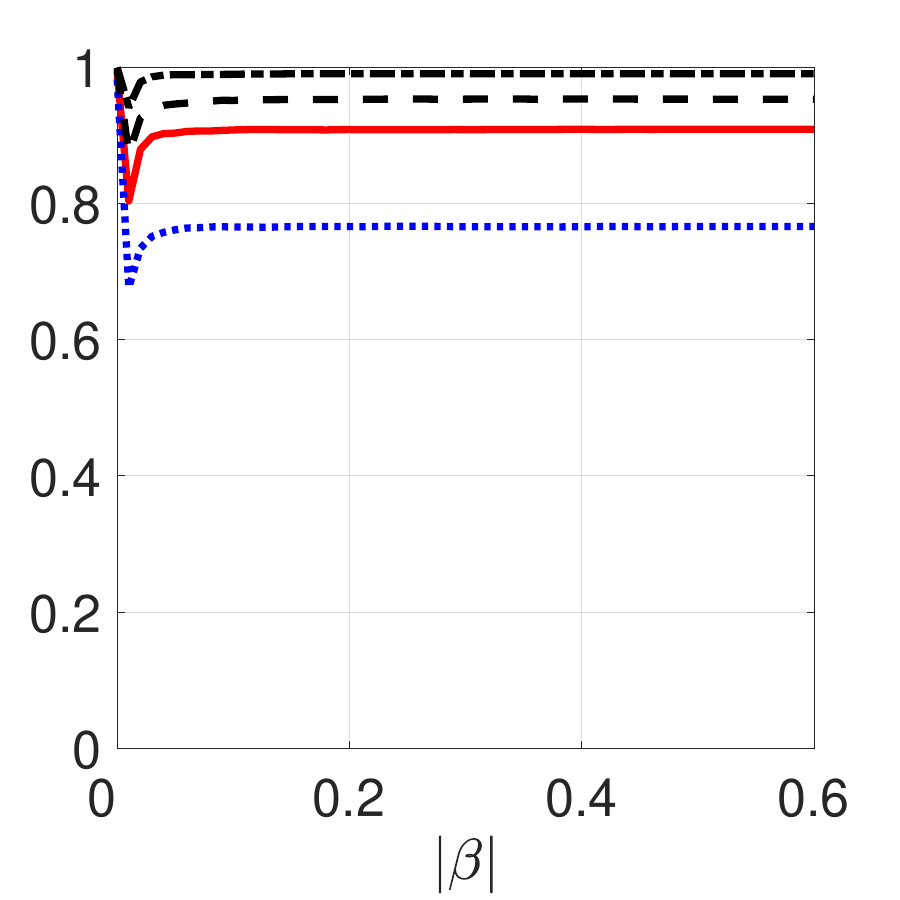}
	\end{subfigure}
	
	\caption*{$\lambda_T = T^{1/2}$}
	\vspace{-1.5ex}
	\begin{subfigure}{0.2\textwidth}
		\centering
		\includegraphics[trim={0cm 0cm 0.50cm 0.5cm},width=\textwidth,clip]
		{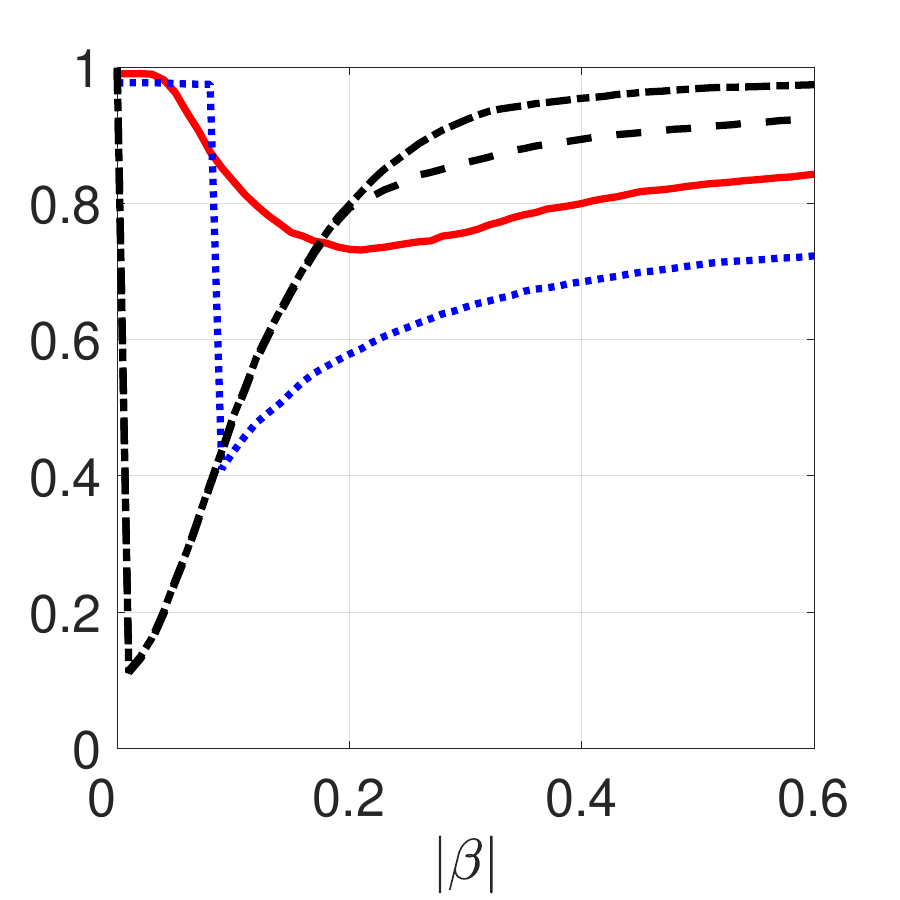}
	\end{subfigure}\begin{subfigure}{0.2\textwidth}
		\centering
		\includegraphics[trim={0cm 0cm 0.50cm 0.5cm},width=\textwidth,clip]
		{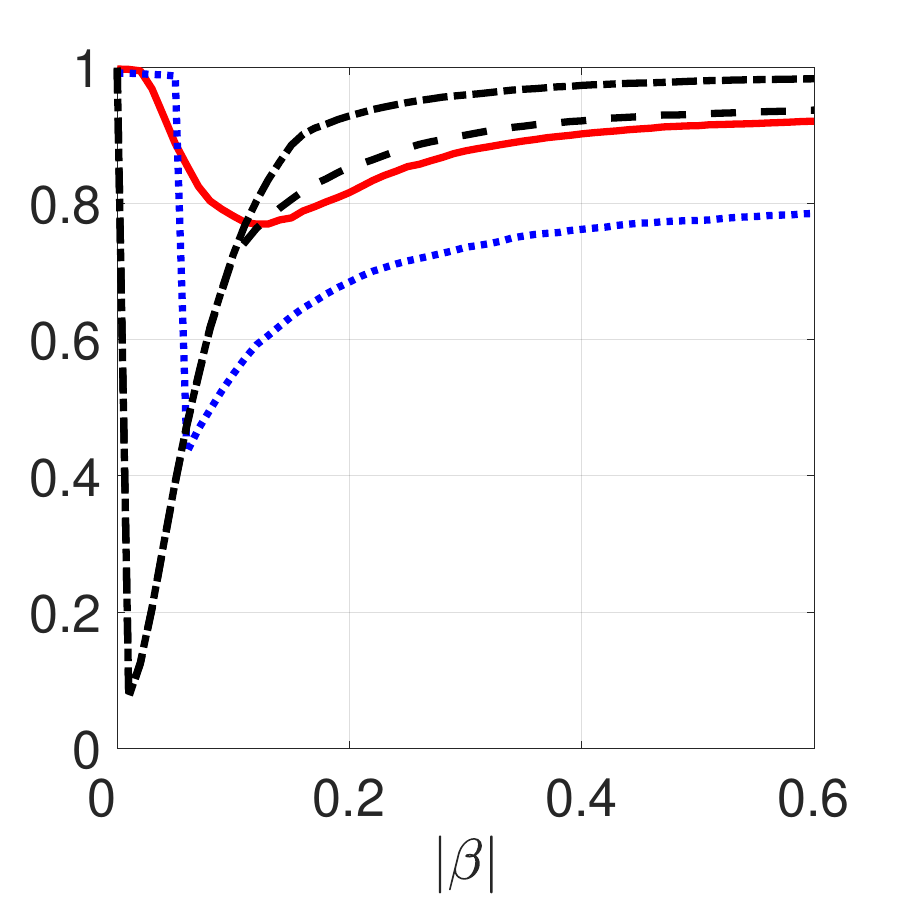}
	\end{subfigure}\begin{subfigure}{0.2\textwidth}
		\centering
		\includegraphics[trim={0cm 0cm 0.50cm 0.5cm},width=\textwidth,clip]
		{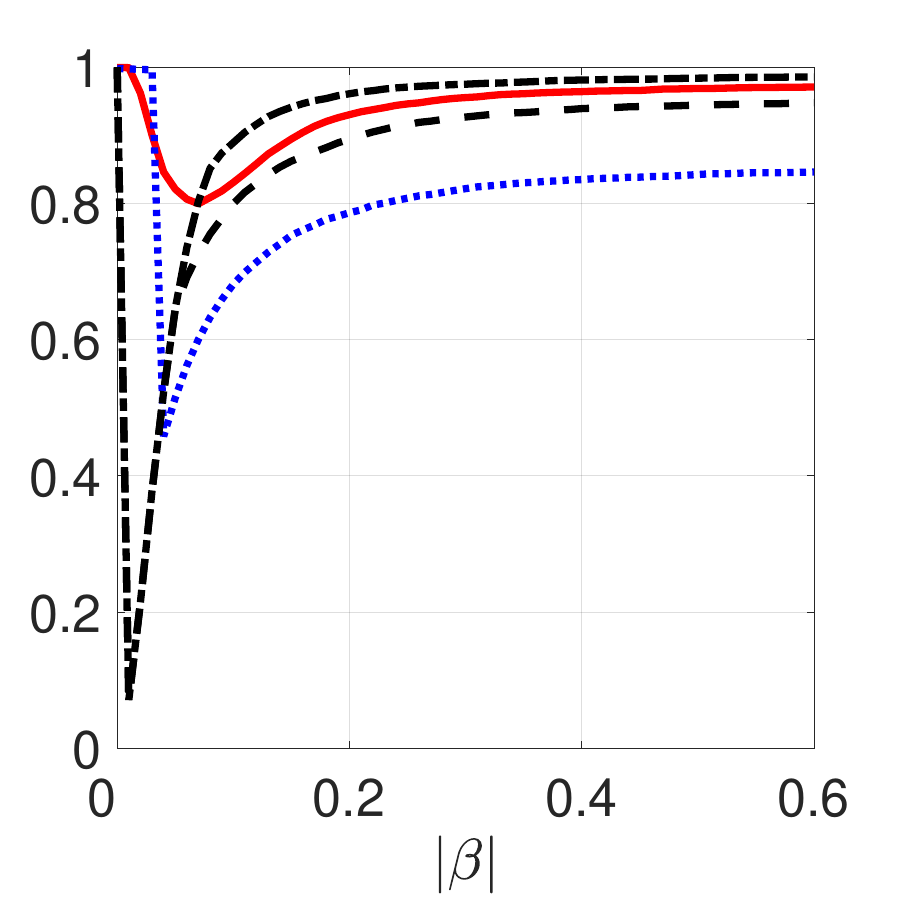}
	\end{subfigure}\begin{subfigure}{0.2\textwidth}
		\centering
		\includegraphics[trim={0cm 0cm 0.50cm 0.5cm},width=\textwidth,clip]
		{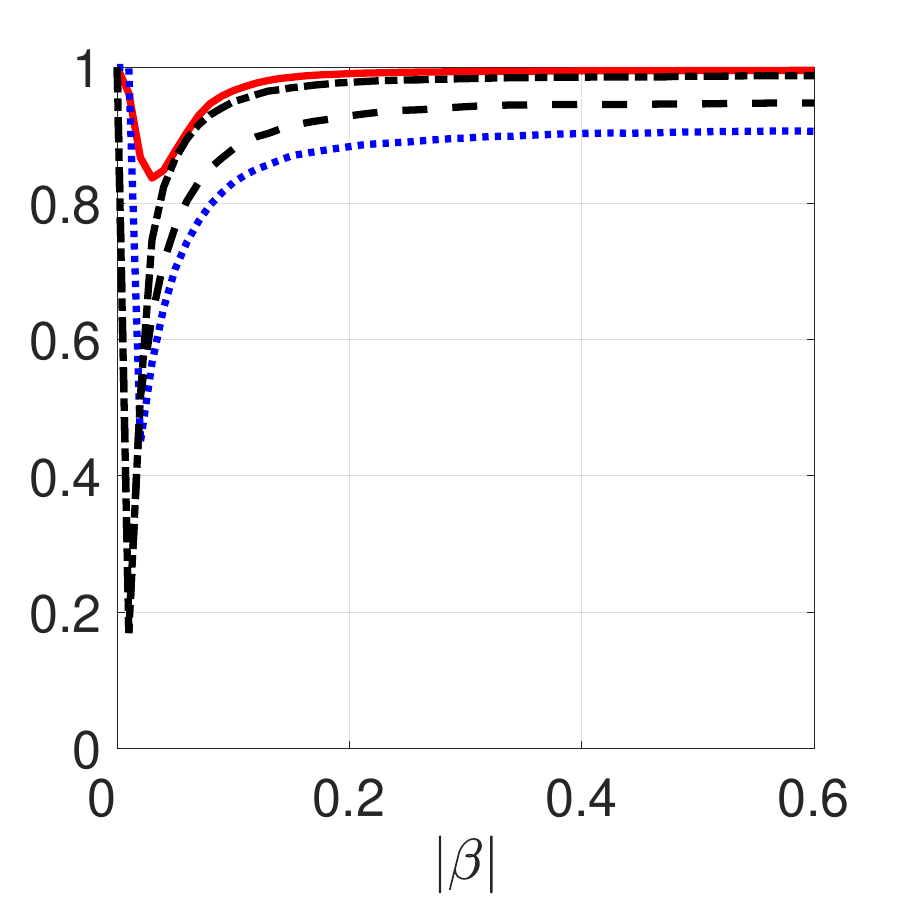}
	\end{subfigure}\begin{subfigure}{0.2\textwidth}
		\centering
		\includegraphics[trim={0cm 0cm 0.50cm 0.5cm},width=\textwidth,clip]
		{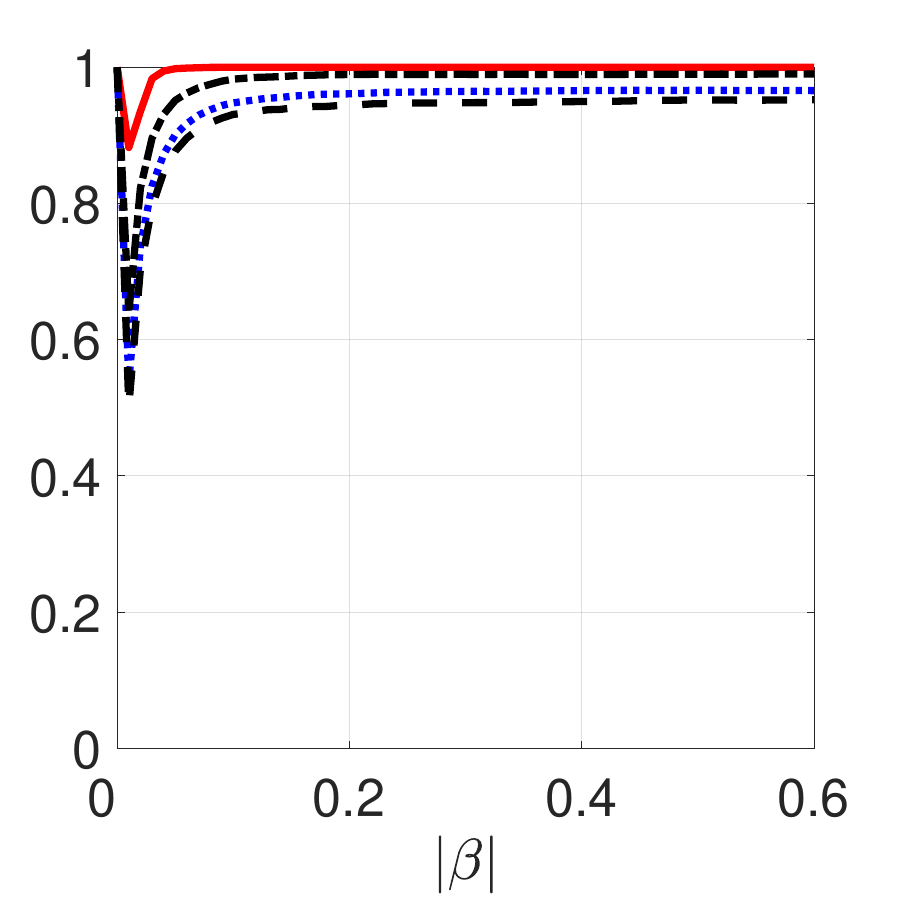}
	\end{subfigure}
	
	\caption*{$\lambda_T = T$}
	\vspace{-1.5ex}
	\begin{subfigure}{0.2\textwidth}
		\centering
		\includegraphics[trim={0cm 0cm 0.50cm 0.5cm},width=\textwidth,clip]
		{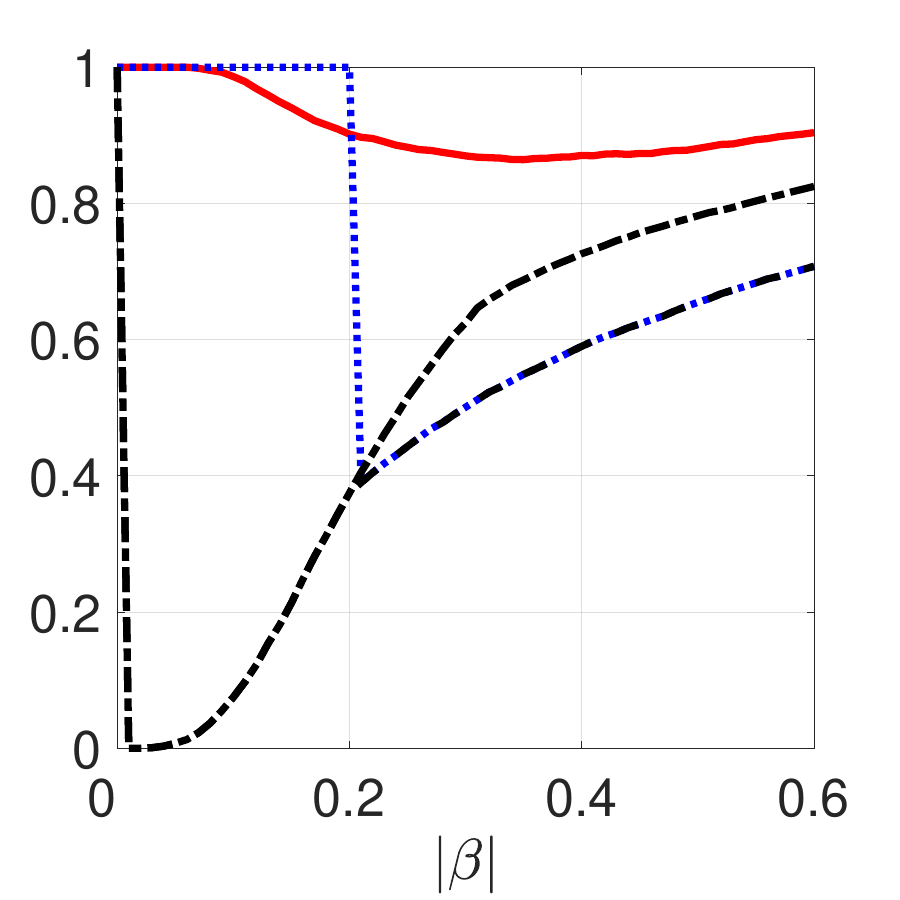}
	\end{subfigure}\begin{subfigure}{0.2\textwidth}
		\centering
		\includegraphics[trim={0cm 0cm 0.50cm 0.5cm},width=\textwidth,clip]
		{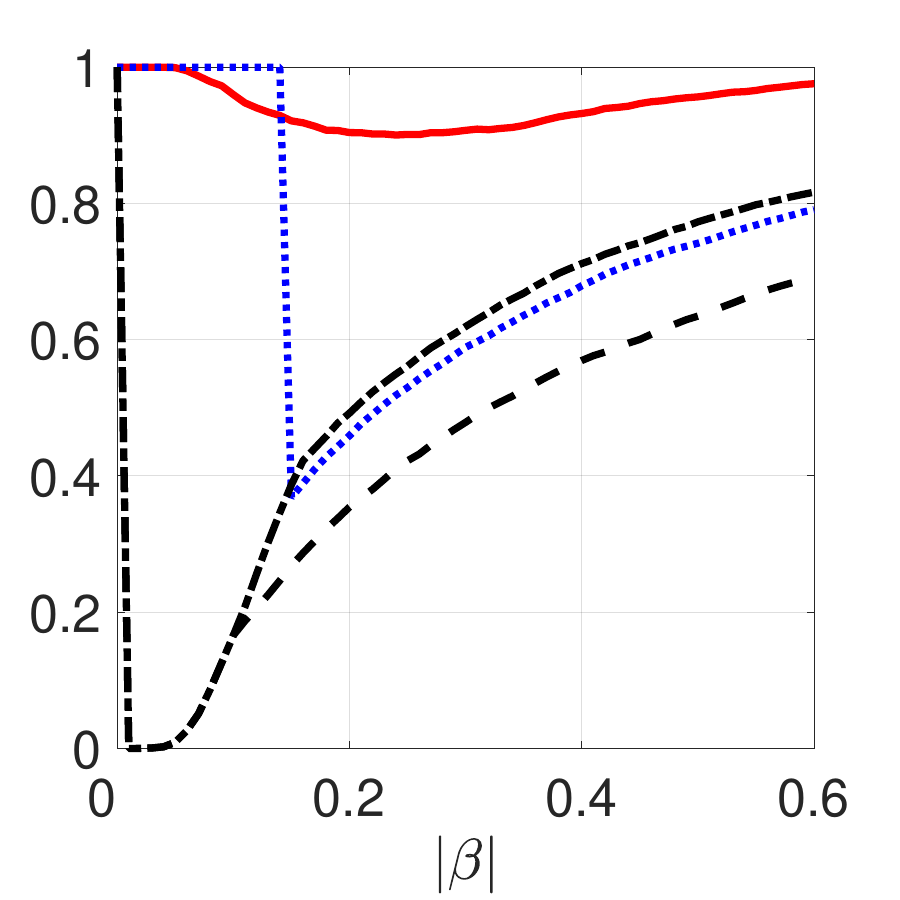}
	\end{subfigure}\begin{subfigure}{0.2\textwidth}
		\centering
		\includegraphics[trim={0cm 0cm 0.50cm 0.5cm},width=\textwidth,clip]
		{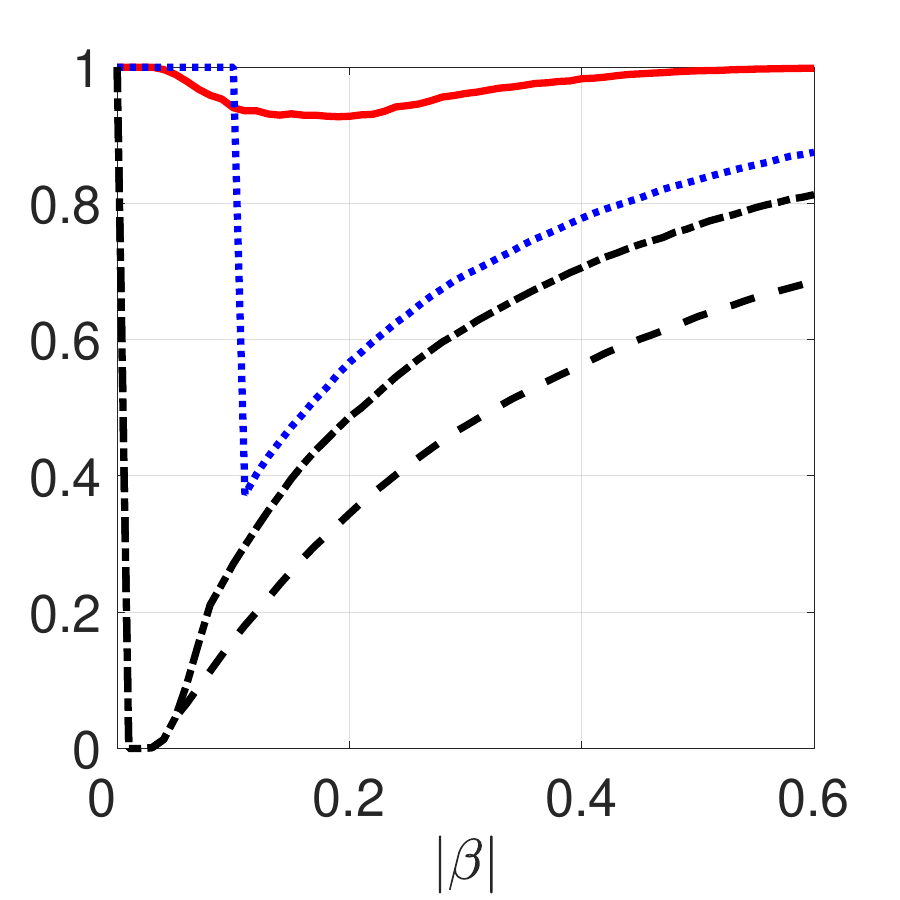}
	\end{subfigure}\begin{subfigure}{0.2\textwidth}
		\centering
		\includegraphics[trim={0cm 0cm 0.50cm 0.5cm},width=\textwidth,clip]
		{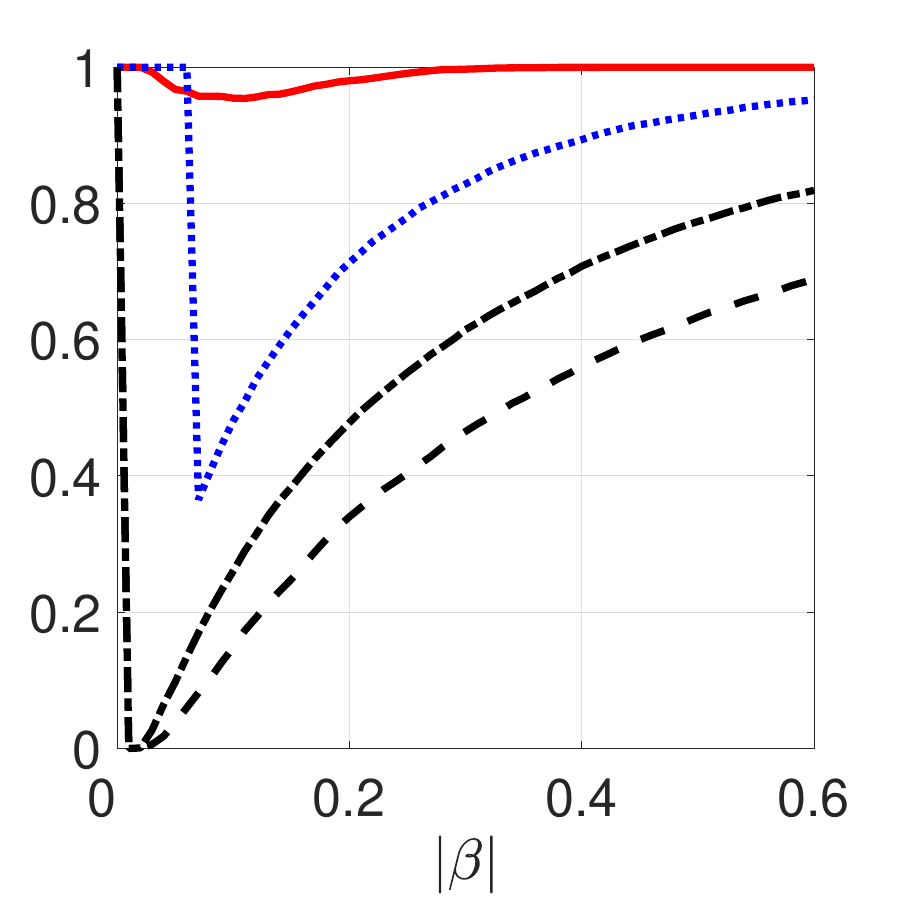}
	\end{subfigure}\begin{subfigure}{0.2\textwidth}
		\centering
		\includegraphics[trim={0cm 0cm 0.50cm 0.5cm},width=\textwidth,clip]
		{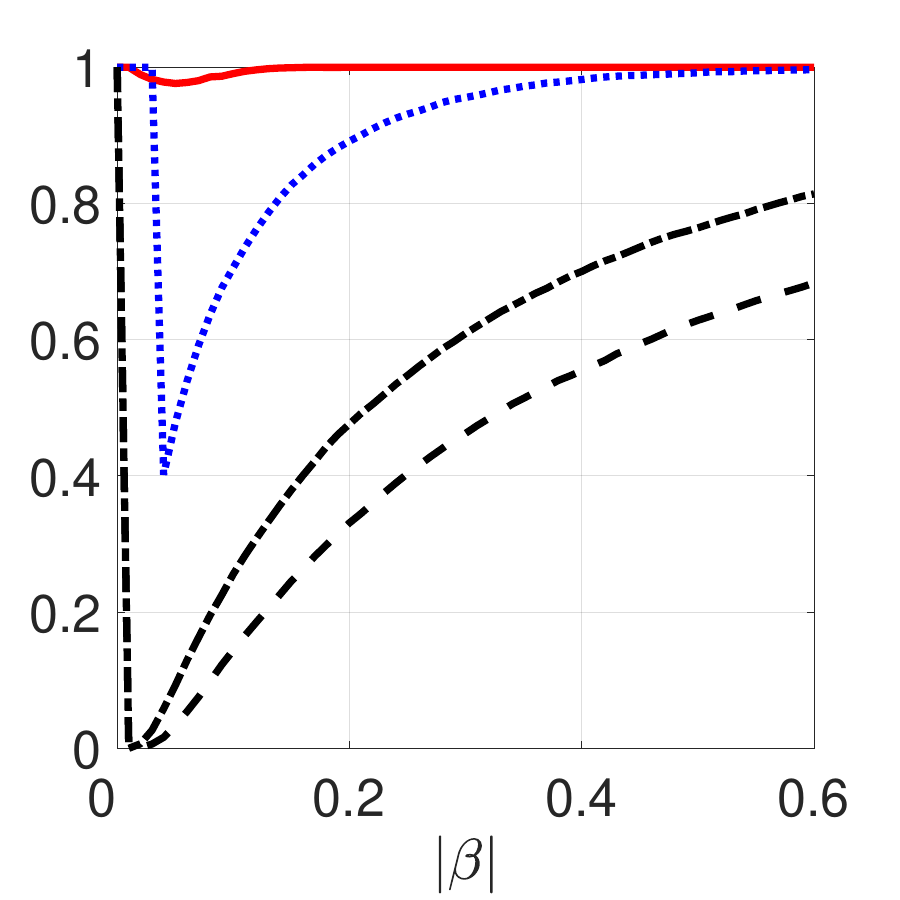}
	\end{subfigure}
	
\end{center}

\vspace{-2ex} 

\caption{Coverage probabilities of confidence intervals}

\label{fig:coverage_lambda}

\end{figure}

\begin{figure}[ht]
\begin{center}	
	\caption*{$\lambda_T = 4\times T^{1/4}$}
	\vspace{-1.5ex}
	\begin{subfigure}{0.2\textwidth}
		\centering
		\caption*{$T = 25$}
		\vspace{-1.5ex}
		\includegraphics[trim={0cm 0cm 0.50cm 0.5cm},width=\textwidth,clip]
		{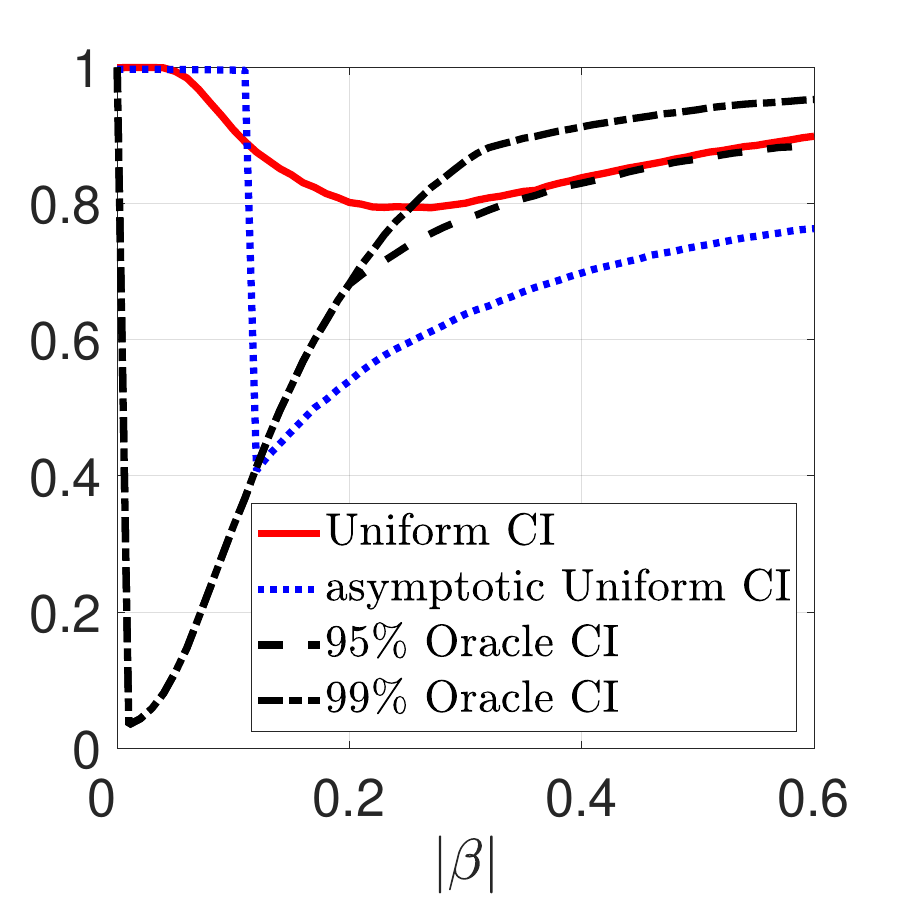}
	\end{subfigure}\begin{subfigure}{0.2\textwidth}
		\centering
		\caption*{$T = 50$}
		\vspace{-1.5ex}
		\includegraphics[trim={0cm 0cm 0.50cm 0.5cm},width=\textwidth,clip]
		{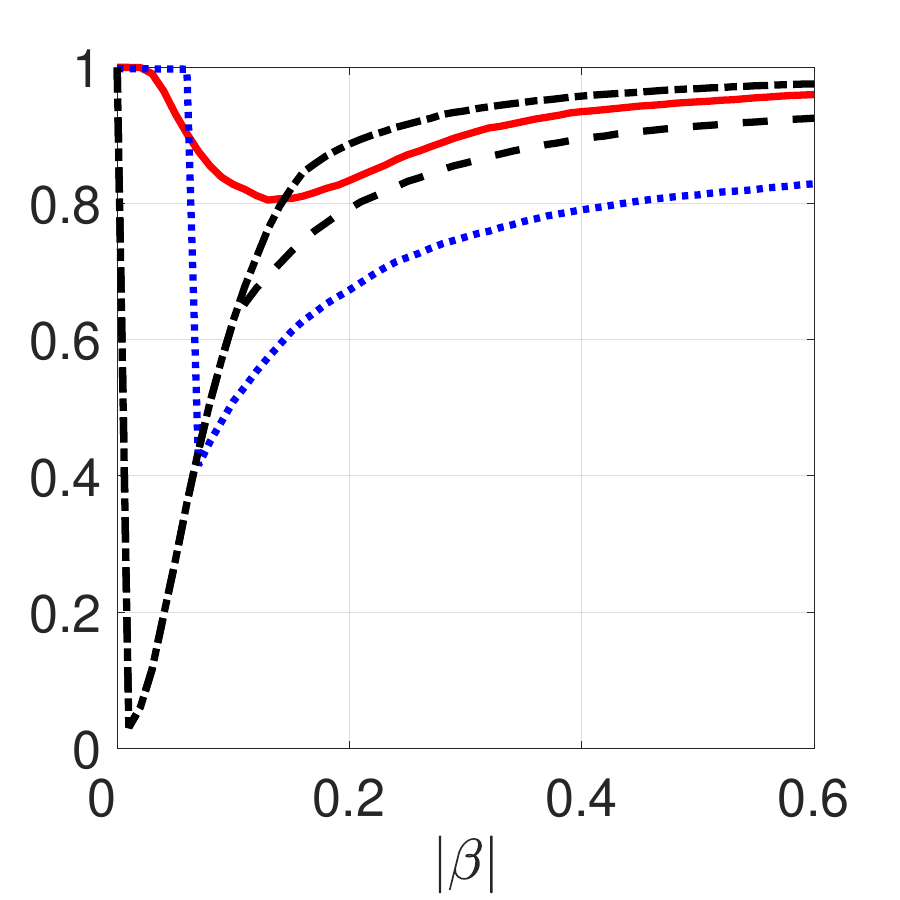}
	\end{subfigure}\begin{subfigure}{0.2\textwidth}
		\centering
		\caption*{$T = 100$}
		\vspace{-1.5ex}
		\includegraphics[trim={0cm 0cm 0.50cm 0.5cm},width=\textwidth,clip]
		{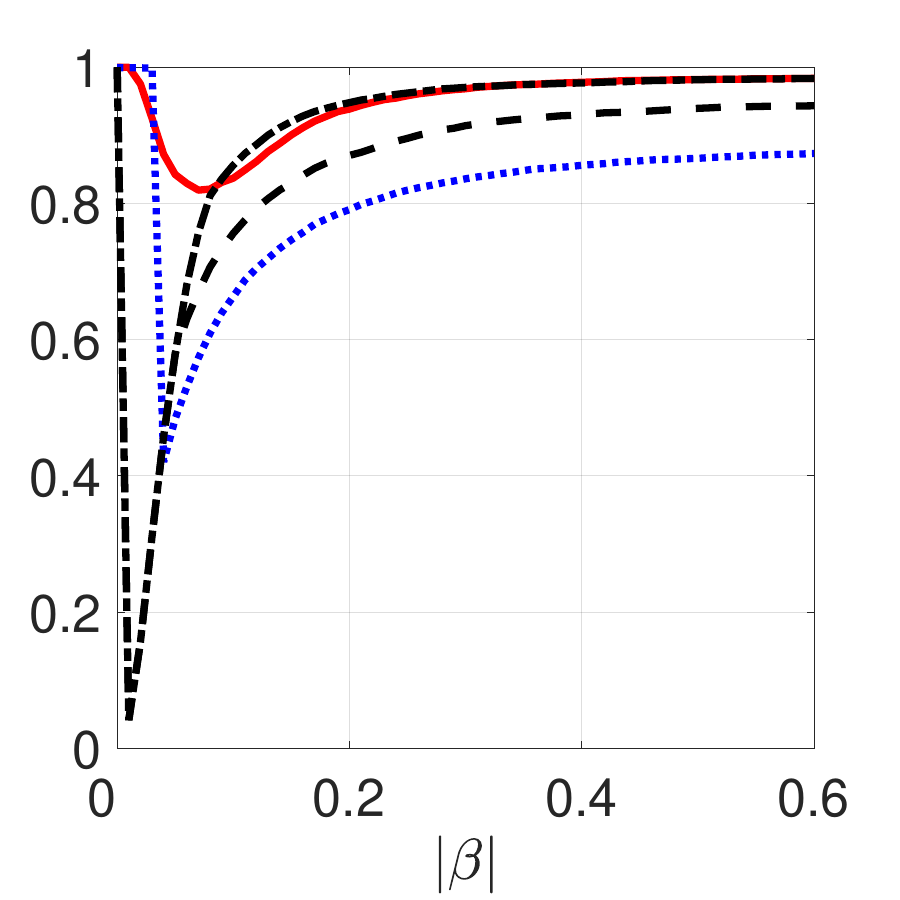}
	\end{subfigure}\begin{subfigure}{0.2\textwidth}
		\centering
		\caption*{$T = 250$}
		\vspace{-1.5ex}
		\includegraphics[trim={0cm 0cm 0.50cm 0.5cm},width=\textwidth,clip]
		{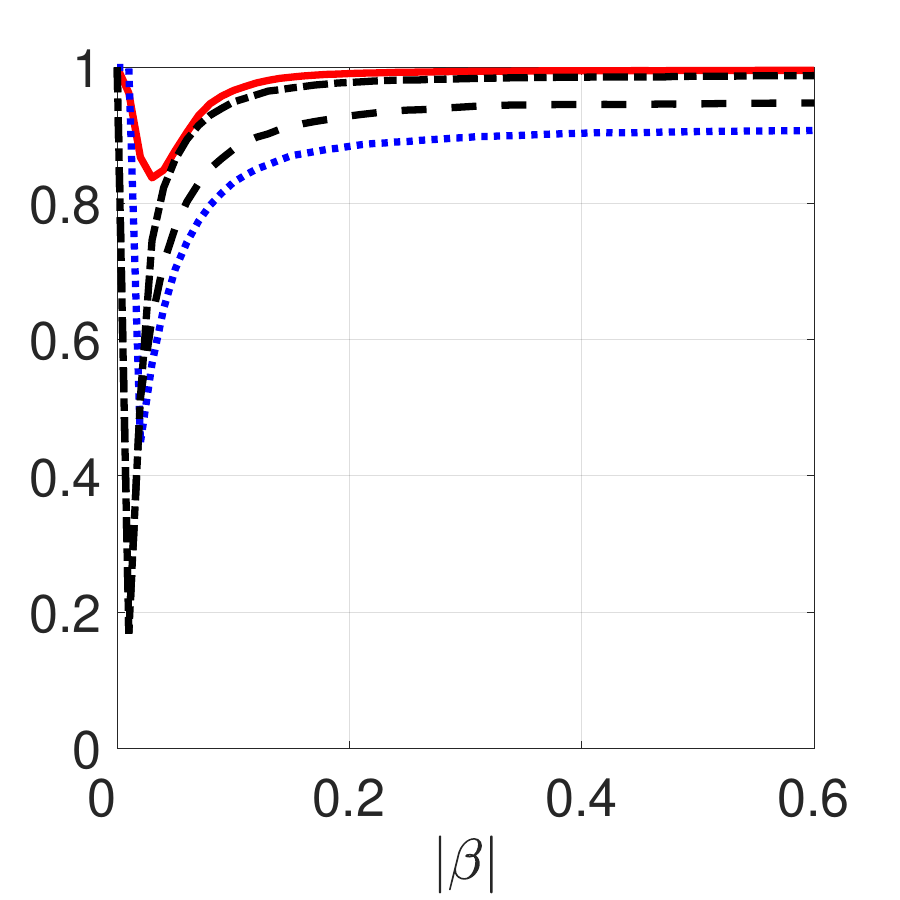}
	\end{subfigure}\begin{subfigure}{0.2\textwidth}
		\centering
		\caption*{$T = 1000$}
		\vspace{-1.5ex}
		\includegraphics[trim={0cm 0cm 0.50cm 0.5cm},width=\textwidth,clip]
		{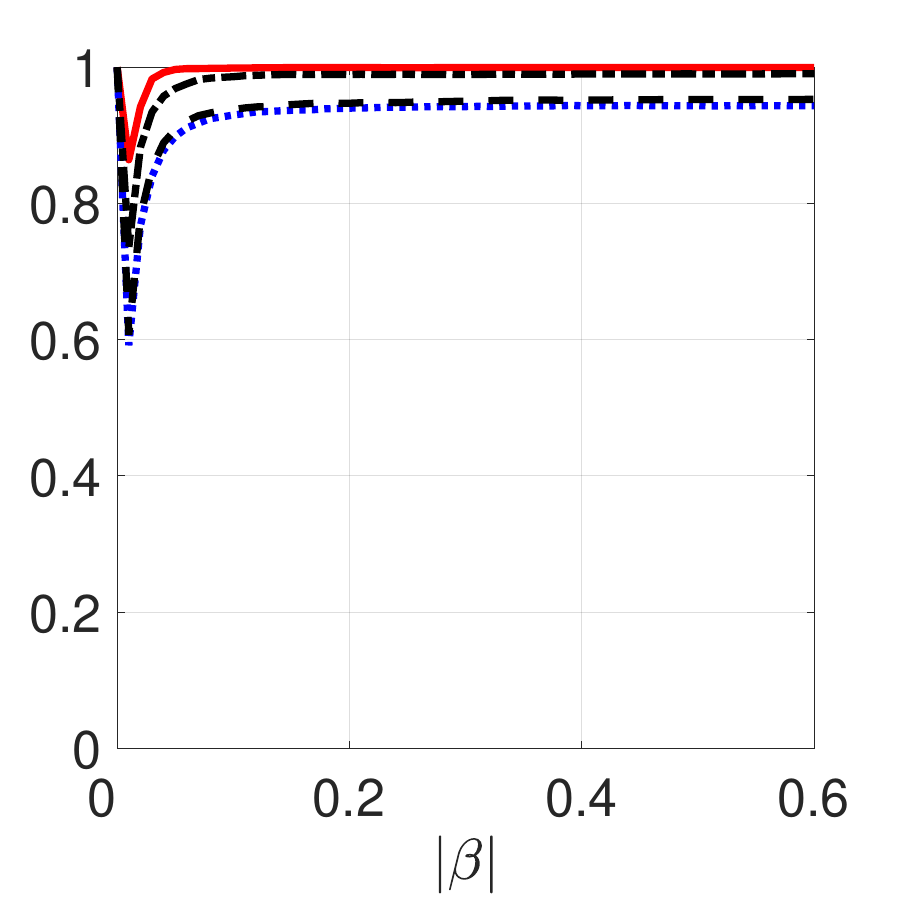}
	\end{subfigure}
	
	\caption*{$\lambda_T = 4\times T^{1/2}$}
	\vspace{-1.5ex}
	\begin{subfigure}{0.2\textwidth}
		\centering
		\includegraphics[trim={0cm 0cm 0.50cm 0.5cm},width=\textwidth,clip]
		{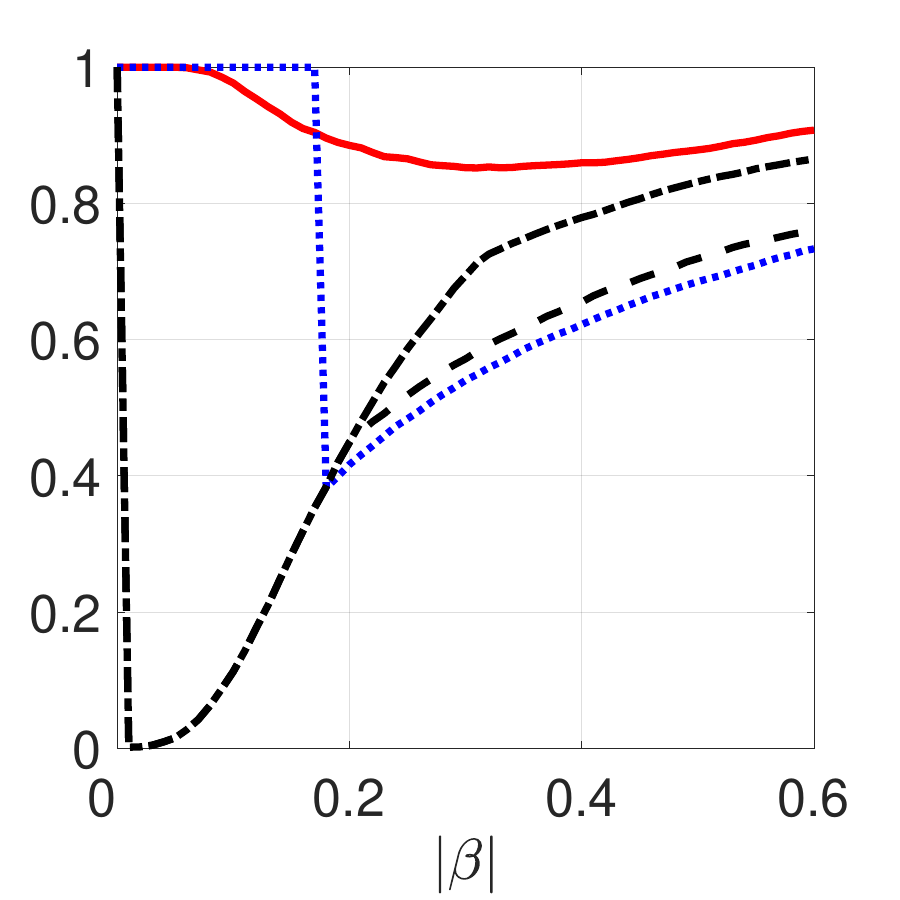}
	\end{subfigure}\begin{subfigure}{0.2\textwidth}
		\centering
		\includegraphics[trim={0cm 0cm 0.50cm 0.5cm},width=\textwidth,clip]
		{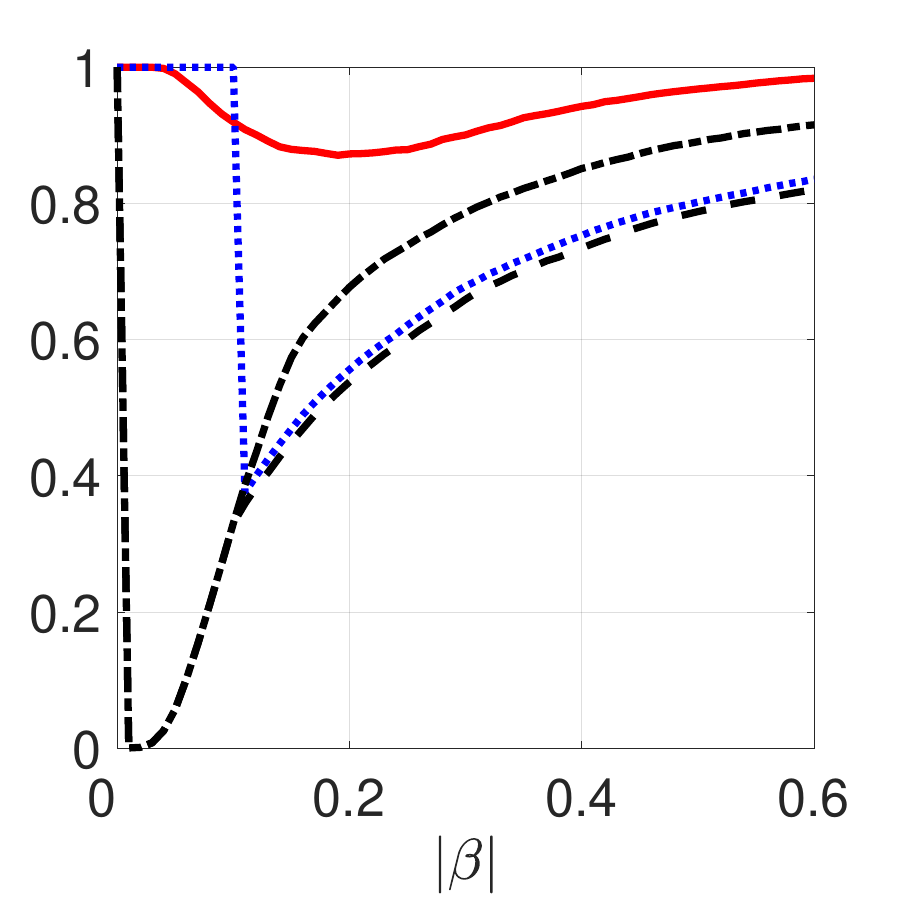}
	\end{subfigure}\begin{subfigure}{0.2\textwidth}
		\centering
		\includegraphics[trim={0cm 0cm 0.50cm 0.5cm},width=\textwidth,clip]
		{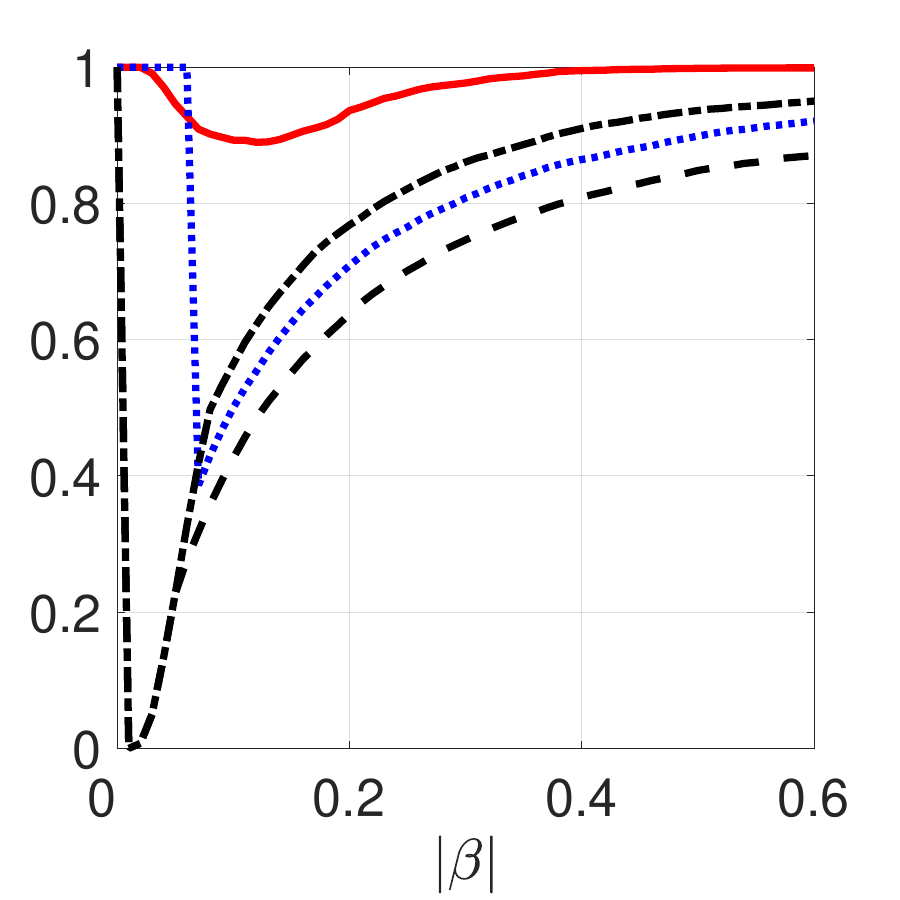}
	\end{subfigure}\begin{subfigure}{0.2\textwidth}
		\centering
		\includegraphics[trim={0cm 0cm 0.50cm 0.5cm},width=\textwidth,clip]
		{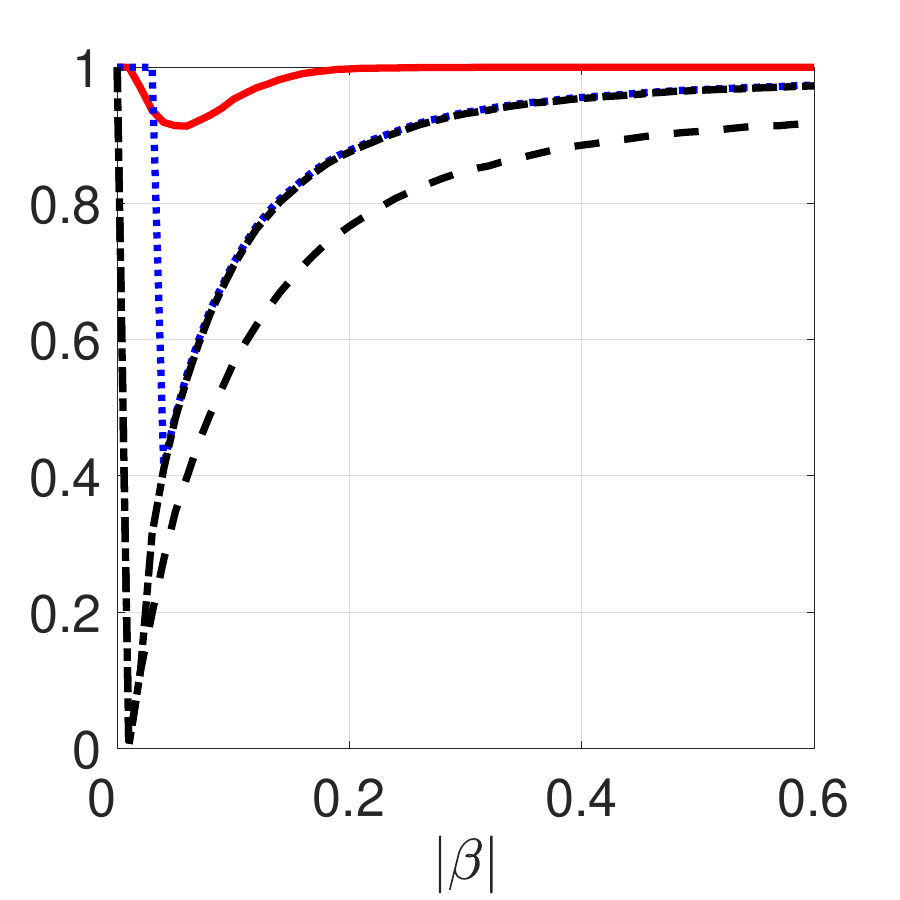}
	\end{subfigure}\begin{subfigure}{0.2\textwidth}
		\centering
		\includegraphics[trim={0cm 0cm 0.50cm 0.5cm},width=\textwidth,clip]
		{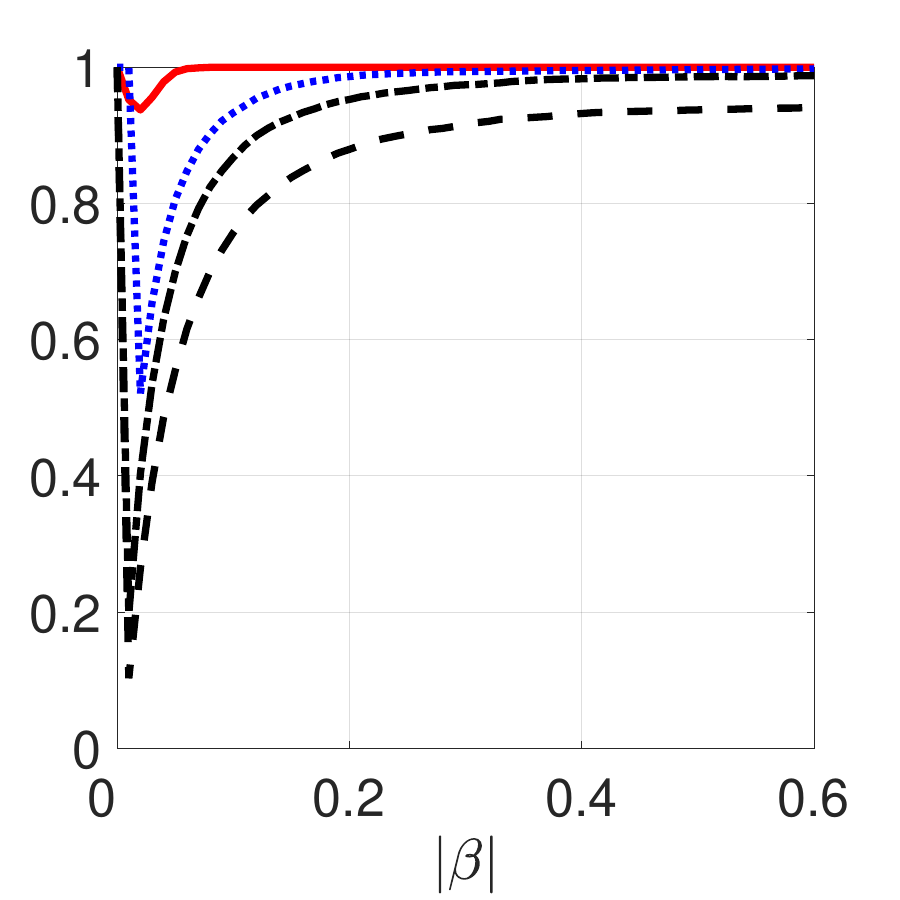}
	\end{subfigure}
	
	\caption*{$\lambda_T = 4\times T$}
	\vspace{-1.5ex}
	\begin{subfigure}{0.2\textwidth}
		\centering
		\includegraphics[trim={0cm 0cm 0.50cm 0.5cm},width=\textwidth,clip]
		{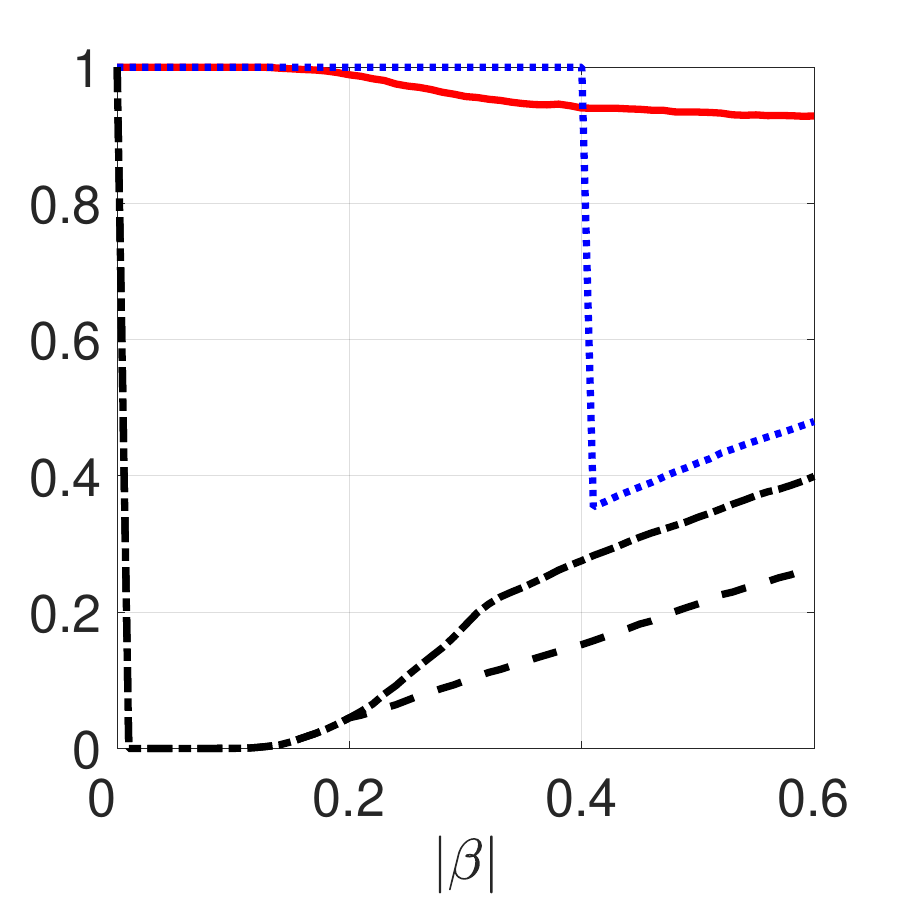}
	\end{subfigure}\begin{subfigure}{0.2\textwidth}
		\centering
		\includegraphics[trim={0cm 0cm 0.50cm 0.5cm},width=\textwidth,clip]
		{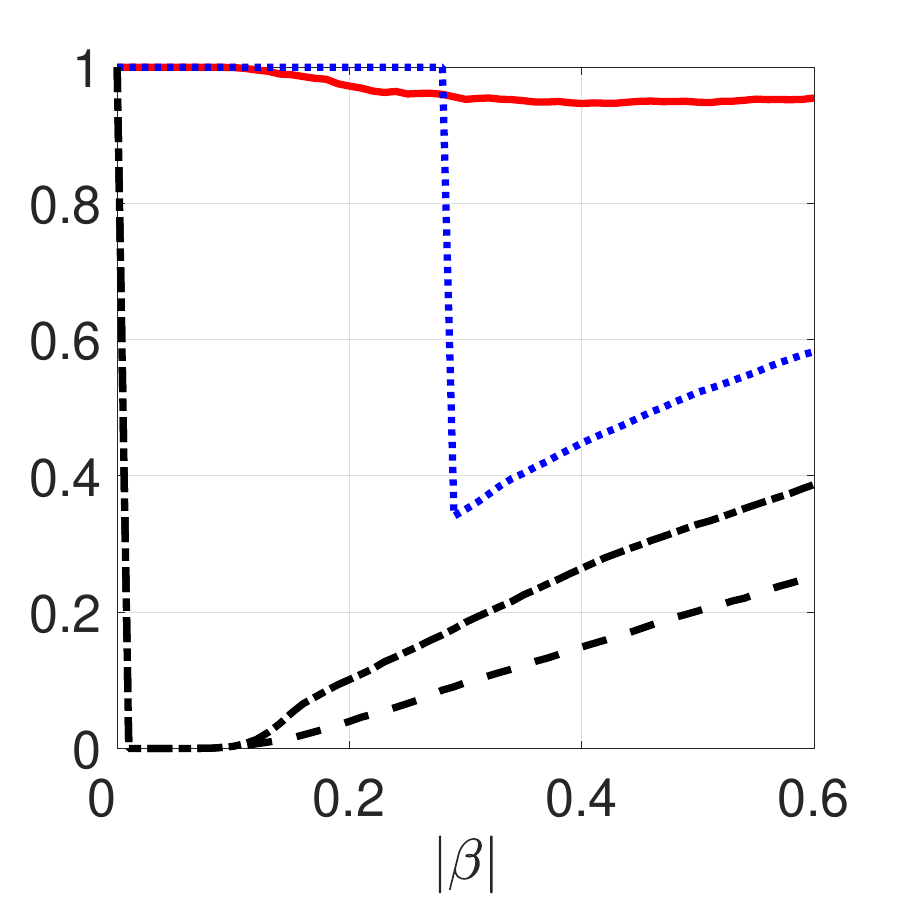}
	\end{subfigure}\begin{subfigure}{0.2\textwidth}
		\centering
		\includegraphics[trim={0cm 0cm 0.50cm 0.5cm},width=\textwidth,clip]
		{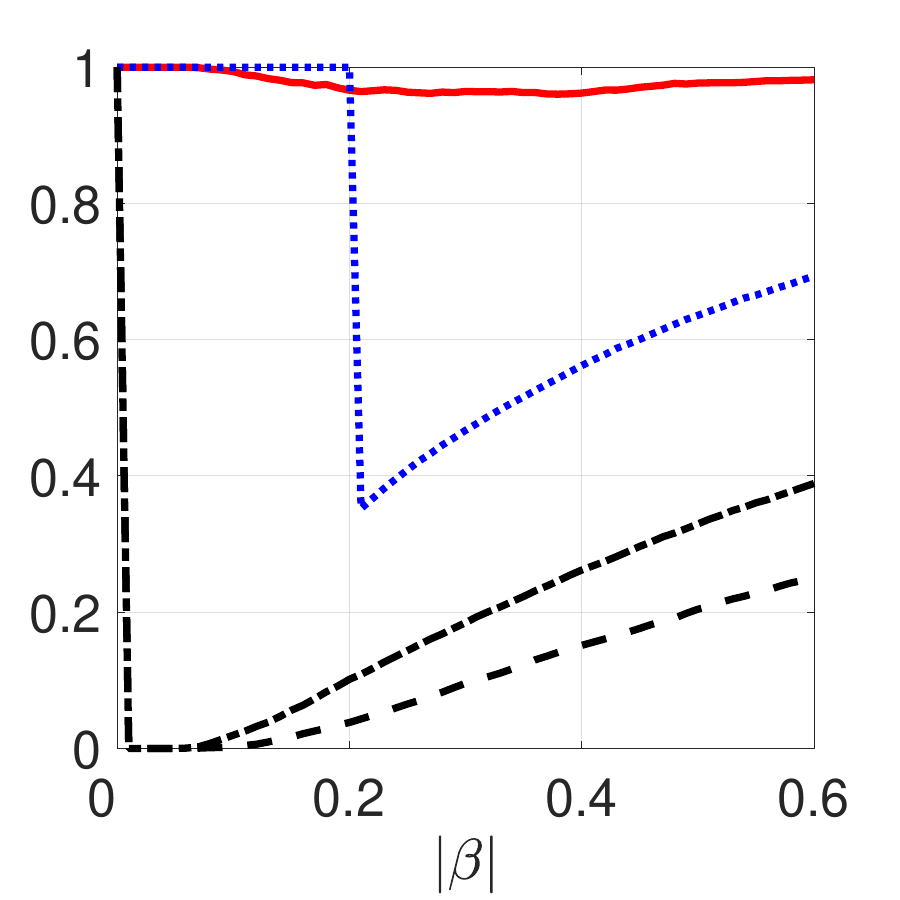}
	\end{subfigure}\begin{subfigure}{0.2\textwidth}
		\centering
		\includegraphics[trim={0cm 0cm 0.50cm 0.5cm},width=\textwidth,clip]
		{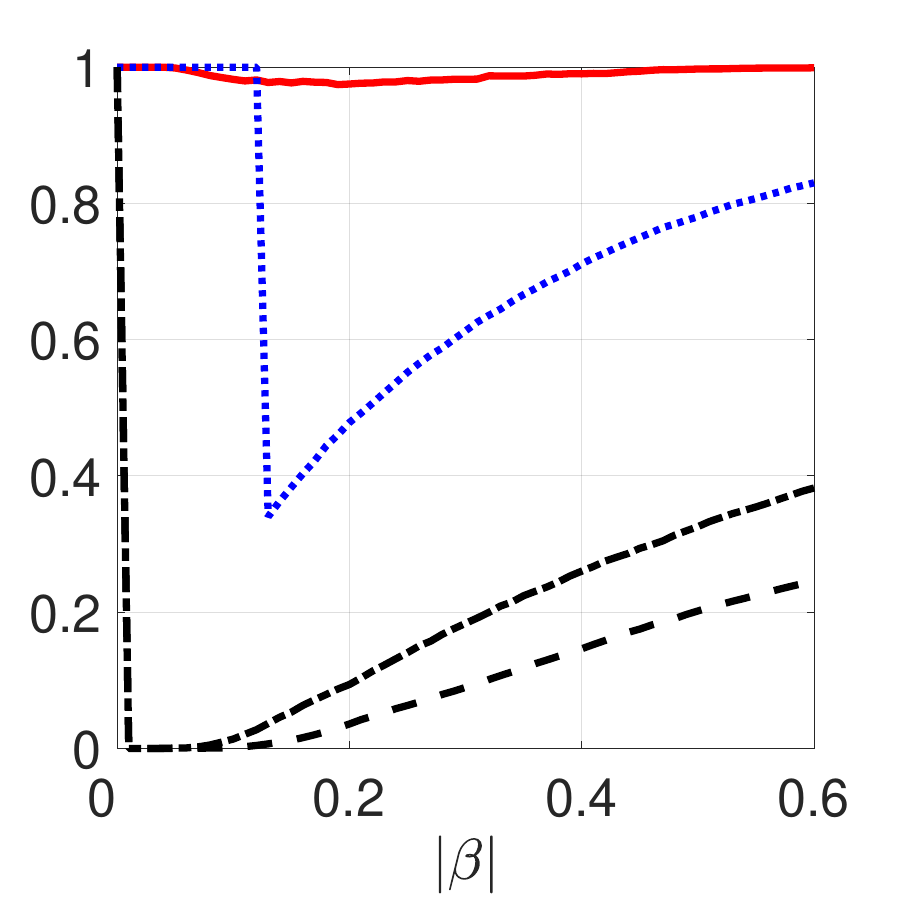}
	\end{subfigure}\begin{subfigure}{0.2\textwidth}
		\centering
		\includegraphics[trim={0cm 0cm 0.50cm 0.5cm},width=\textwidth,clip]
		{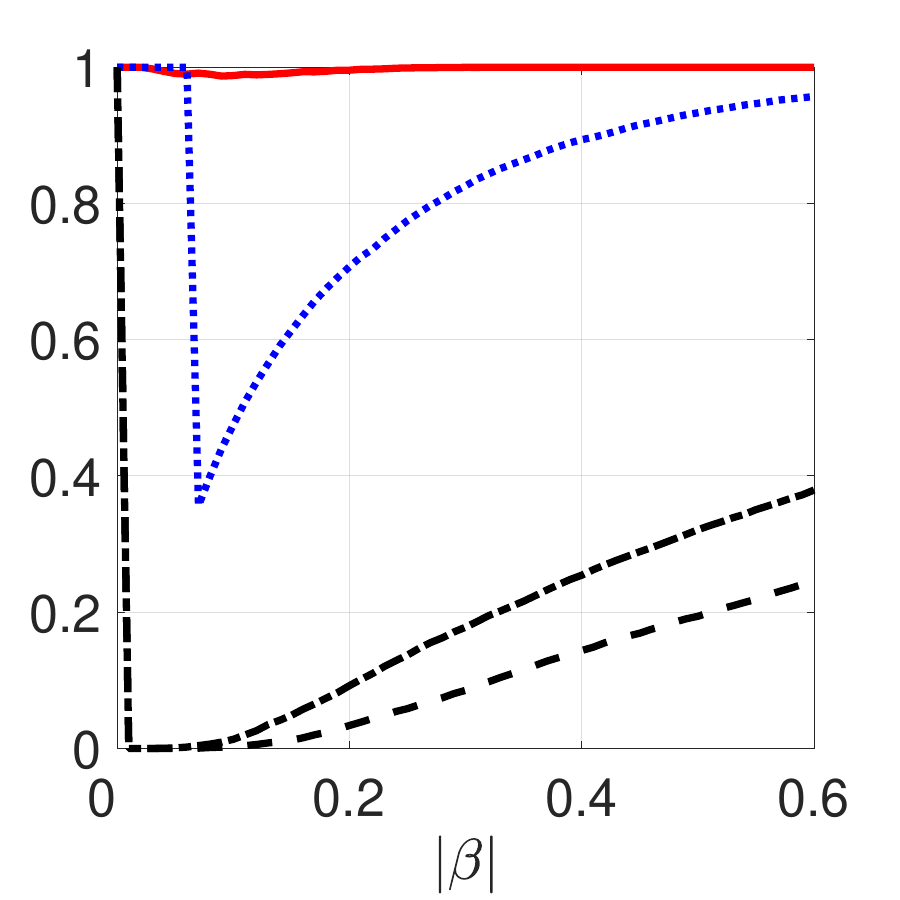}
	\end{subfigure}
	
\end{center}

\vspace{-2ex} 

\caption{Coverage probabilities of confidence intervals for scaled $\lambda_T$}

\label{fig:coverage_4lambda}

\end{figure}

\begin{figure}[ht]
\begin{center}
	\caption*{$\lambda_T = 4\times T^{1/4}$}
	\vspace{-1.5ex}
	\begin{subfigure}{0.2\textwidth}
		\centering
		\caption*{$T = 25$}
		\vspace{-1.5ex}
		\includegraphics[trim={0cm 0cm 0.50cm 0.5cm},width=\textwidth,clip]
		{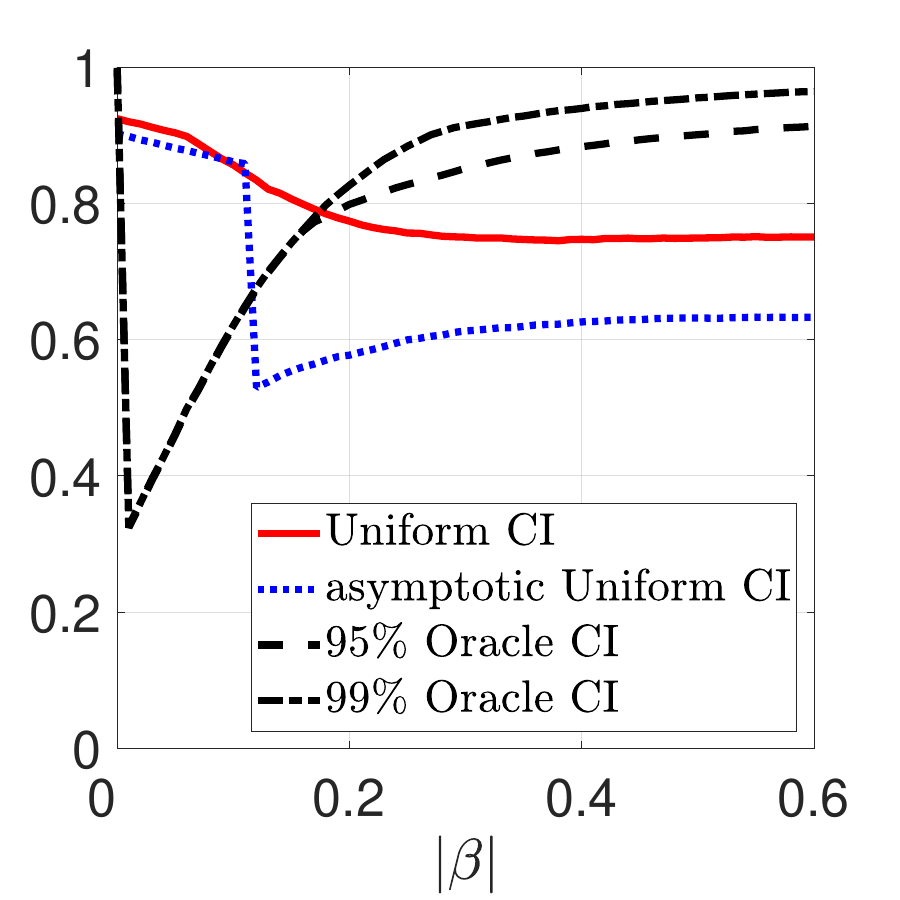}
	\end{subfigure}\begin{subfigure}{0.2\textwidth}
		\centering
		\caption*{$T = 50$}
		\vspace{-1.5ex}
		\includegraphics[trim={0cm 0cm 0.50cm 0.5cm},width=\textwidth,clip]
		{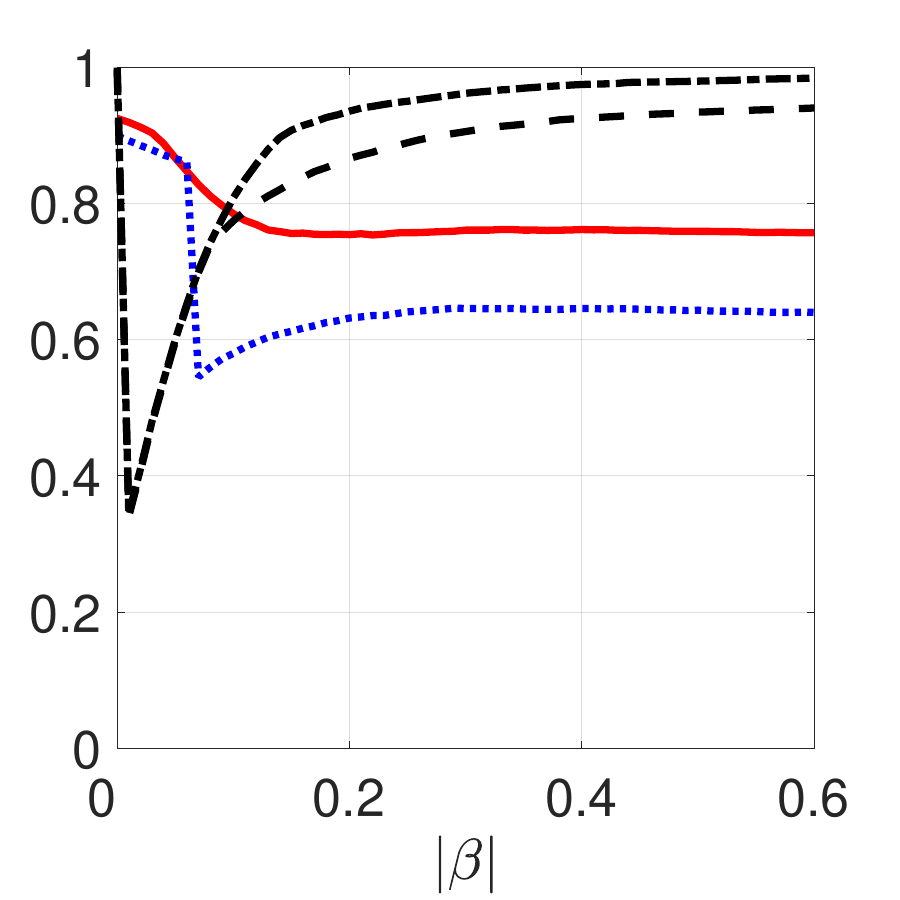}
	\end{subfigure}\begin{subfigure}{0.2\textwidth}
		\centering
		\caption*{$T = 100$}
		\vspace{-1.5ex}
		\includegraphics[trim={0cm 0cm 0.50cm 0.5cm},width=\textwidth,clip]
		{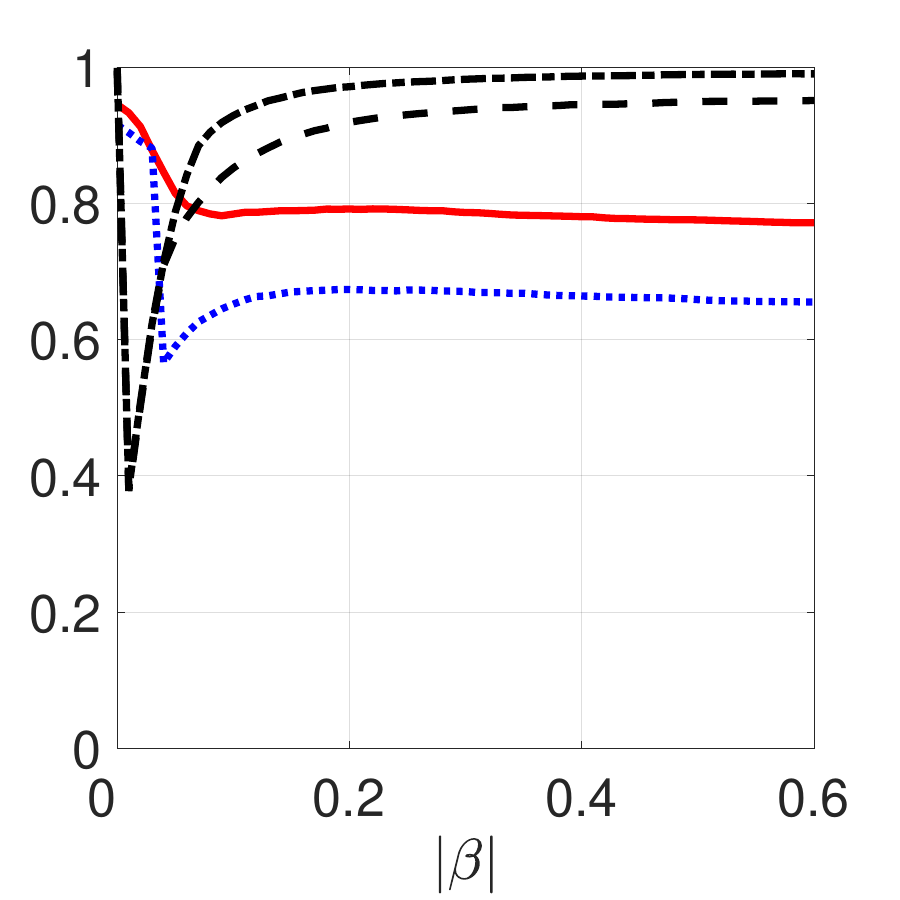}
	\end{subfigure}\begin{subfigure}{0.2\textwidth}
		\centering
		\caption*{$T = 250$}
		\vspace{-1.5ex}
		\includegraphics[trim={0cm 0cm 0.50cm 0.5cm},width=\textwidth,clip]
		{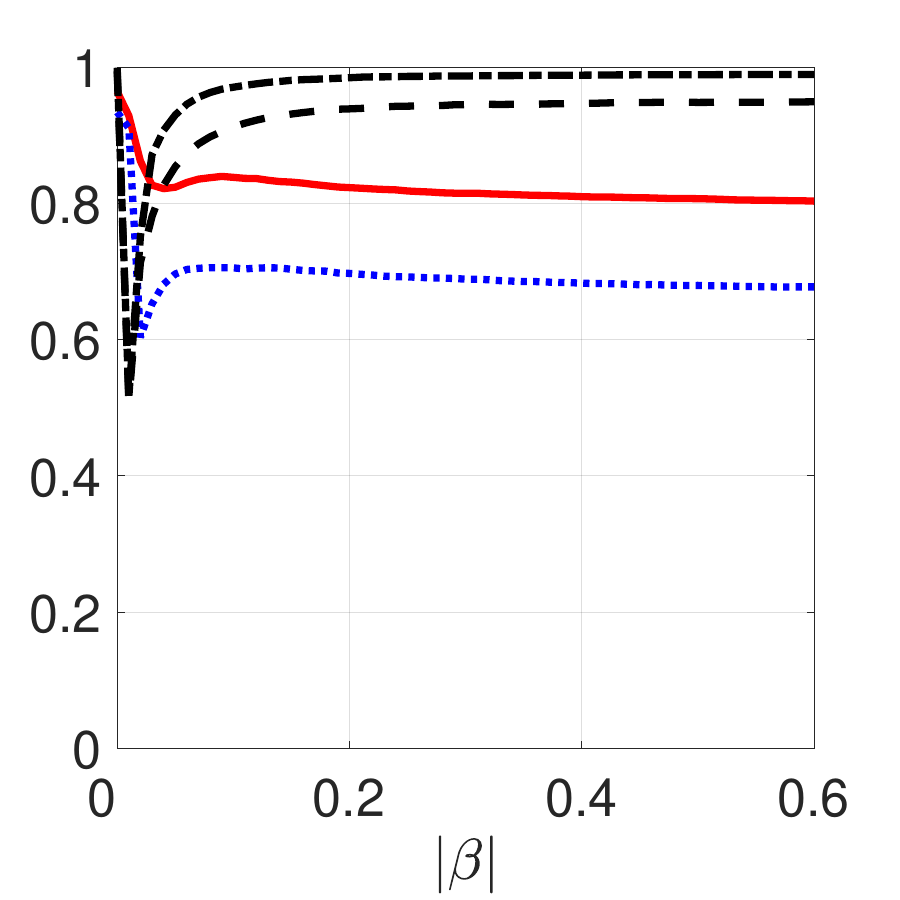}
	\end{subfigure}\begin{subfigure}{0.2\textwidth}
		\centering
		\caption*{$T = 1000$}
		\vspace{-1.5ex}
		\includegraphics[trim={0cm 0cm 0.50cm 0.5cm},width=\textwidth,clip]
		{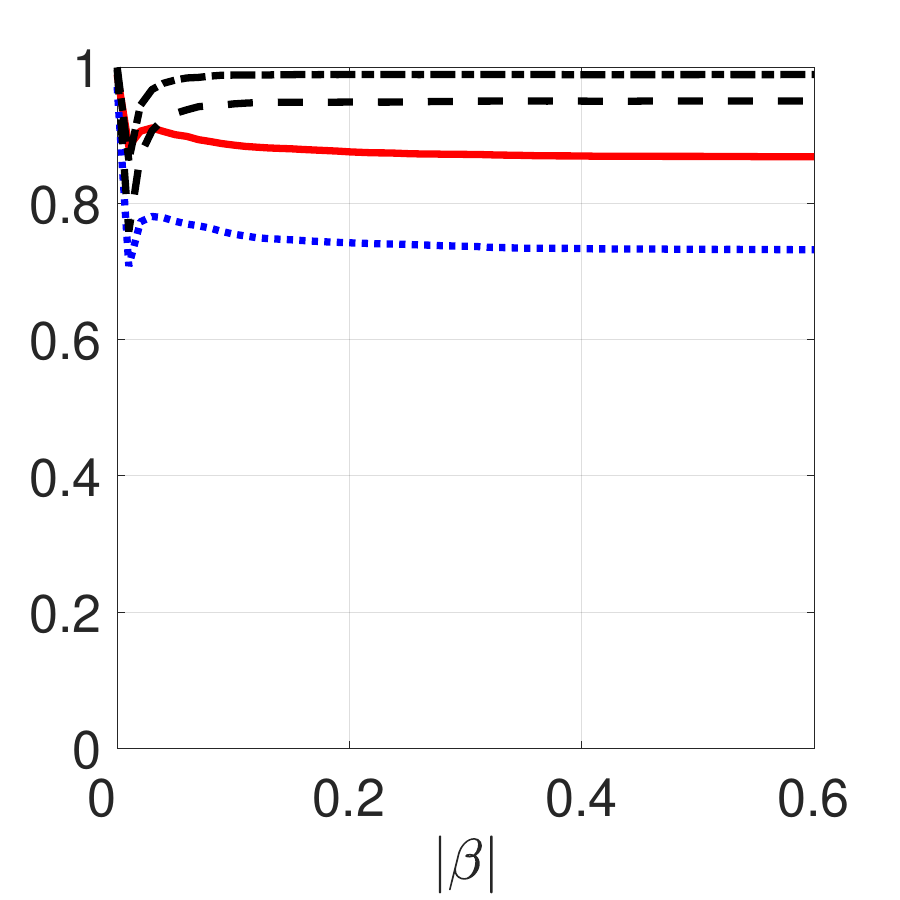}
	\end{subfigure}
	
	\caption*{$\lambda_T = 4\times T^{1/2}$}
	\vspace{-1.5ex}
	\begin{subfigure}{0.2\textwidth}
		\centering
		\includegraphics[trim={0cm 0cm 0.50cm 0.5cm},width=\textwidth,clip]
		{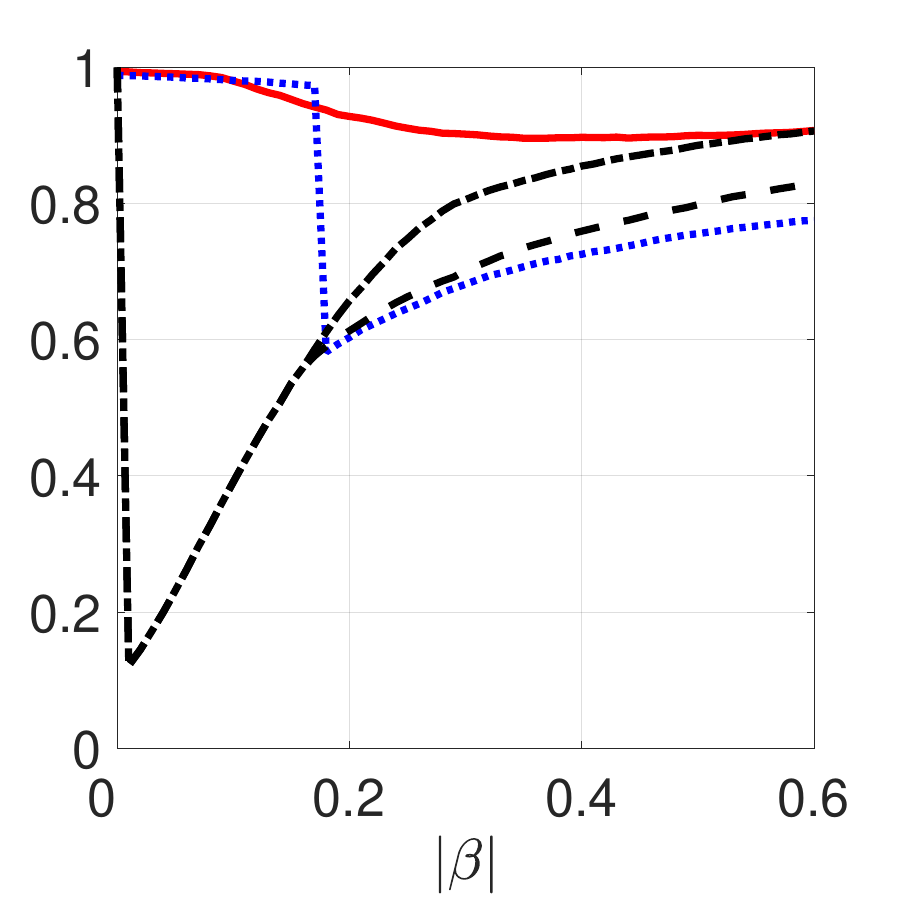}
	\end{subfigure}\begin{subfigure}{0.2\textwidth}
		\centering
		\includegraphics[trim={0cm 0cm 0.50cm 0.5cm},width=\textwidth,clip]
		{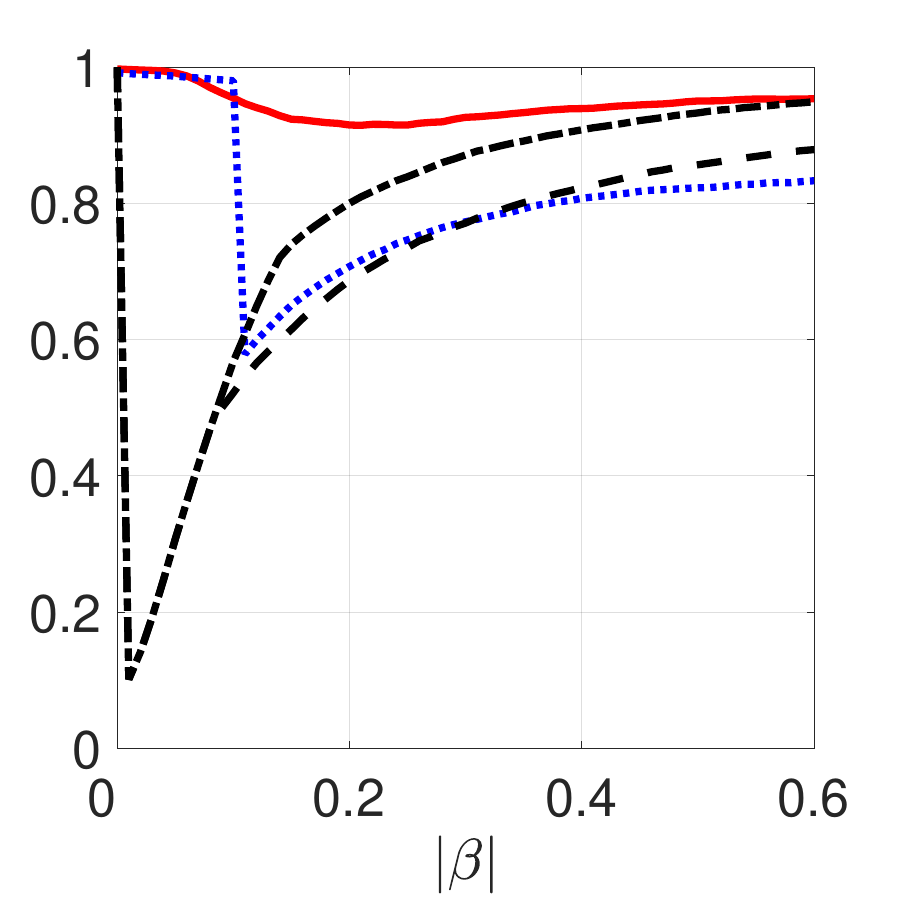}
	\end{subfigure}\begin{subfigure}{0.2\textwidth}
		\centering
		\includegraphics[trim={0cm 0cm 0.50cm 0.5cm},width=\textwidth,clip]
		{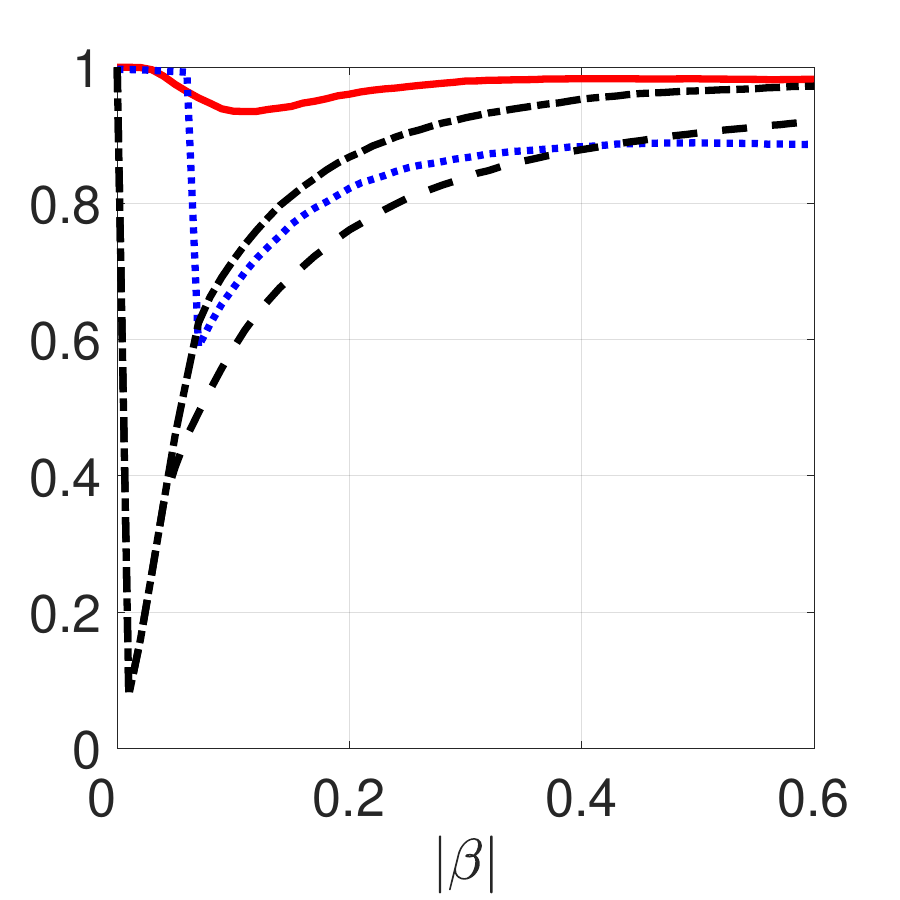}
	\end{subfigure}\begin{subfigure}{0.2\textwidth}
		\centering
		\includegraphics[trim={0cm 0cm 0.50cm 0.5cm},width=\textwidth,clip]
		{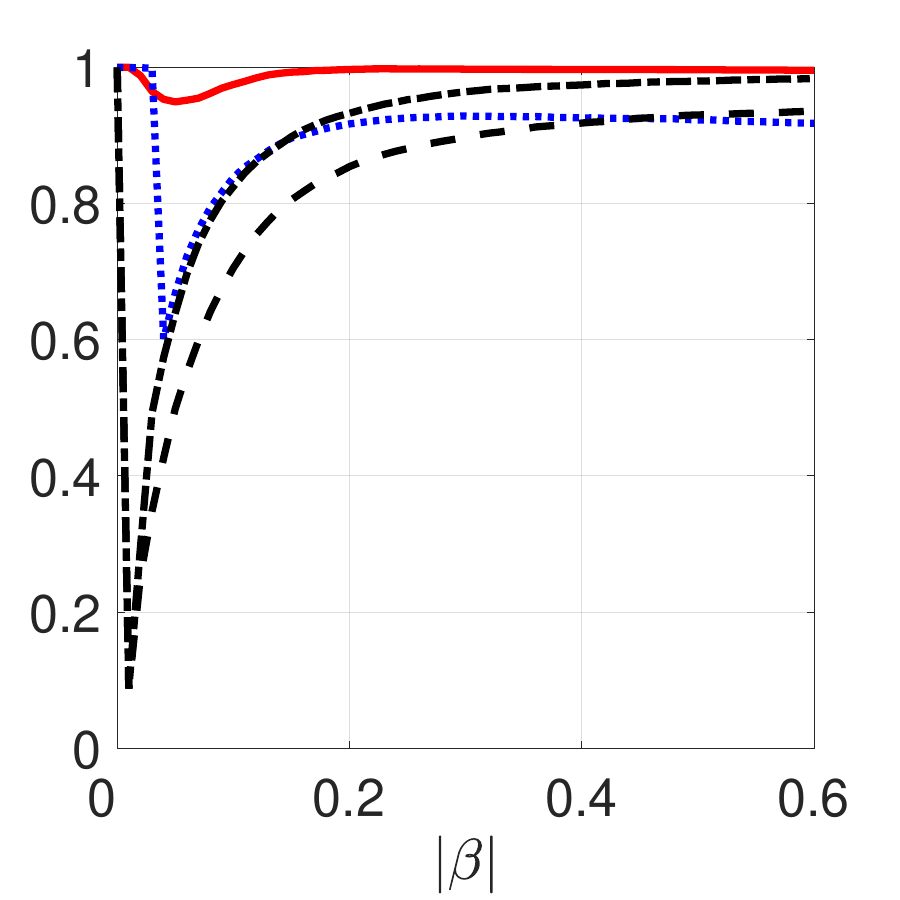}
	\end{subfigure}\begin{subfigure}{0.2\textwidth}
		\centering
		\includegraphics[trim={0cm 0cm 0.50cm 0.5cm},width=\textwidth,clip]
		{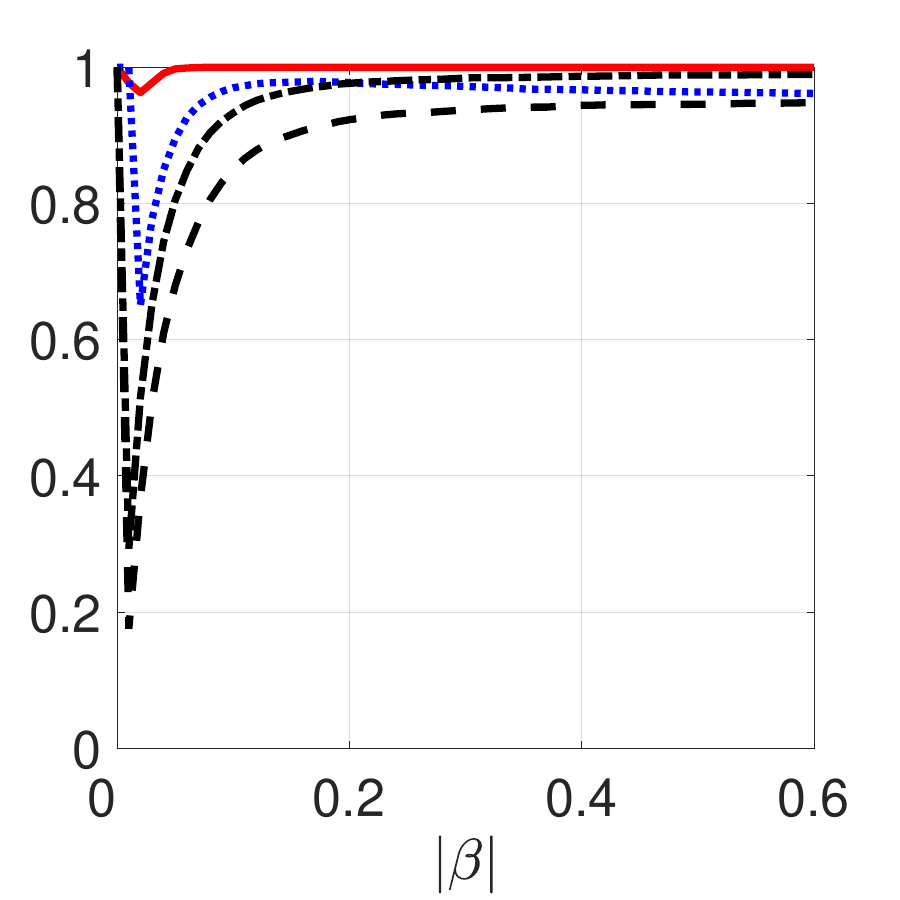}
	\end{subfigure}
	
	\caption*{$\lambda_T = 4\times T$}
	\vspace{-1.5ex}
	\begin{subfigure}{0.2\textwidth}
		\centering
		\includegraphics[trim={0cm 0cm 0.50cm 0.5cm},width=\textwidth,clip]
		{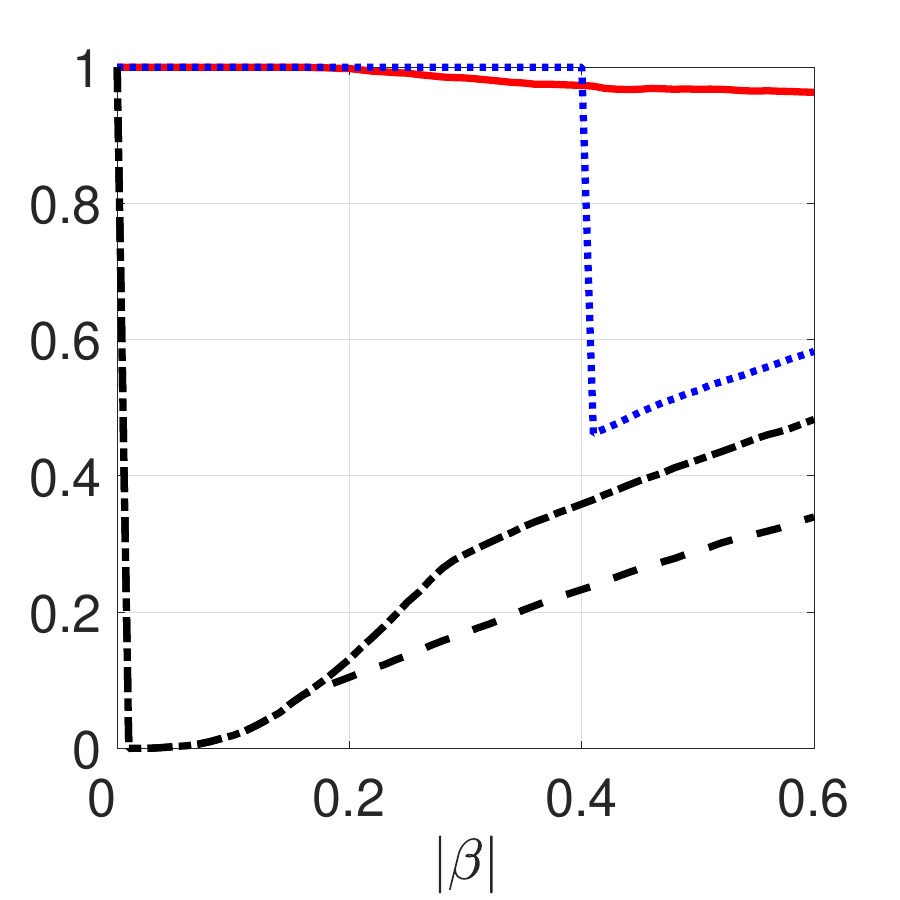}
	\end{subfigure}\begin{subfigure}{0.2\textwidth}
		\centering
		\includegraphics[trim={0cm 0cm 0.50cm 0.5cm},width=\textwidth,clip]
		{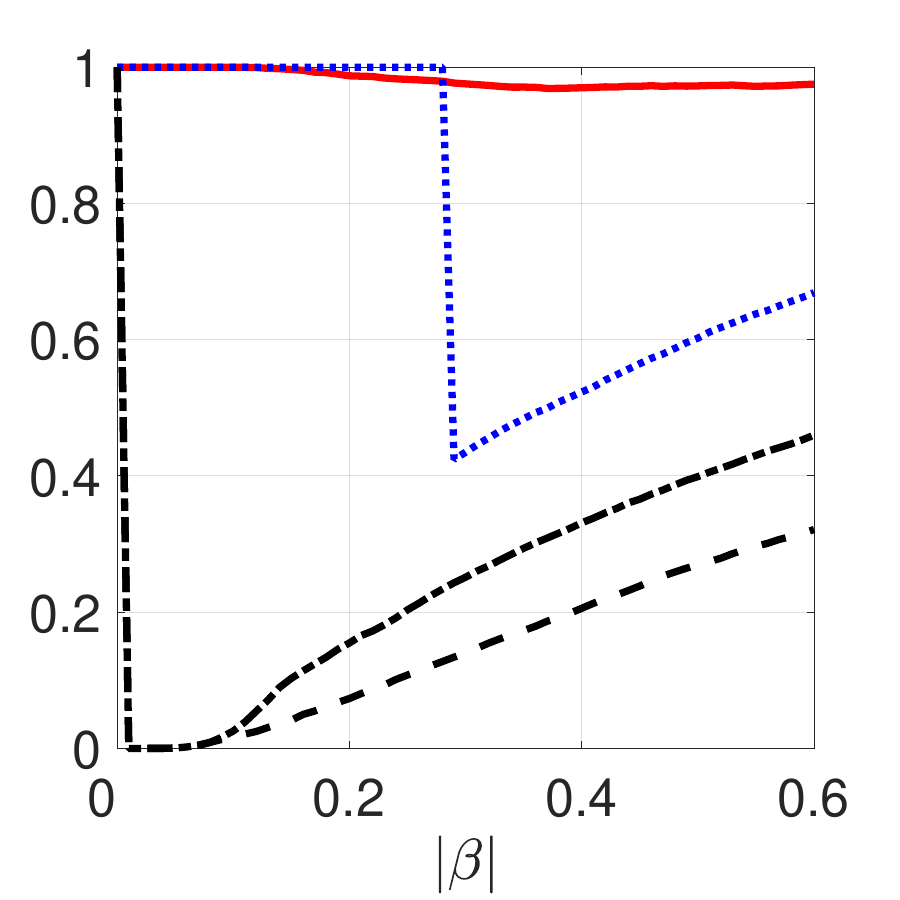}
	\end{subfigure}\begin{subfigure}{0.2\textwidth}
		\centering
		\includegraphics[trim={0cm 0cm 0.50cm 0.5cm},width=\textwidth,clip]
		{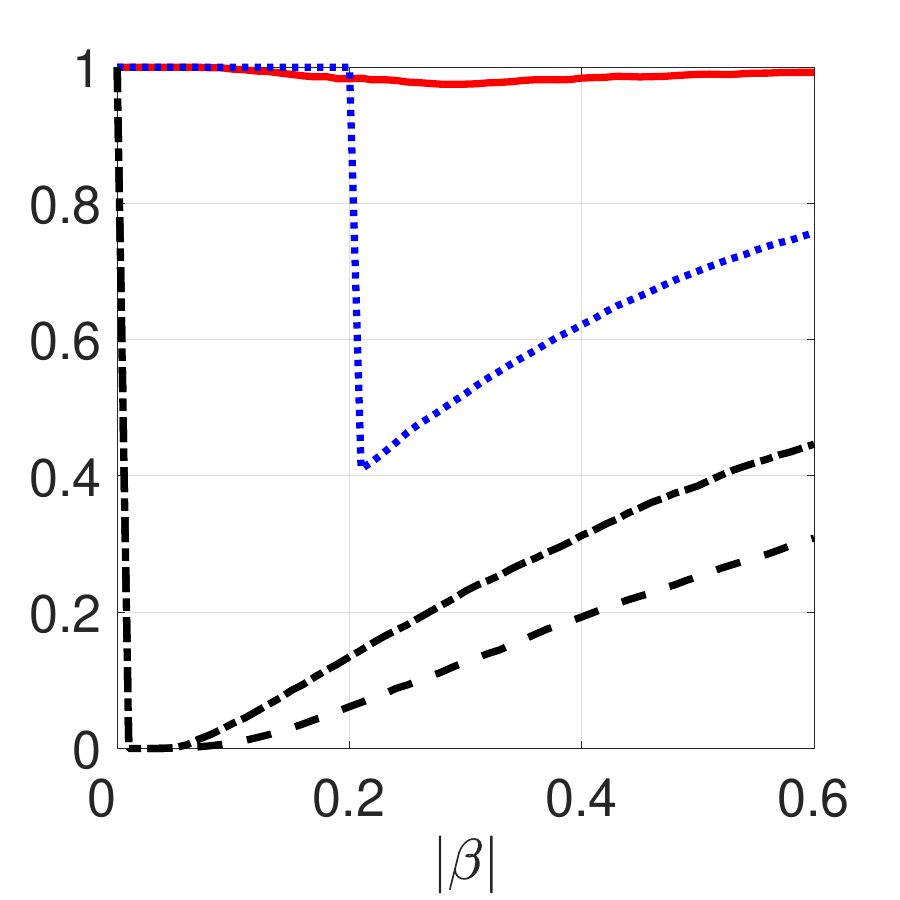}
	\end{subfigure}\begin{subfigure}{0.2\textwidth}
		\centering
		\includegraphics[trim={0cm 0cm 0.50cm 0.5cm},width=\textwidth,clip]
		{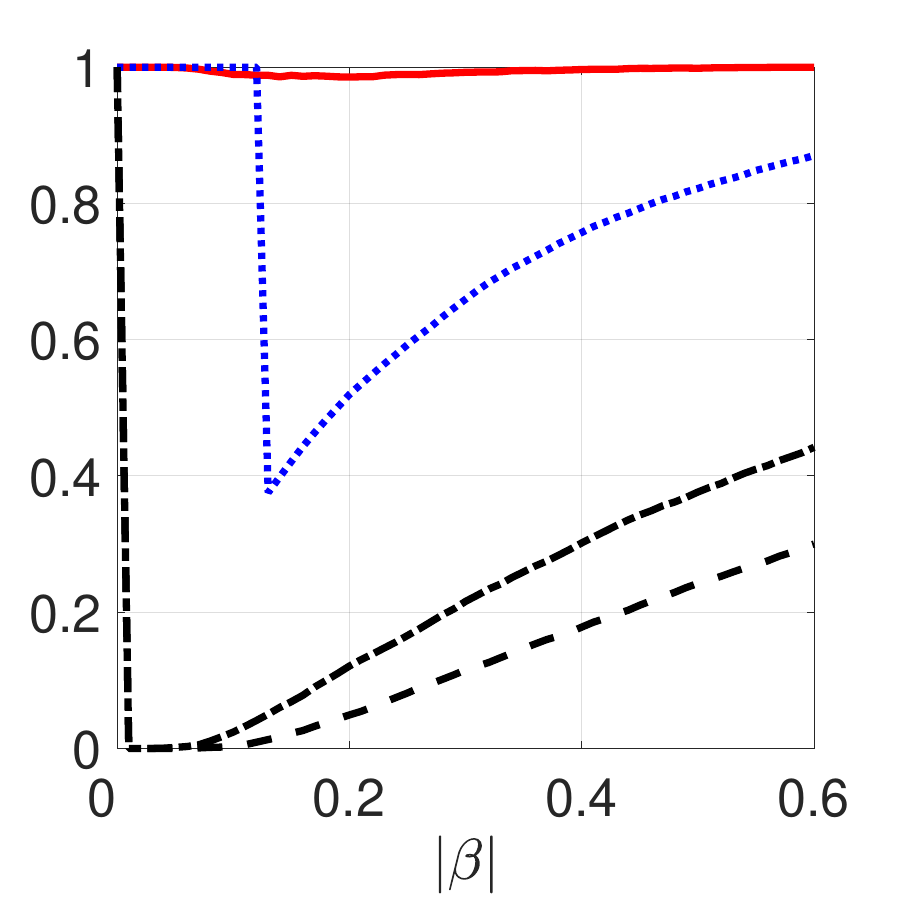}
	\end{subfigure}\begin{subfigure}{0.2\textwidth}
		\centering
		\includegraphics[trim={0cm 0cm 0.50cm 0.5cm},width=\textwidth,clip]
		{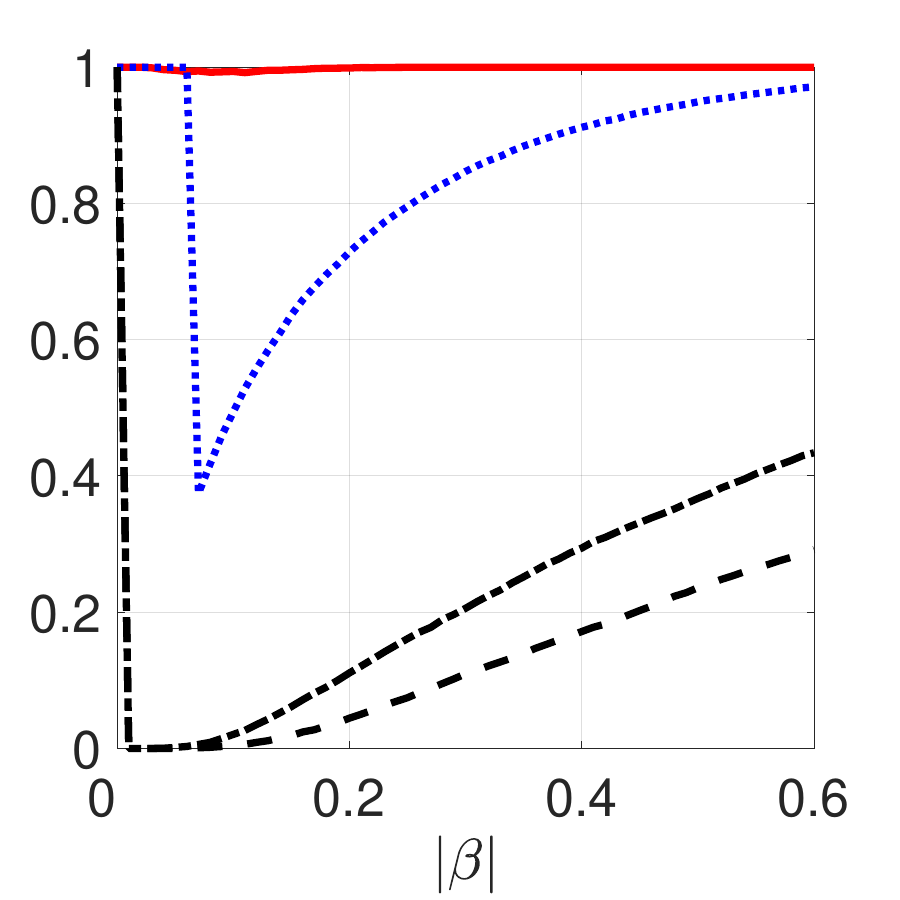}
	\end{subfigure}
	
\end{center}

\vspace{-2ex} 

\caption{Coverage probabilities of confidence intervals for scaled
$\lambda_T$ in the presence of error serial correlation and regressor
endogeneity}

\label{fig:coverage_4lambda_rho6}

\end{figure}

\newpage
\clearpage

\section{Empirical Illustration} \label{sec:emp}

We apply the adaptive LASSO estimator within a predictive regression
framework to forecast the U.S. monthly unemployment rate
(\texttt{UNRATE}). As potential predictors, we include the variables
considered by \citet[Table~1]{BuckmannJoseph23}: the 3-month Treasury
bill (\texttt{TB3MS}), real personal income (\texttt{RPI}), industrial
production (\texttt{INDPRO}), consumption (\texttt{DPCERA3M086SBEA}),
the S\&P 500 price index (\texttt{S\&P 500}), business loans
(\texttt{BUSLOANS}), the consumer price index (\texttt{CPIAUCSL}), the
oil price (\texttt{OILPRICEx}), and the M2 money stock
(\texttt{M2SL}). In addition, we include the four variables most
frequently selected by the standardized LASSO in
\citet[Table~4]{MeiShi24} when predicting the U.S. unemployment rate
one month ahead using a 20-year rolling window: initial jobless claims
(\texttt{CLAIMSx}), the number of unemployed less than 5 weeks
(\texttt{UEMPLT5}), the number of unemployed 5 to 14 weeks
(\texttt{UEMP5TO14}), and the number of unemployed 15 weeks and over
(\texttt{UEMP15OV}). All series are obtained from the FRED-MD
macroeconomic database \citep{McCrackenNg16}, with their respective
FRED-MD codes indicated in parentheses. Throughout, we use the raw
data without applying any transformations.

\cite{MeiShi24} demonstrate the usefulness of LASSO-type methods for
predicting the U.S. unemployment rate using the full set of FRED-MD
variables, considering multiple horizons (1, 2, and 3 months) and
rolling windows of 10, 20, and 30 years, and benchmarking their
results against a random walk with drift and an autoregressive model.
In contrast, our focus is not on relative predictive performance, but
on quantifying uncertainty around adaptive LASSO estimates using the
uniformly valid confidence intervals proposed in
Section~\ref{subsec:ci}. Although we compare the magnitudes of
adaptive LASSO coefficients to their OLS counterparts, oracle-based
confidence intervals for OLS are infeasible because the limiting
distribution of the OLS estimator is distorted by nuisance parameters
arising from endogeneity, serial correlation, and local-to-unity
predictors.

We report results for one-month-ahead out-of-sample forecasts based on
a 20-year rolling window using data from January 1959 to December
2024. The sample from January 1959 to December 1979 is used for
initial estimation, while forecasts are evaluated over the period
January 1980 to December 2024.

Within each rolling window, the penalization parameter $\lambda_T$ is
selected via time-series cross-validation following
\citet[Chapter~5.10]{HyndmanAthanasopoulos18}. For each candidate
value of $\lambda_T$, the adaptive LASSO is calculated using the first
60\% of observations in the window and then used to generate a
one-month-ahead forecast. The estimation sample is then expanded
recursively by one observation at a time until the end of the window,
producing a sequence of forecasts. The value of $\lambda_T$ is chosen
to minimize the resulting root mean squared forecasting error across
those forecasts. The candidate set for $\lambda_T$ consists of a grid
ranging from zero to the smallest value that shrinks all coefficients
to zero when estimated on the full window.

The left panel of Figure~\ref{fig:emp_forecasts} in
Appendix~\ref{app:emp} shows the adaptive LASSO and OLS forecasts
alongside the observed unemployment rate, while the right panel
reports the corresponding forecast errors. Over the full evaluation
period, OLS attains a root mean squared forecasting error of $0.82$,
whereas the adaptive LASSO reduces this by $11\%$ to $0.73$. This
improvement largely reflects the adaptive LASSO’s ability to
accommodate structural changes during and in the aftermath of the
COVID-19 pandemic. The spike in forecast errors both for OLS and
adaptive LASSO in the beginning of COVID-19 is driven by the abrupt
surge in \texttt{CLAIMSx}. While the adaptive LASSO adjusts rapidly in
subsequent months, the OLS estimator fails to do so, as illustrated in
Figure~\ref{fig:emp_coeffs1}.

In line with the motivation of this paper, we find that the
coefficients for the four labor-market variables (\texttt{CLAIMSx},
\texttt{UEMPLT5}, \texttt{UEMP5TO14}, \texttt{UEMP15OV}) are
frequently estimated to be small but non-zero, whereas the remaining
coefficients are often shrunk to zero by the adaptive LASSO.
Figure~\ref{fig:emp_coeffs1} reports the rolling-window adaptive LASSO
estimates for the coefficients corresponding to the labor-market
variables together with their uniformly valid confidence intervals,
benchmarked against the corresponding OLS estimates. The resulting
confidence intervals appear plausible, often widening during and in
the aftermath of crisis episodes. This behavior can be partly
attributed to increases in the penalization parameter $\lambda_T$ (see
Figure~\ref{fig:emp_lambda} in Appendix~\ref{app:emp}) and partly to
changes in the underlying variables. Importantly, a larger $\lambda_T$
does not necessarily imply wider confidence intervals. For example,
during and in the aftermath of the COVID-19 crisis, $\lambda_T$ is
elevated, yet the confidence intervals for the coefficients on
\texttt{CLAIMSx} and \texttt{UEMP5TO14} become noticeably narrower.

Figure~\ref{fig:emp_coeffs2} in Appendix~\ref{app:emp} shows the
results for the remaining variables. Although the adaptive LASSO often
sets their coefficients to zero, the associated uncertainty can still
be relatively large, particularly during crises.

Overall, the application highlights the usefulness of the adaptive
LASSO for estimating relationships among economic variables in the
presence of structural changes or shocks. The confidence intervals
proposed in this paper are plausible and allow to quantify uncertainty
around adaptive LASSO estimates. Although the confidence intervals can
occasionally be wide, they are robust to endogeneity, serial
correlation, and local-to-unity parameters, making them a valuable
tool for empirical applications.

\begin{figure}[ht]
\begin{center}
	
	\begin{subfigure}{0.4\textwidth}
		\centering
		\includegraphics[trim={0cm 0cm 0.50cm 0cm},width=\textwidth,clip]
		{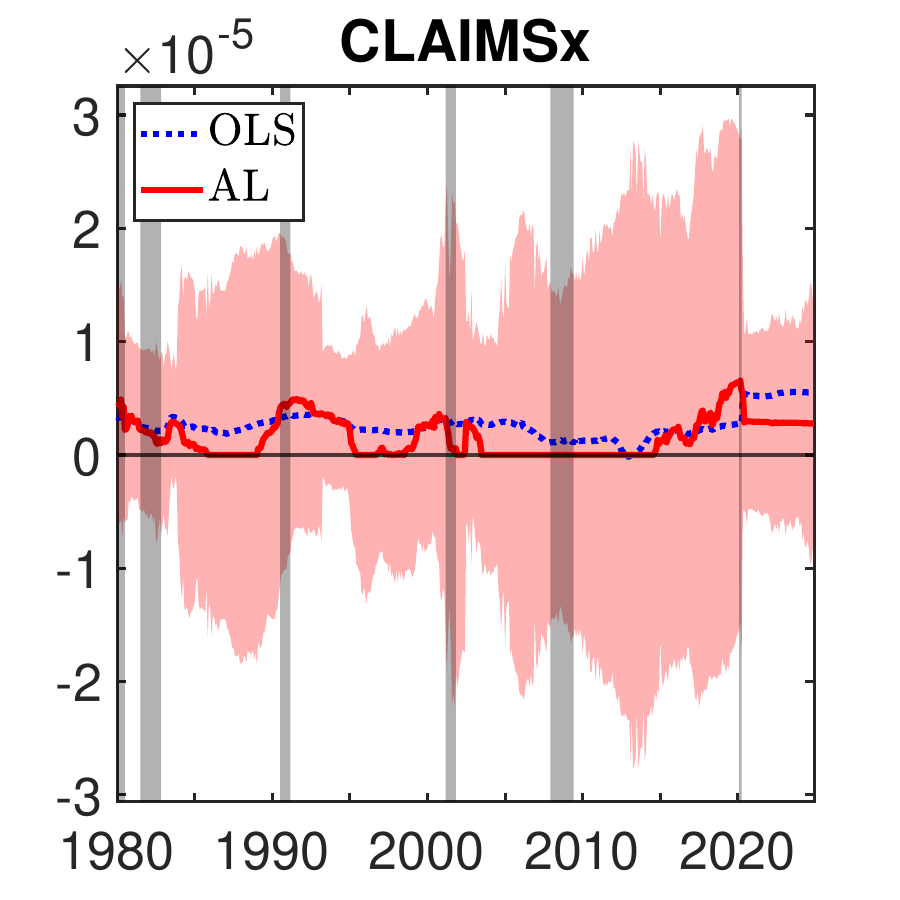}
	\end{subfigure}\begin{subfigure}{0.4\textwidth}
		\centering
		\includegraphics[trim={0cm 0cm 0.50cm 0cm},width=\textwidth,clip]
		{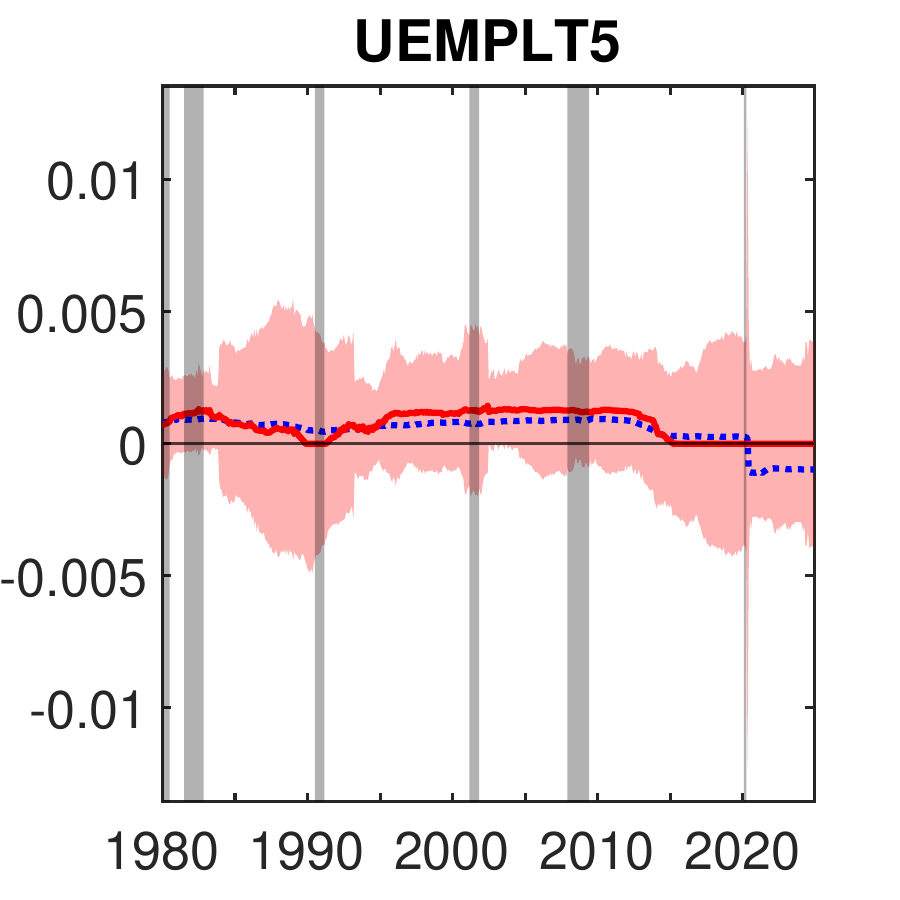}
	\end{subfigure}
	
	\begin{subfigure}{0.4\textwidth}
		\centering
		\includegraphics[trim={0cm 0cm 0.50cm 0cm},width=\textwidth,clip]
		{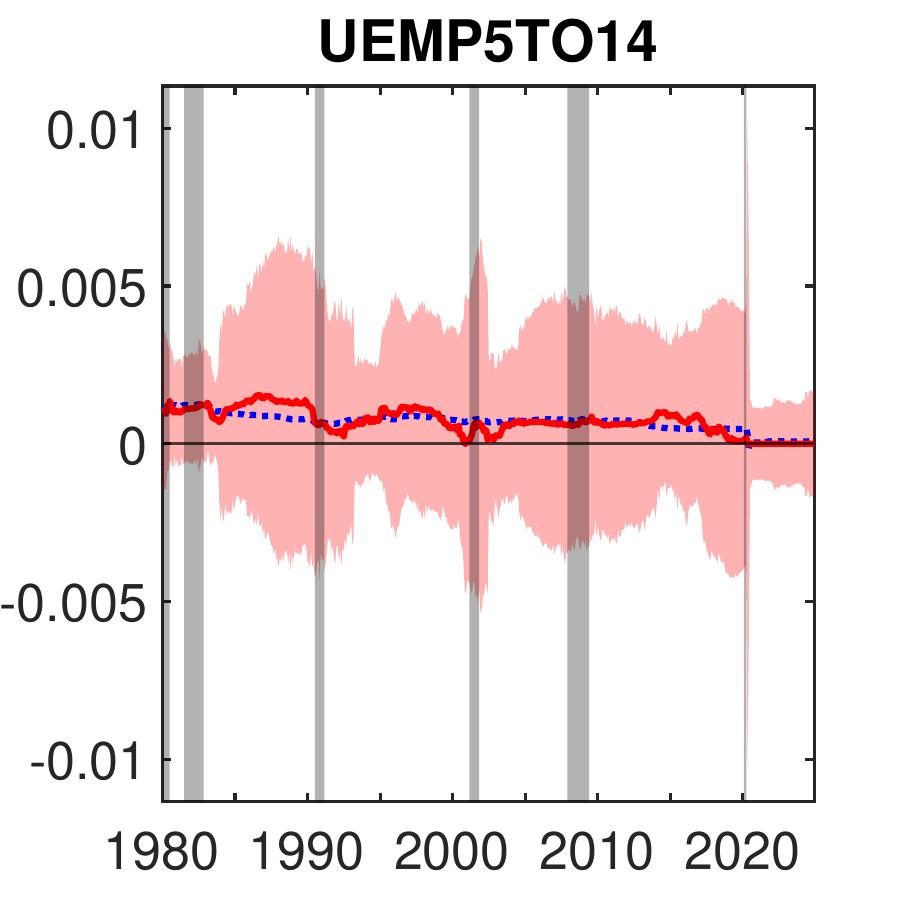}
	\end{subfigure}\begin{subfigure}{0.4\textwidth}
		\centering
		\includegraphics[trim={0cm 0cm 0.50cm 0cm},width=\textwidth,clip]
		{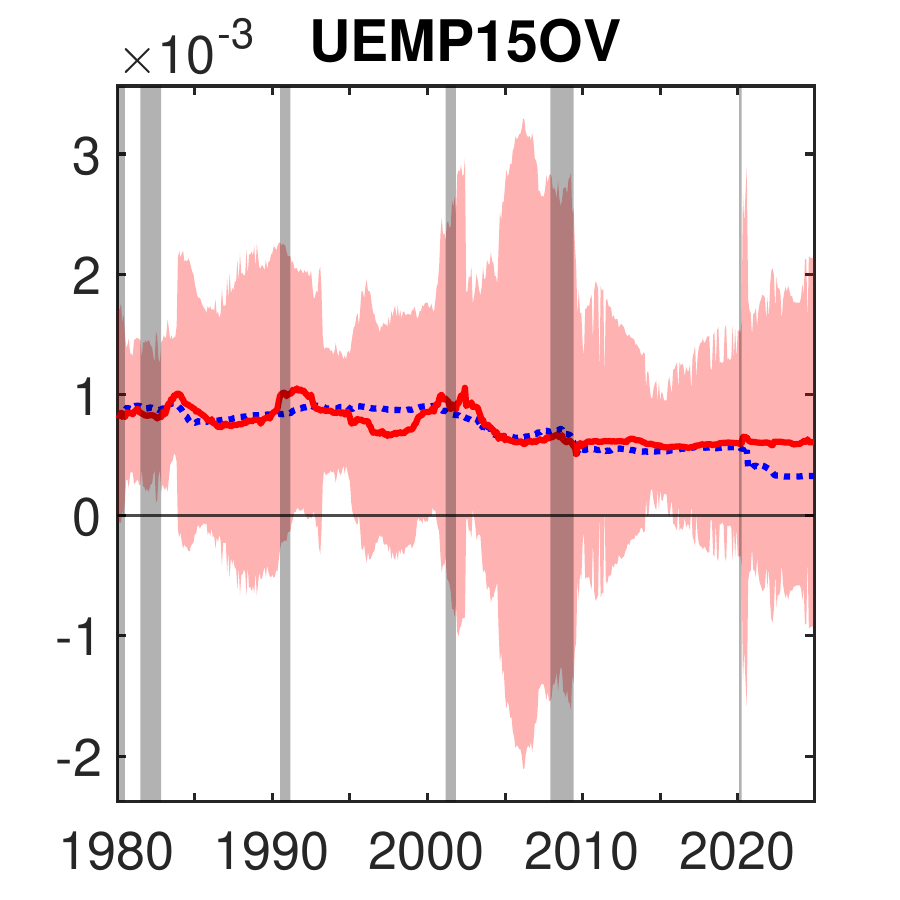}
	\end{subfigure}
	
\end{center}
\vspace{-2ex} 
\caption{20-year rolling window coefficient estimates and adaptive LASSO-based confidence intervals for labor-market variables.}
\label{fig:emp_coeffs1}
\end{figure}

\section{Summary and Conclusions} \label{sec:conclusion}

This paper analyzes the asymptotic behavior of the adaptive LASSO
estimator in cointegrating regressions with local-to-unity regressors
under moving-parameter asymptotics. We establish model selection
probabilities, estimation consistency, limiting distributions, and
uniform convergence rates, as well as the fastest local-to-zero rates
that remain detectable by the estimator. As these rates depend
critically on the tuning regime, the results characterize the smallest
signal-to-noise ratios that can reliably be detected under both
conservative and consistent tuning. In addition, under consistent
tuning, we construct uniformly valid confidence regions for the
regression coefficients that are straightforward to compute and do not
require any knowledge or estimation of nuisance parameters associated
with long-run covariance matrices or local-to-unity parameters.

Our simulation study demonstrates that the finite-sample distribution
of the adaptive LASSO estimator often differs substantially from what
is implied by the oracle property. In contrast, the limiting
distributions derived under moving-parameter asymptotics provide
accurate approximations and also successfully capture the empirical
frequency with which coefficients are set to zero. As a result, the
proposed uniform confidence regions exhibit stable coverage
probabilities across the parameter space, whereas confidence regions
based on the oracle property perform poorly when the true coefficients
are close to zero. The empirical application complements these
findings by illustrating the usefulness of the proposed confidence
regions for quantifying uncertainty around adaptive LASSO estimates in
practice.

Several promising directions for future research emerge. First, an
extension of the analysis to the twin adaptive LASSO proposed in
\cite{LeeEtAl22} for settings that allow also for stationary
regressors and cointegration among local-to-unity regressors would
extend the framework to a broader class of econometric models. Second,
further work should focus on the properties of the proposed confidence
regions, including their potential usefulness for hypothesis testing.
Third, theoretical guidance on choosing the penalization parameter
that balances the performance of the adaptive LASSO estimator and the
size of the confidence regions would further enhance the empirical
applicability of the proposed methods. Finally, extending the results
to high-dimensional regressions represents a natural next step toward
data-rich applications.

\section*{Acknowledgements} \label{sec:acknowledge}

We thank the participants of a research seminar at the Vienna
University of Economics and Business in 2025 and of the 2025
Econometrics Workshop at TU Dortmund University for their valuable
comments and suggestions.

\paragraph*{Declaration of Interest} The authors have no conflicts of
interest to declare.




\begin{appendices}

\setcounter{equation}{0}
\renewcommand\theequation{\Alph{section}.\arabic{equation}}	
\setcounter{table}{0}
\renewcommand\thetable{\Alph{section}.\arabic{table}}
\setcounter{figure}{0}
\renewcommand\thefigure{\Alph{section}.\arabic{figure}}

\setcounter{lemma}{0}
\renewcommand\thelemma{\Alph{section}.\arabic{lemma}}

\section{Preparation} \label{app:prep}

In preparation for some of the asymptotic derivations in the
univariate case, we provide finite-sample expressions of model
selection probabilities and appropriately scaled estimation errors in
the following lemma.

\begin{lemma}[Finite sample results] \label{lem:fs_results}
With $\mathcal{Z}^c_T \coloneqq T(\hat\beta - \beta_T)$, $\zeta_{vv,T}
\coloneqq T^{-2} \sum_{t=1}^T x_t^2$, $\beta_{0,T} \coloneqq
T\beta_T$, $\tilde\beta_{0,T} \coloneqq \lambda_T^{-1/2}T\beta_T$, and
$\bar\beta_{0,T} \coloneqq \lambda_T^{-1}T\beta_T$, we obtain the
following expressions for the model selection probabilities and the
distributions of the adaptive LASSO estimator in finite samples.

\begin{enumerate}

\item \label{enum:fs_ms} $\mP\left(\betaAL = 0\right)$ can be written as
\begin{align} \label{eq:fs-ms_conserv}
& \mP\left(\zeta_{vv,T}^{1/2}\left|\mathcal{Z}^c_T + \beta_{0,T} \right| \leq 
\sqrt{\frac{\lambda_T}{2}}\right) \\ \label{eq:fs-ms_consist}
= \; & \mP\left(\zeta_{vv,T}^{1/2} \leq 
\frac{1}{\sqrt{2}}\left|\lambda_T^{-1/2}\mathcal{Z}^c_T + \tilde\beta_{0,T}\right|^{-1}\right).
\end{align}
\item \label{enum:fs_dist_conserv} $T(\betaAL - \beta_T)$ can be decomposed into
\begin{align} \label{eq:fs-dist_conserv}
 & \ind\left\{\betaAL \neq 0\right\} \left(
\mathcal{Z}^c_T - \frac{\lambda_T}{2\zeta_{vv,T}}(\mathcal{Z}^c_T + \beta_{0,T})^{-1}\right) - 
\ind\left\{\betaAL = 0\right\}\beta_{0,T} \\  \label{eq:fs-dist_consist_T}
= \; & \ind\left\{\betaAL \neq 0\right\} \left(\mathcal{Z}^c_T -  
\frac{1}{2\zeta_{vv,T}}\left(\lambda_T^{-1}\mathcal{Z}^c_T + \bar\beta_{0,T}\right)^{-1}\right) 
- \ind\left\{\betaAL = 0\right\}\beta_{0,T}.
\end{align}
\item \label{enum:fs_dist_consist} $\lambda_T^{-1/2}T(\betaAL - \beta_T)$ can be decomposed into
\begin{align} \label{eq:fs-dist_consist}
\ind\left\{\betaAL \neq 0\right\} \left(\lambda_T^{-1/2}\mathcal{Z}^c_T -  
\frac{1}{2\zeta_{vv,T}}\left(\lambda_T^{-1/2}\mathcal{Z}^c_T + \tilde\beta_{0,T}\right)^{-1} \right) -
\ind\left\{\betaAL = 0\right\}\tilde\beta_{0,T}.
\end{align}
\end{enumerate}

\end{lemma}

\begin{proof} To show~\ref{enum:fs_ms}, note that from~\eqref{eq:betaAL}, we get 
\begin{align*}
\mP\left(\betaAL = 0\right) &= \mP\left(\left|\hat\beta\right| 
\leq \sqrt{\tilde{\lambda}_T}\right) 
= \mP\left(\left|T\left(\hat \beta-\beta_T\right) + T\beta_T\right| \leq
\sqrt{\frac{\lambda_T}{2T^{-2}\sum_{t=1}^T x_t^2}}\right) \\
&= \mP\left(\left| \mathcal{Z}^c_T + \beta_{0,T} \right| \leq 
\sqrt{\frac{\lambda_T}{2\zeta_{vv,T}}}\right)
= \mP\left(\left| \lambda_T^{-1/2}\mathcal{Z}^c_T + \tilde\beta_{0,T} \right| 
\leq \sqrt{\frac{1}{2\zeta_{vv,T}}}\right),
\end{align*}
with the last two expressions immediately yielding
\eqref{eq:fs-ms_conserv} and \eqref{eq:fs-ms_consist}, respectively.

To show~\ref{enum:fs_dist_conserv}, note that again
from~\eqref{eq:betaAL}, it follows that 
\begin{align*}
T(\betaAL - \beta_T) &= \ind\left\{\betaAL \neq 0\right\} \left(T(\hat\beta - \beta_T) 
- \tilde\lambda_T\hat\beta^{-1}\right) - \ind\left\{\betaAL = 0\right\} T\beta_T \\
&= \ind\left\{\betaAL \neq 0\right\} \left(\mathcal{Z}^c_T  
- \frac{\lambda_T}{2\zeta_{vv,T}}\left(\mathcal{Z}^c_T + \beta_{0,T}\right)^{-1}\right) 
- \ind\left\{\betaAL = 0\right\} \beta_{0,T} \\
&= \ind\left\{\betaAL \neq 0\right\} \left(\mathcal{Z}^c_T  
- \frac{1}{2\zeta_{vv,T}}\left(\lambda_T^{-1}\mathcal{Z}^c_T + \tilde\beta_{0,T}\right)^{-1}\right) 
- \ind\left\{\betaAL = 0\right\} \beta_{0,T},
\end{align*}
where the last two expressions prove \eqref{eq:fs-dist_conserv} and
\eqref{eq:fs-dist_consist_T}, respectively.

Finally,~\ref{enum:fs_dist_consist} can be shown by appropriately
scaling \eqref{eq:fs-dist_consist_T}.
\end{proof}

For the proofs using Lemma~\ref{lem:fs_results}, note that
$\mathcal{Z}^c_T \Rightarrow \mathcal{Z}^c$ and $\zeta_{vv,T} \Rightarrow
\zeta_{vv}^c$, as well as $\beta_{0,T} \to \beta_0$, $\tilde\beta_{0,T}
\to \tilde\beta_0$, and  $\bar\beta_{0,T} \to \bar\beta_0$ in
$\mRquer$, the latter three in the notation of
Theorems~\ref{thm:ms-unif}--\ref{thm:ls_dist-unif}.

For the proofs in both the univariate as well as the multivariate
case, we repeatedly use the joint convergence of $\left(T^{-2}\sum_{t=1}^T
x_tx_t', T^{-1}\sum_{t=1}^T x_t u_t\right) \Rightarrow \left(\zeta_{vv}^c,\int_0^1
J_v^c(r)dB_u(r) + \Delta_{vu}\right)$, see \citet[Lemma~3.1]{Phillips88} together with the continuous mapping theorem, without stating it explicitly. For the unit root case see \citet[Lemma~2.1]{ParkPhillips88}.

\section{Proofs for Section~\ref{subsec:fixed}} \label{app:proofs-fixed}

\begin{proof}[Proof of Proposition~\ref{prop:ms-fixed}]
To prove \ref{enum:basic-noFN}, note that just like in the proof of
Lemma~\ref{lem:fs_results}, we have
$$
\mP\left(\betaAL = 0\right) = \mP\left(\left| \hat\beta\right| \leq \sqrt{\tilde{\lambda}_T}\right) 
= \mP\left(\left|\left(\hat\beta - \beta\right) + \beta\right| \leq
\sqrt{\frac{T^{-2}\lambda_T}{2T^{-2}\sum_{t=1}^T x_t^2}}\right).
$$
The result therefore follows from observing that $\hat\beta - \beta =
o_\mP(1)$, $\beta\neq 0$, $T^{-2}\sum_{t=1}^T x_t^2=O_\mP(1)$, and
$T^{-2}\lambda_T\to 0$.

To show (b), use \eqref{eq:fs-ms_conserv} in
Lemma~\ref{lem:fs_results} with $\beta_T \equiv 0$ to see that
$$
\mP\left(\betaAL = 0\right) \to \mP\left((\zeta_{vv}^c)^{1/2}\left|\mathcal{Z}^c\right|
\leq \sqrt{\frac{\lim_{T \to \infty}\lambda_T}{2}}\right).
$$
For (b1), observe that the limiting probability in the above display
is strictly less than 1 since $\lambda_T \to \lambda_0 < \infty$ in
this case and the distribution of $\mathcal{Z}^c$ has unbounded support.
For (b2), note that $\lambda_T \to \infty$ in  this case and that
$(\zeta_{vv}^c)^{1/2}\left|\mathcal{Z}^c\right|$ is bounded in
probability.
\end{proof}


\begin{proof}[Proof of Proposition~\ref{prop:param_consist-fixed}] Part~\ref{enum:param_consist-fixed} 
is a special case of
Theorem~\ref{thm:param_consist-unif}\ref{enum:param_consist-unif}.
Part~\ref{enum:rateT-fixed} follows directly from
Proposition~\ref{prop:ls_dist-fixed}(\ref{enum:ls_dist_conserv-fixed}
and (b1)) and~\ref{enum:slower_rate-fixed} is a direct consequence of
Proposition~\ref{prop:ls_dist-fixed}(b2).
\end{proof}


\begin{proof}[Proof of Proposition~\ref{prop:ls_dist-fixed}] 

Parts~(a1) and (a2) are special cases of
Theorem~\ref{thm:ls_dist-unif}\ref{enum:ls_dist_conserv-unif}, with
$|\beta_0| = \infty$ and $\beta_0 = 0$, respectively.

To prove (b1), note that from~\eqref{eq:betaAL}, we can deduce that
\begin{equation} \label{eq:ProofProp3-b1}
T(\betaAL - \beta) = \ind\left\{\betaAL \neq 0\right\} \left(T\left(\hat \beta - \beta\right) 
- T\tilde \lambda_T \hat\beta^{-1} \right) - \ind\left\{\betaAL = 0\right\}T\beta.
\end{equation}
In case $\beta \neq 0$, we get from
Proposition~\ref{prop:ms-fixed}\ref{enum:basic-noFN} that
$\ind\left\{\betaAL \neq 0\right\} \plim 1$ as well as
$\ind\left\{\betaAL = 0\right\}T\beta = o_\mP(1)$, where the latter holds since 
$$
\mP\left(\ind\left\{\betaAL = 0\right\} |T\beta|> \eps \right) 
= \ind\left\{ |T\beta| > \eps \right\} \mP\left(\betaAL=0\right) \to 0.
$$
Thus, using estimation consistency of $\hat \beta$, $T\left(\hat \beta
- \beta\right) \Rightarrow \mathcal{Z}^c$, and $T\tilde \lambda_T =
\left(T^{-1}\lambda_T\right)/\left(2T^{-2}\sum_{t=1}^T x_t^2\right)
\Rightarrow \tilde \lambda_0/(2\zeta_{vv}^c)$, it follows that 
$$
T(\betaAL - \beta) \Rightarrow \mathcal{Z}^c - (\zeta_{vv}^c)^{-1}\frac{\tilde\lambda_0}{2\beta}.
$$
In case $\beta = 0$, \eqref{eq:ProofProp3-b1} reduces to
$$
T(\betaAL - \beta) = \ind\left\{\betaAL \neq 0\right\} \left(T\left(\hat \beta - \beta\right) 
- T\tilde \lambda_T \hat\beta^{-1} \right),
$$
which is $o_\mP(1)$ since
\begin{align*}
\mP\left(\ind\left\{\betaAL \neq 0 \right. \right. & \left. \left. \right\} 
\left| \left(\hat\beta - \beta\right) - 
T\tilde\lambda_T\hat\beta^{-1}\right| > \eps\right) 
\leq \mP\left(\ind\left\{\betaAL \neq 0 \right\} = 1\right) \\ 
& = \mP\left(\betaAL \neq 0\right) \to 0
\end{align*} 
by Proposition~\ref{prop:ms-fixed}(b2).

The proof of~(b2) is similar. First, again from~\eqref{eq:betaAL}, we get that
\begin{align} \label{eq:ProofProp3-b2}
\begin{split}
\lambda_T^{-1}T^2(\betaAL - \beta) = \ind&\left\{\betaAL \neq 0\right\} \left(\lambda_T^{-1}T^2
\left(\hat \beta - \beta\right) - \lambda_T^{-1}T^2\tilde \lambda_T \hat\beta^{-1} \right) \\
& \quad - \ind\left\{\betaAL = 0\right\}\lambda_T^{-1}T^2\beta.
\end{split}
\end{align}
As before, in case $\beta \neq 0$, we get that $\ind\left\{\betaAL
\neq 0\right\} \plim 1$ and $\ind\left\{\betaAL =
0\right\}\lambda_T^{-1}T^2\beta = o_\mP(1)$. Thus, from estimation
consistency of $\hat \beta$, $T\left(\hat \beta - \beta\right) =
O_\mP(1)$, $\lambda_T^{-1}T\to 0$, and $\lambda_T^{-1}T^2\tilde
\lambda_T = 1/\left(2T^{-2}\sum_{t=1}^T x_t^2\right) \Rightarrow
1/(2\zeta_{vv}^c)$, it follows that
$$
\lambda_T^{-1}T^2(\betaAL - \beta) \Rightarrow - (\zeta_{vv}^c)^{-1}\frac{1}{2\beta}.
$$
In case $\beta=0$,~\eqref{eq:ProofProp3-b2} reduces to
$$
\lambda_T^{-1}T^2(\betaAL - \beta) = 
\ind\left\{\betaAL \neq 0\right\} \left(\lambda_T^{-1}T^2\left(\hat \beta - \beta\right) 
- \lambda_T^{-1}T^2\tilde \lambda_T \hat\beta^{-1} \right),
$$
which is $o_\mP(1)$ by the same arguments as above.
\end{proof}


\begin{proof}[Proof of Corollary~\ref{cor:summary_fixed}]
The corollary follows directly from \ref{enum:basic-noFN} and (b2) in
Proposition~\ref{prop:ms-fixed} and from
Proposition~\ref{prop:ls_dist-fixed}(b1).
\end{proof}

\section{Proofs for Section~\ref{subsec:unif}} \label{app:proofs-unif}

\begin{proof}[Proof of Theorem~\ref{thm:ms-unif}] Part~\ref{enum:ms_conserv-unif}  
immediately follows from \eqref{eq:fs-ms_conserv} in
Lemma~\ref{lem:fs_results}\ref{enum:fs_ms} after letting $\mathcal{Z}^c_T$,
$\zeta_{vv,T}$, $\beta_{0,T}$, and $\lambda_T$ settle at their (weak) limits
$\mathcal{Z}^c$, $\zeta_{vv}^c$, $\beta_0$, and $\lambda_0$,
respectively. Also note that since $\mathcal{Z}^c$ has a distribution
with unbounded support, the limiting probability is smaller than one
for all $0 \leq \lambda_0 < \infty$ and $\beta_0 \in \mRquer$.

Similarly, \ref{enum:ms_consist-unif} can be deduced from
\eqref{eq:fs-ms_consist} in Lemma~\ref{lem:fs_results}\ref{enum:fs_ms}
in a straightforward manner after letting $\zeta_{vv,T}$ and
$\tilde\beta_{0,T}$ settle at their (weak) limits $\zeta_{vv}^c$ and
$\tilde\beta_0$, respectively, and noting that
$\lambda_T^{-1/2}\mathcal{Z}^c_T = o_{\mP}(1)$.
\end{proof}


\begin{proof}[Proof of Theorem~\ref{thm:param_consist-unif}] To show 
part~\ref{enum:param_consist-unif}, observe that \eqref{eq:betaAL}
implies
$$
\left|\betaAL - \beta_T\right| \leq \ind\left\{\left|
\hat\beta\right| > \sqrt{\tilde\lambda_T} \right\}\left|
\hat\beta - \tilde\lambda_T\hat\beta^{-1} - \beta_T\right| +
\ind\left\{\left|\hat\beta\right| \leq \sqrt{\tilde\lambda_T} \right\}\left|\beta_T \right|.
$$
The first term on the right-hand side of the above display is bounded by
$$
\ind\left\{\left| \hat\beta \right| > \sqrt{\tilde\lambda_T}
\right\} \left(\left| \hat\beta - \beta_T \right| + \left|
\tilde\lambda_T\hat\beta^{-1} \right|\right) \leq o_\mP(1) +
\sqrt{\tilde\lambda_T} = o_\mP(1),
$$
where the last equality follows from the fact that $\tilde\lambda_T =
0.5T^{-2}\lambda_T (T^{-2}\sum_{t=1}^Tx_t^2)^{-1} = o_\mP(1)$ since
$T^{-2}\sum_{t=1}^Tx_t^2 = O_\mP(1)$ and $T^{-2}\lambda_T \to 0$. For
the second term, we have
$$
\mP\left(\ind\left\{\left| \hat\beta \right| \leq \sqrt{\tilde\lambda_T} \right\}
\left| \beta_T \right| > \eps\right) 
=  \ind\left\{\left| \beta_T \right| > \eps\right\} 
\mP\left(\left| \hat\beta - \beta_T + \beta_T \right| \leq \sqrt{\tilde\lambda_T} \right) \to 0,
$$
as $\hat\beta - \beta_T = o_\mP(1)$, $\left| \beta_T \right| >
\eps$ if $\ind\left\{\left| \beta_T \right| > \eps\right\} =
1$, and $\tilde\lambda_T = o_\mP(1)$. Hence, $\ind\left\{\left|
\hat\beta \right| \leq \sqrt{\tilde\lambda_T} \right\} \left|
\beta_T \right| = o_\mP(1)$ also and we can conclude that
$\left| \betaAL - \beta_T \right| = o_\mP(1)$.

Part~\ref{enum:rate_conserv-unif} can be learned from
Theorem~\ref{thm:ls_dist-unif}\ref{enum:ls_dist_conserv-unif} since the
limiting distribution is stochastically bounded. For this, note that
$$
\ind\left\{(\zeta_{vv}^c)^{1/2}\left|\mathcal{Z}^c + \beta_0\right|
\leq \sqrt{\frac{\lambda_0}{2}}\right\} = 0
$$
for $|\beta_0| = \infty$ since $\lambda_0 < \infty$.

Part~\ref{enum:rate_consist-unif} follows directly from
Theorem~\ref{thm:ls_dist-unif}\ref{enum:ls_dist_consist-unif}.
\end{proof}


\begin{proof}[Proof of Theorem~\ref{thm:ls_dist-unif}] To show~\ref{enum:ls_dist_conserv-unif}, 
note that \eqref{eq:fs-dist_conserv} in Lemma~\ref{lem:fs_results}\ref{enum:fs_dist_conserv} states
$$
T(\betaAL - \beta_T) = \ind\left\{\betaAL \neq 0\right\} \left(\mathcal{Z}^c_T -
\frac{\lambda_T}{2\zeta_{vv,T}}\left(\mathcal{Z}^c_T+\beta_{0,T}\right)^{-1}\right) - 
\ind\left\{\betaAL = 0\right\}\beta_{0,T}.
$$
Observing that 
$$
\ind\left\{\betaAL = 0\right\} = \ind\left\{\zeta_{vv,T}^{1/2}\left|\mathcal{Z}^c_T
 + \beta_{0,T} \right| \leq \sqrt{\frac{\lambda_T}{2}}\right\}
$$
by \eqref{eq:fs-ms_conserv} in
Lemma~\ref{lem:fs_results}\ref{enum:fs_ms} and letting
$\zeta_{vv,T}$, $\mathcal{Z}^c_T$, $\beta_{0,T}$, and $\lambda_T$
settle at their respective (weak) limits yields the desired result.

To show~\ref{enum:ls_dist_consist-unif}, note that
\eqref{eq:fs-dist_conserv} in
Lemma~\ref{lem:fs_results}\ref{enum:fs_dist_conserv} states
\begin{align*}
\lambda_T^{-1/2}T(\betaAL - \beta_T) 
= & \ind\left\{\betaAL \neq 0\right\} \left(\lambda_T^{-1/2}\mathcal{Z}^c_T -  
\frac{1}{2\zeta_{vv,T}}\left(\lambda_T^{-1/2}\mathcal{Z}^c_T + \tilde\beta_{0,T}\right)^{-1} \right) \\
& - \ind\left\{\betaAL = 0\right\}\tilde\beta_{0,T}.
\end{align*}
Observing that 
$$
\ind\left\{\betaAL = 0\right\} =  \ind\left\{\zeta_{vv,T}^{1/2} \leq 
\frac{1}{\sqrt{2}}\left|\lambda_T^{-1/2}\mathcal{Z}^c_T + \tilde\beta_{0,T}\right|^{-1}\right\}
$$
by \eqref{eq:fs-ms_consist} in
Lemma~\ref{lem:fs_results}\ref{enum:fs_ms}, letting $\zeta_{vv,T}$
and $\tilde\beta_{0,T}$ settle at their respective (weak) limits and
recalling that $\lambda_T^{-1/2}\mathcal{Z}^c_T = o_\mP(1)$ immediately
yields the desired result when $\tilde\beta_0 \in \mR$. If $|\tilde\beta_0| =
\infty$, the fact that $\mP(\betaAL = 0) \to 0$ ensures
$\ind\{\betaAL = 0\}\tilde\beta_{0,T} = o_\mP(1)$, which
proves the claim for this case in a straight-forward manner also.
\end{proof}


\begin{proof}[Proof of Remark~\ref{rem:ls_dist_consist_rateT-unif}] 
Using Lemma~\ref{lem:fs_results}\ref{enum:fs_dist_conserv},
Equation~\eqref{eq:fs-dist_consist} states that $T(\betaAL - \beta_T)$
is given by
$$
\ind\left\{\betaAL \neq 0\right\} \left(\mathcal{Z}^c_T -  
\frac{1}{2\zeta_{vv,T}}\left(\lambda_T^{-1}\mathcal{Z}^c_T + \bar\beta_{0,T}\right)^{-1}\right) 
- \ind\left\{\betaAL = 0\right\}\beta_{0,T}.
$$
By Theorem~\ref{thm:ms-unif}\ref{enum:ms_consist-unif}, $\mP(\betaAL =
0)$ converges to one for $\tilde\beta_0 = 0$, to $0 < p < 1$ for $0 <
|\tilde \beta_0| < \infty$, and to zero for $|\tilde \beta_0|=
\infty$. Hence, for $\tilde\beta_0 = 0$, the first summand in the
above display is $o_\mP(1)$. Consequently, $T(\betaAL - \beta_T)
\Rightarrow  -\beta_0$. In case $0 < |\tilde \beta_0| < \infty$, note
that necessarily, $\beta_0 = \sign(\tilde\beta_0)\infty$ and
$\bar\beta_0 = 0$ hold. Since $\lambda_T^{-1}\mathcal{Z}^c_T =
o_\mP(1)$, the terms next to both indicator functions in the above
display tend to $-\sign(\tilde\beta_0)\infty$, allowing to deduce that
altogether $T(\betaAL - \beta_T)\Rightarrow
-\sign(\tilde\beta_0)\infty$. Finally, if $|\tilde\beta_0| = \infty$,
the second summand in the above display is $o_\mP(1)$, whereas in the
first summand, again, $\lambda_T^{-1}\mathcal{Z}^c = o_\mP(1)$, so that
$T(\betaAL - \beta_T) \Rightarrow \mathcal{Z}^c -
0.5(\zeta_{vv}^c\bar\beta_0)^{-1}$.
\end{proof}

\section{Proofs for Section~\ref{subsec:multi}} \label{app:proofs_multi}

For the proofs in the multivariate case, we introduce some additional
notation. Let $y \coloneqq [y_1,\ldots,y_T]' \in \mR^T$, $u \coloneqq
[u_1,\ldots,u_T]' \in \mR^T$, $X_j \coloneqq [x_{1,j},\ldots,x_{T,j}]'
\in \mR^T$, and $X=[X_1,\ldots,X_k] \in \mR^{T\times k}$, with
$x_{t,j}$ denoting the $j$-th element of $x_t$.


\begin{proof}[Proof of Theorem~\ref{thm:ms-multi}] 
To prove~\ref{enum:ms_conserv-multi}, note that, clearly, the
event $\{\betaALj = 0\}$ is a subset of the event $\{2|\betaALj -
\beta_{T,j}|>|\beta_{T,j}|\}$. Hence, 
$$
\mP\left(\betaALj = 0\right) \leq \mP\left(2|T(\betaALj-\beta_{T,j})|>|\beta_{0,T,j}|\right) \to 0,
$$
since $T(\betaALj-\beta_{T,j}) = O_\mP(1)$ by
Theorem~\ref{thm:param_consist-multi}\ref{enum:rate_conserv-multi} and
$|\beta_{0,T,j}| \to |\beta_{0,j}| = \infty$.

For~\ref{enum:ms_consist-multi}, we first show $\mP\left(\betaALj
= 0\right) \to 1$ if $\tilde\beta_{0,j} = 0$. It follows from the
Karush-Kuhn-Tucker optimality conditions that when $\betaALj \neq 0$ we have
$$
2X_j'\left(y - X\betaAL\right) = \lambda_T \frac{1}{|\hat\beta_j|}\sign\left(\betaALj\right).
$$
Multiplying both sides of the above display by $\lambda_T^{-1/2}T^{-1}$ yields
\begin{equation} \label{eq:KKT}
2\lambda_T^{-1/2}T^{-1} X_j'\left(y - X\betaAL\right) = \frac{1}{\lambda_T^{-1/2}T|\hat\beta_j|}
\sign\left(\betaALj\right).
\end{equation}
Using Assumption~\ref{ass:w1}, we can rewrite the left-hand side
of~\eqref{eq:KKT} as
$$
\lambda_T^{-1/2}T^{-1} X_j'u - (T^{-2}X_j'X) \lambda_T^{-1/2}T(\betaAL - \beta_T) = O_\mP(1),
$$
since $T^{-1} X_j'u = O_\mP(1)$, $T^{-2}X_j'X = O_\mP(1)$, and
$\lambda_T^{-1/2}T(\betaAL - \beta_T) = O_\mP(1)$ by
Theorem~\ref{thm:param_consist-multi}\ref{enum:rate_consist-multi}.
However, for the absolute value of the right-hand side
of~\eqref{eq:KKT} we get that 
$$
\frac{1}{\lambda_T^{-1/2}T|\hat \beta_j|} = 
\frac{1}{|\lambda_T^{-1/2}\mathcal{Z}^c_{T,j} + \tilde\beta_{0,T,j}|} 
\Rightarrow \frac{1}{|\tilde\beta_{0,j}|},
$$
which is unbounded if $\tilde\beta_{0,j} = 0$. Hence,
$\mP\left(\betaALj \neq 0\right) \to 0$ if $\tilde\beta_{0,j} = 0$. To
complete the proof, it remains to show that $\mP\left(\betaALj =
0\right) \to 0$ if $|\tilde\beta_{0,j}| = \infty$. Similar arguments
as used to prove~\ref{enum:ms_conserv-multi} yield 
$$
\mP\left(\betaALj = 0\right) \leq 
\mP\left(2|\lambda_T^{-1/2}T(\betaALj-\beta_{T,j})| > |\tilde\beta_{0,T,j}|\right) \to 0,
$$
since $\lambda_T^{-1/2}T(\betaALj-\beta_{T,j}) = O_\mP(1)$ by
Theorem~\ref{thm:param_consist-multi}\ref{enum:rate_consist-multi} and
$|\tilde\beta_{0,T,j}| \to |\tilde\beta_{0,j}| = \infty$.
\end{proof}


\begin{proof}[Proof of Remark~\ref{rem:ms_conserv-multi}] To prove~\ref{enum:ms_conserv_ge0-multi}, define $\cA_0 \coloneqq \{j
: \beta_{0,j} \neq 0\}$ and $\cA^c_0 \coloneqq \{j:\beta_{0,j}= 0\}$.
We assume,  without loss of generality, that the elements in $\beta_0$
are ordered such that the non-zero elements come first. We can
therefore partition $\beta_0 =
[\beta_{0,\cA_0}',\beta_{0,\cA_0^c}']'$, $\hat \beta =
[\hat\beta_{\cA_0}',\hat\beta_{\cA_0^c}']'$, $\betaAL =
[\betaALA',\betaALAc']'$, and $X = [X_{\cA_0},X_{\cA_0^c}]$. In
generic notation, for sets $\cI$ and $\cJ$, a vector $V$, and a
matrix $M$, we define $D_{\cI}(V) \coloneqq \diag\left(|V_j| : j
\in \cI\right)$, and $M_{[\cI,\cJ]}$ to be the $|\cI| \times
|\cJ|$-matrix containing the elements of $M$ with row indices in $\cI$
and column indices in $\cJ$ only. For $j \in \cA_0^c$, we clearly have
that
$$
\mP\left(\betaALj = 0\right) \geq \mP\left(\betaALAcz=0, \betaALi\neq 0 \text{ for all } i \in \cA_0\right).
$$
The corresponding Karush-Kuhn-Tucker optimality conditions can be written as
\begin{align} \label{eq:KKT_A}
	X_{\mathcal{A}_0}'X\left(\betaAL - \hat\beta\right) & \; = \;
	- \frac{\lambda_T}{2}D_{\mathcal{A}_0}^{-1}(\hat\beta)\sign(\betaALAz) \\
	\label{eq:KKT_Ac} 
	\left\|D_{\cA_0^c}(\hat\beta) X_{\cA_0^c}'X\left(\betaAL - \hat\beta\right)\right\|_\infty
	& \; \leq \; \frac{\lambda_T}{2}.
\end{align}
Since $\betaALAcz = 0$, we can rewrite \eqref{eq:KKT_A} and
\eqref{eq:KKT_Ac} as %
\begin{align}
	\label{eq:KKT_A2}
	& \betaALAz = \hat\beta_{\cA_0} + \left(X_{\cA_0}'X_{\cA_0}\right)^{-1}
	\left[X_{\cA_0}'X_{\cA_0^c}\hat\beta_{\cA_0^c} 
	- \frac{\lambda_T}{2} D_{\cA_0}^{-1}(\hat \beta)\sign(\betaALAz)\right] \\
	\label{eq:KKT_Ac2}
	& \left\|D_{\cA_0^c}(\hat\beta) X_{\cA_0^c}'
	\left[X_{\cA_0} \hat\beta_{\cA_0} + X_{\cA_0^c}\hat\beta_{\cA_0^c} - X_{\cA_0}\betaALAz\right]\right\|_\infty 
	\leq \frac{\lambda_T}{2}.
\end{align}
Plugging-in~\eqref{eq:KKT_A2} into~\eqref{eq:KKT_Ac2}, after some rearrangement, leads to
\begin{align*}
	\left\| D_{\cA_0^c}(\hat\beta) \left[\right.\right. & \left(\left(X'X\right)^{-1}_{[\cA_0^c,\cA_0^c]}\right)^{-1}
	\hat\beta_{\cA_0^c} \\[0.5ex]
	+ & \left.\left. \frac{\lambda_T}{2}\left(X_{\cA_0^c}'X_{\cA_0}\right)
	\left(X_{\cA_0}'X_{\cA_0}\right)^{-1} D_{\cA_0}^{-1}(\hat\beta) \sign(\betaALAz)\right]\right\|_\infty 
	\leq \frac{\lambda_T}{2},
\end{align*}
where $\left(X'X\right)^{-1}_{[\cA_0^c,\cA_0^c]}$ denotes the
bottom-right block-element of $(X'X)^{-1}$. Hence, 
\begin{align*}
	\liminf_{T \to \infty} \; & \mP\left(\betaALj = 0\right) 
	\geq \; \mP\left(\left\|D_{\cA_0^c}(\mathcal{Z}^c) \left[\left(\zeta_{vv[\cA_0^c,\cA_0^c]}^{-1}\right)^{-1}
	\mathcal{Z}^c_{\cA_0^c} \right.\right.\right. \\[0.5ex]
	& + \left.\left.\left. \frac{\lambda_0}{2} \zeta_{vv[\cA_0^c,\cA_0]}
	\left(\zeta_{vv[\cA_0,\cA_0]}\right)^{-1} D_{\cA_0}^{-1}(\mathcal{Z}^c + \beta_0)
	\sign(\beta_{0,\mathcal{A}_0})\right]\right\|_\infty \leq \frac{\lambda_0}{2}\right) > 0,
\end{align*}
as the random variable in the left-hand side of the above display has support $(0,\infty)$.

To prove~\ref{enum:ms_conserv_le1-multi}, note that $\mP\left(\hat\cA
= \cA \right) =  \mP\left(\betaALAc = 0, \betaALi \neq 0 \text{ for
	all } i \in \cA\right)$ and, for $T\beta_T = T\beta \to
\beta_0 \in \mRquer^k$, we have $\beta_{0,\cA} = \sign(\beta_\cA)\infty$ and $\beta_{0,\cA^c} = 0$, with
$\cA^c \coloneqq \{j: \beta_j = 0\}$. Therefore, similar calculations as above yield
$$
\limsup_{T \to \infty} \mP\left(\hat\cA = \cA\right) 
\leq \; \mP\left(\left\|D_{\cA^c}(\mathcal{Z}^c) \left[\left(\zeta_{vv[\cA^c,\cA^c]}^{-1}\right)^{-1}
\mathcal{Z}^c_{\cA^c} \right]\right\Vert_{\infty} \leq \frac{\lambda_0}{2}\right) < 1,
$$
as the random variable in the left-hand side of the above display has support $(0,\infty)$.
\end{proof}


To prove Theorem~\ref{thm:param_consist-multi}, we first need the following lemma.

\begin{lemma} \label{lem:KKT}
We have
$$
\left(\betaAL - \hat\beta\right)'\left(X'X\right)\left(\betaAL - \hat\beta\right) = 
\sum_{j=1}^k \left(\betaAL - \hat\beta\right)_j\left(X'X\left(\betaAL - \hat\beta\right)\right)_j 
\leq \frac{k}{2} \lambda_T,
$$
where the inequality holds surely, i.e., for all $\omega$ in the sample space of the underlying probability space.
\end{lemma}

\begin{proof}
The proof is similar to the proof of Lemma~1 in
\cite{AmannSchneider23}. The starting point are the Karush-Kuhn-Tucker
optimality conditions, which can be written as 
\begin{align}
\label{eq:KKT_beta_neq0}
\left(X'X\left(\betaAL - \hat\beta\right)\right)_j & \; = \;
-\frac{1}{2}\frac{\lambda_T}{|\hat\beta_j|}\sign(\betaALj) & \text{if } \betaALj \neq 0\, \\
\label{eq:KKT_beta_eq0}
\left|\left(X'X\left(\betaAL - \hat\beta\right)\right)_j\right| & \; \leq \;
\frac{1}{2}\frac{\lambda_T}{|\hat\beta_j|} & \text{if } \betaALj = 0,
\end{align}
using $X'\hat u=0$, where $\hat u \coloneqq y - X\hat\beta$ denote the
OLS residuals. When $\betaALj = 0$, \eqref{eq:KKT_beta_eq0} yields 
$$
\left|\left(\betaAL - \hat\beta\right)_j\left(X'X\left(\betaAL - \hat\beta\right)\right)_j\right| 
\leq \frac{1}{2}\lambda_T.
$$
We now consider the case $\betaALj \neq 0$. If $|\betaALj - \hat\beta_j| \leq |\hat\beta_j|$, 
\eqref{eq:KKT_beta_neq0} implies that
$$
\left| \left(\betaAL - \hat\beta\right)_j\left(X'X\left(\betaAL - \hat\beta\right)\right)_j \right|
\leq |\hat\beta_j|\left|\left(X'X\left(\betaAL - \hat\beta\right)\right)_j \right|
= \frac{1}{2}\lambda_T.
$$
On the other hand, if $|\betaALj - \hat\beta_j| > |\hat\beta_j|$, we have
$\sign(\betaALj - \hat\beta_j) = \sign(\betaALj) \neq 0$. Therefore, since
$\lambda_T > 0$,
\begin{align*}
\left(\betaAL - \hat\beta\right)_j \left(X'X\left(\betaAL - \hat\beta\right)\right)_j & = 
-\frac{1}{2}\frac{\lambda_T}{|\hat\beta_j|}\sign(\betaALj)(\betaALj - \hat\beta_j) \\
& = -\frac{1}{2}\frac{\lambda_T}{|\hat\beta_j|}\sign(\betaALj - \hat\beta_j)(\betaALj - \hat\beta_j) \\
& = -\frac{1}{2}\lambda_T\frac{|\betaALj-\hat\beta_j|}{|\hat\beta_j|} \leq 0 \leq \frac{1}{2}\lambda_T,
\end{align*}
which completes the proof.
\end{proof}


\begin{proof}[Proof of Theorem~\ref{thm:param_consist-multi}]
Part~\ref{enum:param_consist-multi} follows directly
from~\ref{enum:rate_conserv-multi} and \ref{enum:rate_consist-multi}.
To prove these parts, let $\mu_{\min,T}$ denote the smallest
eigenvalue of $T^{-2}X'X = T^{-2}\sum_{t=1}^T x_t x_t'$. The
continuity of eigenvalues \cite[Proof of Lemma~5]{SaikkonenChoi04} and
the fact that the limit of $T^{-2}\sum_{t=1}^T x_t x_t'$ is positive
definite a.s.~by Assumption~\ref{ass:w1} ensure that
$\mu_{\min,T}^{-1} = O_\mP(1)$.

For~\ref{enum:rate_conserv-multi}, first note that
$$
T\left(\betaAL - \beta_T\right) = T\left(\betaAL - \hat\beta\right) + T\left(\hat\beta - \beta_T\right) 
= T\left(\betaAL - \hat\beta\right) + O_\mP(1).
$$
It thus remains to show that $T\left(\betaAL - \hat\beta\right) =
O_\mP(1)$, which follows from  $\mu_{\min,T}^{-1} = O_\mP(1)$ since
\begin{align*}
T^2\left\|\betaAL - \hat\beta\right\|^2 & = 
T^2 \left(\betaAL - \hat\beta\right)'\left(\betaAL - \hat\beta\right) \\[0.5ex]
& \leq \frac{1}{\mu_{\min,T}} T^2 \left(\betaAL - \hat\beta\right)' \frac{X'X}{T^2} 
\left(\betaAL - \hat\beta\right) \\[0.5ex]
& = \frac{1}{\mu_{\min,T}} \left(\betaAL - \hat\beta\right)'X'X\left(\betaAL - \hat\beta\right) \\[0.5ex]
& \leq \frac{k\lambda_T}{2\mu_{\min,T}} = O_\mP(1),
\end{align*}
where the last inequality follows from Lemma~\ref{lem:KKT}.

For~\ref{enum:rate_consist-multi}, analogously to the proof of 
\ref{enum:rate_conserv-multi}, we have 
\begin{align*}
\lambda_T^{-1/2}T\left(\betaAL - \beta_T\right) &= \lambda_T^{-1/2}T\left(\betaAL - \hat\beta\right) 
+ \lambda_T^{-1/2}T\left(\hat\beta - \beta_T\right) \\
&= \lambda_T^{-1/2}T\left(\betaAL-\hat\beta\right) + o_\mP(1).
\end{align*}
It therefore remains to show that $\lambda_T^{-1/2}T\left(\betaAL -
\hat\beta\right) = O_\mP(1)$, which follows again from
$\mu_{\min,T}^{-1} = O_\mP(1)$ since 
\begin{align*}
\lambda_T^{-1}T^2\left\|\betaAL - \hat\beta\right\|^2 
& = \lambda_T^{-1}T^2 \left(\betaAL - \hat\beta\right)'\left(\betaAL - \hat\beta\right) \\[0.5ex]
& \leq \frac{1}{\mu_{\min,T}} \lambda_T^{-1}T^2
\left(\betaAL - \hat\beta\right)'\frac{X'X}{T^2}\left(\betaAL - \hat\beta\right) \\[0.5ex]
& = \frac{1}{\mu_{\min,T}}\lambda_T^{-1} 
\left(\betaAL - \hat\beta\right)'X'X\left(\betaAL - \hat\beta\right)\\[0.5ex]
& \leq \frac{k}{2\mu_{\min,T}} = O_\mP(1),
\end{align*}
where the last inequality again follows from Lemma~\ref{lem:KKT}.
\end{proof}


\begin{proof}[Proof of Theorem~\ref{thm:ls_dist-multi}] For~\ref{enum:ls_dist_conserv-multi}, 
note the following: Since $\betaAL$ is the solution to the
minimization problem in \eqref{eq:ALASSO} (with $\gamma = 1$ and
$\hat\beta^0 = \hat\beta$), $T(\betaAL - \beta_T)$ is the minimizer of
$$
\Psi_T(z) \coloneqq \sum_{t=1}^T \left(y_t - x_t'(\beta_T + T^{-1}z)\right)^2
 + \lambda_T \sum_{j=1}^k \frac{|T^{-1}z_j + \beta_{T,j}|}{|\hat\beta_j|}.
$$
Therefore, $T(\betaAL - \beta_T)$ also minimizes
\begin{align*}
V_T(z) & \coloneqq \Psi_T(z) - \Psi_T(0) \\
& = z'\left(T^{-2}\sum_{t=1}^T x_tx_t'\right)z - 2z'\left(T^{-1}\sum_{t=1}^T x_tu_t\right) 
+ \lambda_T \sum_{j=1}^k \frac{|T^{-1}z_j + \beta_{T,j}|-|\beta_{T,j}|}{|\hat\beta_j|} \\
& = z'\left(T^{-2}\sum_{t=1}^T x_tx_t'\right)z - 2z'\left(T^{-1}\sum_{t=1}^T x_tu_t\right) 
+ \lambda_T \sum_{j=1}^k \frac{|z_j + \beta_{0,T,j}| - |\beta_{0,T,j}|}
{|\mathcal{Z}^c_{T,j} + \beta_{0,T,j}|}.
\end{align*}
Under Assumption~\ref{ass:w1}, we get that $V_T(z) \Rightarrow V_{\beta_0}^c(z)$
for all $z \in \mR^k$. Using Proposition~2.2 and Theorem~3.2 in
\cite{Geyer96TR}, we may deduce that also the minimizer of $V_T(z)$
converges weakly to the minimizer of $V_{\beta_0}^c(z)$, i.e., $T(\betaAL -
\beta_T)\Rightarrow \argmin_{z\in\mR^k} V_{\beta_0}^c(z)$.

For~\ref{enum:ls_dist_consist-multi}, analogously to the proof
of \ref{enum:ls_dist_conserv-multi}, we get that
$\lambda_T^{-1/2}T(\betaAL - \beta_T)$ minimizes 
$$
\tilde{\Psi}_T(z) \coloneqq \sum_{t=1}^T \left(y_t - x_t'(\beta_T + \lambda_T^{1/2}T^{-1}z)\right)^2
+ \lambda_T \sum_{j=1}^k \frac{|\lambda_T^{1/2}T^{-1}z_j + \beta_{T,j}|}{|\hat\beta_j|}.
$$
Therefore, $\lambda_T^{-1/2}T(\betaAL - \beta_T)$ also minimizes
\begin{align*}
\tilde{V}_T(z) & \coloneqq \lambda_T^{-1}\left(\tilde{\Psi}_T(z) - \tilde{\Psi}_T(0)\right) \\
& = z'\left(T^{-2}\sum_{t=1}^T x_tx_t'\right)z 
- 2z'\left(\lambda_T^{-1/2}T^{-1}\sum_{t=1}^T x_tu_t\right) 
+ \sum_{j=1}^k \frac{|\lambda_T^{1/2}T^{-1}z_j + \beta_{T,j}| - |\beta_{T,j}|}{|\hat\beta_j|} \\
& = z'\left(T^{-2}\sum_{t=1}^T x_tx_t'\right)z 
- 2z'\left(\lambda_T^{-1/2}T^{-1}\sum_{t=1}^T x_tu_t\right) 
+ \sum_{j=1}^k \frac{|z_j + \tilde\beta_{0,T,j}| - |\tilde\beta_{0,T,j}|}
{|\lambda_T^{-1/2}\mathcal{Z}^c_{T,j} + \tilde\beta_{0,T,j}|}.
\end{align*}
Under Assumption~\ref{ass:w1}, we get that $\tilde{V}_T(z) \Rightarrow
\tilde{V}_{\tilde\beta_0}^c(z)$ for all $z \in \mR^k$. However, since
$\tilde{V}_{\tilde\beta_0}^c(z)$ is not finite on an open subset of
$\mR^k$, we cannot use the same arguments as in
\ref{enum:ls_dist_conserv-multi} to deduce that also the minimizer of
$\tilde{V}_T(z)$ converges weakly to the minimizer of
$\tilde{V}_{\tilde\beta_0}^c(z)$. Nevertheless, the result follows from
similar arguments as used in \citet[Proof of
Theorem~7]{AmannSchneider23}.

For~\ref{enum:ls_dist_consist_rateT-multi}, note that we already
know from \ref{enum:ls_dist_conserv-multi} that $T(\betaAL - \beta_T)$
minimizes $V_T(z)$, which can also be written as 
$$
V_T(z) = z'\left(T^{-2}\sum_{t=1}^T x_tx_t'\right)z - 2z'\left(T^{-1}\sum_{t=1}^T x_tu_t\right) 
+ \sum_{j=1}^k \frac{|z_j + \beta_{0,T,j}|-|\beta_{0,T,j}|}
{|\lambda_T^{-1}\mathcal{Z}^c_{T,j} + \bar{\beta}_{0,T,j}|}.
$$
Under Assumption~\ref{ass:w1}, we get that $\tilde{V}_T(z) \Rightarrow
\bar{V}_{\bar\beta_0}^c(z)$ for all $z \in \mR^k$. To derive the limits $\bar
A_j(z_j,\beta_{0,j},\bar\beta_{0,j})$, a tedious case-by-case analysis
is necessary when $\bar\beta_{0,j} = 0$ and $z_j \neq 0$. When $0 <
|\bar\beta_{0,j}| < \infty$ and $z_j \neq 0$ (``otherwise''), note
that $\beta_{0,j} = \sign(\bar\beta_{0,j})\infty$ must hold. The
result then follows from analogous arguments as used in
\ref{enum:ls_dist_consist-multi}.
\end{proof}


\begin{proof}[Proof of Proposition~\ref{prop:M-argmin}]
Clearly, $m_j = 0$ must hold if $\tilde\beta_{0,j} = 0$, as the
objective function would be infinite otherwise. If
$|\tilde\beta_{0,j}| = \infty$ or  $\{0 < |\tilde\beta_{0,j}| < \infty$
and $m_j \neq -\tilde\beta_{0,j}\}$, the partial derivative of $\tilde
V_{\tilde\beta_0}$ with respect to $z_j$ exists and regular
first-order conditions yield the required result for these cases. For
$\{0 < |\tilde\beta_{0,j}| < \infty$ and $m_j = -\tilde\beta_{0,j}\}$, we
make use of the fact that $m$ is a minimizer if and only if zero is a
subgradient of $\tilde V_{\tilde\beta_0}$ at $m$.
\end{proof}


\begin{proof}[Proof of Theorem~\ref{thm:setM}] Let $m = \argmin_{z \in \mR^k} 
\tilde V_{\beta_0}^c(z)$. We need to show that $m \in \Mc$ also, i.e., that
$m(\omega)_j(\zeta_{vv}^c(\omega)m(\omega))_j \leq 1/2$ is
satisfied for all $\omega$ and for all $j = 1,\dots,k$. Using
Proposition~\ref{prop:M-argmin}, we get $m_j \equiv 0$ if
$\tilde\beta_0 = 0$, so that the required inequality surely holds. If
$|\tilde\beta_0| = \infty$, the same proposition yields
$(\zeta_{vv}^cm)_j \equiv 0$, again ensuring that the inequality
holds for all $\omega$. When $0 < |\tilde\beta_{0,j}| < \infty$, we
look at the following cases:

If $m_j(\omega) = - \beta_{0,j}$, Proposition~\ref{prop:M-argmin} shows that
$$
|(\zeta_{vv}^c(\omega)m(\omega))_j| \leq \frac{1}{2|\tilde\beta_{0,j}|}
$$
and therefore 
$$
m(\omega)_j(\zeta_{vv}^c(\omega)m(\omega))_j \leq |m(\omega)_j(\zeta_{vv}^c(\omega)m(\omega))_j| 
\leq \frac{|\tilde\beta_{0,j}|}{2|\tilde\beta_{0,j}|} = \frac{1}{2}.
$$

If $m_j(\omega) \neq - \beta_{0,j}$, again from
Proposition~\ref{prop:M-argmin}, we get
$$
(\zeta_{vv}^c(\omega)m(\omega))_j = -\frac{\sign(m(\omega)_j + \beta_{0,j})}{2|\tilde\beta_{0,j}|}.
$$
If $|m(\omega))_j| > |\beta_{0,j}|$, then $\sign(m(\omega)_j + \beta_{0,j}) = \sign(m(\omega)_j)$ 
and 
$$
m(\omega)_j(\zeta_{vv}^c(\omega)m(\omega))_j = - \frac{|m(\omega)_j|}{2|\tilde\beta_{0,j}|} 
\leq 0 < \frac{1}{2}, 
$$
whereas for $|m(\omega))_j| \leq |\beta_{0,j}|$, we have
$$
m(\omega)_j(\zeta_{vv}^c(\omega)m(\omega))_j \leq |m(\omega)_j(\zeta_{vv}^c(\omega)m(\omega))_j| =
\frac{|m(\omega)_j|}{2|\tilde\beta_{0,j}|} \leq \frac{1}{2}.
$$
\end{proof}


\section{Proofs for Section~\ref{subsec:ci}}

\begin{proof}[Proof of Theorem~\ref{thm:confM}] 
Let $g_T(\beta) = \mP_\beta(\beta \in \betaAL -
T^{-1}\lambda_T^{1/2}\Mhateps$ and $c_T = \inf_{\beta \in \mR^k}
g_T(\beta)$. We need to show that $c_T \to 1$ as $T \to \infty$. Since
$c_T$ are the infima of $g_T$, we can choose sequences
$(\beta_{T,n})_{n \in \mN} \subseteq \mR^k$ such that
$$
|c_T - g_T(\beta_{T,n})| \leq \frac{1}{n}
$$
for all $T, k \in \mN$. Now define $\breve\beta_T = \beta_{T,T}$.
Since $|c_T - g_T(\breve\beta_T)| \leq 1/T = o(1)$ as $T \to \infty$,
we can look at the limiting behavior of $g_T(\breve\beta_T)$ instead
of $c_T$. Now let $\tilde\beta_0 \in \mRquer^p$ such that
$T\lambda_T^{-1/2}\breve\beta_T \to \tilde\beta_0$.\footnote{If this quantity
does not converge, simply revert to a convergent subsequence.} Define
$\Mc(\eps) = \{m : m_j(\zeta_{vv}^cm)_j < \frac{1}{2} + \eps\}$
and note that for large enough $T$, we have $\Mc(\eps/2)
\subseteq \Mhateps$ for all $\omega$. We then
get
\begin{align*}
1 \geq \limsup_{T \to \infty} g_T(\breve\beta_T) & \geq 
\liminf_{T \to \infty} g_T(\breve\beta_T) = \liminf_{T \to \infty}
\mP_{\breve\beta_T}\left(T\lambda_T^{-1/2}(\betaAL - \breve\beta) \in 
\Mhateps \right) \\ & \geq
\liminf_{T \to \infty} \mP_{\breve\beta_T}\left(T\lambda_T^{-1/2}(\betaAL - \breve\beta) \in 
\Mc(\eps/2)\right) \\[0.5ex] & \geq 
\mP_{\tilde\beta_0}\left(\argmin_{z \in \mR^k} \tilde V_{\tilde\beta_0}^c(z) \in 
\Mc(\eps/2)\right) \\[1ex] & \geq 
\mP_{\tilde\beta_0}\left(\argmin_{z \in \mR^k} \tilde V_{\tilde\beta_0}^c(z) \in  
\Mc\right) 
= 1,
\end{align*}
where second-to-last inequality holds by
Theorem~\ref{thm:ls_dist-multi}\ref{enum:ls_dist_consist-unif} and the
Portmanteau Theorem, and the final equality by Theorem~\ref{thm:setM}
together with Remark~\ref{rem:skorohod}. We therefore get
$$
\lim_{T \to \infty} c_T = \lim_{T \to \infty} g_T(\breve\beta_T) = 1.
$$
\end{proof}

\newpage
\clearpage

\section{Additional Finite-Sample Results} \label{app:simulation}

\begin{figure}[ht]
\begin{center}
\caption*{$\lambda_T = T^{1/4}$}
\vspace{-1.5ex}
\begin{subfigure}{0.2\textwidth}
	\centering
	\caption*{$T = 25$}
	\vspace{-1.5ex}
	\includegraphics[trim={0cm 0cm 0.50cm 0.5cm},width=\textwidth,clip]
	{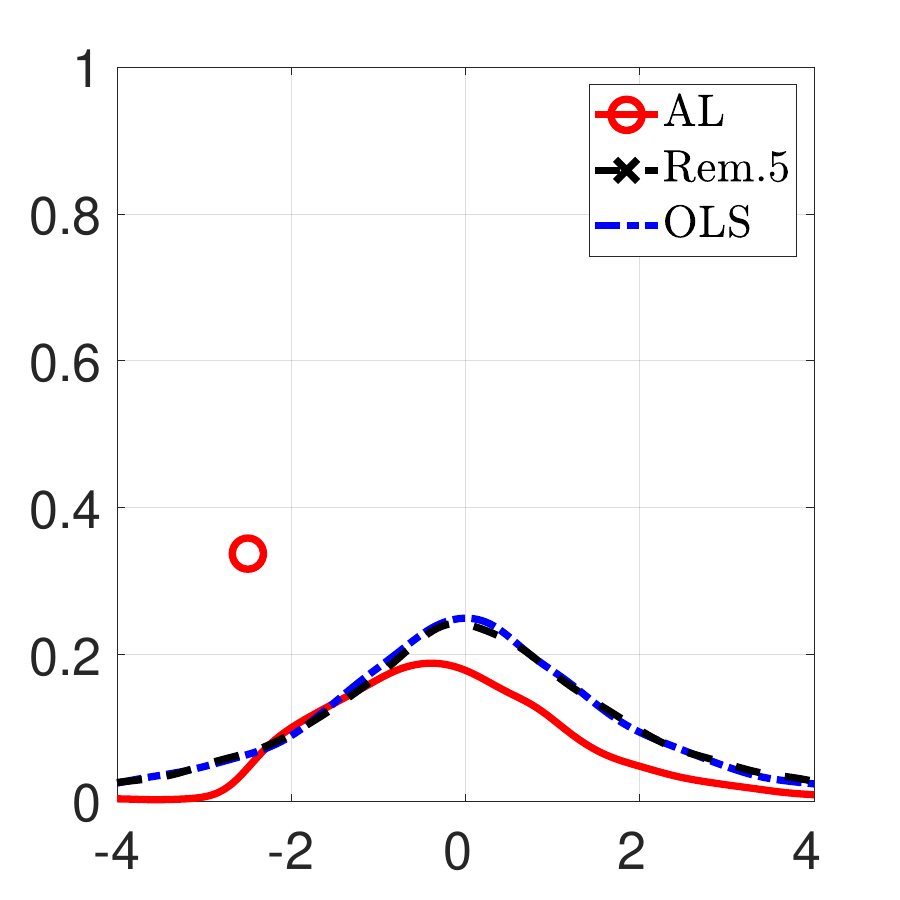}
\end{subfigure}\begin{subfigure}{0.2\textwidth}
	\centering
	\caption*{$T = 50$}
	\vspace{-1.5ex}
	\includegraphics[trim={0cm 0cm 0.50cm 0.5cm},width=\textwidth,clip]
	{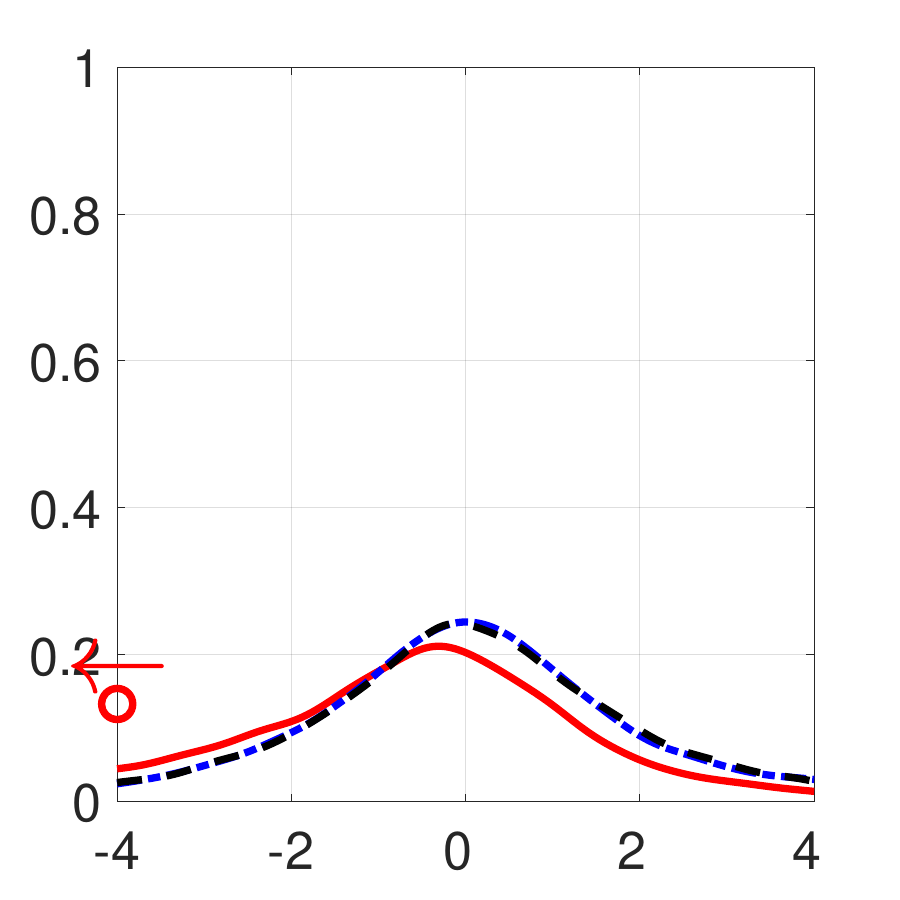}
\end{subfigure}\begin{subfigure}{0.2\textwidth}
	\centering
	\caption*{$T = 100$}
	\vspace{-1.5ex}
	\includegraphics[trim={0cm 0cm 0.50cm 0.5cm},width=\textwidth,clip]
	{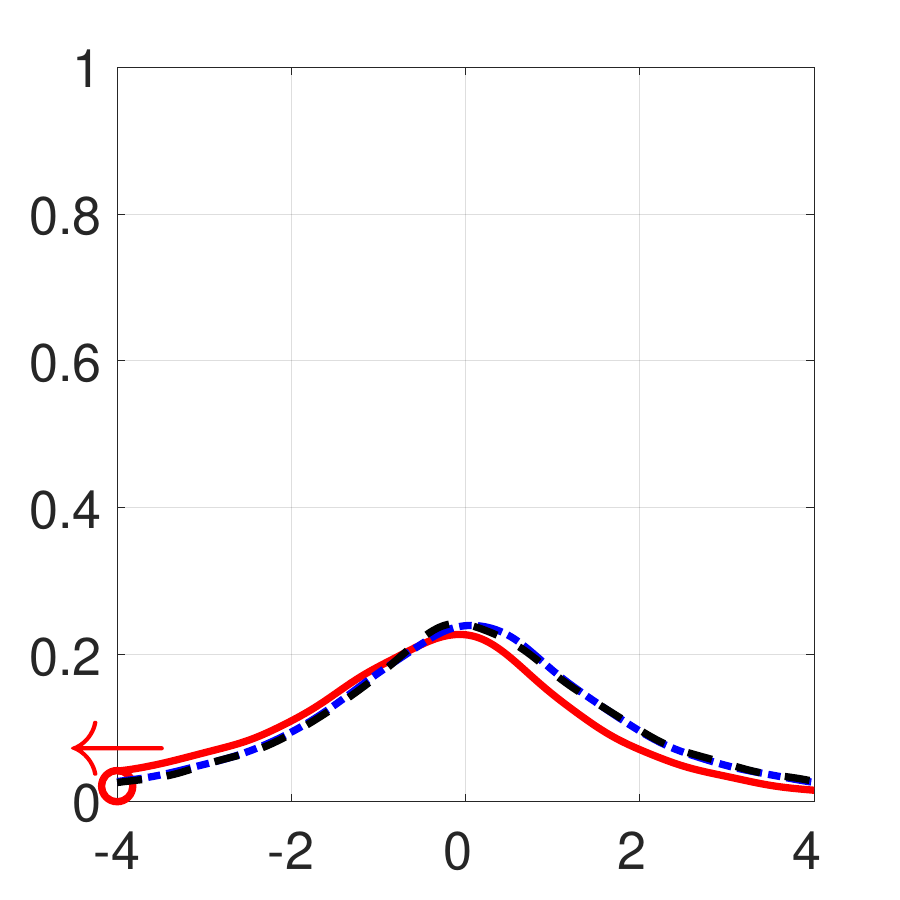}
\end{subfigure}\begin{subfigure}{0.2\textwidth}
	\centering
	\caption*{$T = 250$}
	\vspace{-1.5ex}
	\includegraphics[trim={0cm 0cm 0.50cm 0.5cm},width=\textwidth,clip]
	{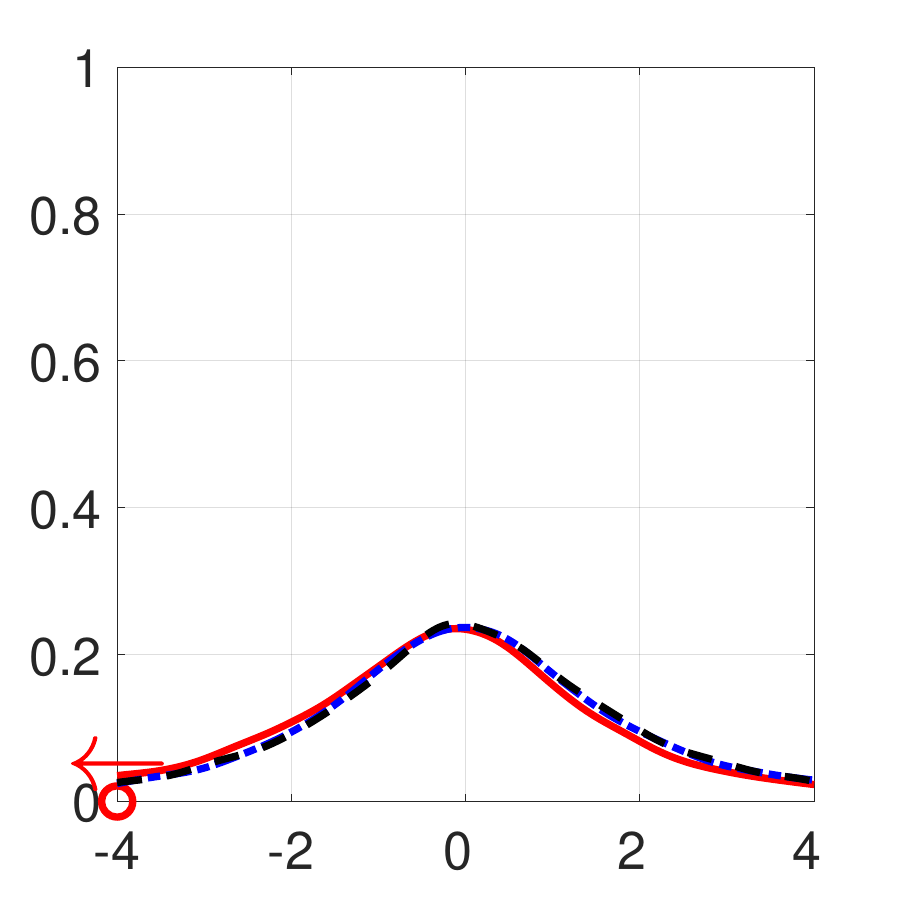}
\end{subfigure}\begin{subfigure}{0.2\textwidth}
	\centering
	\caption*{$T = 1000$}
	\vspace{-1.5ex}
	\includegraphics[trim={0cm 0cm 0.50cm 0.5cm},width=\textwidth,clip]
	{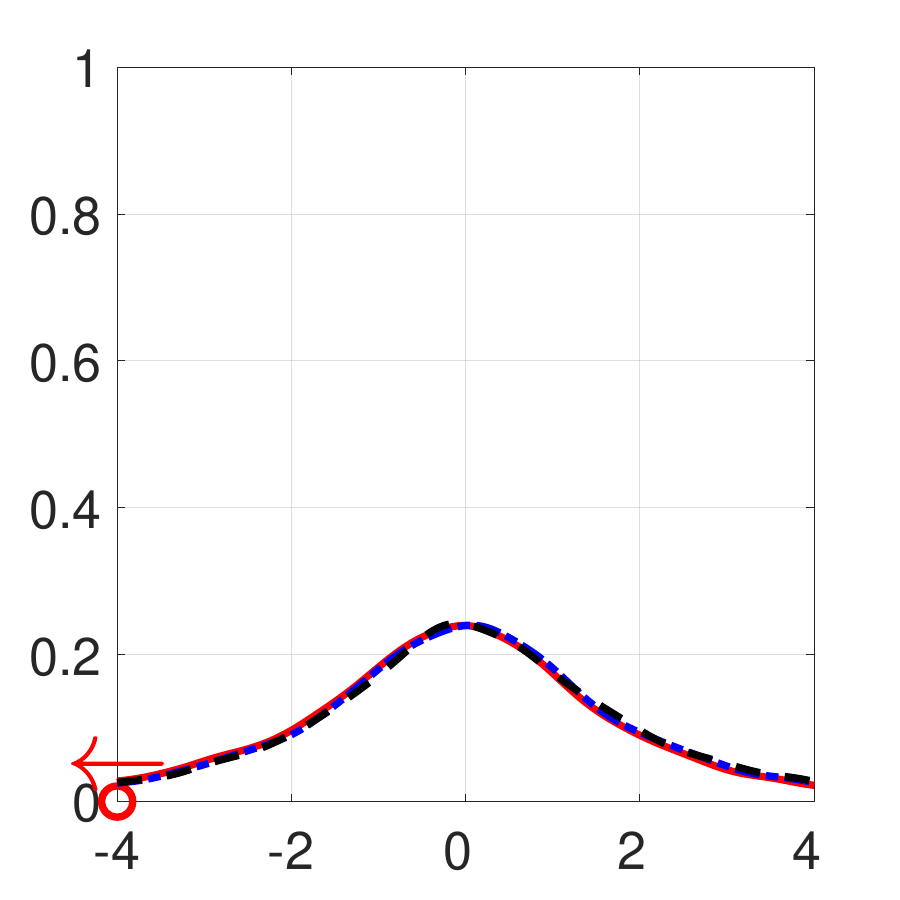}
\end{subfigure}

\caption*{$\lambda_T = T^{1/2}$}
\vspace{-1.5ex}
\begin{subfigure}{0.2\textwidth}
	\centering
	\includegraphics[trim={0cm 0cm 0.50cm 0.5cm},width=\textwidth,clip]
	{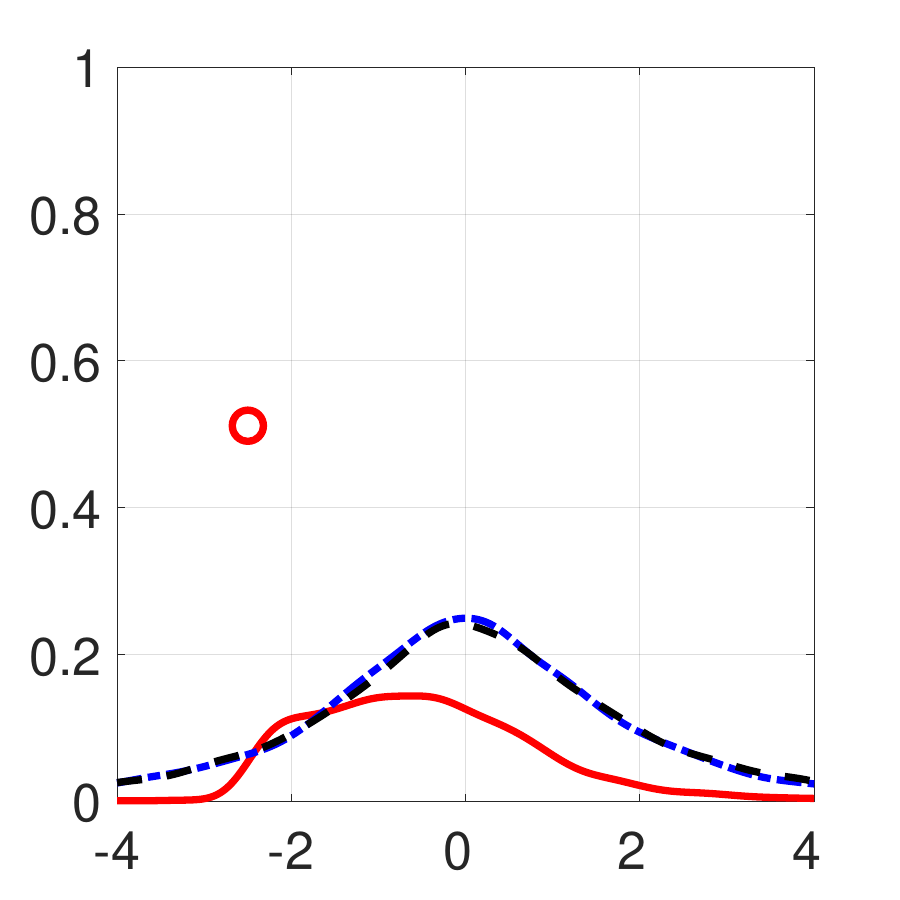}
\end{subfigure}\begin{subfigure}{0.2\textwidth}
	\centering
	\includegraphics[trim={0cm 0cm 0.50cm 0.5cm},width=\textwidth,clip]
	{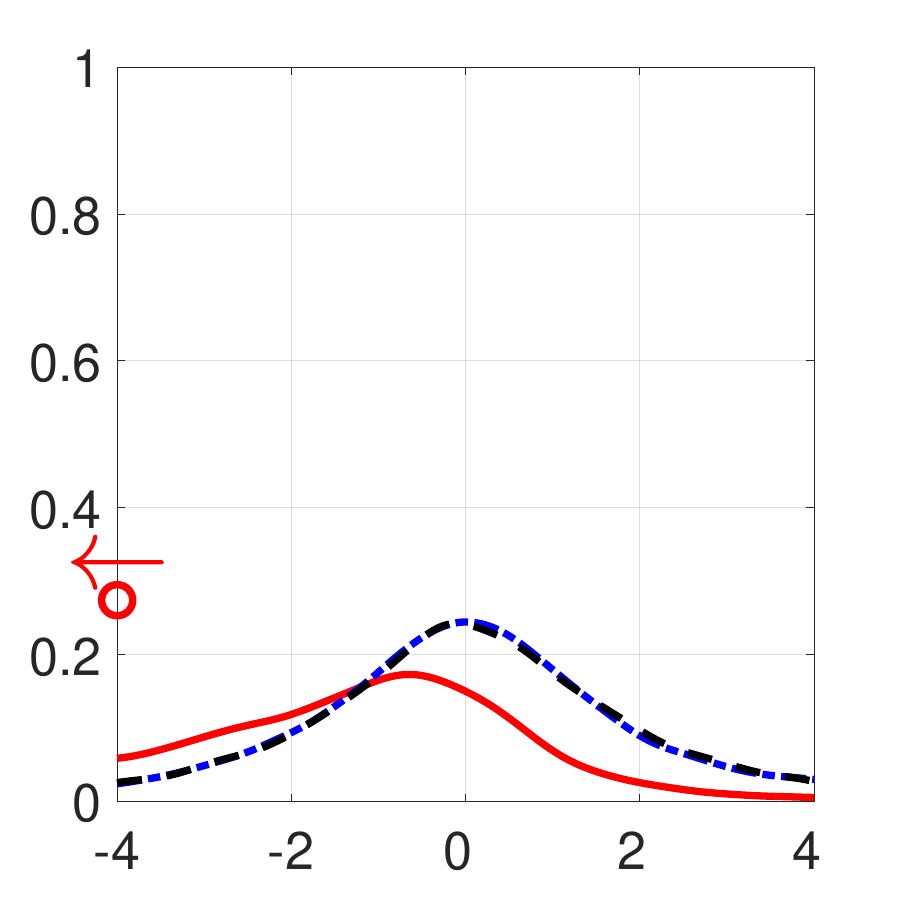}
\end{subfigure}\begin{subfigure}{0.2\textwidth}
	\centering
	\includegraphics[trim={0cm 0cm 0.50cm 0.5cm},width=\textwidth,clip]
	{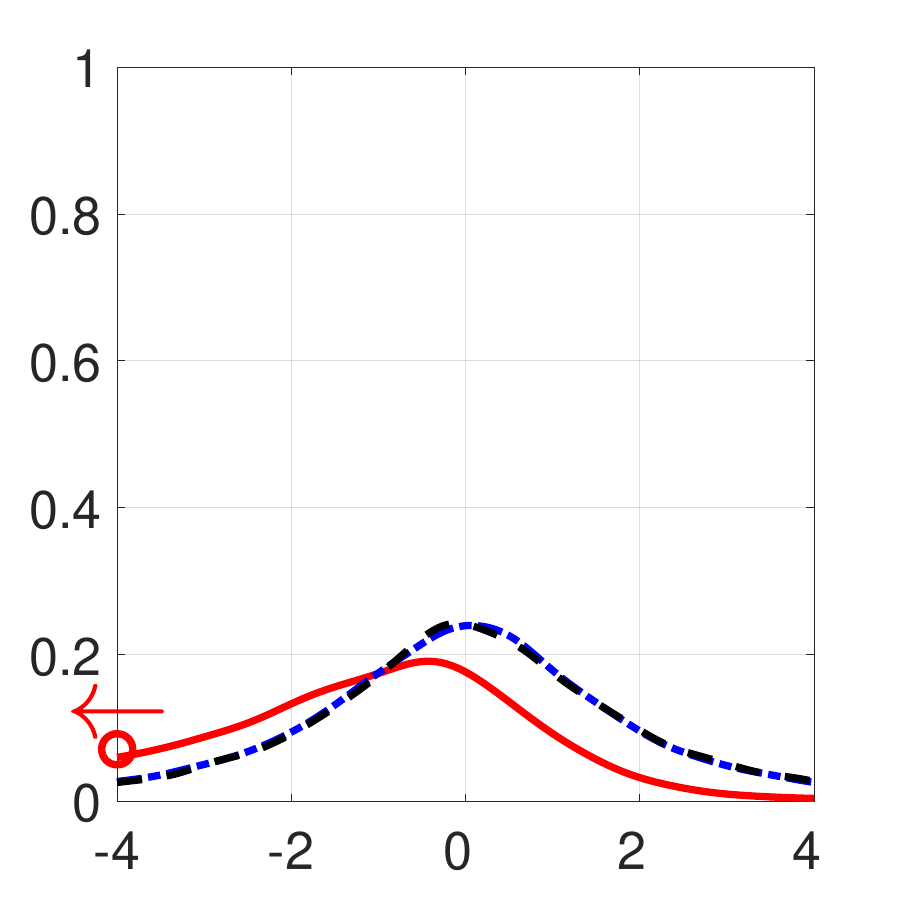}
\end{subfigure}\begin{subfigure}{0.2\textwidth}
	\centering
	\includegraphics[trim={0cm 0cm 0.50cm 0.5cm},width=\textwidth,clip]
	{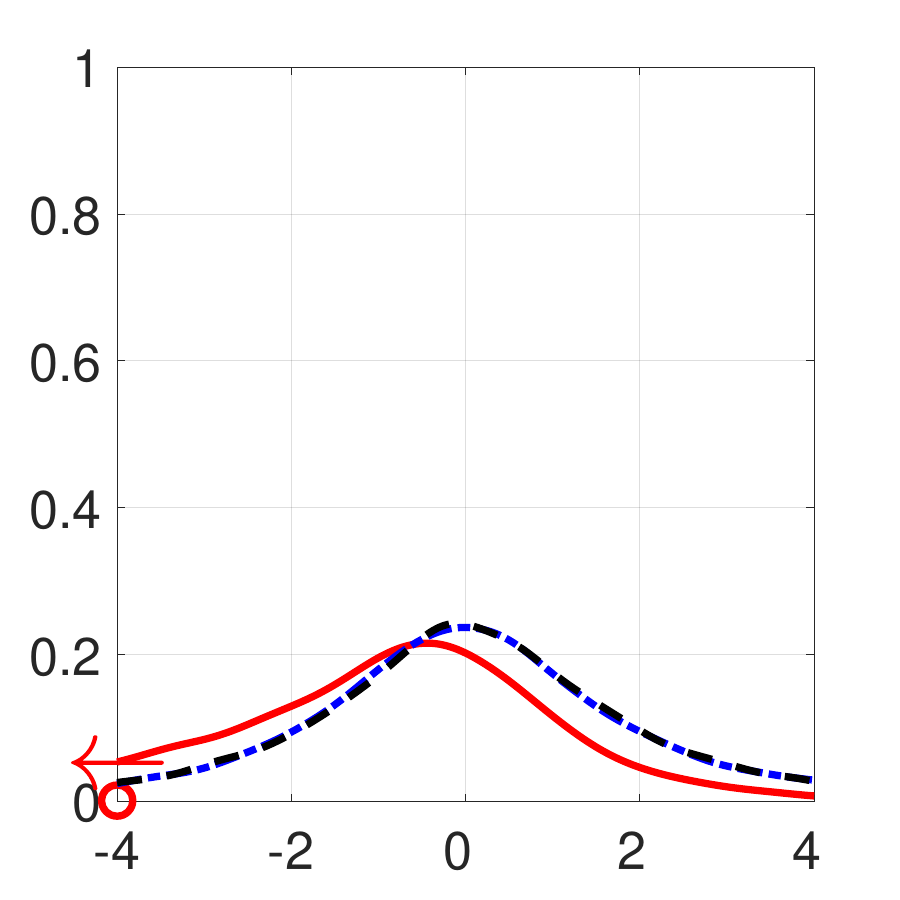}
\end{subfigure}\begin{subfigure}{0.2\textwidth}
	\centering
	\includegraphics[trim={0cm 0cm 0.50cm 0.5cm},width=\textwidth,clip]
	{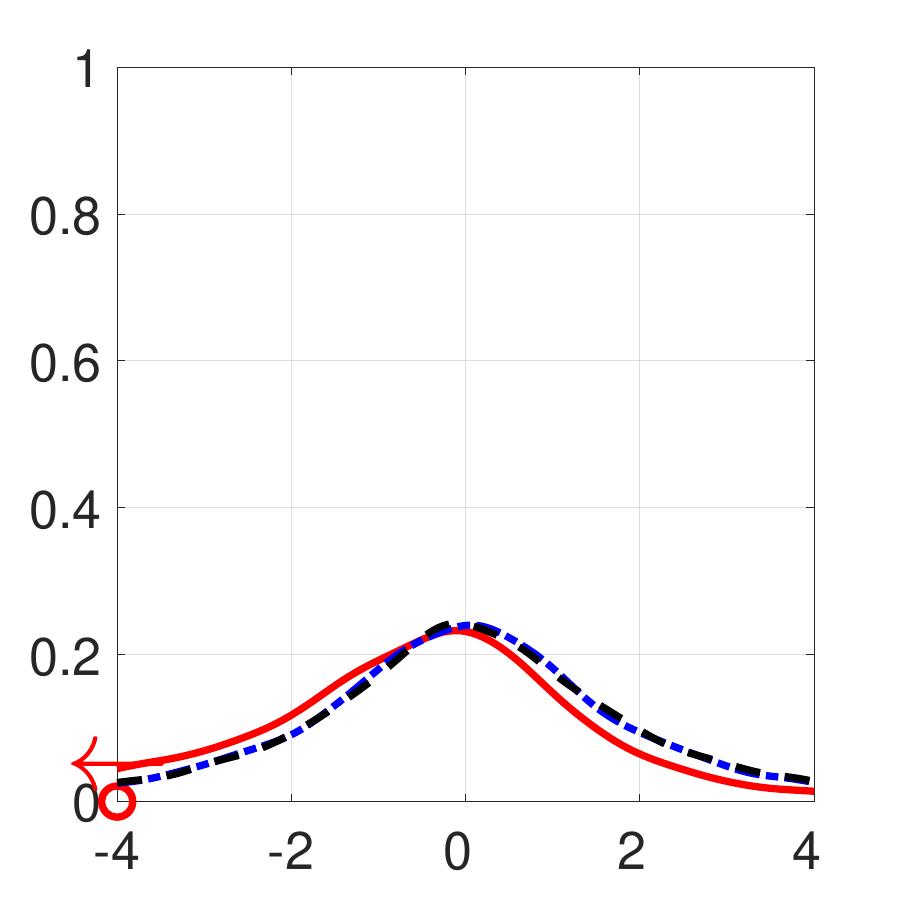}
\end{subfigure}

\caption*{$\lambda_T = T$}
\vspace{-1.5ex}
\begin{subfigure}{0.2\textwidth}
	\centering
	\includegraphics[trim={0cm 0cm 0.50cm 0.5cm},width=\textwidth,clip]
	{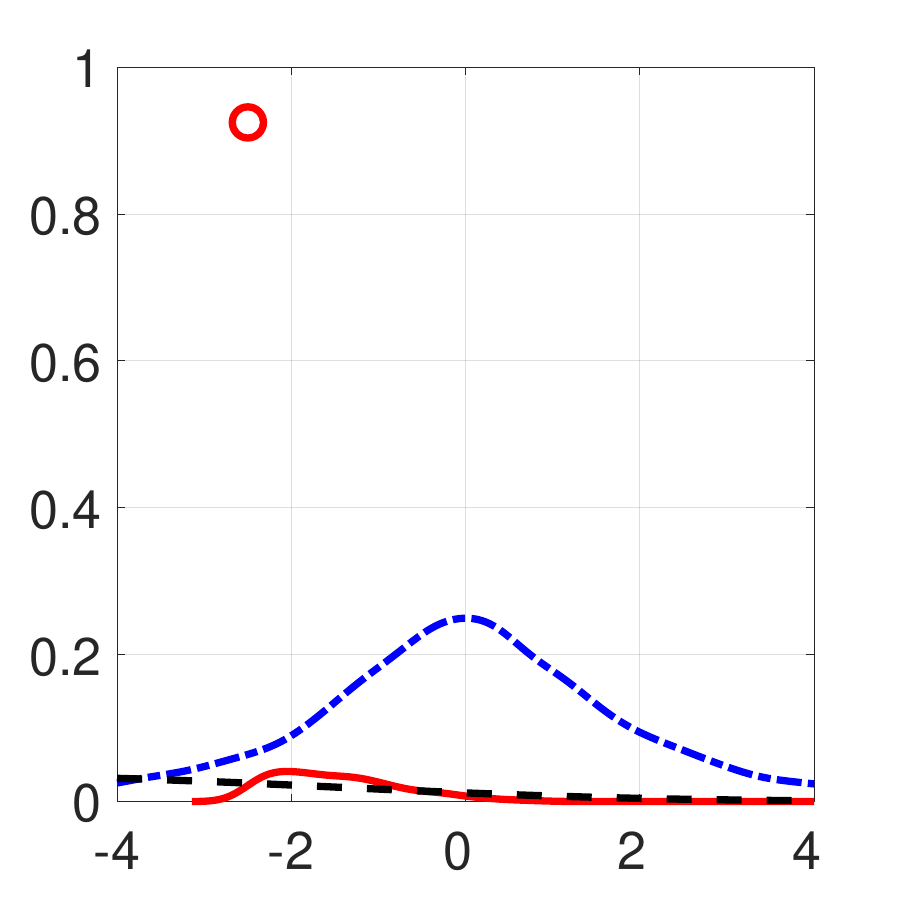}
\end{subfigure}\begin{subfigure}{0.2\textwidth}
	\centering
	\includegraphics[trim={0cm 0cm 0.50cm 0.5cm},width=\textwidth,clip]{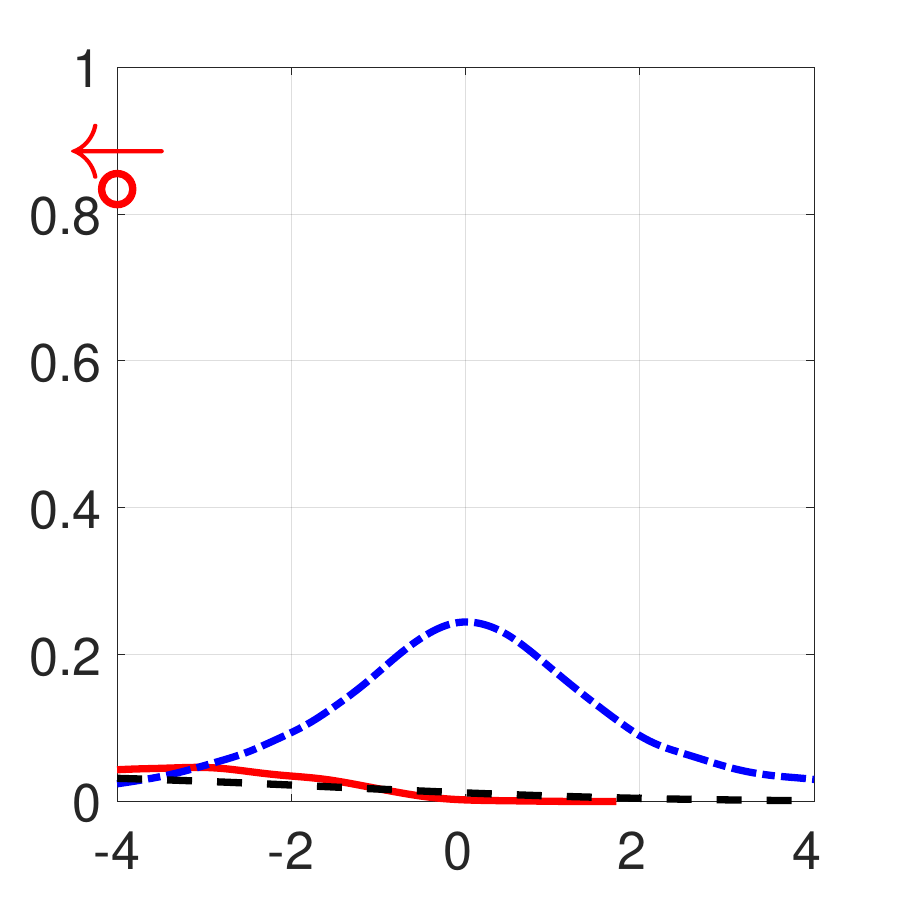}
\end{subfigure}\begin{subfigure}{0.2\textwidth}
	\centering
	\includegraphics[trim={0cm 0cm 0.50cm 0.5cm},width=\textwidth,clip]
	{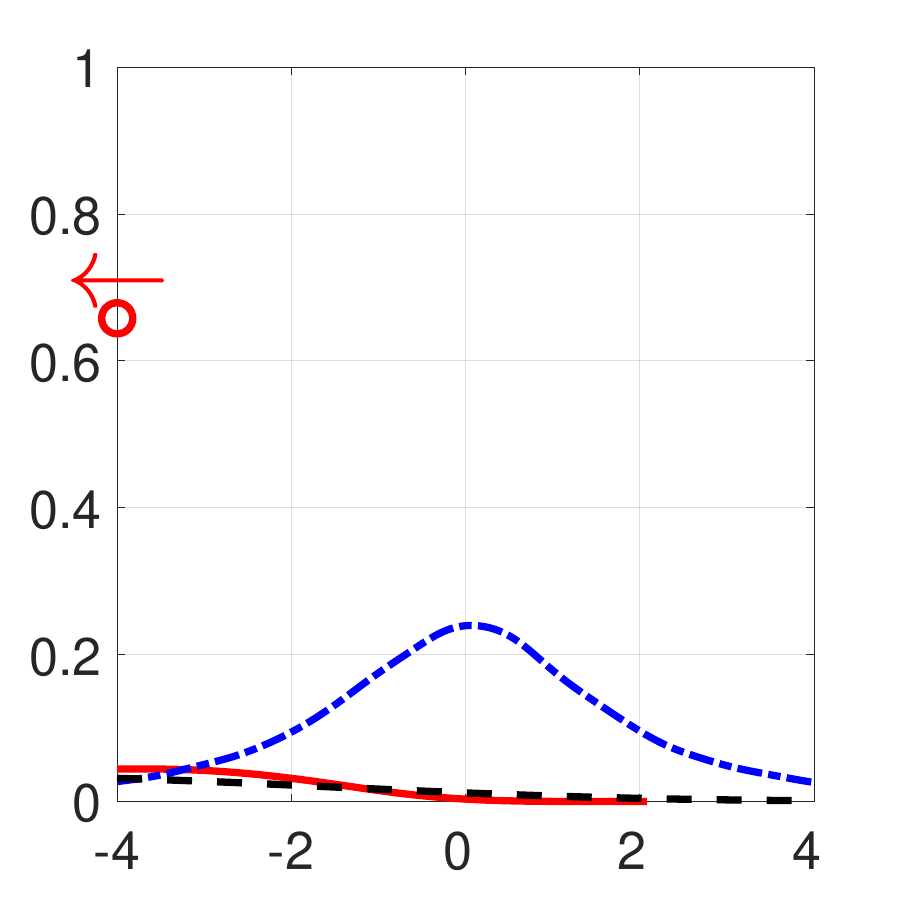}
\end{subfigure}\begin{subfigure}{0.2\textwidth}
	\centering
	\includegraphics[trim={0cm 0cm 0.50cm 0.5cm},width=\textwidth,clip]
	{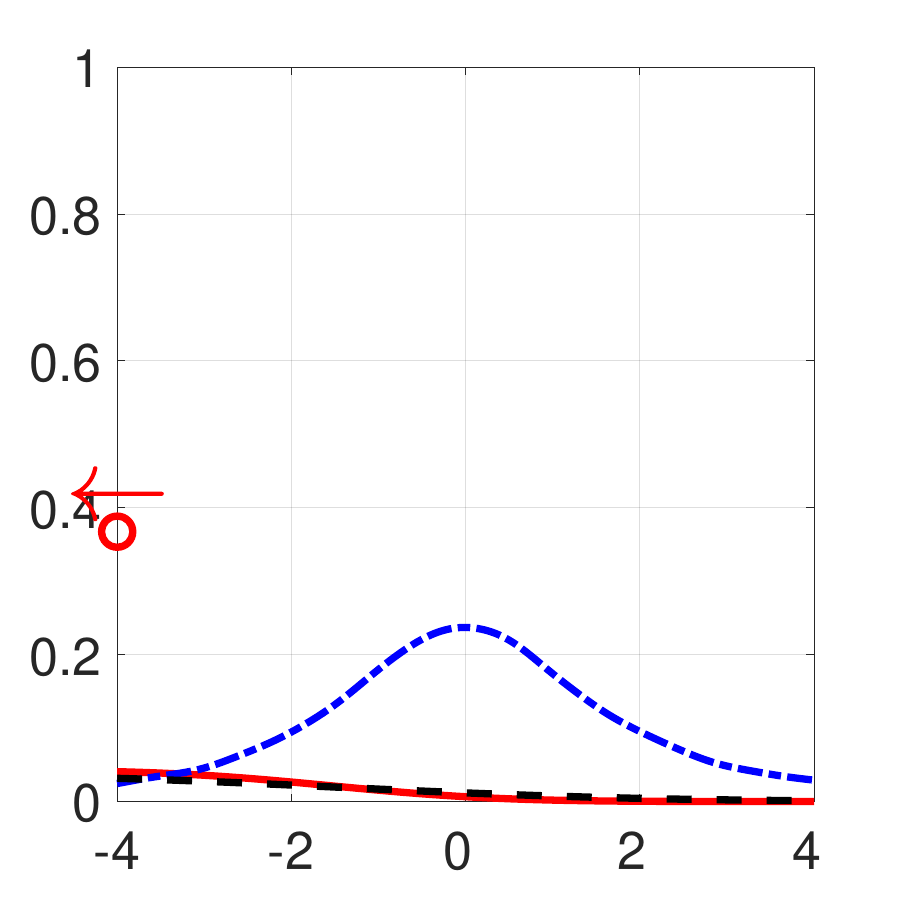}
\end{subfigure}\begin{subfigure}{0.2\textwidth}
	\centering
	\includegraphics[trim={0cm 0cm 0.50cm 0.5cm},width=\textwidth,clip]
	{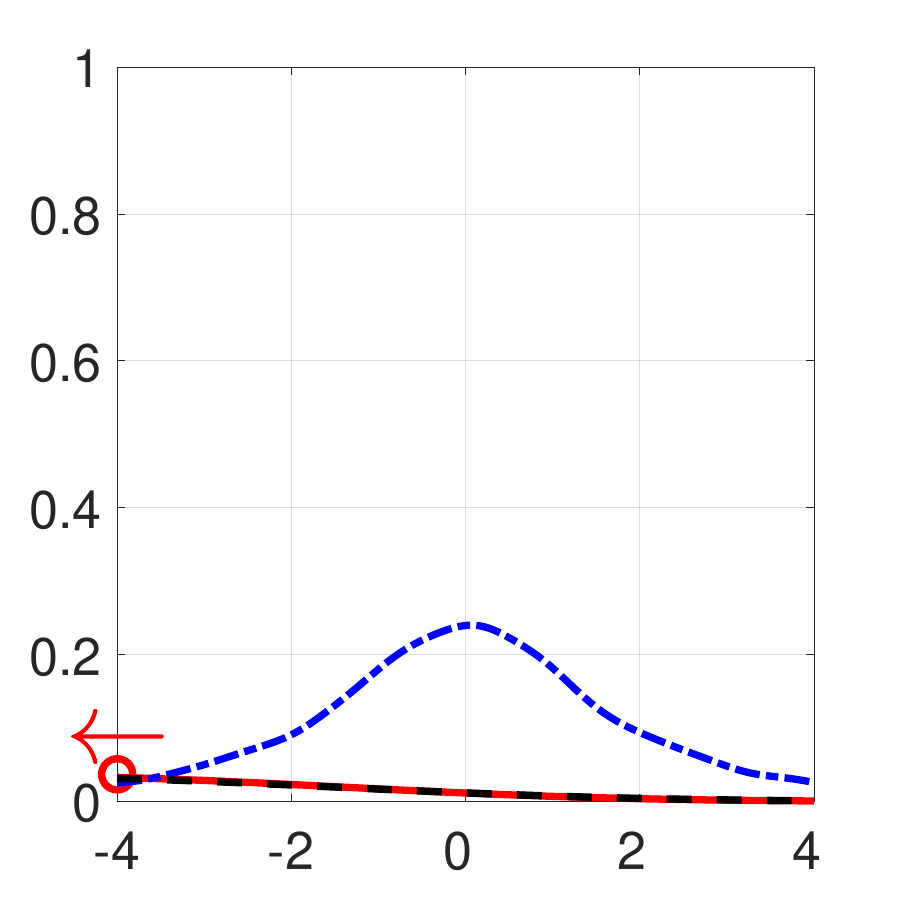}
\end{subfigure}

\end{center}

\vspace{-2ex} 

\caption{Finite-sample distribution of $T(\betaAL - \beta_T)$ under
consistent tuning (in all rows) in case $\beta_T \equiv 0.1\beta$
(labeled ``AL''), and limiting distribution from
Remark~\ref{rem:ls_dist_consist_rateT-unif} (labeled ``Rem.5'').
\emph{Notes}: See notes to Figure~\ref{fig:densities_thm3_1}.}

\label{fig:densities_rem5_1}

\end{figure}

\begin{figure}[ht]
\begin{center}
\caption*{$\lambda_T = T^{1/4}$}
\vspace{-1.5ex}
\begin{subfigure}{0.2\textwidth}
	\centering
	\caption*{$T = 25$}
	\vspace{-1.5ex}
	\includegraphics[trim={0cm 0cm 0.50cm 0.5cm},width=\textwidth,clip]
	{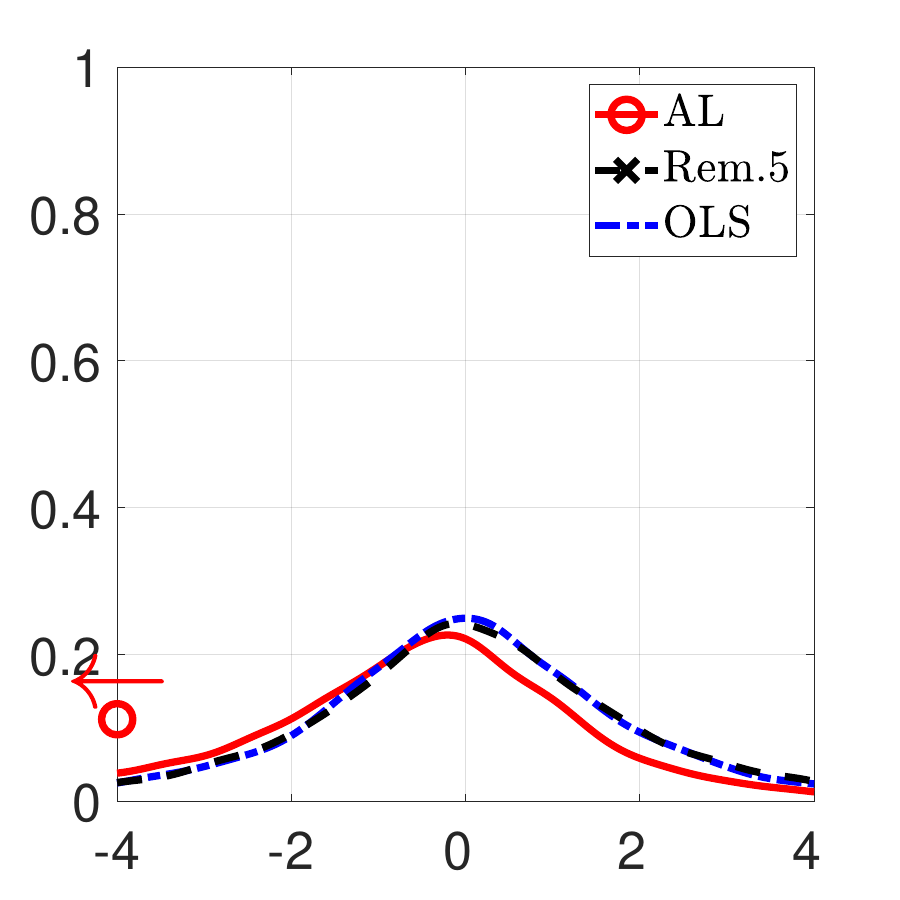}
\end{subfigure}\begin{subfigure}{0.2\textwidth}
	\centering
	\caption*{$T = 50$}
	\vspace{-1.5ex}
	\includegraphics[trim={0cm 0cm 0.50cm 0.5cm},width=\textwidth,clip]
	{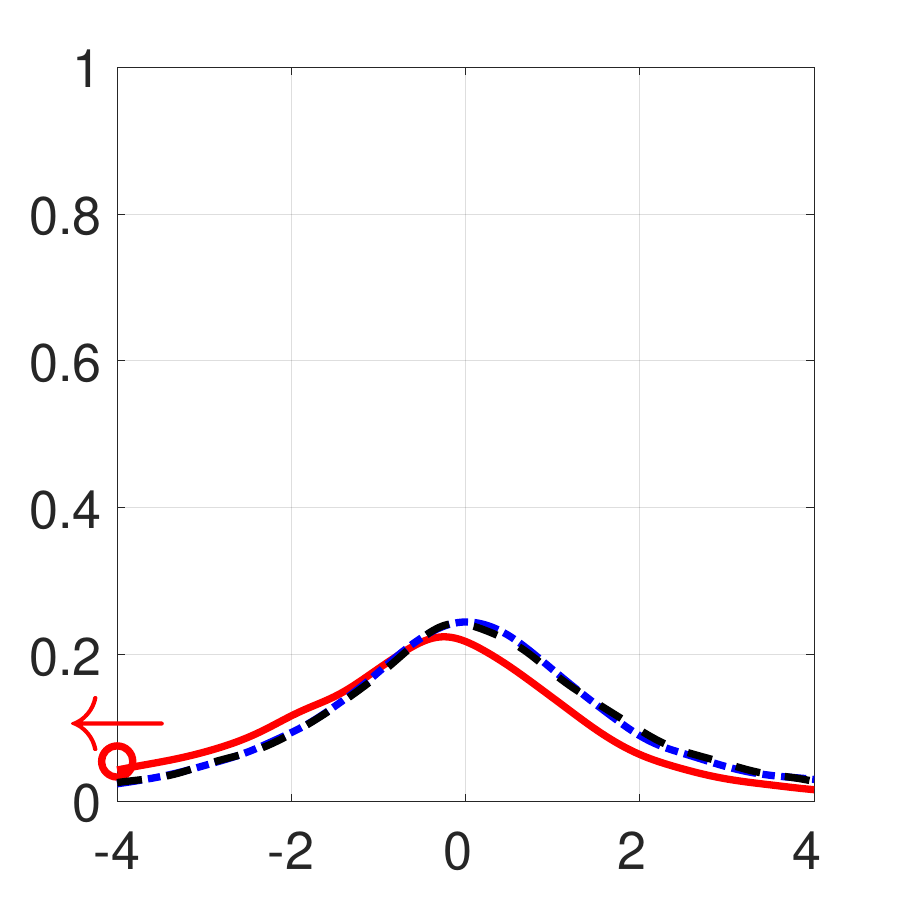}
\end{subfigure}\begin{subfigure}{0.2\textwidth}
	\centering
	\caption*{$T = 100$}
	\vspace{-1.5ex}
	\includegraphics[trim={0cm 0cm 0.50cm 0.5cm},width=\textwidth,clip]
	{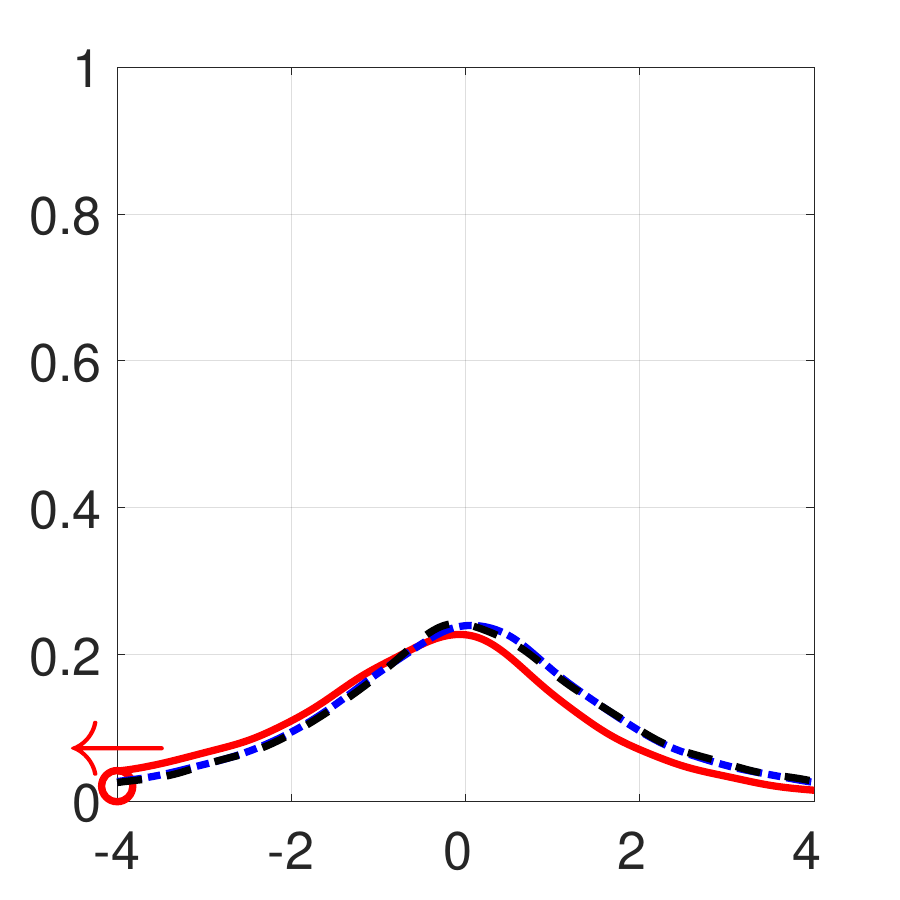}
\end{subfigure}\begin{subfigure}{0.2\textwidth}
	\centering
	\caption*{$T = 250$}
	\vspace{-1.5ex}
	\includegraphics[trim={0cm 0cm 0.50cm 0.5cm},width=\textwidth,clip]
	{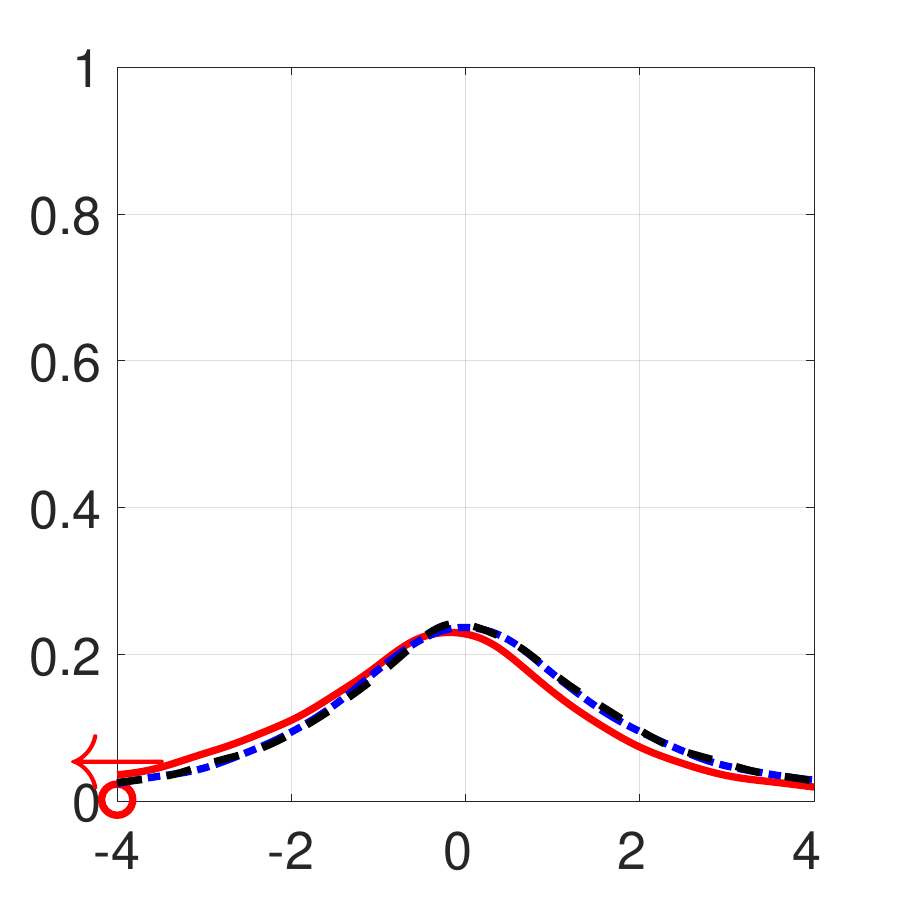}
\end{subfigure}\begin{subfigure}{0.2\textwidth}
	\centering
	\caption*{$T = 1000$}
	\vspace{-1.5ex}
	\includegraphics[trim={0cm 0cm 0.50cm 0.5cm},width=\textwidth,clip]
	{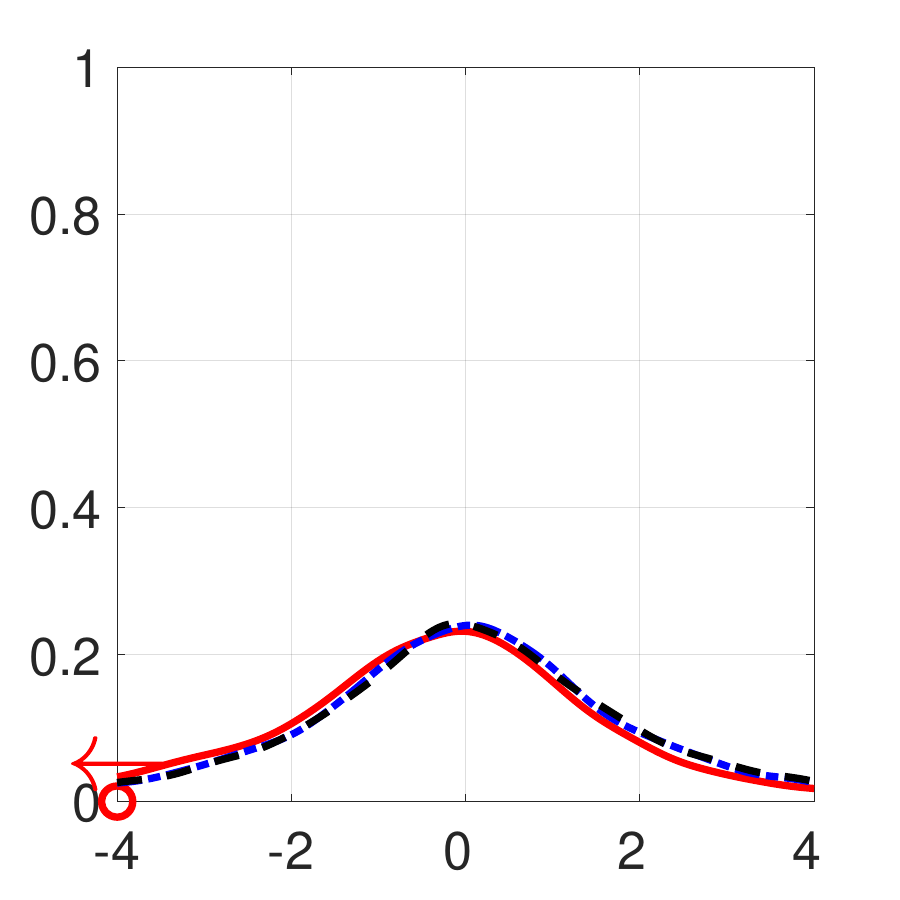}
\end{subfigure}

\caption*{$\lambda_T = T^{1/2}$}
\vspace{-1.5ex}
\begin{subfigure}{0.2\textwidth}
	\centering
	\includegraphics[trim={0cm 0cm 0.50cm 0.5cm},width=\textwidth,clip]
	{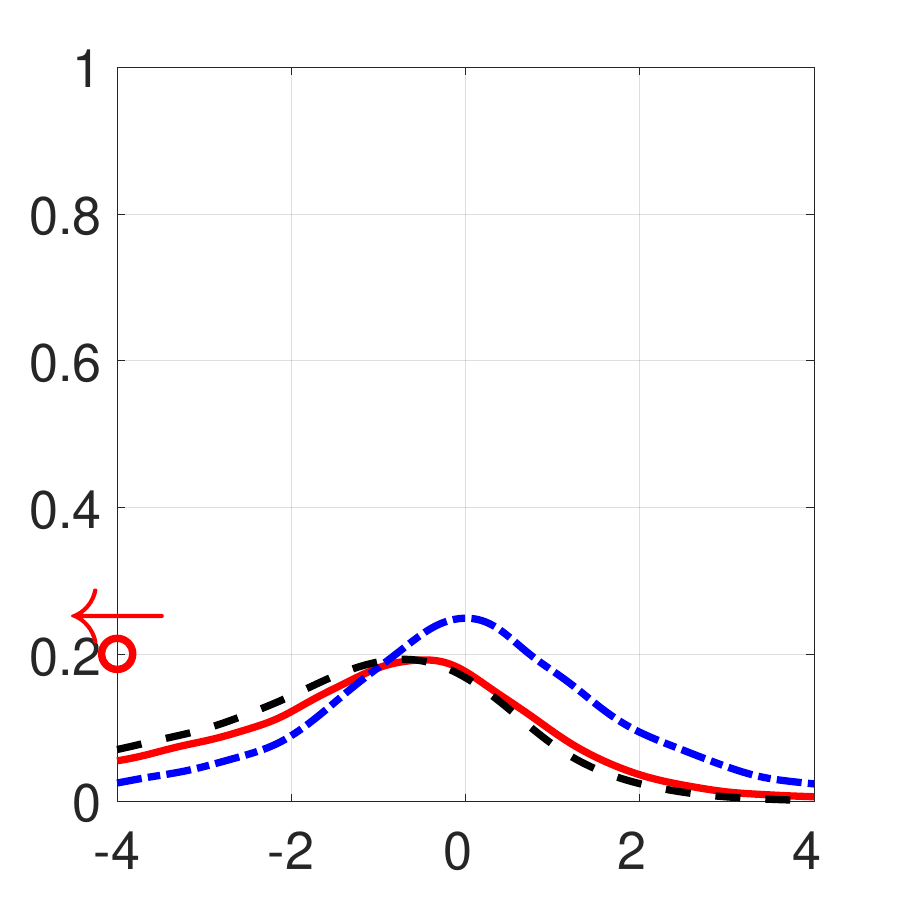}
\end{subfigure}\begin{subfigure}{0.2\textwidth}
	\centering
	\includegraphics[trim={0cm 0cm 0.50cm 0.5cm},width=\textwidth,clip]
	{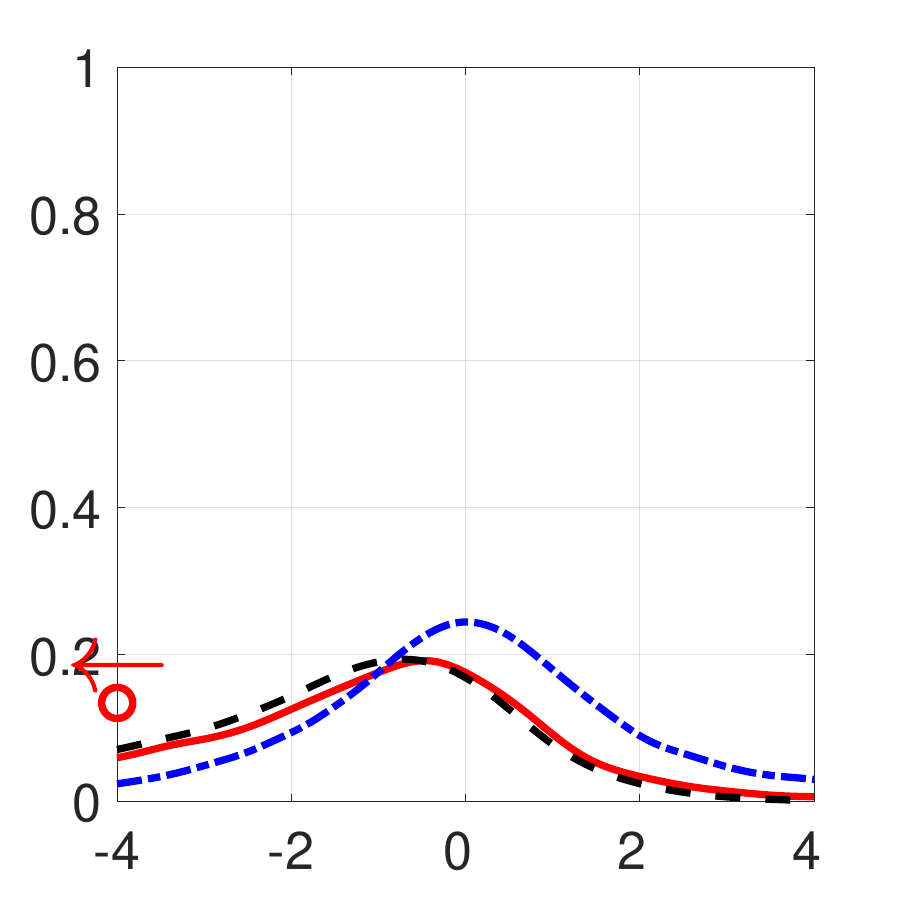}
\end{subfigure}\begin{subfigure}{0.2\textwidth}
	\centering
	\includegraphics[trim={0cm 0cm 0.50cm 0.5cm},width=\textwidth,clip]
	{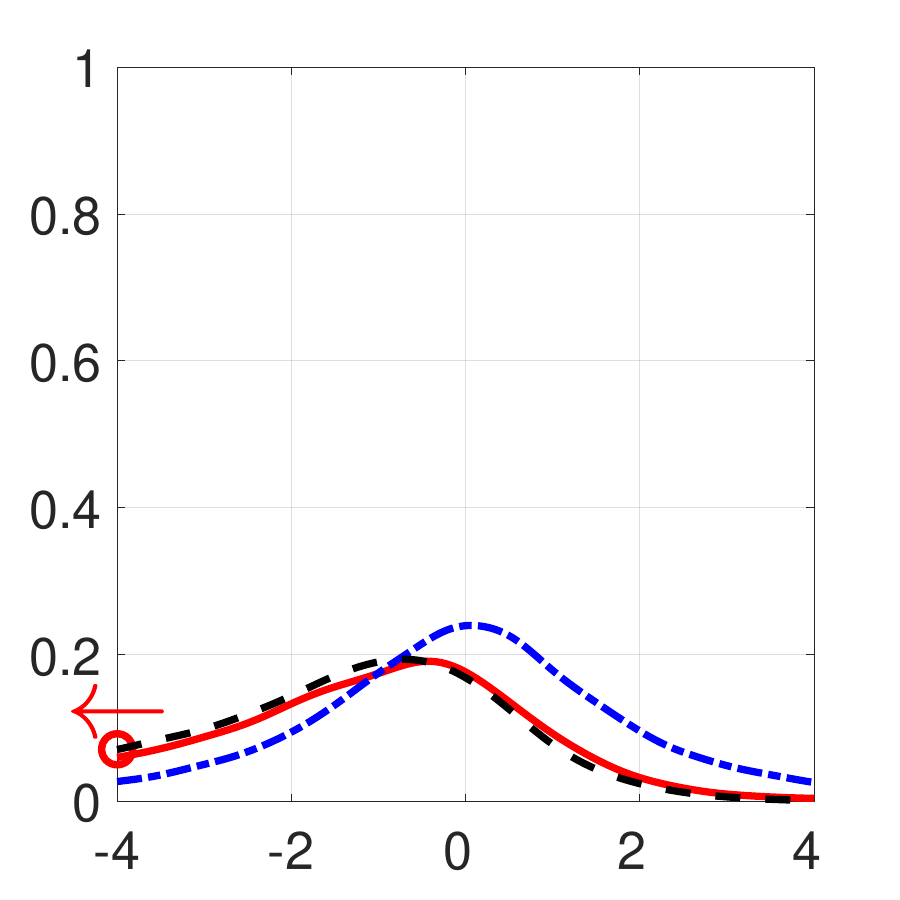}
\end{subfigure}\begin{subfigure}{0.2\textwidth}
	\centering
	\includegraphics[trim={0cm 0cm 0.50cm 0.5cm},width=\textwidth,clip]
	{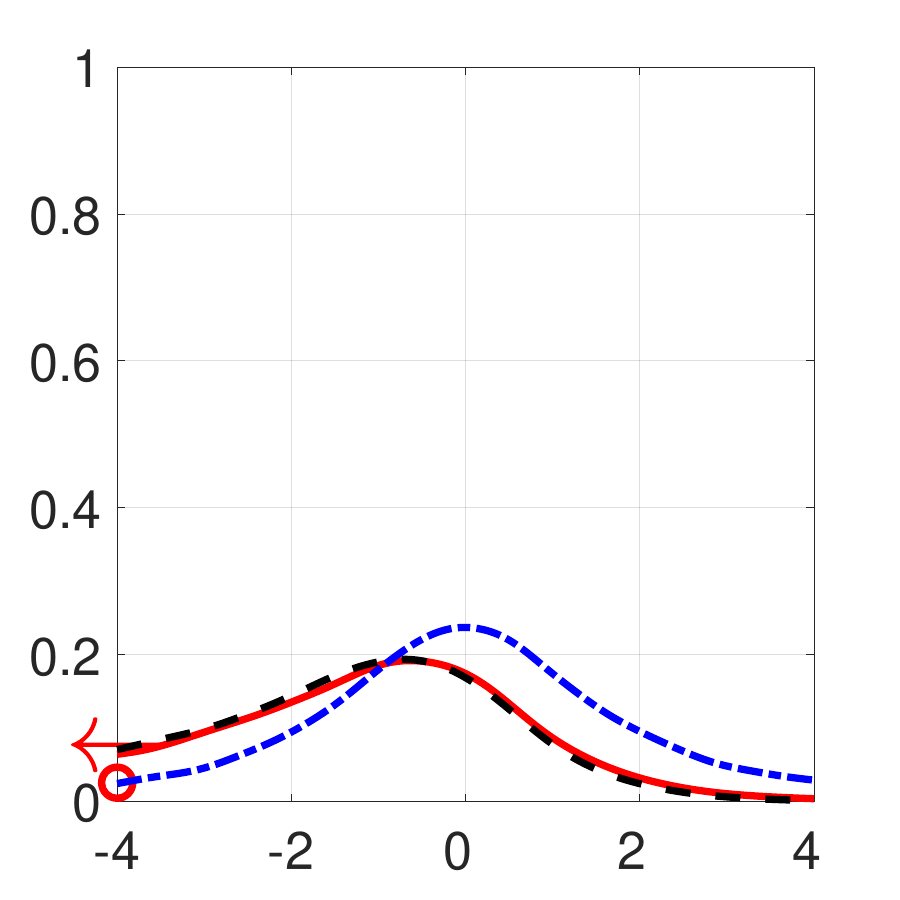}
\end{subfigure}\begin{subfigure}{0.2\textwidth}
	\centering
	\includegraphics[trim={0cm 0cm 0.50cm 0.5cm},width=\textwidth,clip]
	{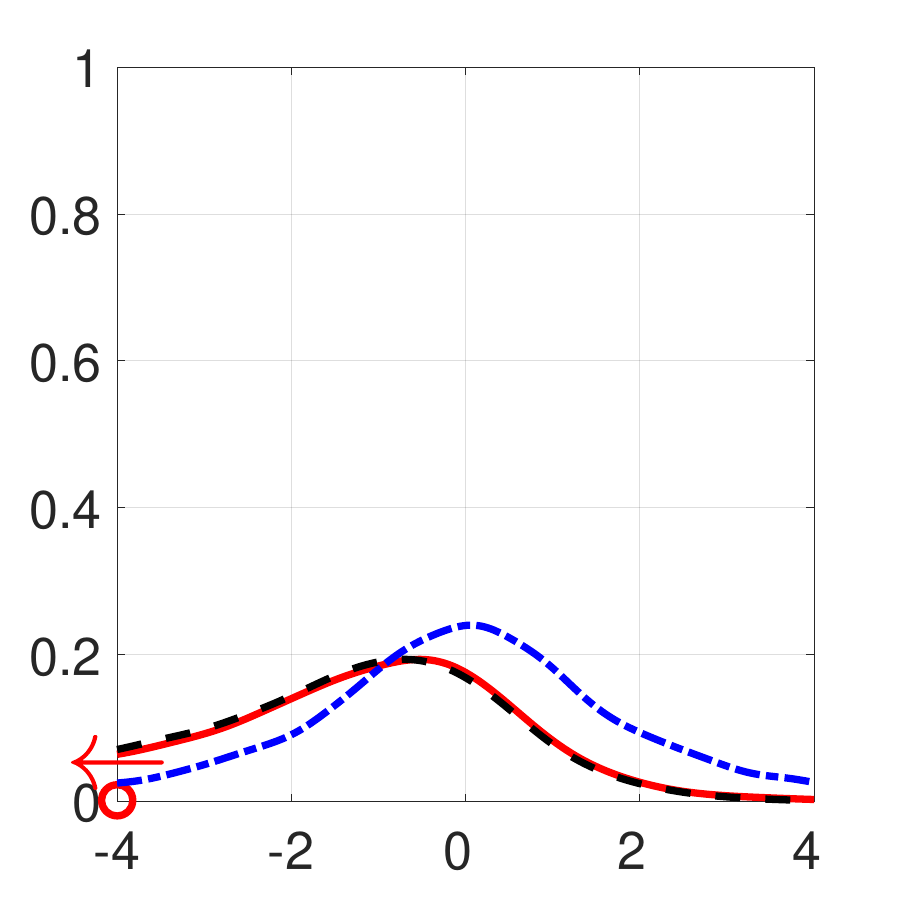}
\end{subfigure}

\caption*{$\lambda_T = T$}
\vspace{-1.5ex}
\begin{subfigure}{0.2\textwidth}
	\centering
	\includegraphics[trim={0cm 0cm 0.50cm 0.5cm},width=\textwidth,clip]
	{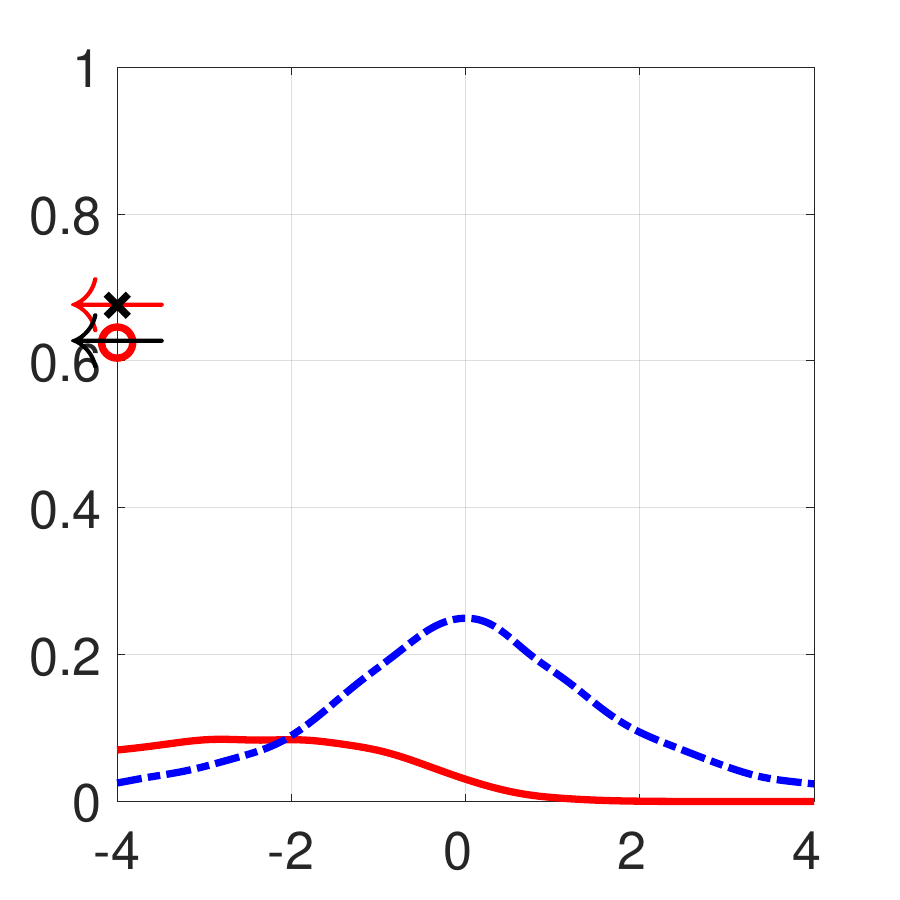}
\end{subfigure}\begin{subfigure}{0.2\textwidth}
	\centering
	\includegraphics[trim={0cm 0cm 0.50cm 0.5cm},width=\textwidth,clip]
	{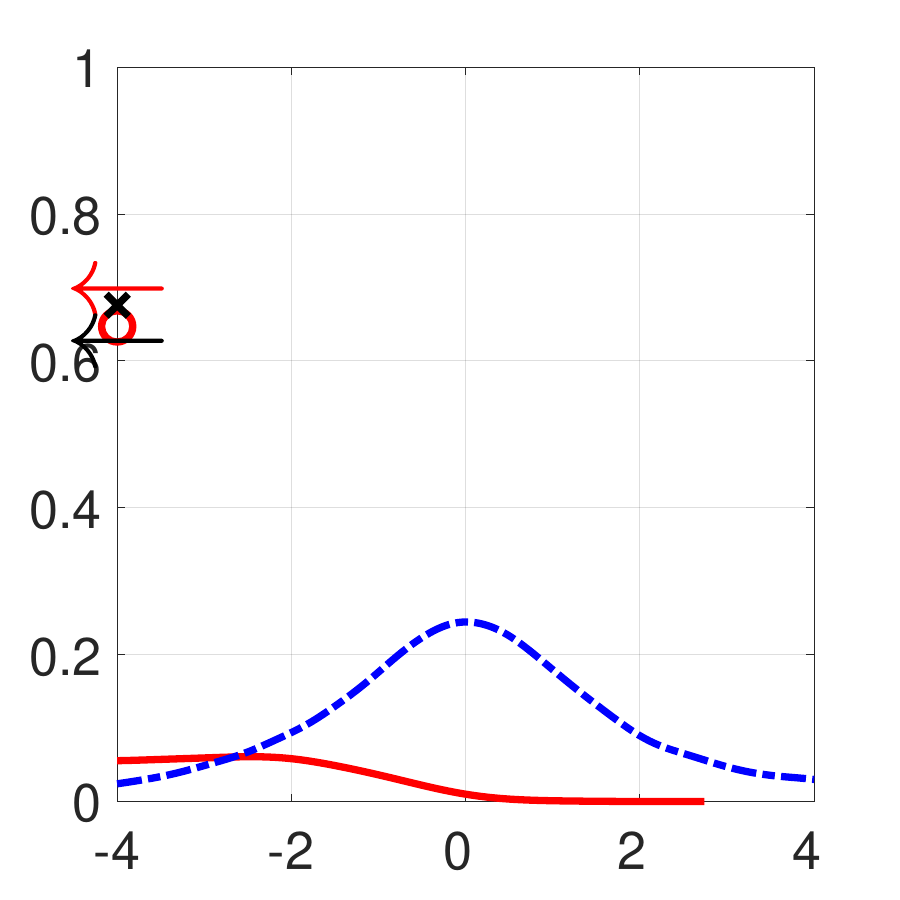}
\end{subfigure}\begin{subfigure}{0.2\textwidth}
	\centering
	\includegraphics[trim={0cm 0cm 0.50cm 0.5cm},width=\textwidth,clip]
	{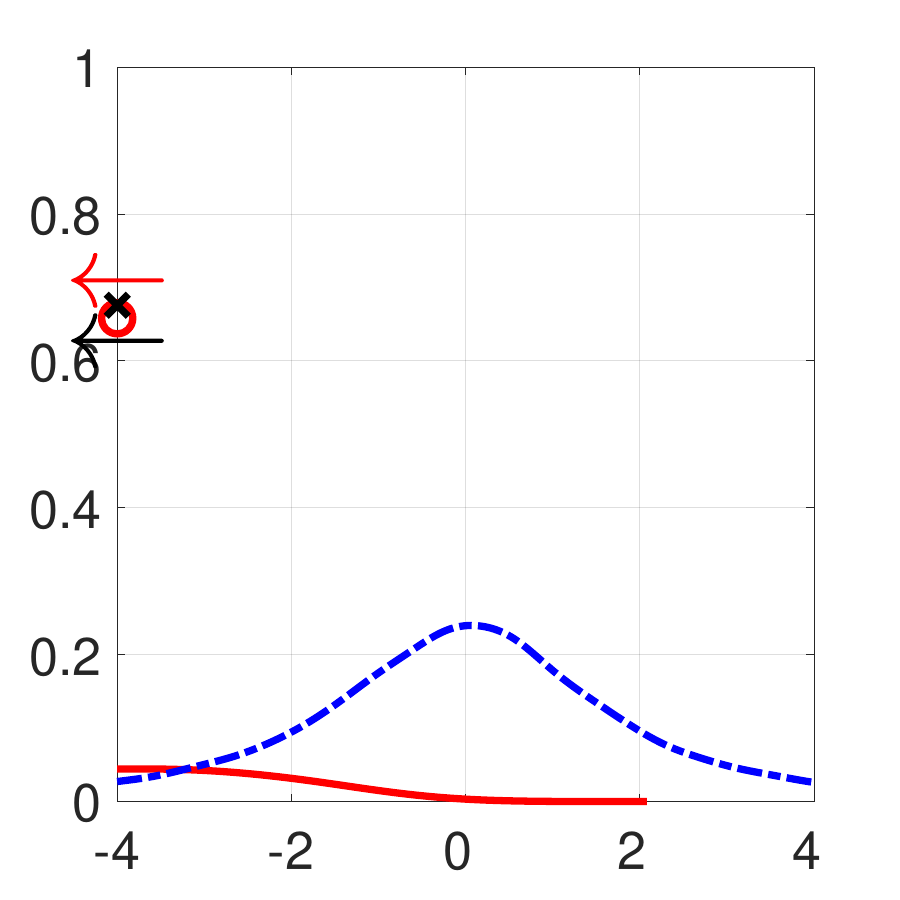}
\end{subfigure}\begin{subfigure}{0.2\textwidth}
	\centering
	\includegraphics[trim={0cm 0cm 0.50cm 0.5cm},width=\textwidth,clip]
	{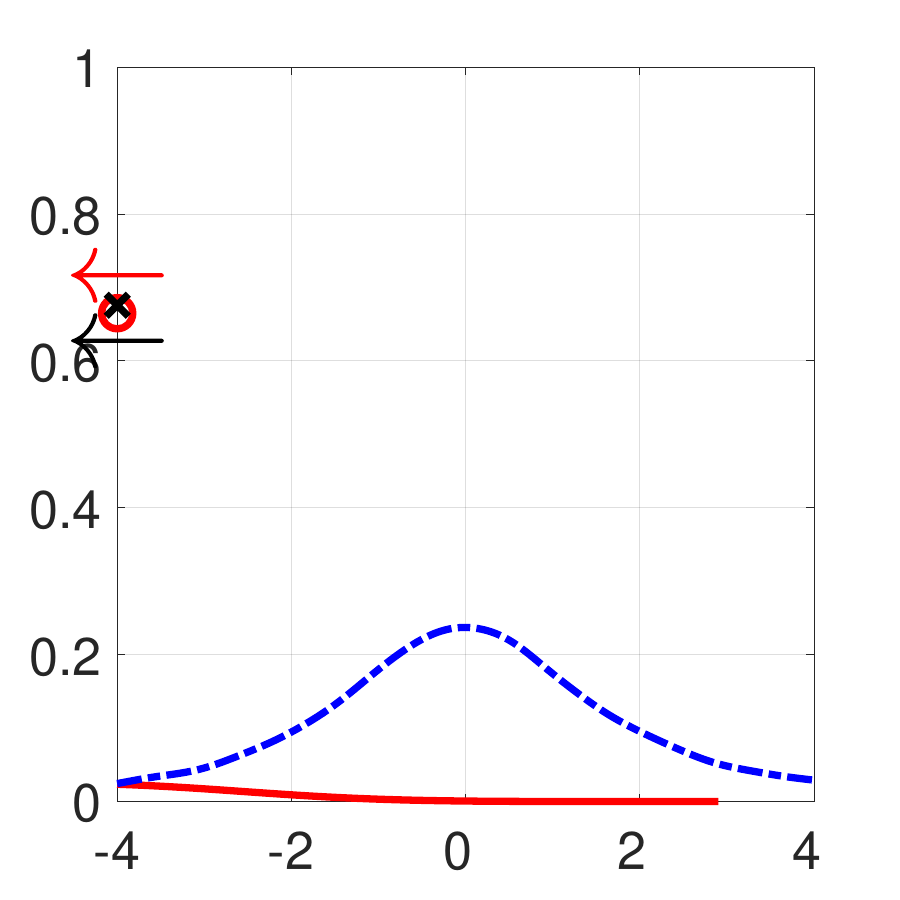}
\end{subfigure}\begin{subfigure}{0.2\textwidth}
	\centering
	\includegraphics[trim={0cm 0cm 0.50cm 0.5cm},width=\textwidth,clip]
	{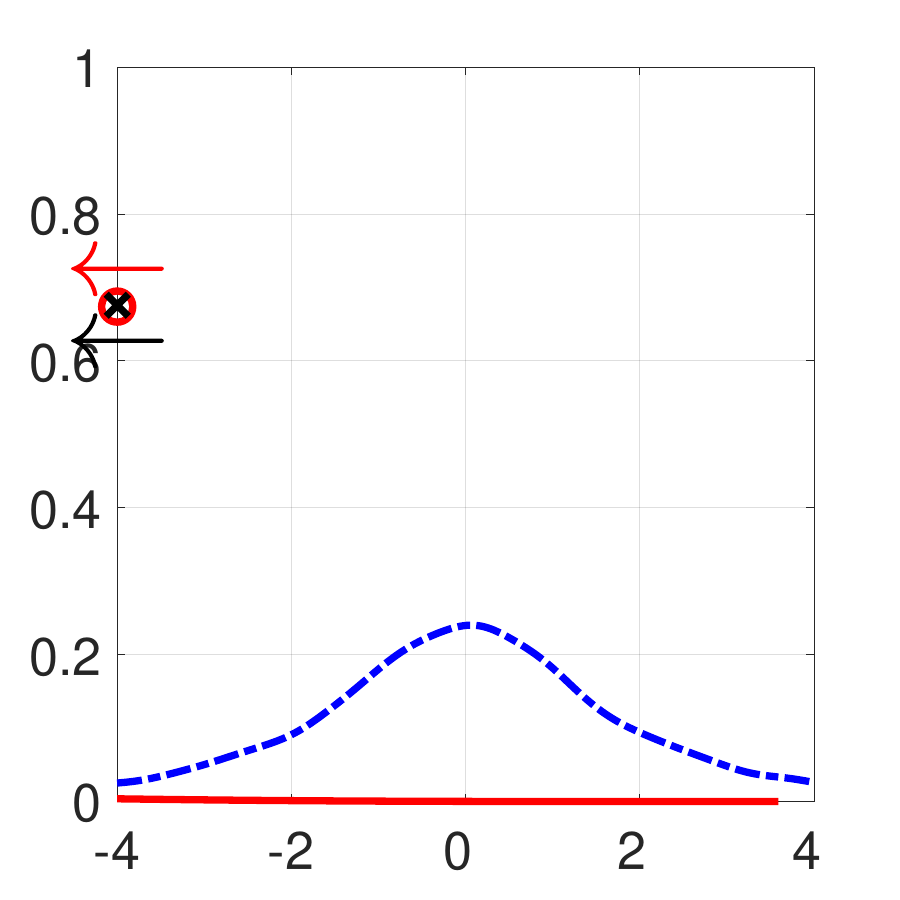}
\end{subfigure}

\end{center}

\vspace{-2ex} 

\caption{Finite-sample distribution of $T(\betaAL - \beta_T)$ under
consistent tuning (in all rows) in case $\beta_T = \beta/T^{1/2}$
(labeled ``AL''), and limiting distribution from
Remark~\ref{rem:ls_dist_consist_rateT-unif} (labeled ``Rem.5'').
\emph{Notes}: See notes to Figure~\ref{fig:densities_thm3_1}.}

\label{fig:densities_rem5_2}

\end{figure}

\begin{figure}[ht]
\begin{center}
\caption*{$\lambda_T = T^{1/4}$}
\vspace{-1.5ex}
\begin{subfigure}{0.2\textwidth}
	\centering
	\caption*{$T = 25$}
	\vspace{-1.5ex}
	\includegraphics[trim={0cm 0cm 0.50cm 0.5cm},width=\textwidth,clip]
	{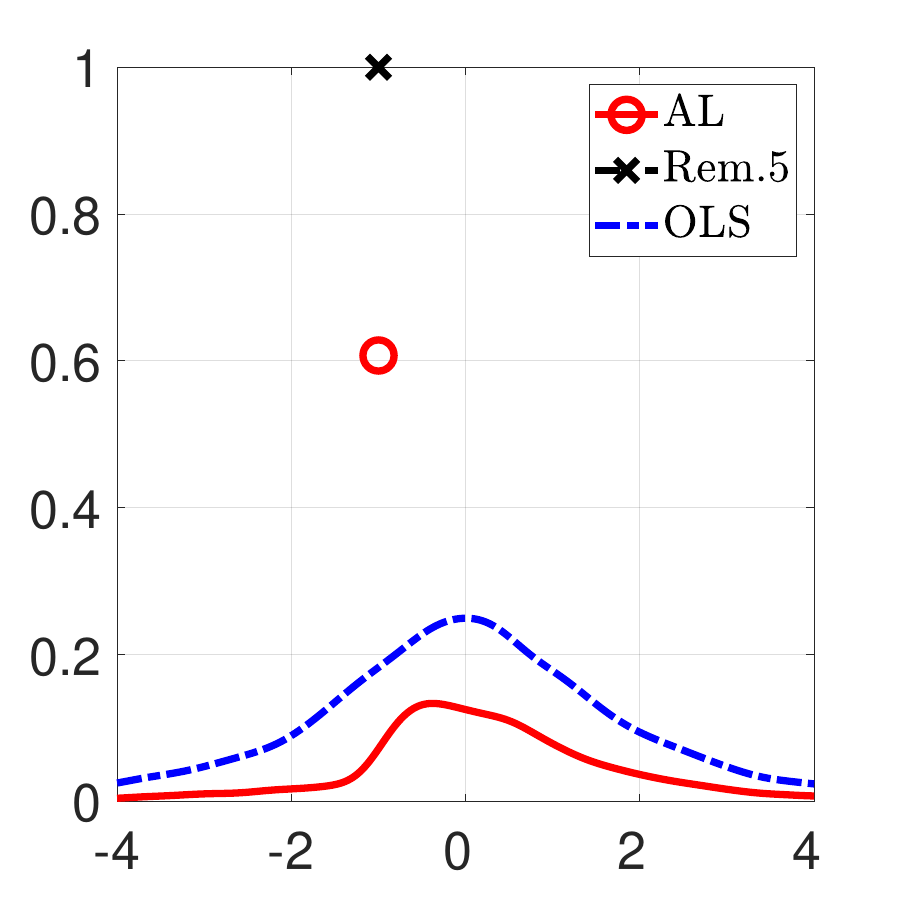}
\end{subfigure}\begin{subfigure}{0.2\textwidth}
	\centering
	\caption*{$T = 50$}
	\vspace{-1.5ex}
	\includegraphics[trim={0cm 0cm 0.50cm 0.5cm},width=\textwidth,clip]
	{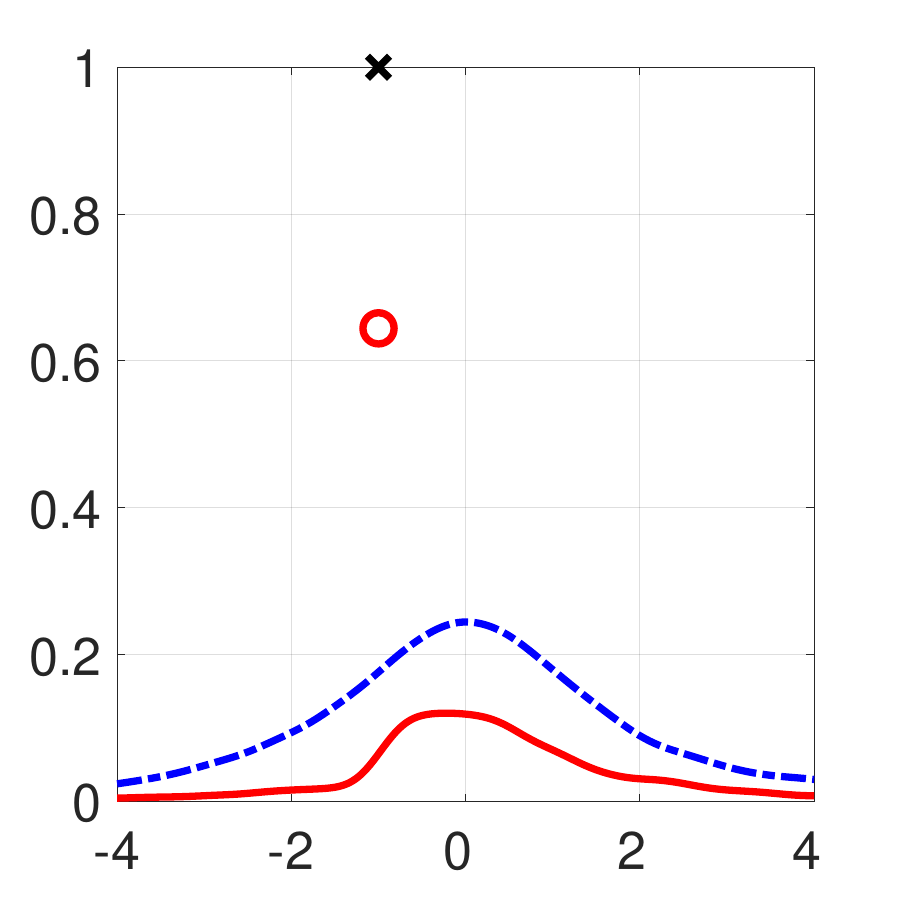}
\end{subfigure}\begin{subfigure}{0.2\textwidth}
	\centering
	\caption*{$T = 100$}
	\vspace{-1.5ex}
	\includegraphics[trim={0cm 0cm 0.50cm 0.5cm},width=\textwidth,clip]
	{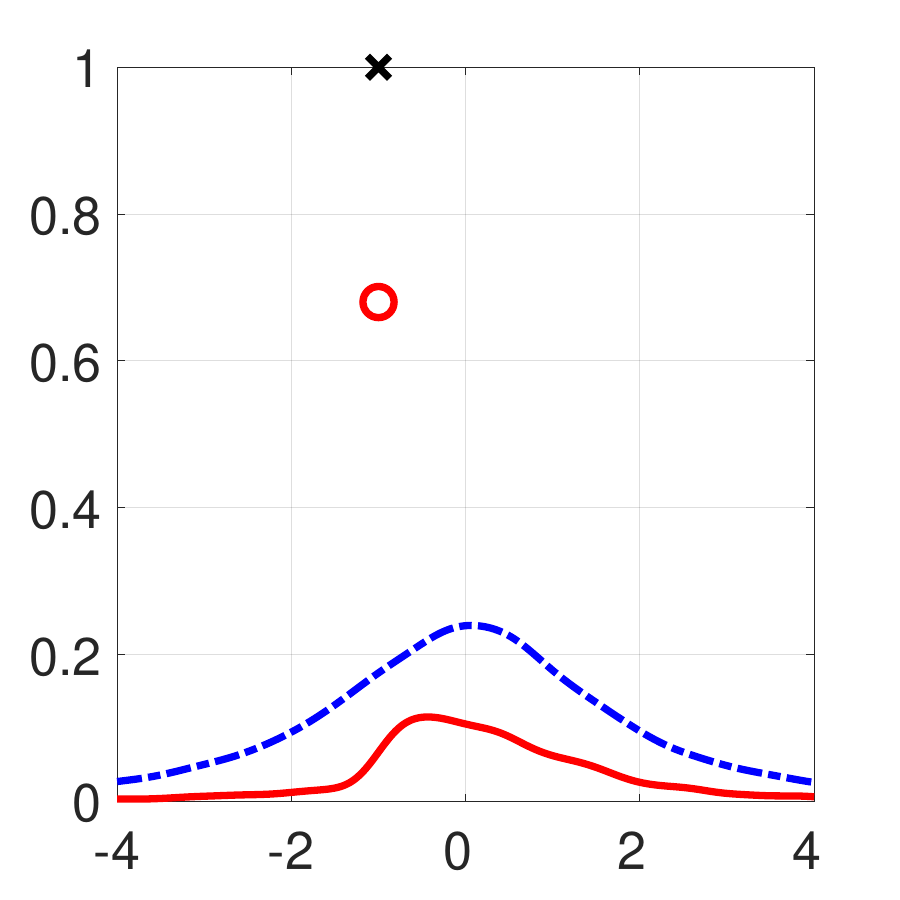}
\end{subfigure}\begin{subfigure}{0.2\textwidth}
	\centering
	\caption*{$T = 250$}
	\vspace{-1.5ex}
	\includegraphics[trim={0cm 0cm 0.50cm 0.5cm},width=\textwidth,clip]
	{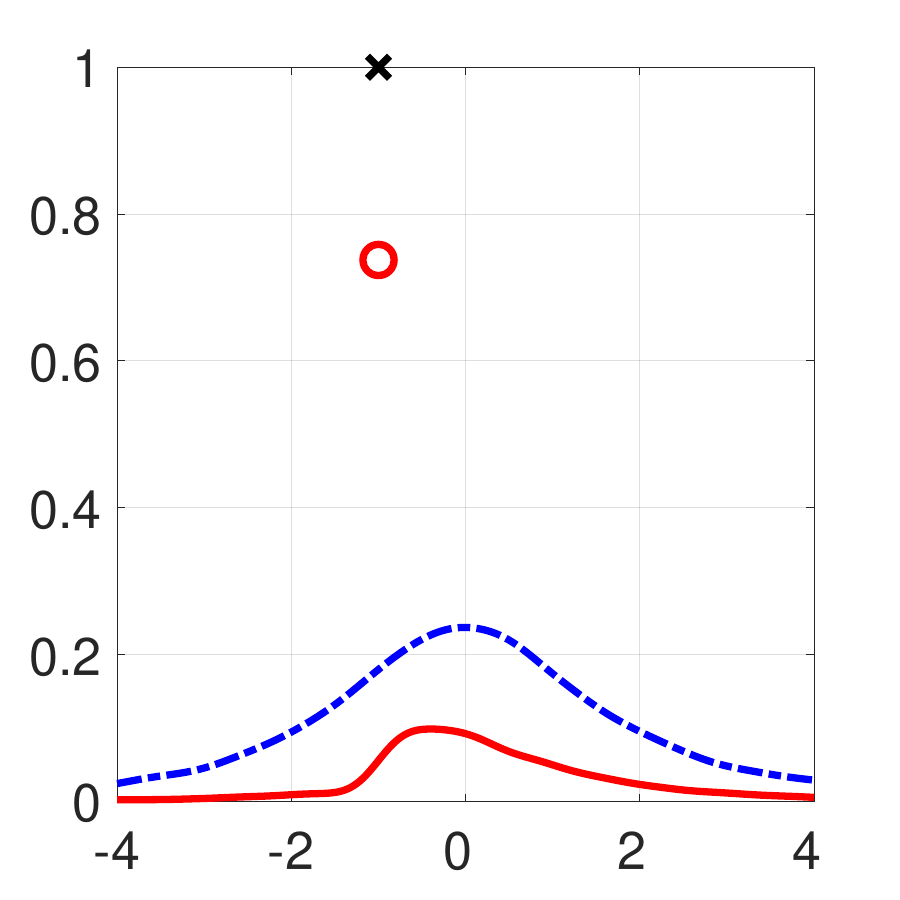}
\end{subfigure}\begin{subfigure}{0.2\textwidth}
	\centering
	\caption*{$T = 1000$}
	\vspace{-1.5ex}
	\includegraphics[trim={0cm 0cm 0.50cm 0.5cm},width=\textwidth,clip]
	{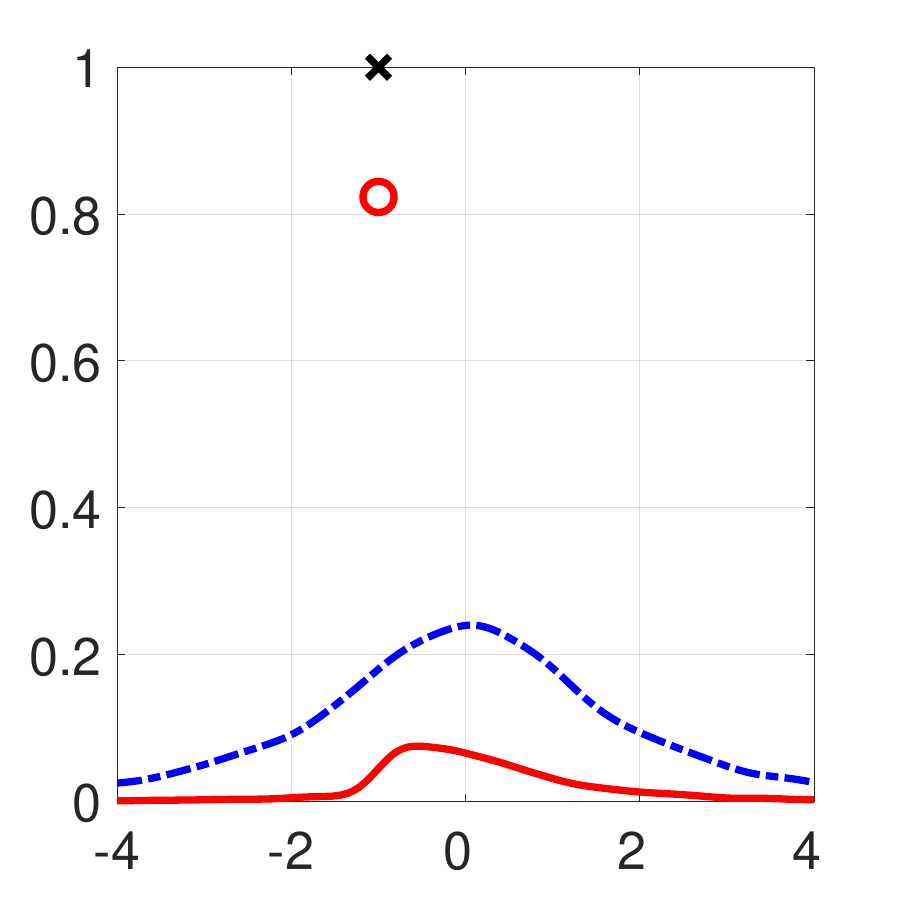}
\end{subfigure}

\caption*{$\lambda_T = T^{1/2}$}
\vspace{-1.5ex}
\begin{subfigure}{0.2\textwidth}
	\centering
	\includegraphics[trim={0cm 0cm 0.50cm 0.5cm},width=\textwidth,clip]
	{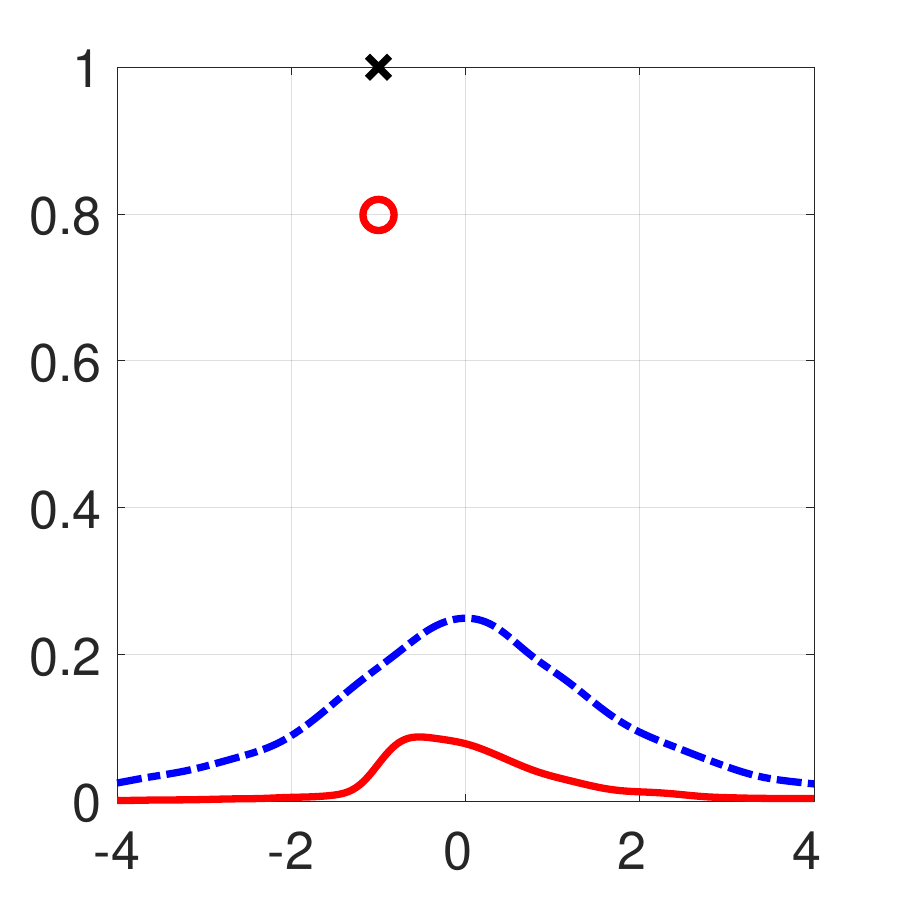}
\end{subfigure}\begin{subfigure}{0.2\textwidth}
	\centering
	\includegraphics[trim={0cm 0cm 0.50cm 0.5cm},width=\textwidth,clip]
	{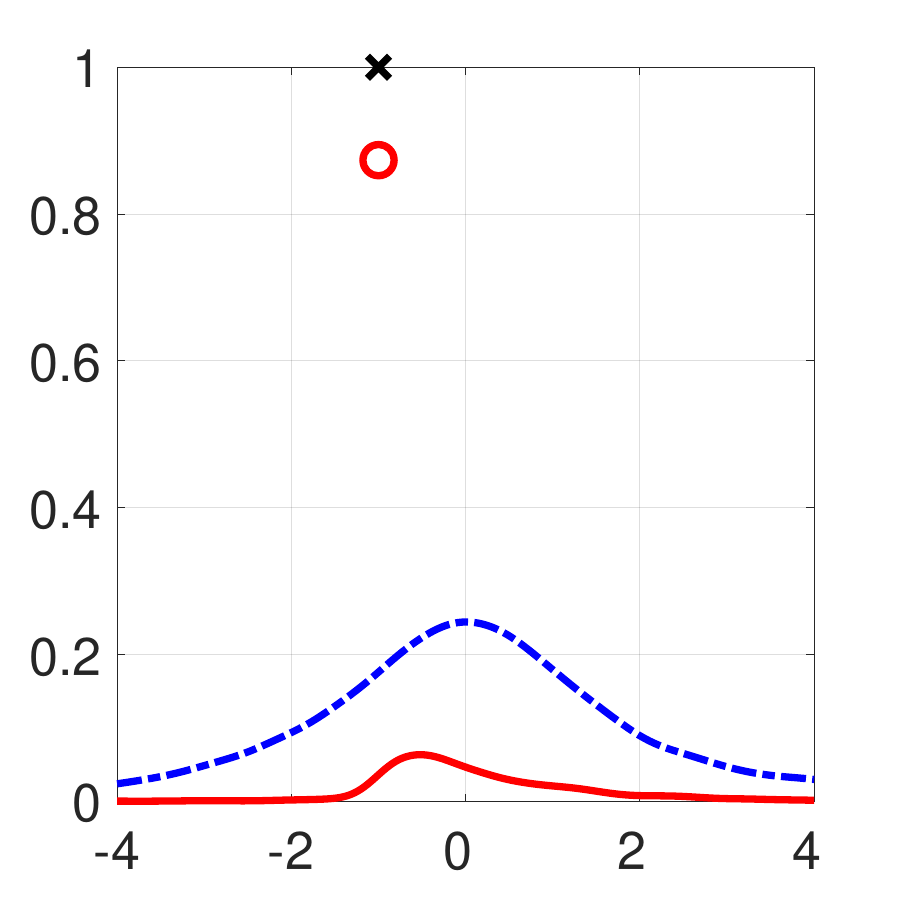}
\end{subfigure}\begin{subfigure}{0.2\textwidth}
	\centering
	\includegraphics[trim={0cm 0cm 0.50cm 0.5cm},width=\textwidth,clip]
	{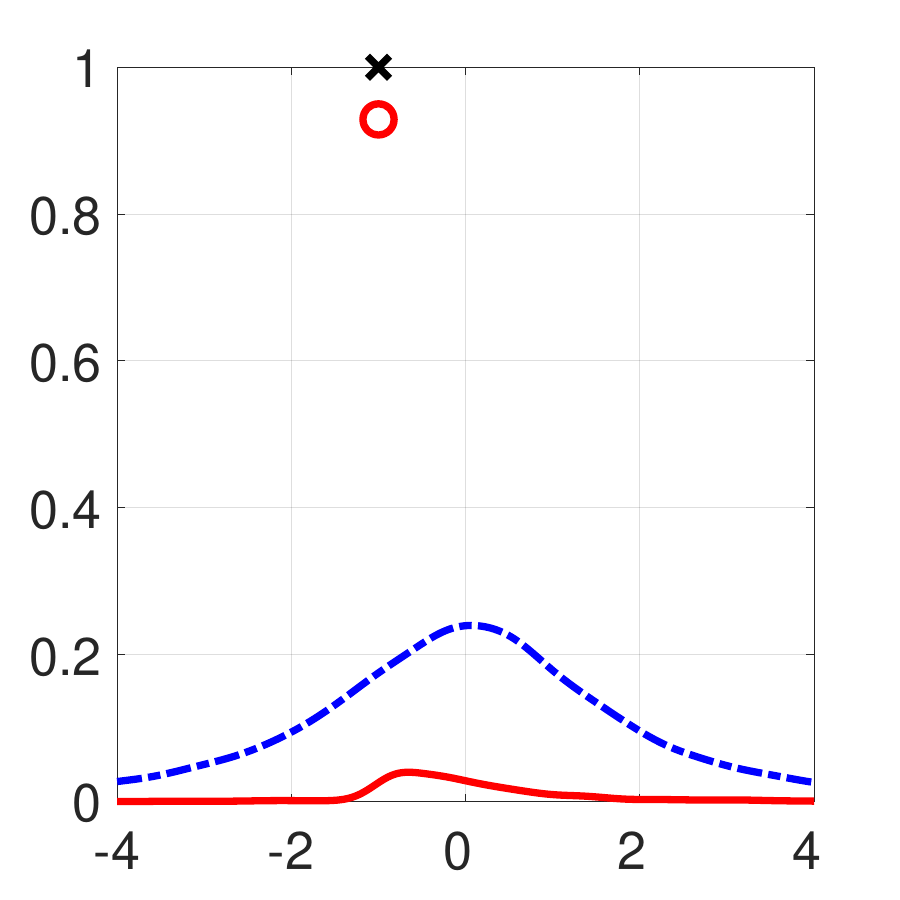}
\end{subfigure}\begin{subfigure}{0.2\textwidth}
	\centering
	\includegraphics[trim={0cm 0cm 0.50cm 0.5cm},width=\textwidth,clip]
	{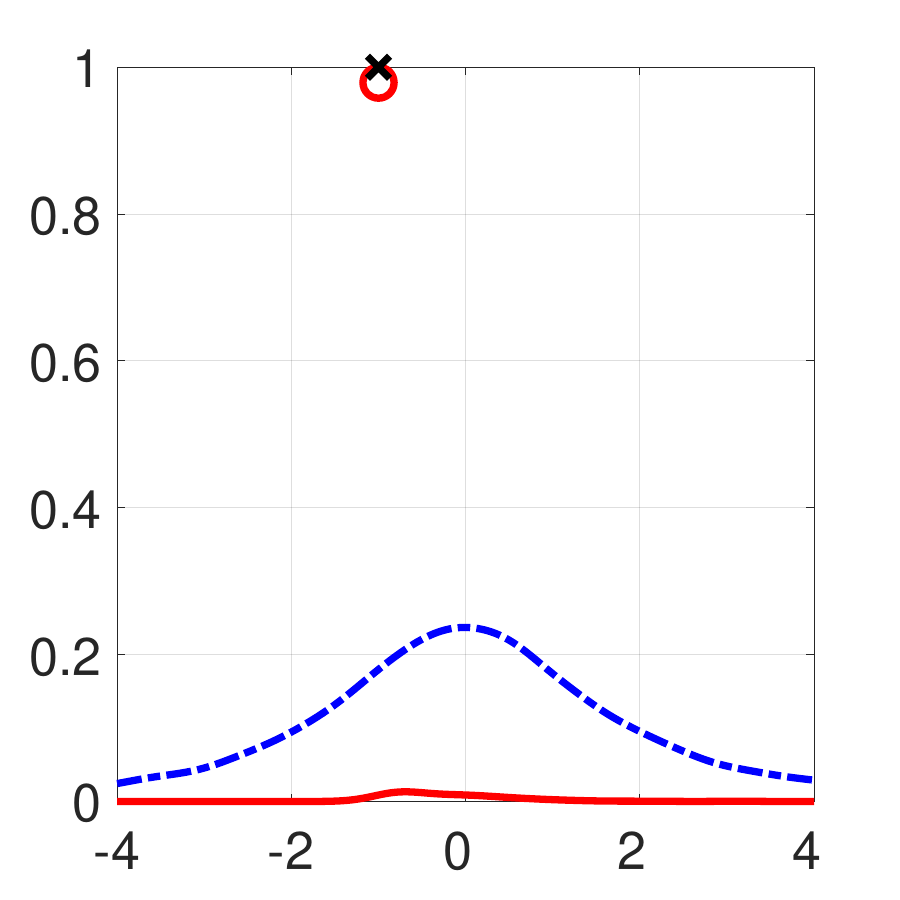}
\end{subfigure}\begin{subfigure}{0.2\textwidth}
	\centering
	\includegraphics[trim={0cm 0cm 0.50cm 0.5cm},width=\textwidth,clip]
	{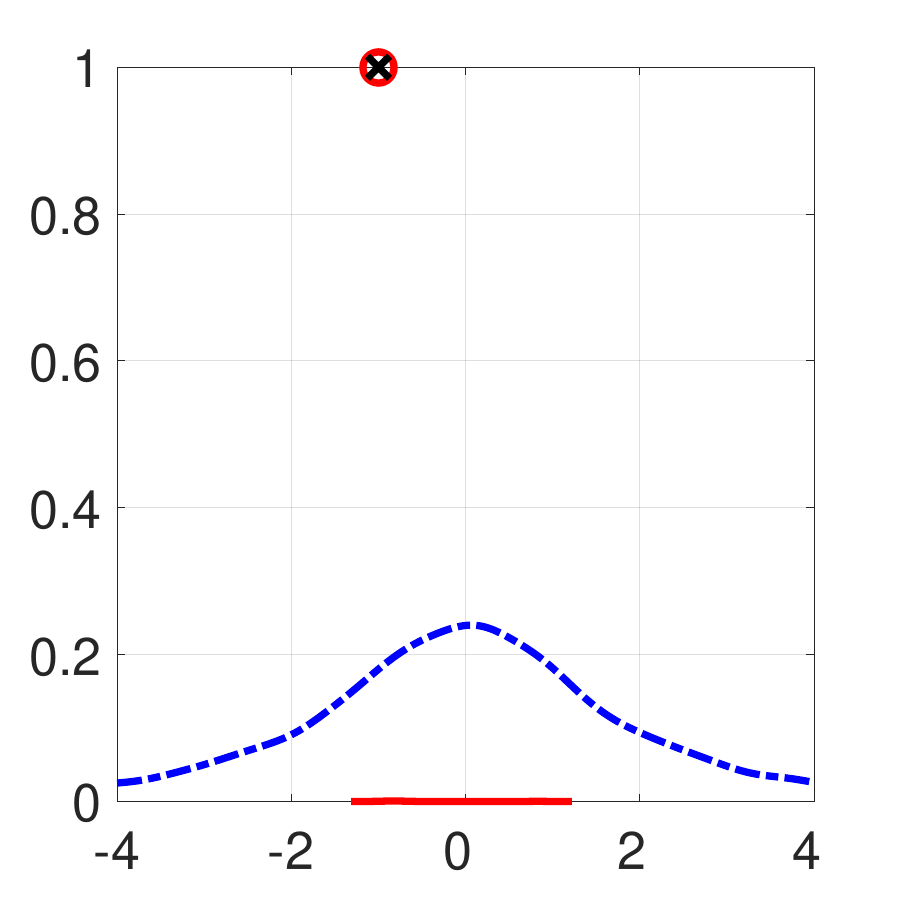}
\end{subfigure}

\caption*{$\lambda_T = T$}
\vspace{-1.5ex}
\begin{subfigure}{0.2\textwidth}
	\centering
	\includegraphics[trim={0cm 0cm 0.50cm 0.5cm},width=\textwidth,clip]
	{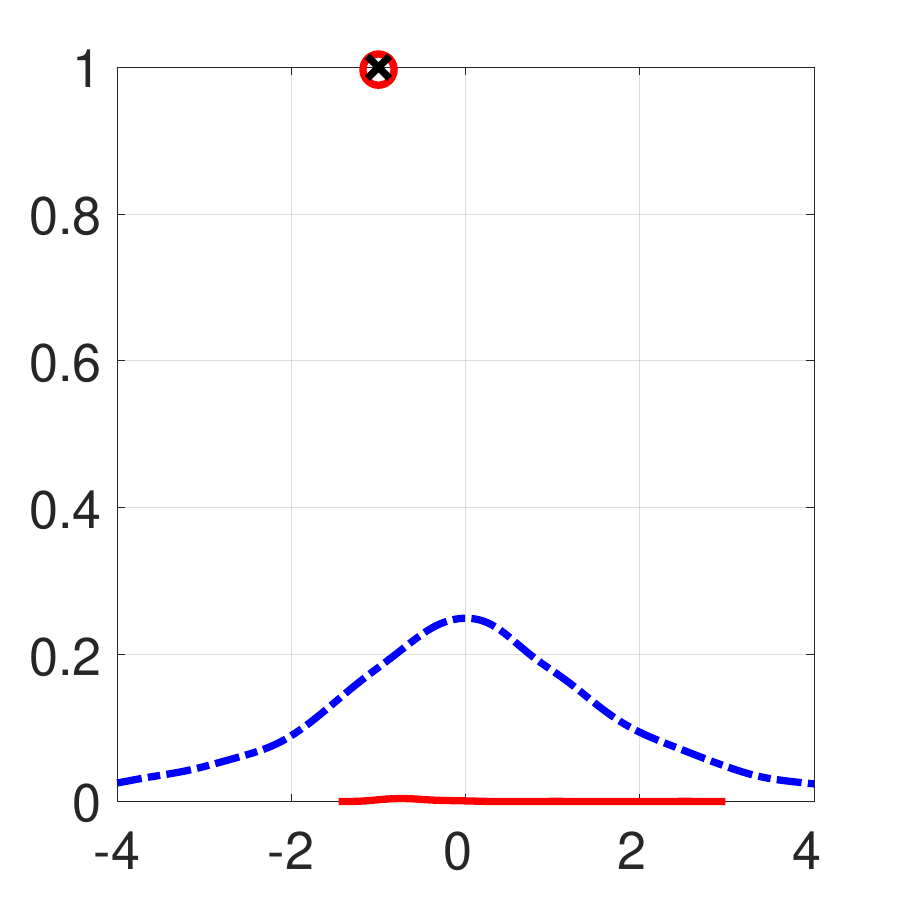}
\end{subfigure}\begin{subfigure}{0.2\textwidth}
	\centering
	\includegraphics[trim={0cm 0cm 0.50cm 0.5cm},width=\textwidth,clip]
	{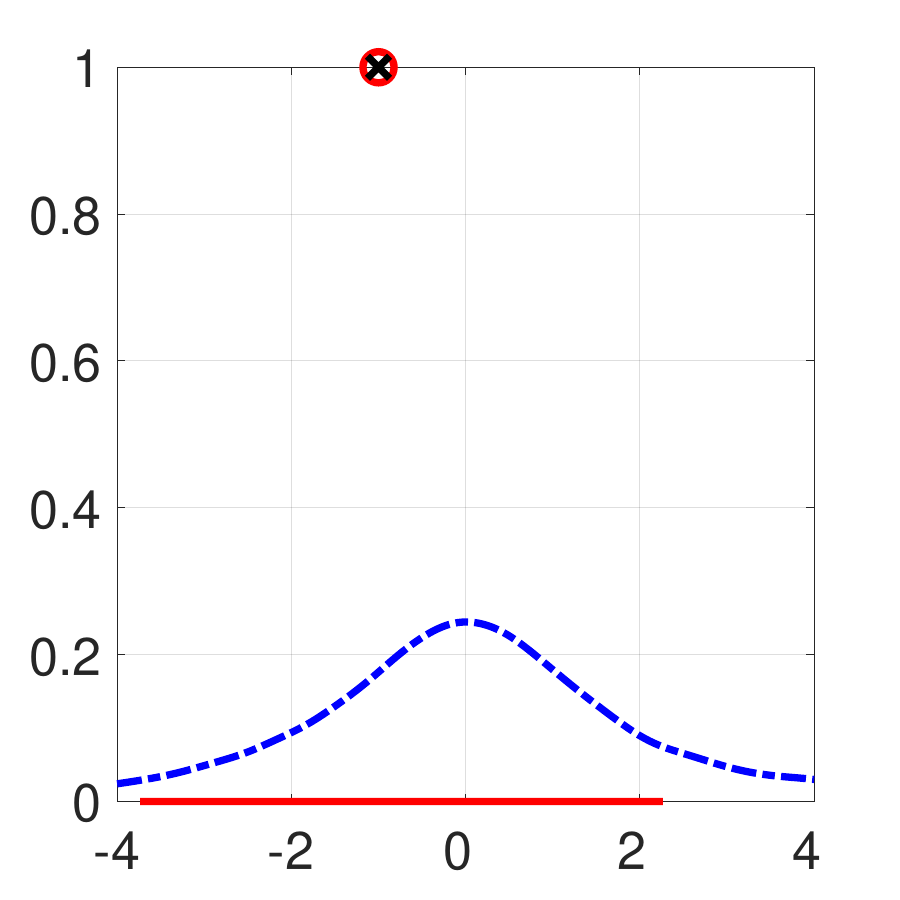}
\end{subfigure}\begin{subfigure}{0.2\textwidth}
	\centering
	\includegraphics[trim={0cm 0cm 0.50cm 0.5cm},width=\textwidth,clip]
	{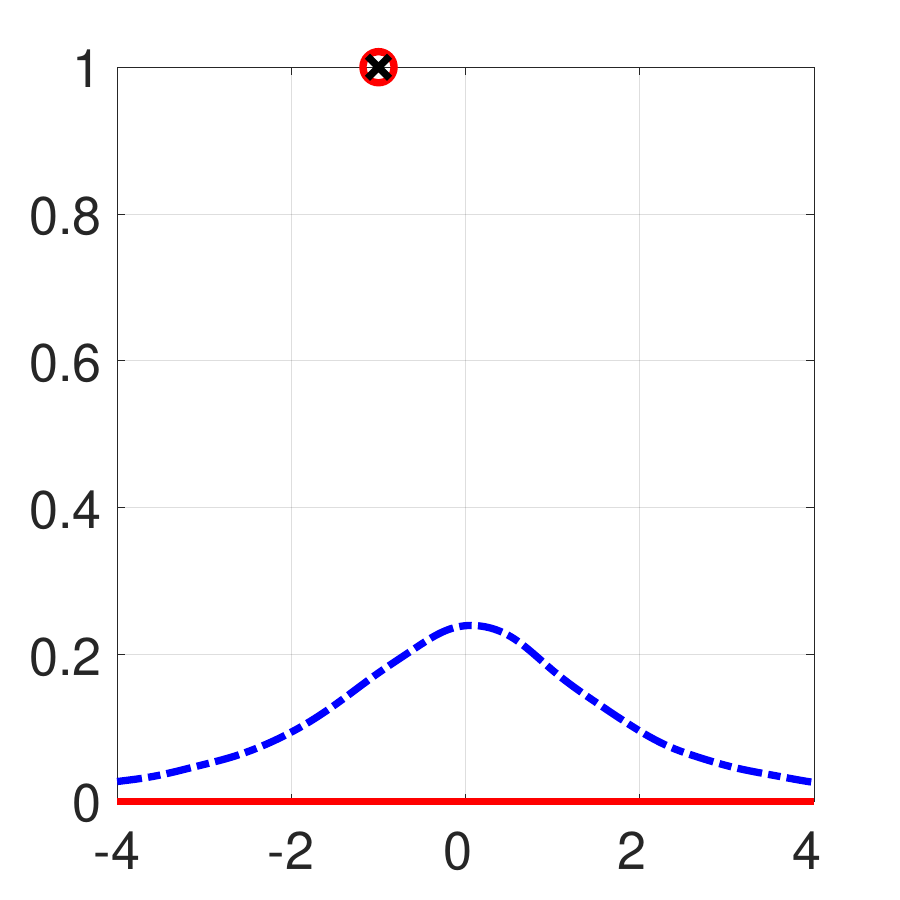}
\end{subfigure}\begin{subfigure}{0.2\textwidth}
	\centering
	\includegraphics[trim={0cm 0cm 0.50cm 0.5cm},width=\textwidth,clip]
	{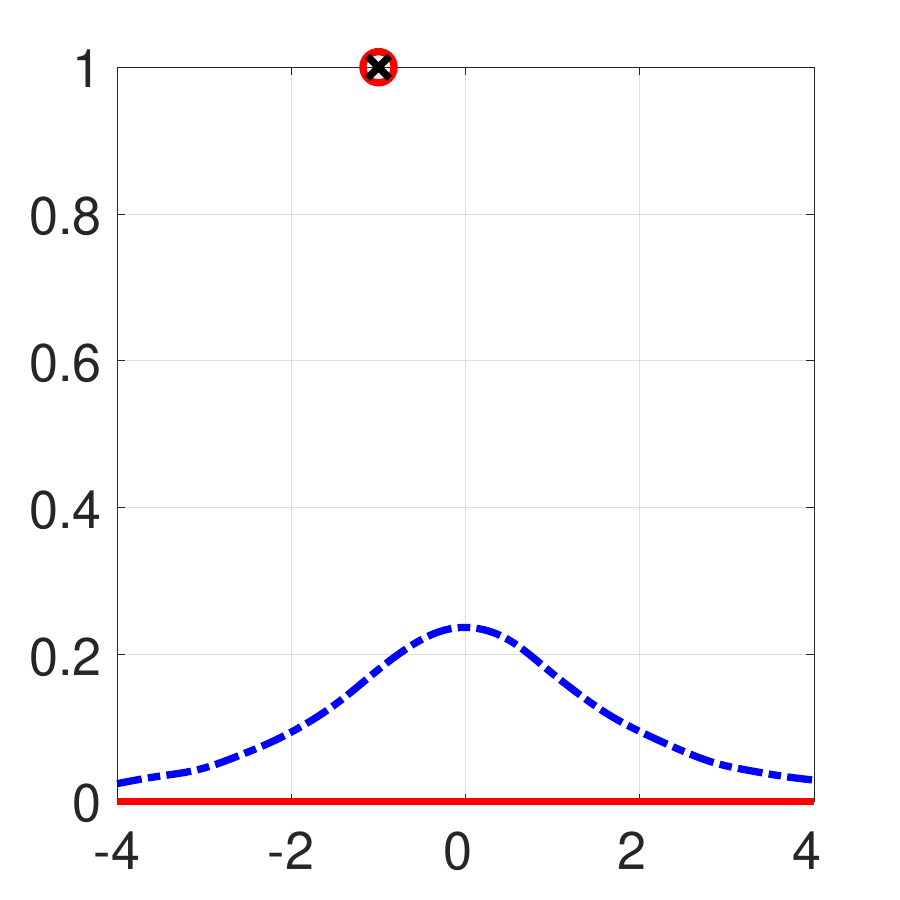}
\end{subfigure}\begin{subfigure}{0.2\textwidth}
	\centering
	\includegraphics[trim={0cm 0cm 0.50cm 0.5cm},width=\textwidth,clip]
	{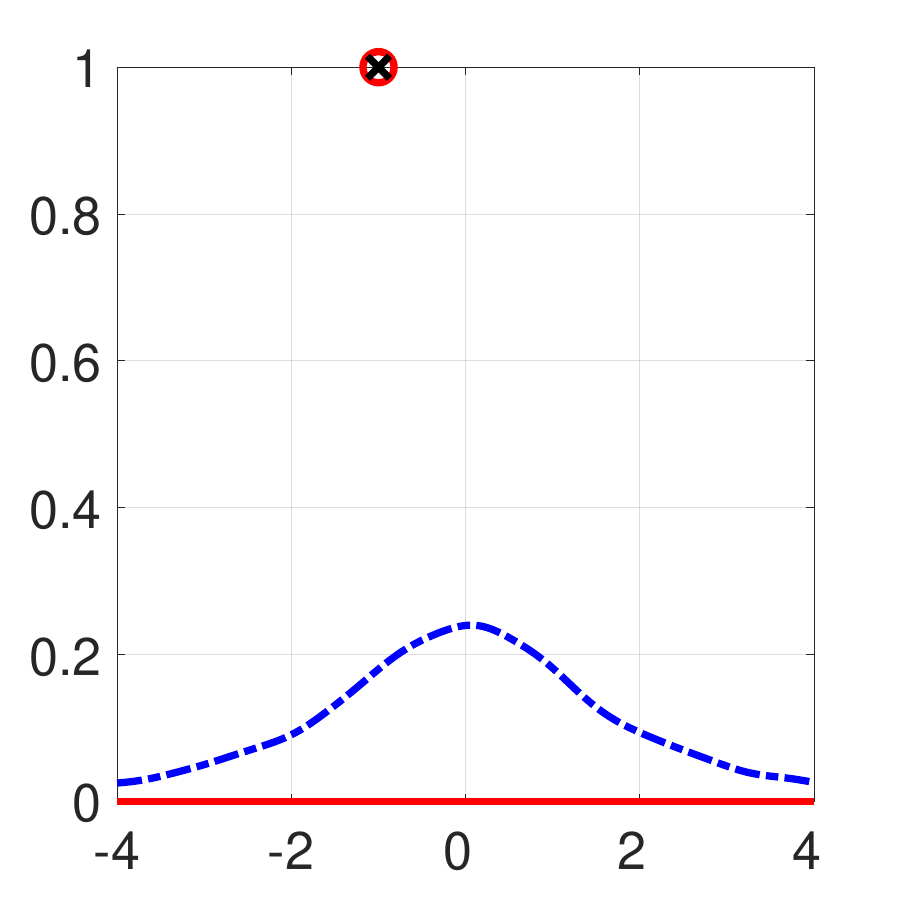}
\end{subfigure}

\end{center}

\vspace{-2ex} 

\caption{Finite-sample distribution of $T(\betaAL - \beta_T)$ under
consistent tuning (in all rows) in case $\beta_T = \beta/T$ (labeled
``AL''), and limiting distribution from
Remark~\ref{rem:ls_dist_consist_rateT-unif} (labeled ``Rem.5'').
\emph{Notes}: See notes to Figure~\ref{fig:densities_thm3_1}.}

\label{fig:densities_rem5_3}

\end{figure}

\begin{figure}[ht]
\begin{center}
\caption*{$\lambda_T = T^{1/4}$}
\vspace{-1.5ex}
\begin{subfigure}{0.2\textwidth}
	\centering
	\caption*{$T = 25$}
	\vspace{-1.5ex}
	\includegraphics[trim={0cm 0cm 0.50cm 0.5cm},width=\textwidth,clip]
	{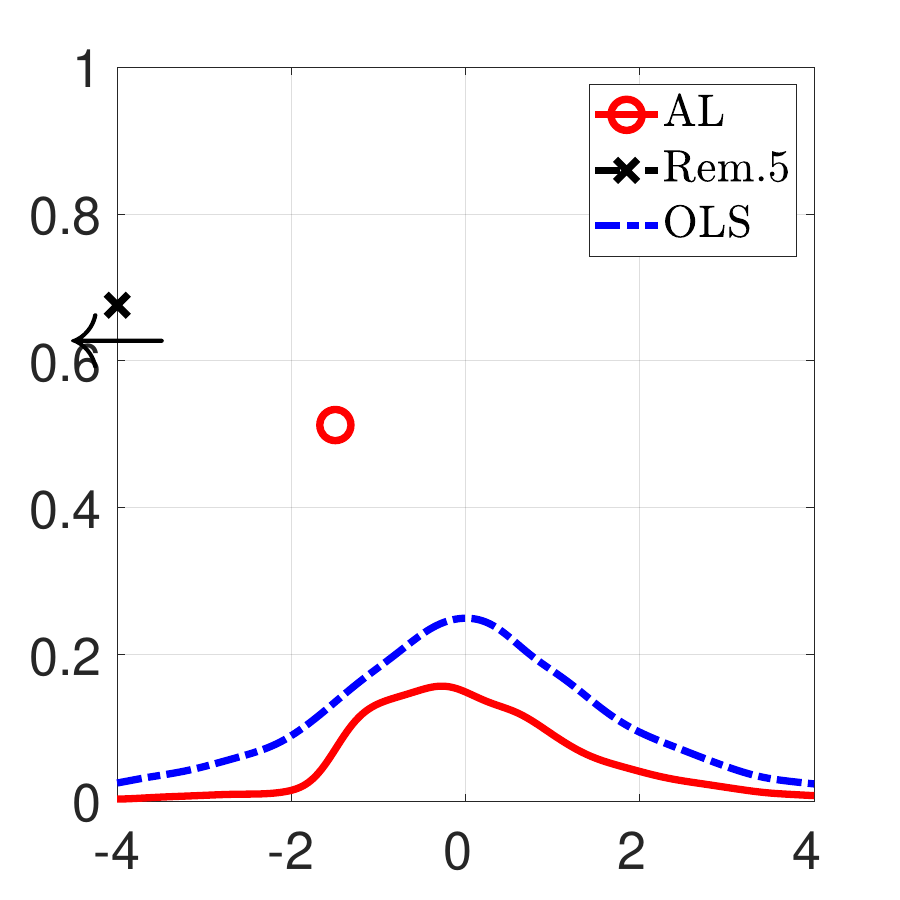}
\end{subfigure}\begin{subfigure}{0.2\textwidth}
	\centering
	\caption*{$T = 50$}
	\vspace{-1.5ex}
	\includegraphics[trim={0cm 0cm 0.50cm 0.5cm},width=\textwidth,clip]
	{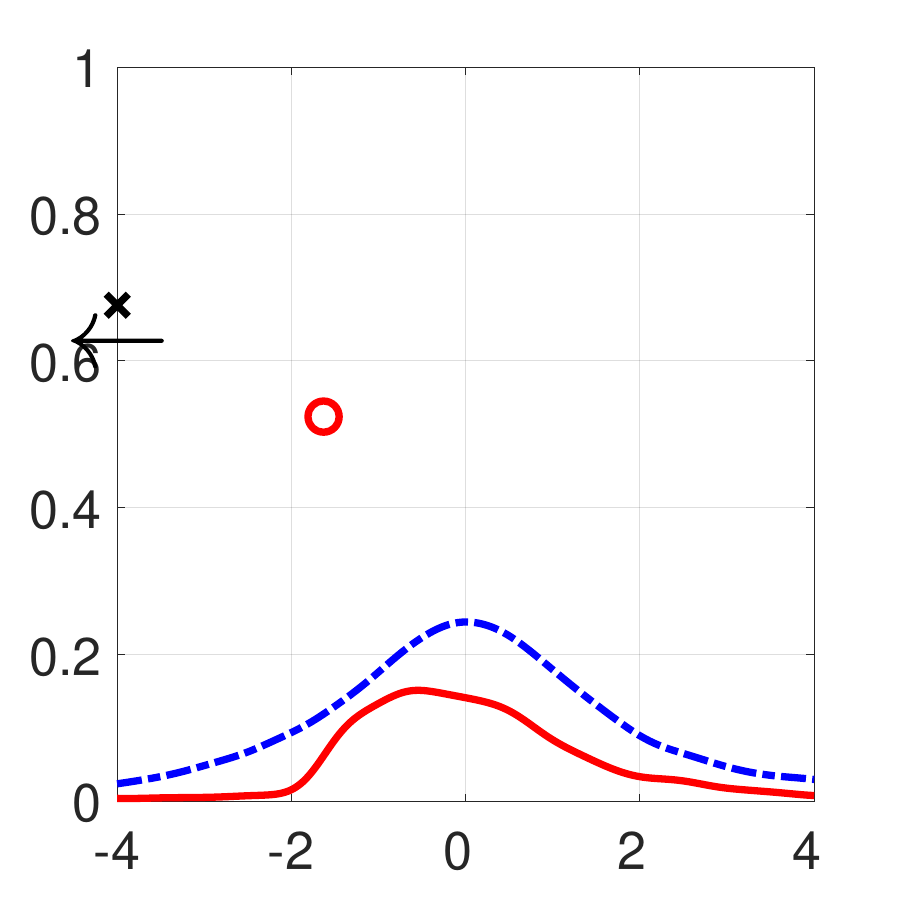}
\end{subfigure}\begin{subfigure}{0.2\textwidth}
	\centering
	\caption*{$T = 100$}
	\vspace{-1.5ex}
	\includegraphics[trim={0cm 0cm 0.50cm 0.5cm},width=\textwidth,clip]
	{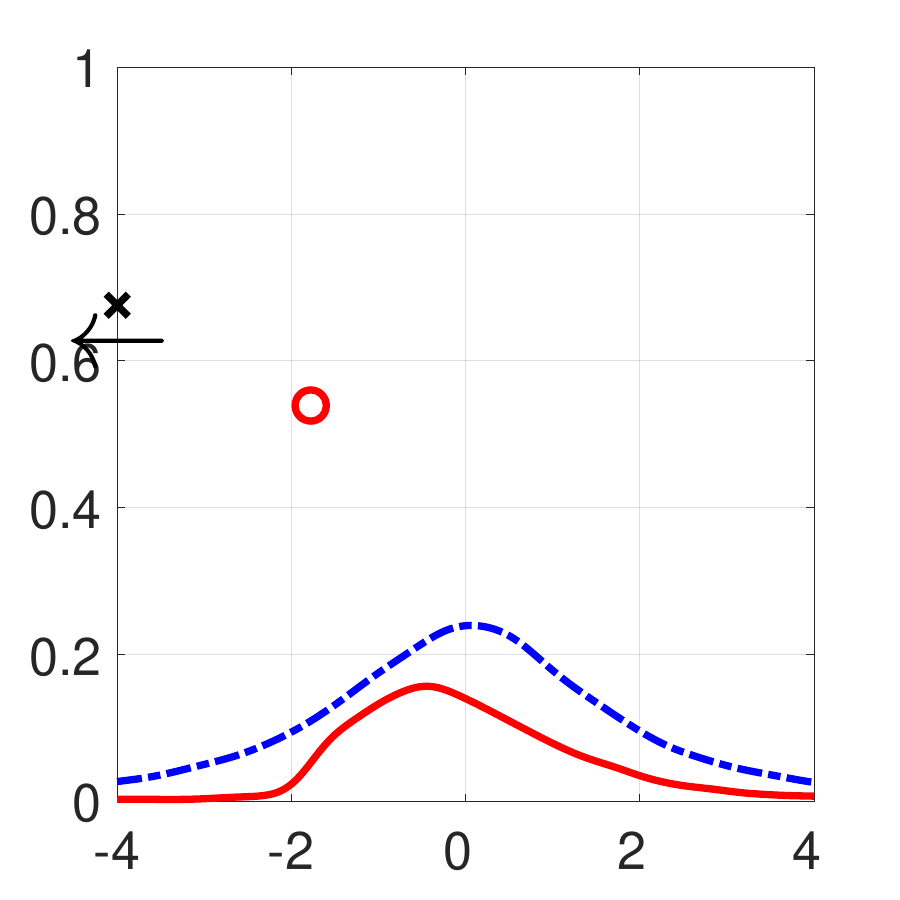}
\end{subfigure}\begin{subfigure}{0.2\textwidth}
	\centering
	\caption*{$T = 250$}
	\vspace{-1.5ex}
	\includegraphics[trim={0cm 0cm 0.50cm 0.5cm},width=\textwidth,clip]
	{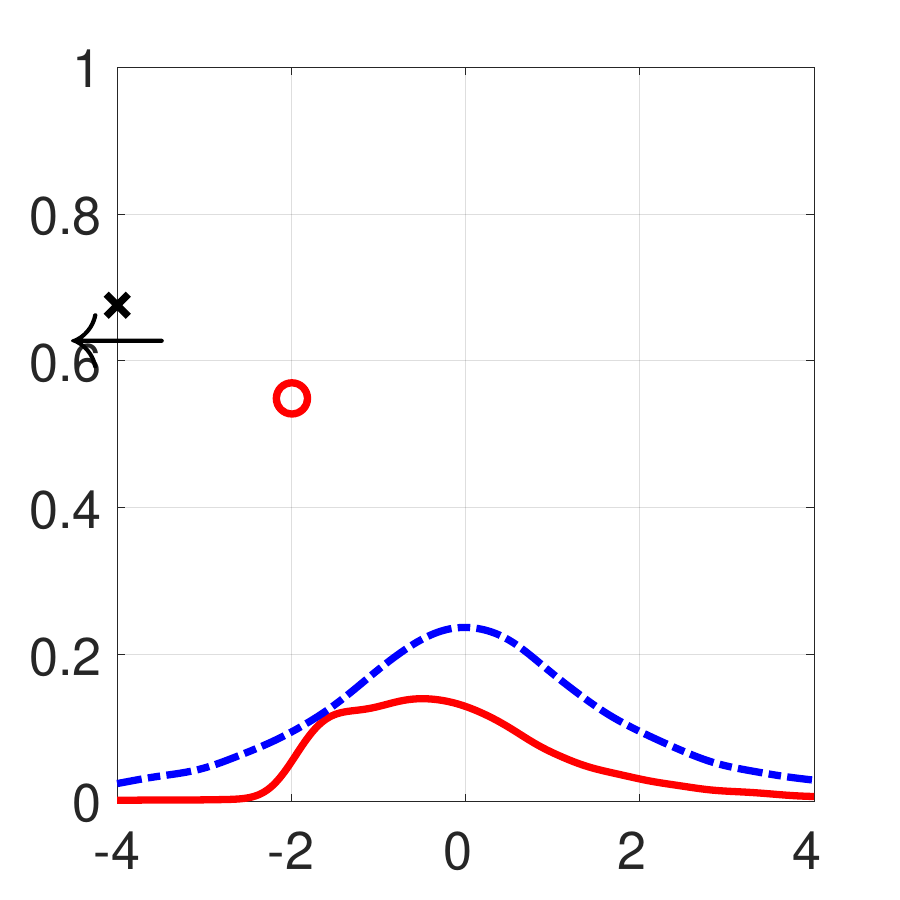}
\end{subfigure}\begin{subfigure}{0.2\textwidth}
	\centering
	\caption*{$T = 1000$}
	\vspace{-1.5ex}
	\includegraphics[trim={0cm 0cm 0.50cm 0.5cm},width=\textwidth,clip]
	{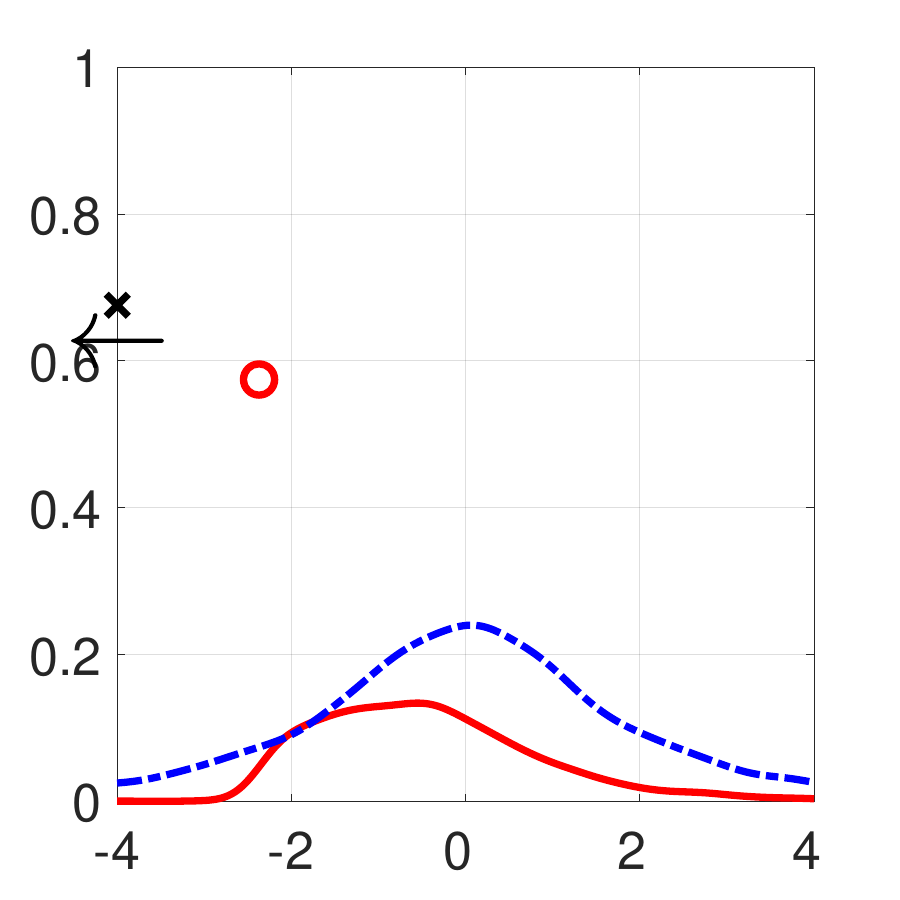}
\end{subfigure}

\caption*{$\lambda_T = T^{1/2}$}
\vspace{-1.5ex}
\begin{subfigure}{0.2\textwidth}
	\centering
	\includegraphics[trim={0cm 0cm 0.50cm 0.5cm},width=\textwidth,clip]
	{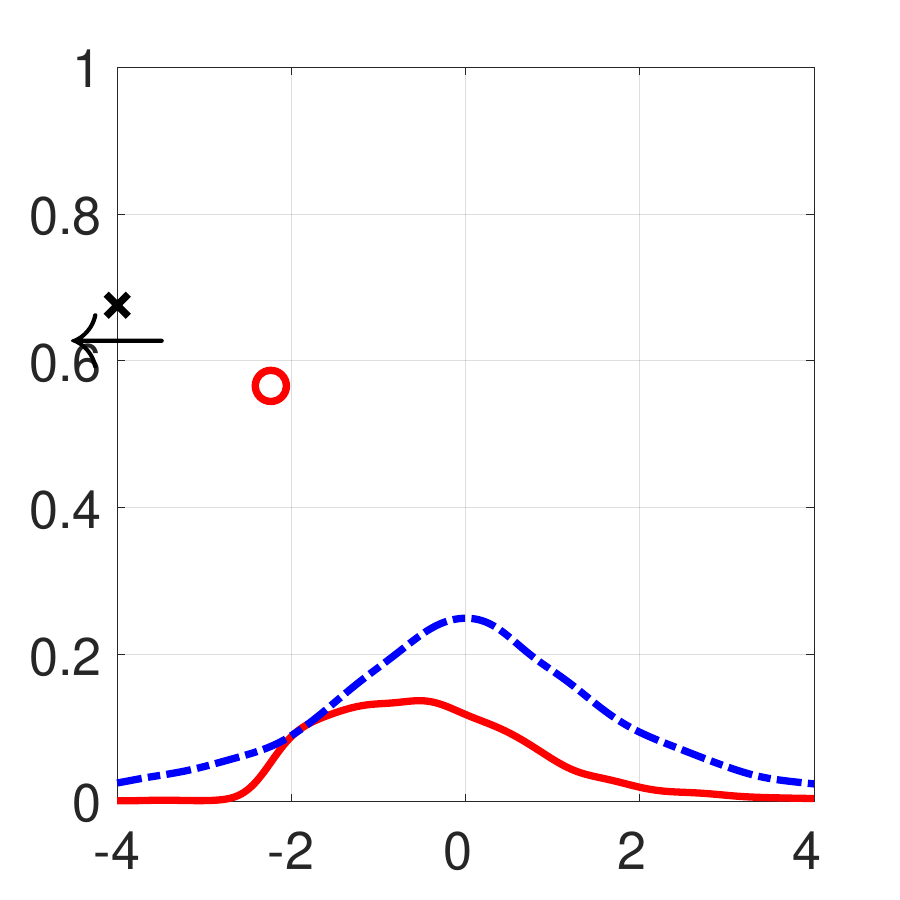}
\end{subfigure}\begin{subfigure}{0.2\textwidth}
	\centering
	\includegraphics[trim={0cm 0cm 0.50cm 0.5cm},width=\textwidth,clip]
	{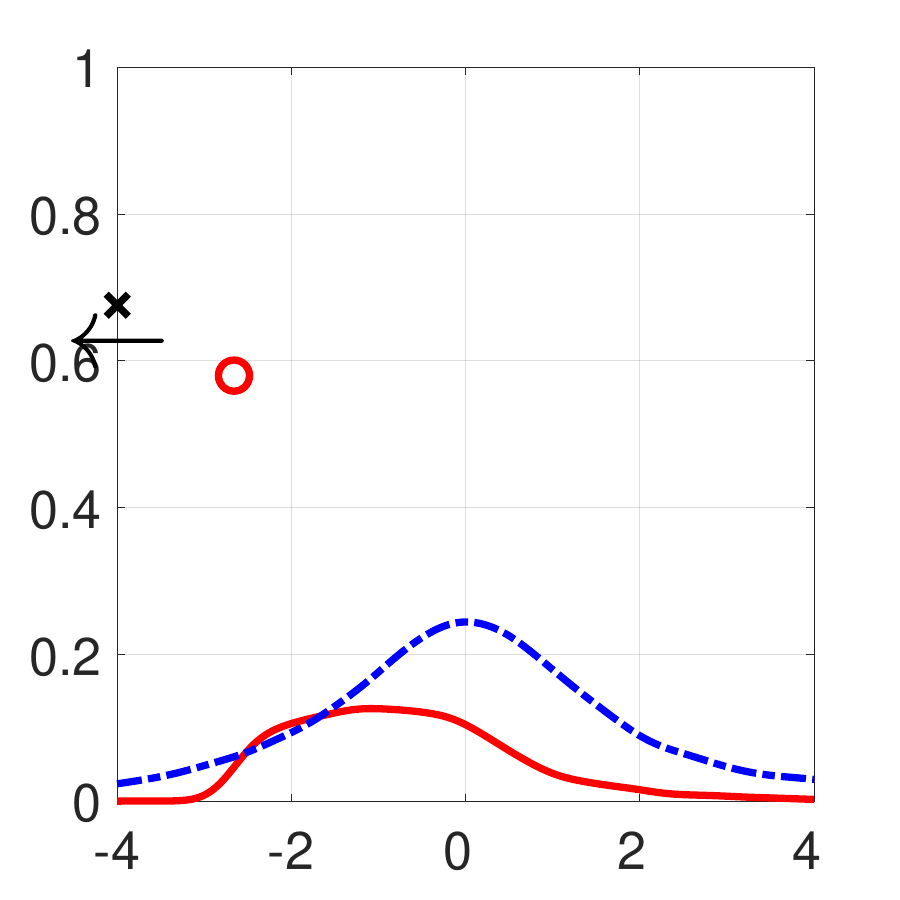}
\end{subfigure}\begin{subfigure}{0.2\textwidth}
	\centering
	\includegraphics[trim={0cm 0cm 0.50cm 0.5cm},width=\textwidth,clip]
	{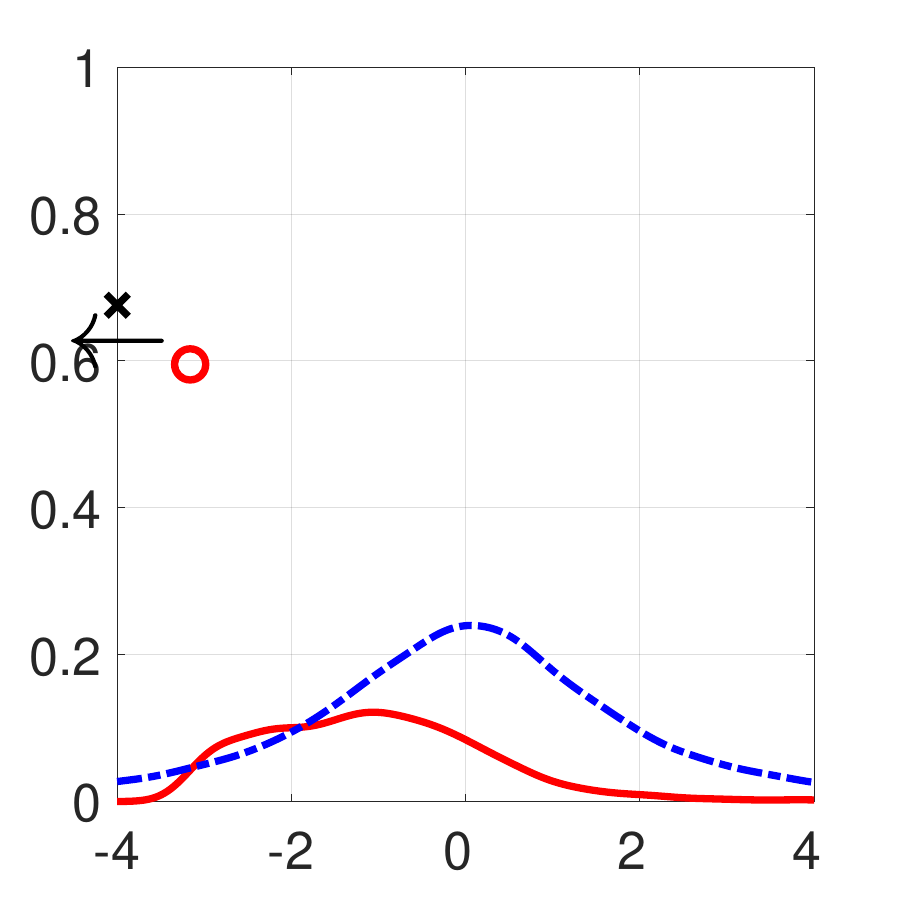}
\end{subfigure}\begin{subfigure}{0.2\textwidth}
	\centering
	\includegraphics[trim={0cm 0cm 0.50cm 0.5cm},width=\textwidth,clip]
	{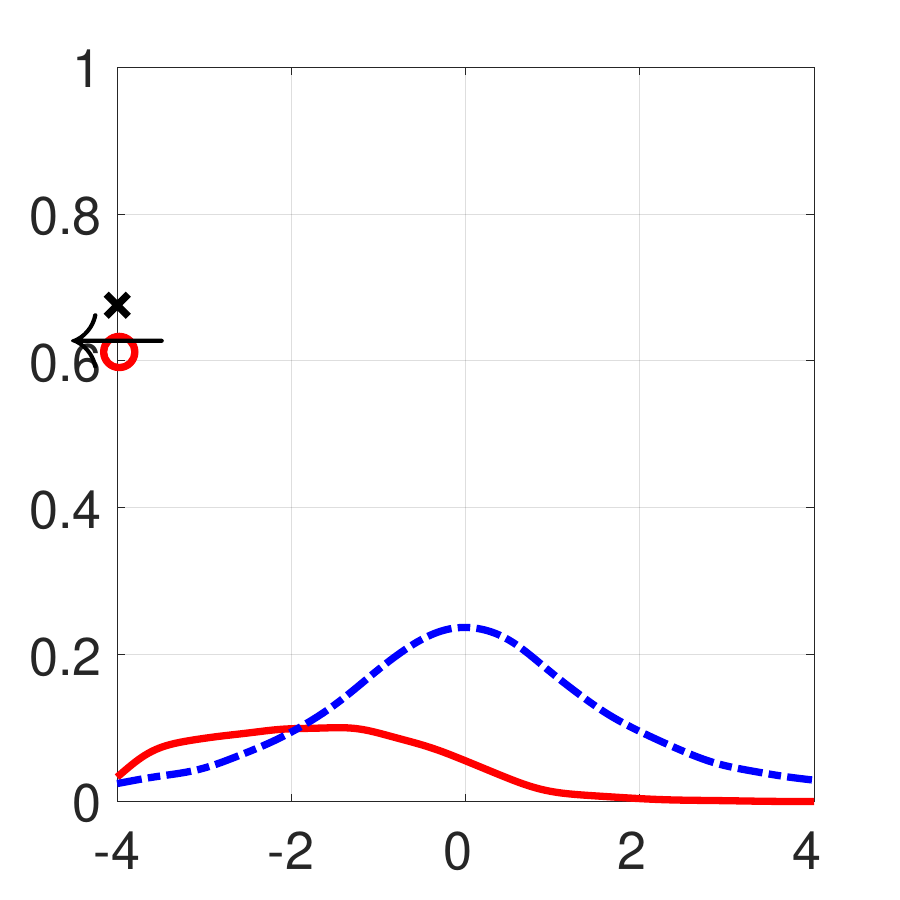}
\end{subfigure}\begin{subfigure}{0.2\textwidth}
	\centering
	\includegraphics[trim={0cm 0cm 0.50cm 0.5cm},width=\textwidth,clip]
	{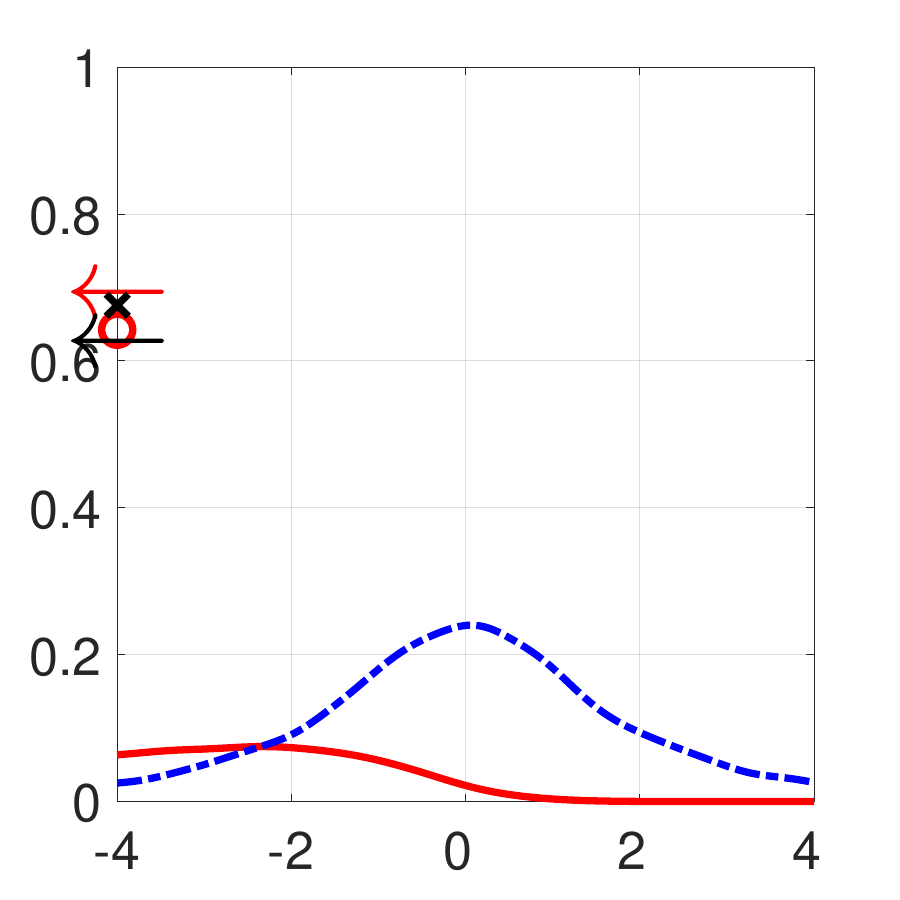}
\end{subfigure}

\caption*{$\lambda_T = T$}
\vspace{-1.5ex}
\begin{subfigure}{0.2\textwidth}
	\centering
	\includegraphics[trim={0cm 0cm 0.50cm 0.5cm},width=\textwidth,clip]
	{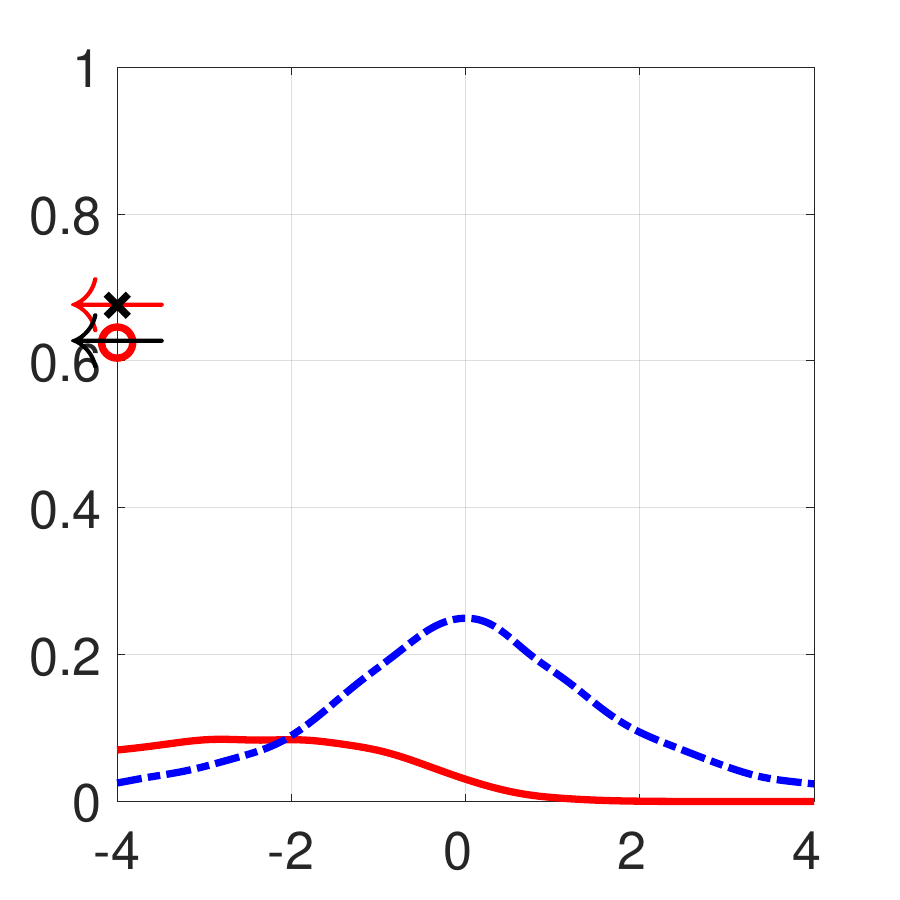}
\end{subfigure}\begin{subfigure}{0.2\textwidth}
	\centering
	\includegraphics[trim={0cm 0cm 0.50cm 0.5cm},width=\textwidth,clip]
	{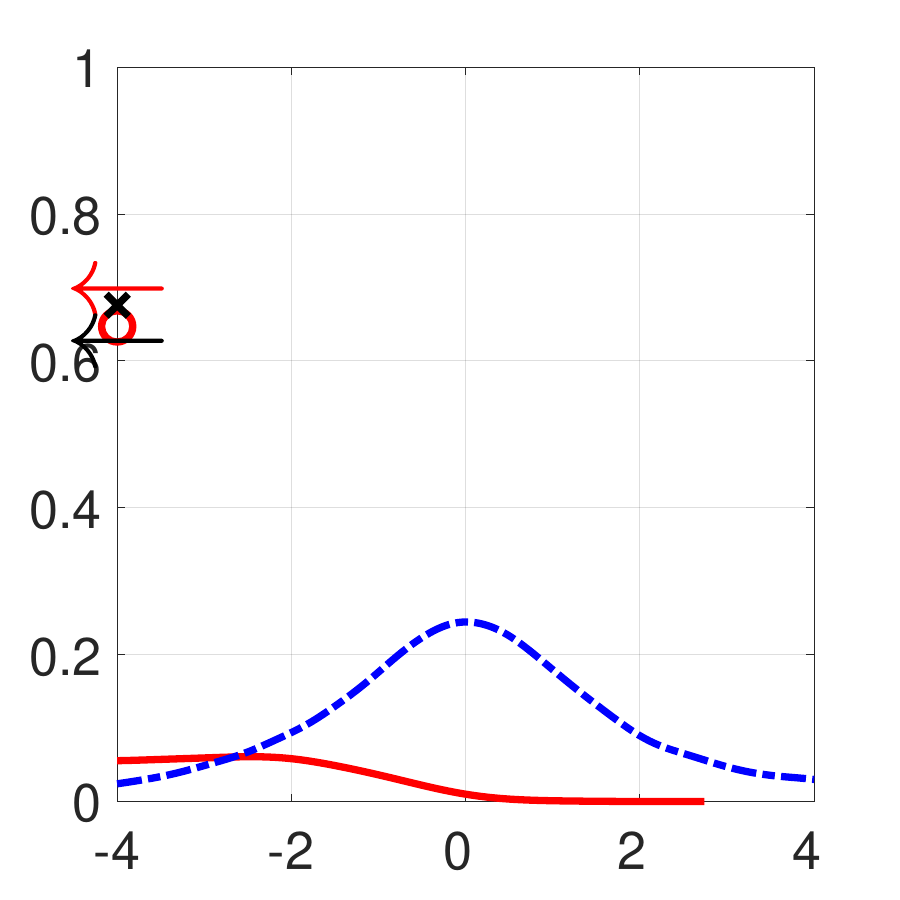}
\end{subfigure}\begin{subfigure}{0.2\textwidth}
	\centering
	\includegraphics[trim={0cm 0cm 0.50cm 0.5cm},width=\textwidth,clip]
	{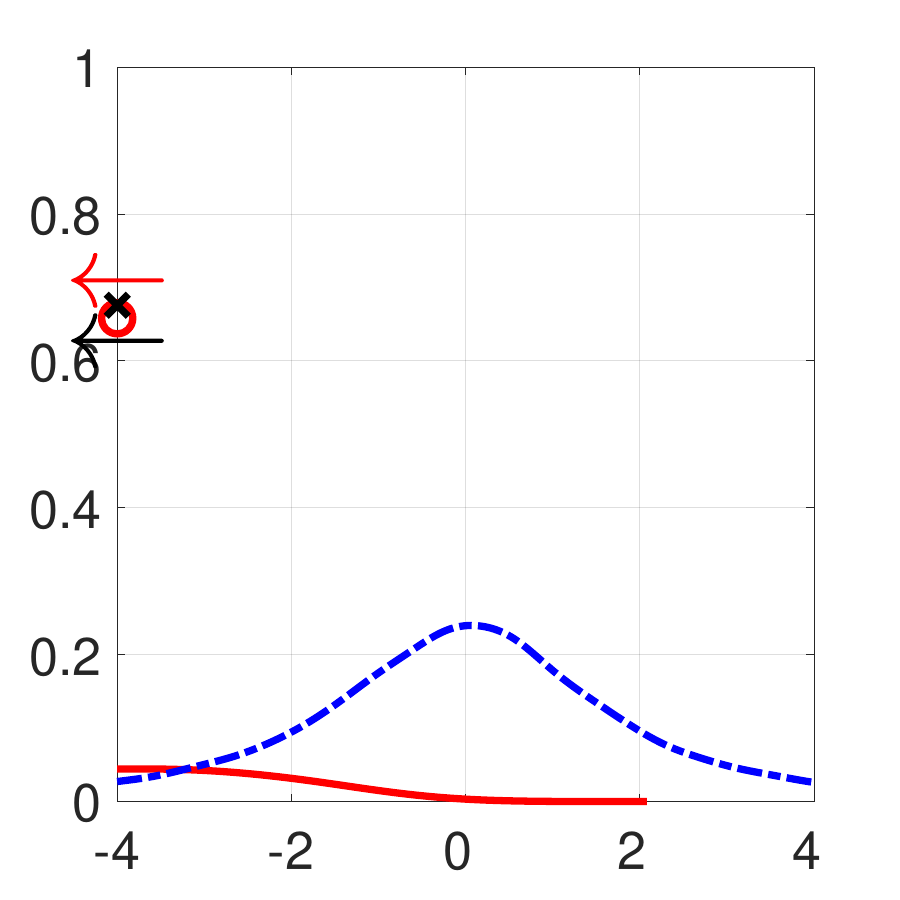}
\end{subfigure}\begin{subfigure}{0.2\textwidth}
	\centering
	\includegraphics[trim={0cm 0cm 0.50cm 0.5cm},width=\textwidth,clip]
	{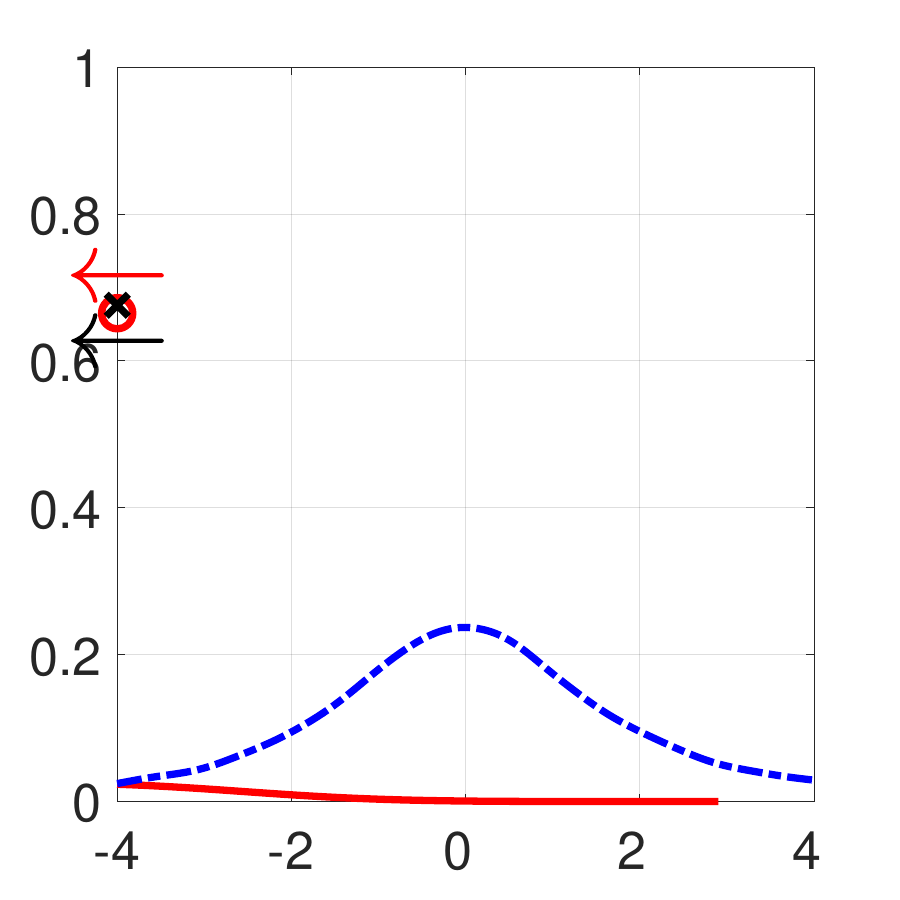}
\end{subfigure}\begin{subfigure}{0.2\textwidth}
	\centering
	\includegraphics[trim={0cm 0cm 0.50cm 0.5cm},width=\textwidth,clip]
	{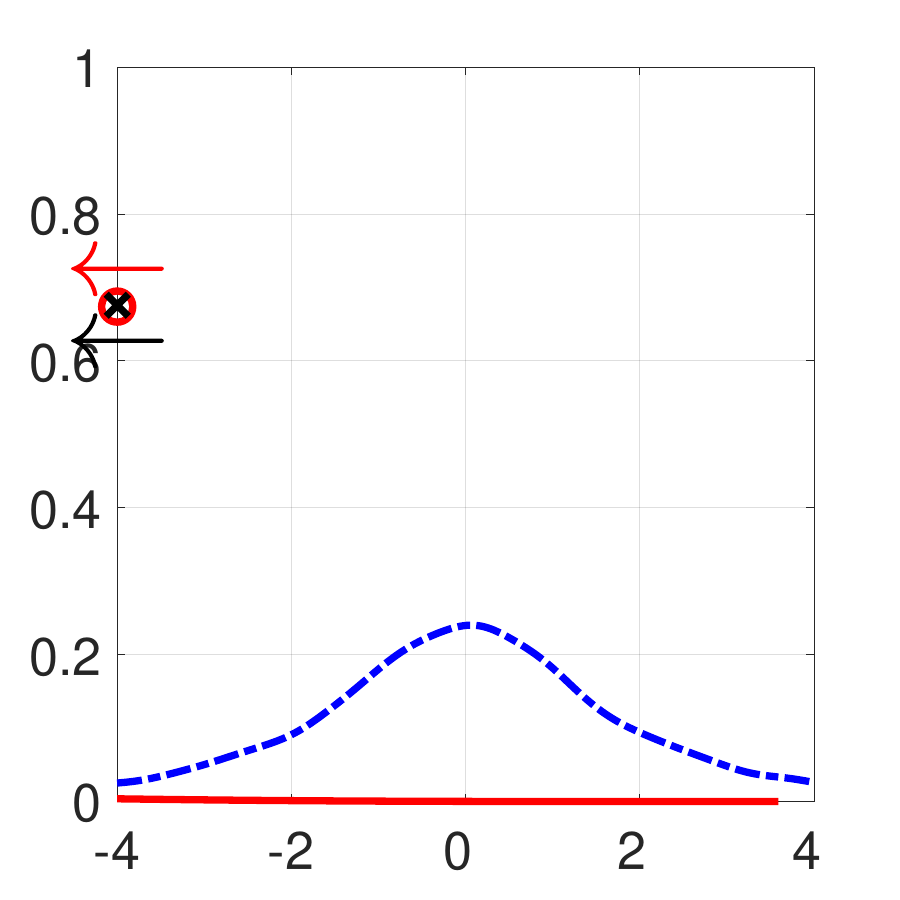}
\end{subfigure}

\end{center}

\vspace{-2ex} 

\caption{Finite-sample distribution of $T(\betaAL - \beta_T)$ under
consistent tuning (in all rows) in case $\beta_T =
\sqrt{\lambda_T}\beta/T$ (labeled ``AL''), and limiting
distribution from Remark~\ref{rem:ls_dist_consist_rateT-unif} (labeled
``Rem.5''). \emph{Notes}: See notes to
Figure~\ref{fig:densities_thm3_1}.}

\label{fig:densities_rem5_4}

\end{figure}

\begin{table}[t] 
\centering
\adjustbox{max width=\textwidth}{\begin{threeparttable}
\caption{Length of confidence intervals corresponding to Figures~\ref{fig:coverage_lambda}--\ref{fig:coverage_4lambda}}
\label{tab:CI_length_iid}
\begin{tabular}{clcccccc}
	\toprule[1pt]\midrule[0.3pt]
	\multicolumn{1}{c}{}&\multicolumn{1}{c}{}&\multicolumn{4}{c}{Uniform CI}&\multicolumn{2}{c}{Oracle CI}\\
	\cmidrule(lr){3-6}
	\cmidrule(lr){7-8}
	T & $\lambda_T$ & min. & median & mean & max. & $95\%$ & $99\%$ \\
	\midrule
	25 & $T^{1/4}$ & 0.032 & 0.153 & 0.172 & 0.764 & 0.400 & 0.627 \\
	& $T^{1/2}$ & 0.047 & 0.228 & 0.257 & 1.143 & 0.400 & 0.627 \\
	& $T$ & 0.106 & 0.510 & 0.574 & 2.555 & 0.400 & 0.627 \\
	\midrule
	50 & $T^{1/4}$ & 0.017 & 0.084 & 0.095 & 0.391 & 0.200 & 0.313 \\
	& $T^{1/2}$ & 0.027 & 0.138 & 0.156 & 0.637 & 0.200 & 0.313 \\
	& $T$ & 0.073 & 0.366 & 0.414 & 1.694 & 0.200 & 0.313 \\
	\midrule
	100 & $T^{1/4}$ & 0.010 & 0.046 & 0.052 & 0.188 & 0.100 & 0.157 \\
	& $T^{1/2}$ & 0.018 & 0.082 & 0.092 & 0.334 & 0.100 & 0.157 \\
	& $T$ & 0.057 & 0.258 & 0.292 & 1.055 & 0.100 & 0.157 \\
	\midrule
	250 & $T^{1/4}$ & 0.005 & 0.021 & 0.023 & 0.109 & 0.040 & 0.063 \\
	& $T^{1/2}$ & 0.009 & 0.041 & 0.047 & 0.217 & 0.040 & 0.063 \\
	& $T$ & 0.037 & 0.165 & 0.186 & 0.864 & 0.040 & 0.063 \\
	\midrule
	1000 & $T^{1/4}$ & 0.001 & 0.006 & 0.007 & 0.028 & 0.010 & 0.016 \\
	& $T^{1/2}$ & 0.003 & 0.015 & 0.017 & 0.066 & 0.010 & 0.016 \\
	& $T$ & 0.017 & 0.082 & 0.093 & 0.371 & 0.010 & 0.016 \\
	\midrule[1pt]
	25 & $4\times T^{1/4}$ & 0.063 & 0.305 & 0.343 & 1.528 & 0.400 & 0.627 \\
	& $4\times T^{1/2}$ & 0.095 & 0.457 & 0.513 & 2.285 & 0.400 & 0.627 \\
	& $4\times T$ & 0.212 & 1.021 & 1.148 & 5.110 & 0.400 & 0.627 \\
	\midrule
	50 & $4\times T^{1/4}$ & 0.034 & 0.169 & 0.191 & 0.781 & 0.200 & 0.313 \\
	& $4\times T^{1/2}$ & 0.055 & 0.275 & 0.311 & 1.274 & 0.200 & 0.313 \\
	& $4\times T$ & 0.146 & 0.732 & 0.827 & 3.388 & 0.200 & 0.313 \\
	\midrule
	100 & $4\times T^{1/4}$ & 0.020 & 0.092 & 0.104 & 0.375 & 0.100 & 0.157 \\
	& $4\times T^{1/2}$ & 0.036 & 0.163 & 0.185 & 0.667 & 0.100 & 0.157 \\
	& $4\times T$ & 0.114 & 0.516 & 0.584 & 2.110 & 0.100 & 0.157 \\
	\midrule
	250 & $4\times T^{1/4}$ & 0.009 & 0.042 & 0.047 & 0.218 & 0.040 & 0.063 \\
	& $4\times T^{1/2}$ & 0.019 & 0.083 & 0.093 & 0.434 & 0.040 & 0.063 \\
	& $4\times T$ & 0.074 & 0.330 & 0.371 & 1.727 & 0.040 & 0.063 \\
	\midrule
	1000 & $4\times T^{1/4}$ & 0.003 & 0.012 & 0.014 & 0.056 & 0.010 & 0.016 \\
	& $4\times T^{1/2}$ & 0.006 & 0.029 & 0.033 & 0.132 & 0.010 & 0.016 \\
	& $4\times T$ & 0.035 & 0.165 & 0.186 & 0.742 & 0.010 & 0.016 \\
	\midrule[0.3pt]\bottomrule[1pt]
\end{tabular}
\begin{tablenotes}
	\item Notes: The length of the Uniform CI depends on $x_t$,
	$t=1,\ldots,T$. The table presents the minimum, median, mean, and
	maximum length of the Uniform CI across $10{,}000$ Monte Carlo
	replications for different values of $T$ and $\lambda_T$. The length
	of the Oracle CI is constant across Monte Carlo replications as it only
	depends on the nominal size and $T$.
\end{tablenotes}
\end{threeparttable}}
\end{table}

\begin{table}[t] 
\centering
\adjustbox{max width=\textwidth}{\begin{threeparttable}
\caption{Length of confidence intervals corresponding to Figure~\ref{fig:coverage_4lambda_rho6}}
\label{tab:CI_length_corr}
\begin{tabular}{clcccccc}
	\toprule[1pt]\midrule[0.3pt]
	\multicolumn{1}{c}{}&\multicolumn{1}{c}{}&\multicolumn{4}{c}{Uniform CI}&\multicolumn{2}{c}{Oracle CI}\\
	\cmidrule(lr){3-6}
	\cmidrule(lr){7-8}
	T & $\lambda_T$ & min. & median & mean & max. & $95\%$ & $99\%$ \\
	\midrule
	25 & $4\times T^{1/4}$ & 0.070 & 0.314 & 0.361 & 1.978 & 0.737 & 1.142 \\
	& $4\times T^{1/2}$ & 0.105 & 0.469 & 0.539 & 2.958 & 0.737 & 1.142 \\
	& $4\times T$ & 0.234 & 1.048 & 1.205 & 6.614 & 0.737 & 1.142 \\
	\midrule
	50 & $4\times T^{1/4}$ & 0.034 & 0.171 & 0.195 & 0.839 & 0.369 & 0.571 \\
	& $4\times T^{1/2}$ & 0.056 & 0.279 & 0.318 & 1.368 & 0.369 & 0.571 \\
	& $4\times T$ & 0.149 & 0.742 & 0.845 & 3.639 & 0.369 & 0.571 \\
	\midrule
	100 & $4\times T^{1/4}$ & 0.020 & 0.093 & 0.105 & 0.382 & 0.184 & 0.285 \\
	& $4\times T^{1/2}$ & 0.036 & 0.165 & 0.187 & 0.679 & 0.184 & 0.285 \\
	& $4\times T$ & 0.114 & 0.523 & 0.591 & 2.147 & 0.184 & 0.285 \\
	\midrule
	250 & $4\times T^{1/4}$ & 0.009 & 0.042 & 0.047 & 0.216 & 0.074 & 0.114 \\
	& $4\times T^{1/2}$ & 0.018 & 0.083 & 0.094 & 0.432 & 0.074 & 0.114 \\
	& $4\times T$ & 0.073 & 0.330 & 0.372 & 1.716 & 0.074 & 0.114 \\
	\midrule
	1000 & $4\times T^{1/4}$ & 0.003 & 0.012 & 0.014 & 0.053 & 0.018 & 0.029 \\
	& $4\times T^{1/2}$ & 0.006 & 0.029 & 0.033 & 0.126 & 0.018 & 0.029 \\
	& $4\times T$ & 0.035 & 0.165 & 0.187 & 0.709 & 0.018 & 0.029 \\
	\midrule[0.3pt]\bottomrule[1pt]
\end{tabular}
\begin{tablenotes}
	\item Notes: See notes to Table~\ref{tab:CI_length_iid}.
\end{tablenotes}
\end{threeparttable}}
\end{table}

\newpage
\clearpage

\section{Additional Empirical Results} \label{app:emp}

\begin{figure}[ht]
\begin{center}
	\begin{subfigure}{0.4\textwidth}
		\centering
		\includegraphics[trim={0cm 0cm 0.50cm 0.5cm},width=\textwidth,clip]
		{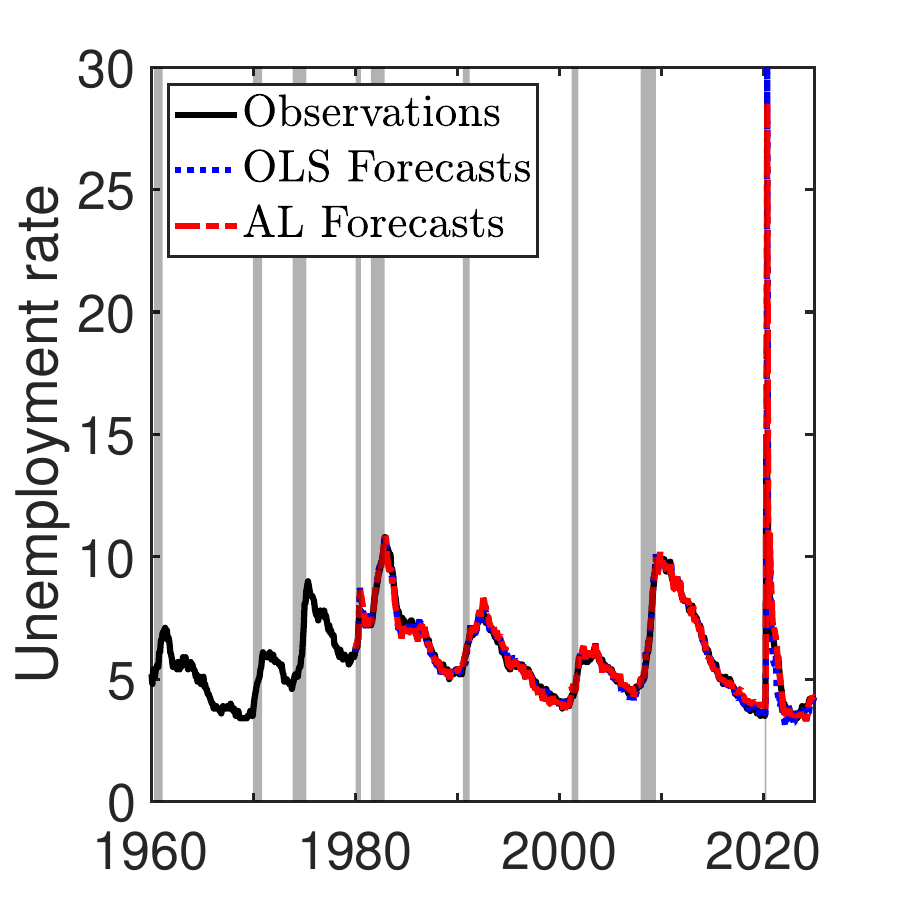}
	\end{subfigure}\begin{subfigure}{0.4\textwidth}
		\centering
		\includegraphics[trim={0cm 0cm 0.50cm 0.5cm},width=\textwidth,clip]
		{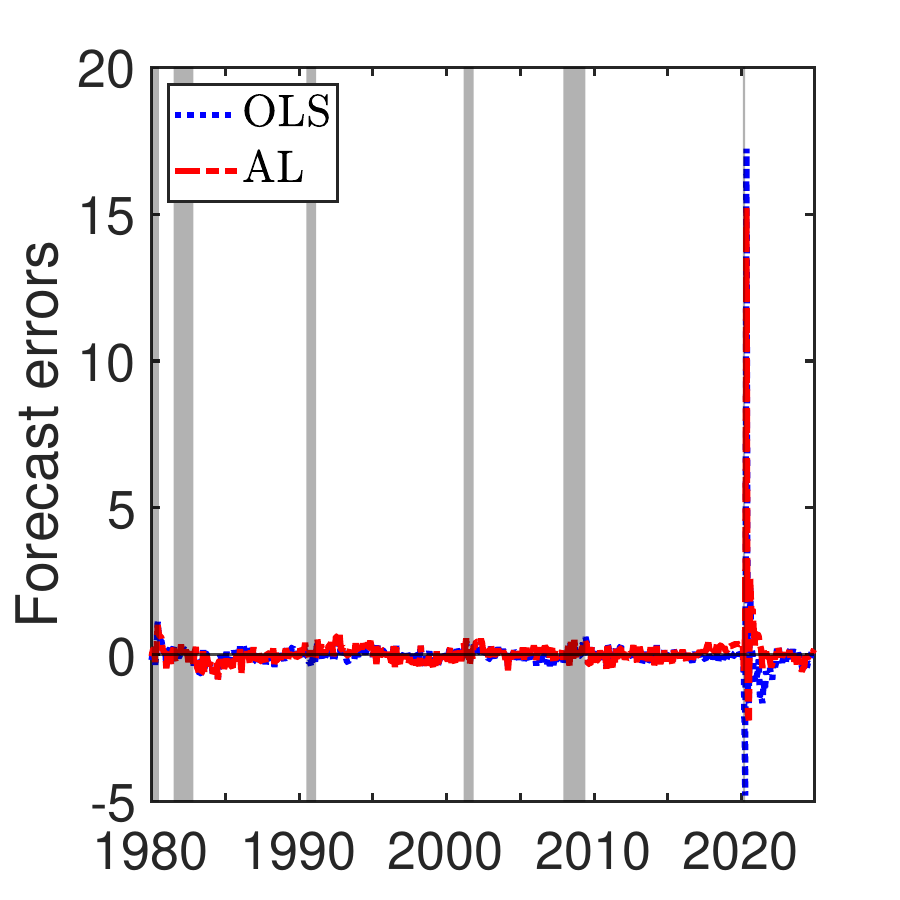}
	\end{subfigure}
\end{center}
\vspace{-2ex} 
\caption{One-month-ahead forecasts (left) and corresponding forecast errors (right) based on a 20-year rolling window.}
\label{fig:emp_forecasts}
\end{figure}

\begin{figure}[ht]
	\begin{center}
		\centering
		\includegraphics[trim={0cm 0cm 0.50cm 0.5cm},width=0.4\textwidth,clip]
		{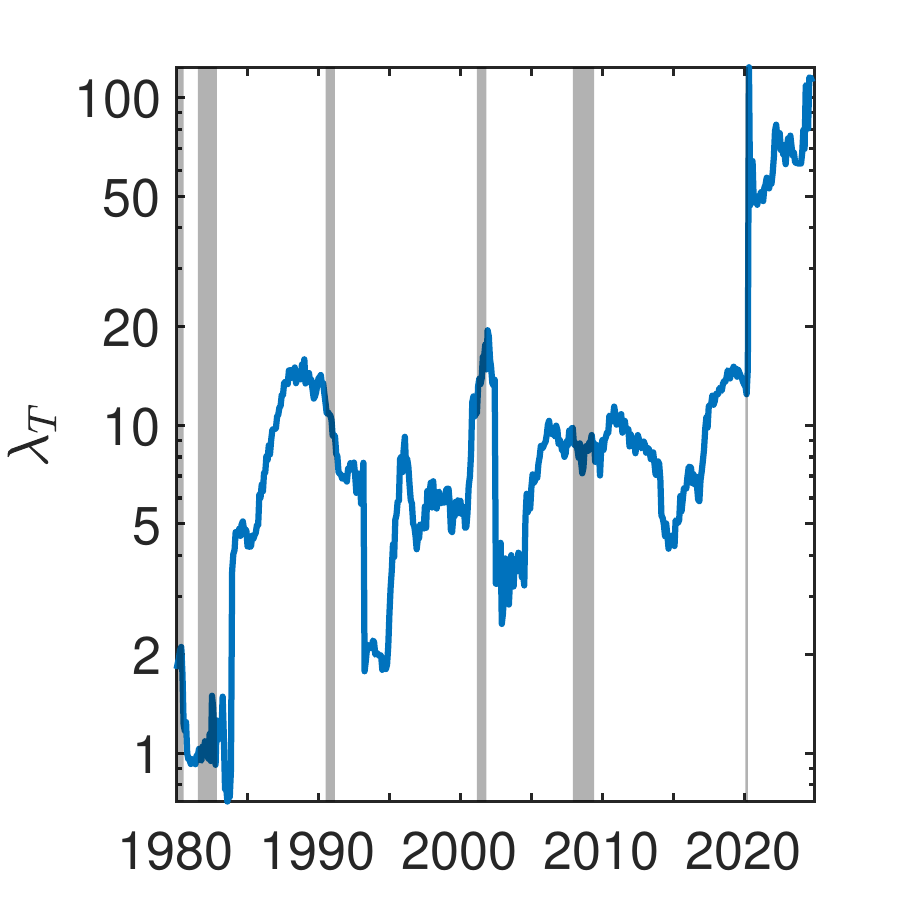}
	\end{center}
	\vspace{-2ex} 
	\caption{Values of the adaptive LASSO penalty parameter $\lambda_T$ selected via time-series cross-validation, plotted on a logarithmic scale.}
	\label{fig:emp_lambda}
\end{figure}

\begin{figure}[ht]
\begin{center}
	\begin{subfigure}{0.3\textwidth}
		\centering
		\includegraphics[trim={0cm 0cm 0.50cm 0cm},width=\textwidth,clip]
		{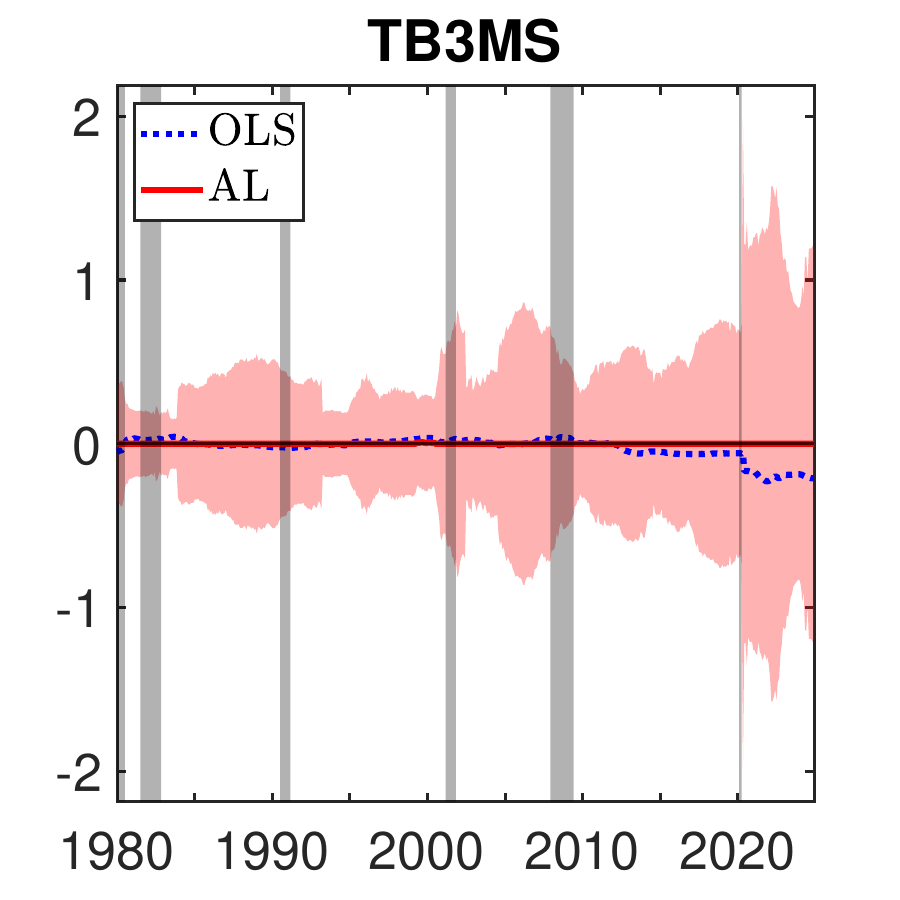}
	\end{subfigure}\begin{subfigure}{0.3\textwidth}
		\centering
		\includegraphics[trim={0cm 0cm 0.50cm 0cm},width=\textwidth,clip]
		{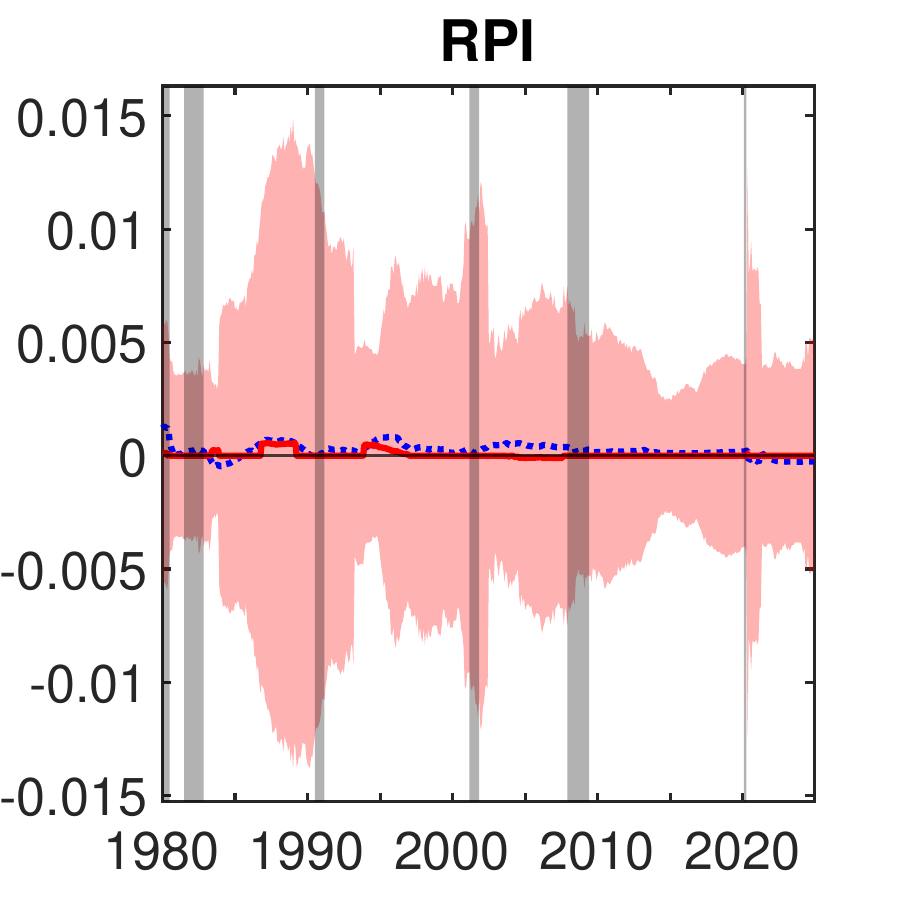}
	\end{subfigure}\begin{subfigure}{0.3\textwidth}
		\centering
		\includegraphics[trim={0cm 0cm 0.50cm 0cm},width=\textwidth,clip]
		{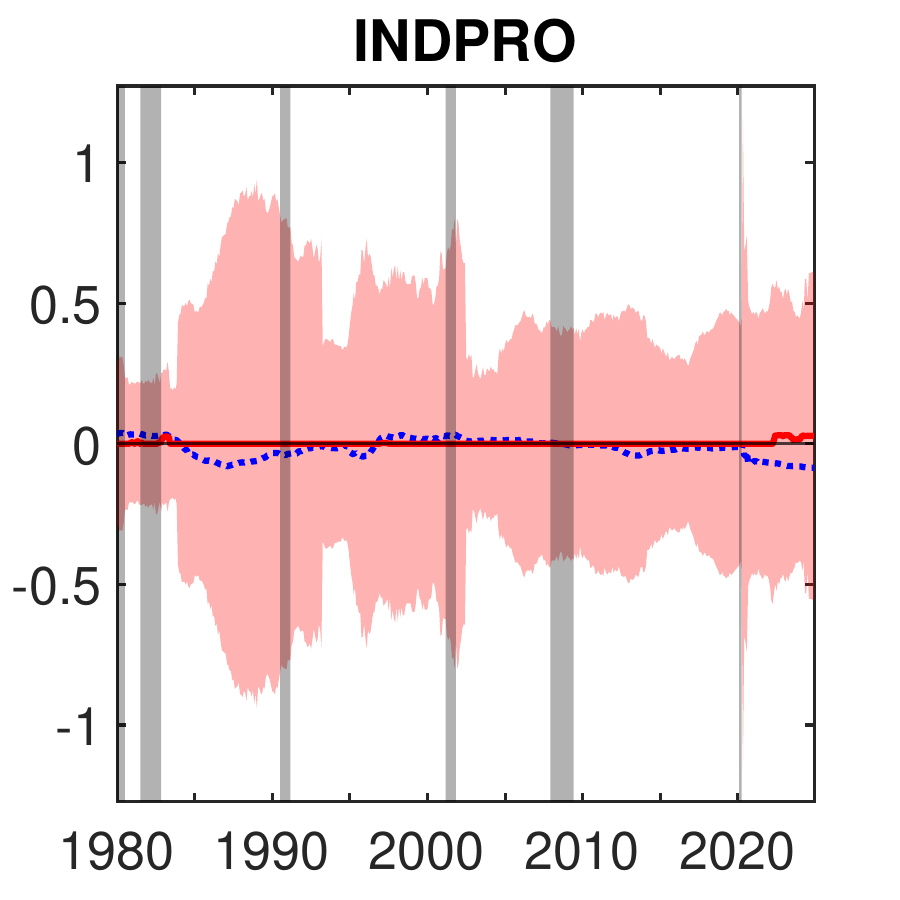}
	\end{subfigure}
	
	\begin{subfigure}{0.3\textwidth}
		\centering
		\includegraphics[trim={0cm 0cm 0.50cm 0cm},width=\textwidth,clip]
		{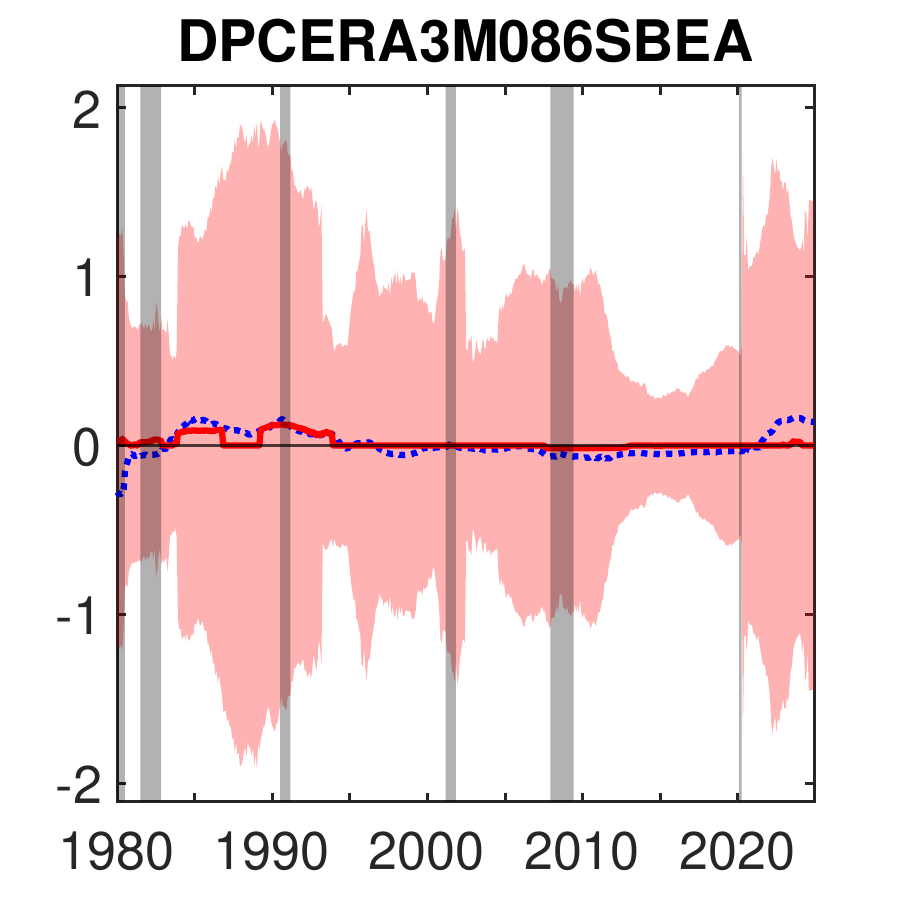}
	\end{subfigure}\begin{subfigure}{0.3\textwidth}
		\centering
		\includegraphics[trim={0cm 0cm 0.50cm 0cm},width=\textwidth,clip]
		{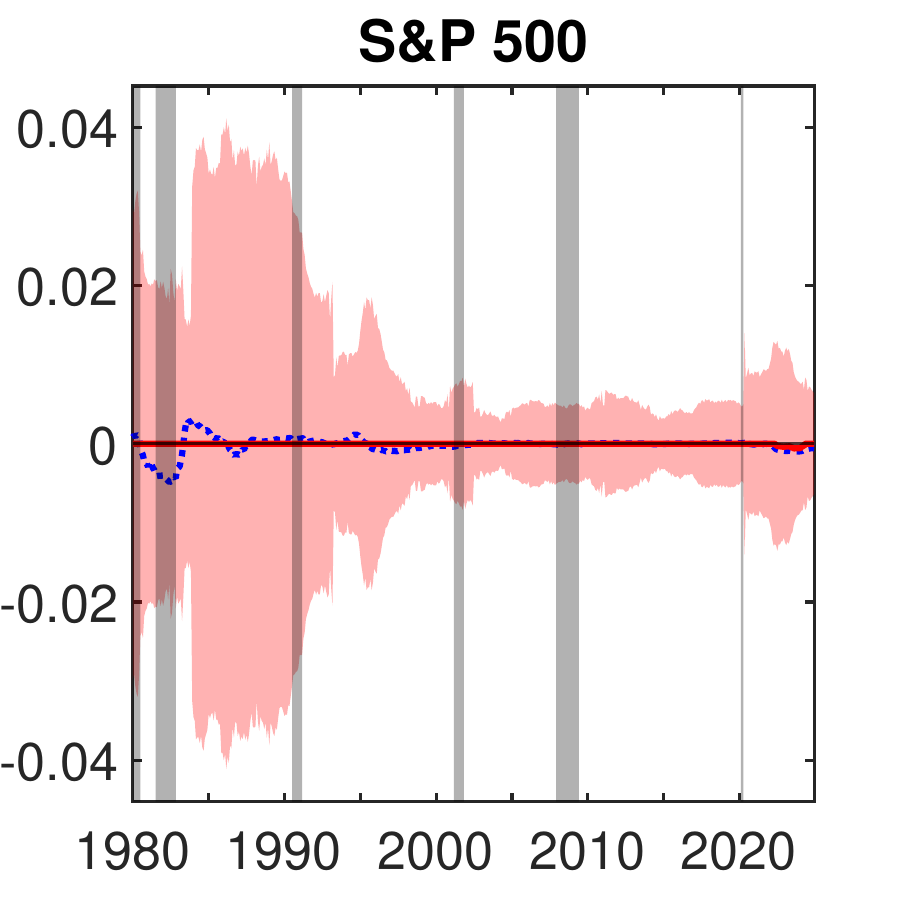}
	\end{subfigure}\begin{subfigure}{0.3\textwidth}
		\centering
		\includegraphics[trim={0cm 0cm 0.50cm 0cm},width=\textwidth,clip]
		{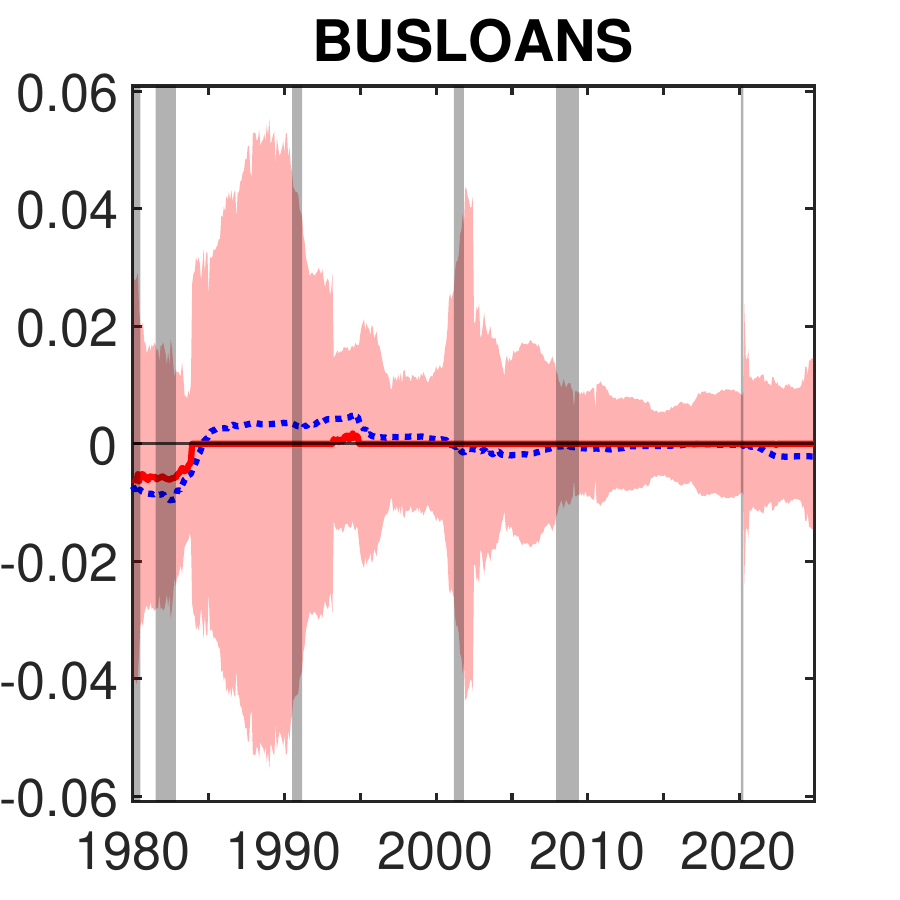}
	\end{subfigure}
	
	\begin{subfigure}{0.3\textwidth}
		\centering
		\includegraphics[trim={0cm 0cm 0.50cm 0cm},width=\textwidth,clip]
		{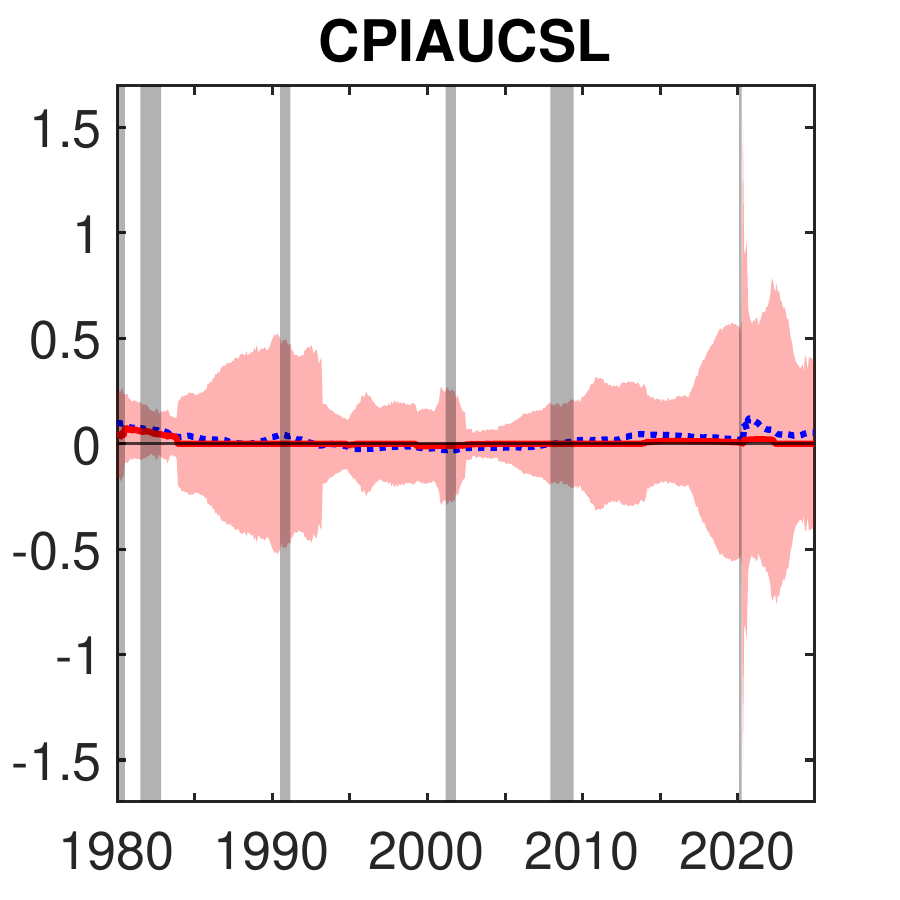}
	\end{subfigure}\begin{subfigure}{0.3\textwidth}
		\centering
		\includegraphics[trim={0cm 0cm 0.50cm 0cm},width=\textwidth,clip]
		{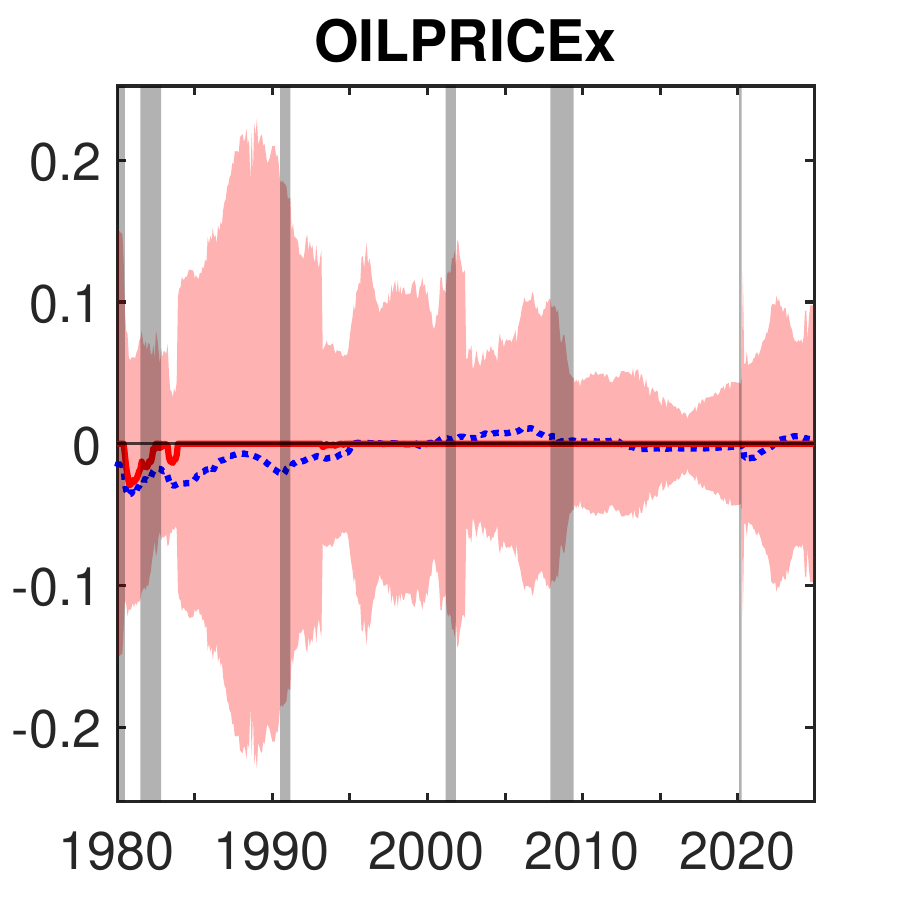}
	\end{subfigure}\begin{subfigure}{0.3\textwidth}
		\centering
		\includegraphics[trim={0cm 0cm 0.50cm 0cm},width=\textwidth,clip]
		{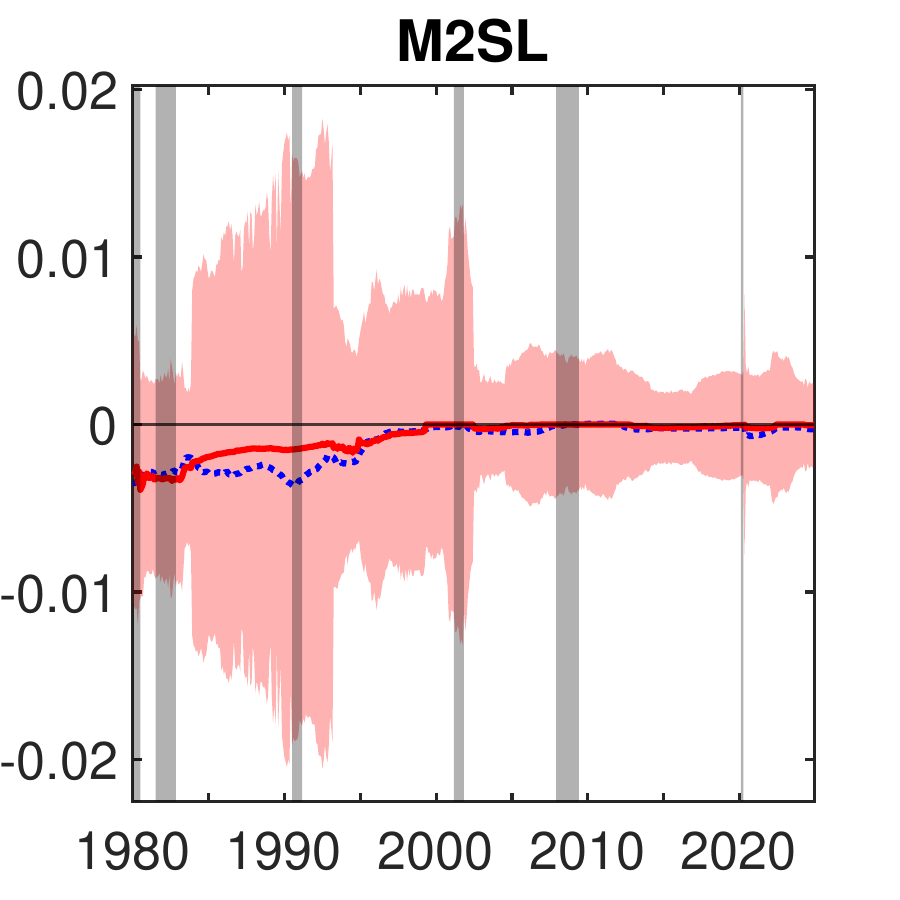}
	\end{subfigure}
\end{center}
\vspace{-2ex} 
\caption{20-year rolling window coefficient estimates and adaptive LASSO-based confidence intervals for all remaining variables.}
\label{fig:emp_coeffs2}
\end{figure}

\end{appendices}

\end{document}